\documentclass[12pt,oneside]{report}

\usepackage[left=2.5cm, right=2.5cm, bottom=3cm, top=2cm, heightrounded]{geometry}

\usepackage{./sty/leosstylefile}
\usepackage{./sty/tikzit}
\usetikzlibrary{positioning,intersections,quotes}

\tikzstyle{node}=[fill=black, draw=black, shape=circle, scale=0.5]
\tikzstyle{wnode}=[fill=white, draw=black, shape=circle, scale=0.5]
\tikzstyle{textbox}=[inner sep=2pt, shape=rectangle, fill=none]
\tikzstyle{textnode}=[inner sep=0mm, shape=circle, fill=white]
\tikzstyle{gnode}=[inner sep=0mm, minimum size=1mm, fill={rgb,255: red,221; green,221; blue,221}, draw={rgb,255: red,221; green,221; blue,221}, shape=circle]
\tikzstyle{refine}=[fill=black, draw=black, shape=regular polygon, regular polygon sides=3, rotate=180, scale=0.5]
\tikzstyle{coarsen}=[fill=white, draw=black, shape=regular polygon, regular polygon sides=3, scale=0.5]
\tikzstyle{bdytextbox}=[fill=white, draw=black, shape=rectangle]
\tikzstyle{redbox}=[fill=white, draw=red, shape=rectangle, text=red]
\tikzstyle{bluecirc}=[inner sep=1mm, fill=white, draw={rgb,255: red,4; green,51; blue,255}, shape=circle, text={rgb,255: red,4; green,51; blue,255}]
\tikzstyle{rednode}=[fill=red, draw=red, shape=circle, scale=0.5]
\tikzstyle{new style 0}=[fill=white, draw=red, shape=circle, scale=0.5]
\tikzstyle{bluenode}=[fill={rgb,255: red,4; green,51; blue,255}, draw={rgb,255: red,4; green,51; blue,255}, shape=circle, scale=0.5]
\tikzstyle{yellownode}=[fill={rgb,255: red,255; green,210; blue,75}, draw={rgb,255: red,255; green,210; blue,75}, shape=circle, scale=0.5]
\tikzstyle{blacksq}=[fill=black, draw=black, shape=rectangle, scale=0.5]
\tikzstyle{bluetext}=[fill=none, draw=none, shape=rectangle, text={rgb,255: red,4; green,51; blue,255}]
\tikzstyle{reg}=[draw, fill=white, rounded rectangle, rounded rectangle left arc=none, minimum height=1em, minimum width=1em, node font={\scriptsize}]
\tikzstyle{coreg}=[draw, fill=white, rounded rectangle, rounded rectangle right arc=none, minimum height=1em, minimum width=1em, node font={\scriptsize}]
\tikzstyle{turquoisenode}=[fill={rgb,255: red,115; green,255; blue,239}, draw=black, shape=circle, scale=0.5]
\tikzstyle{resistor}=[R]
\tikzstyle{inductor}=[L]
\tikzstyle{capacitor}=[C]
\tikzstyle{voltage-source}=[american voltage source]
\tikzstyle{current-source}=[american current source]
\tikzstyle{togray}=[none]

\tikzstyle{edge}=[-, draw=black]
\tikzstyle{diredge}=[->, draw=black]
\tikzstyle{dashed edge}=[-, draw=black]
\tikzstyle{loosedash}=[-, dash pattern=on 1.5pt off 3pt, draw=black]
\tikzstyle{dirdash}=[->, dashed, dash pattern=on 2pt off 0.5pt, draw=black]
\tikzstyle{mapsto}=[{|->}, draw=black]
\tikzstyle{gray diredge}=[draw={rgb,255: red,221; green,221; blue,221}, ->]
\tikzstyle{dark grey dirdash}=[->, dashed, dash pattern=on 2pt off 0.5pt, draw={rgb,255: red,81; green,81; blue,81}]
\tikzstyle{doubedge}=[-, draw=black, double=none, double distance=3pt, inner sep=0pt, thick]
\tikzstyle{thedge}=[-, line width=1.5pt, draw=black]
\tikzstyle{gray dashed}=[-, dashed, dash pattern=on 1pt off 1.5pt, draw={rgb,255: red,128; green,128; blue,128}]
\tikzstyle{rededge}=[-, draw=red]
\tikzstyle{gray edge}=[-, draw={rgb,255: red,128; green,128; blue,128}]
\tikzstyle{blthedge}=[-, thick, draw={rgb,255: red,4; green,51; blue,255}]
\tikzstyle{blthedge-extend}=[-, thick, draw={rgb,255: red,4; green,51; blue,255}, shorten >= -0.3pt, shorten <= -0.3pt]
\tikzstyle{blthdash}=[-, dashed, dash pattern=on 3pt off 1pt, thick, draw={rgb,255: red,4; green,51; blue,255}]
\tikzstyle{dirrededge}=[draw=red, ->]

\tikzstyle{object}=[inner sep=0mm, shape=circle, fill=none]
\tikzstyle{bullet}=[fill=black, draw=black, shape=circle, scale=0.3]
\tikzstyle{circ}=[fill=white, draw=black, shape=circle, scale=0.3]
\tikzstyle{objectbox}=[inner sep=3pt, shape=rectangle, fill=none]
\tikzstyle{bdyobjectbox}=[fill=white, draw=black, shape=rectangle]

\tikzstyle{morphism}=[->, draw=black]
\tikzstyle{dash morphism}=[->, dashed, dash pattern=on 1.5pt off 1pt, draw=black]
\tikzstyle{mapsto}=[{|->}, draw=black]
\tikzstyle{nat transf}=[-implies, double, double distance=3pt, thick]
\tikzstyle{gray nat transf}=[-implies, draw=gray, double, double distance=3pt, thick]
\tikzstyle{equality}=[-, double, double distance=3pt]
\tikzstyle{squig morphism}=[->, draw=black, line join=round, decorate, decoration={zigzag, segment length=4, amplitude=.9,post=lineto, post length=2pt}]
\tikzstyle{hookarrow}=[right hook->, draw=black]
\tikzstyle{hookarrowmirror}=[left hook->, draw=black]
\tikzstyle{dotted line}=[-, dotted, draw=black]

\definecolor{electroblue}{RGB}{4,51,255}
\usepackage{siunitx}
\usepackage[american,siunitx]{circuitikz}
\ctikzset{
    color/.initial=electroblue,
    color=electroblue,
    bipoles/thickness=2.5,
    bipoles/length=1.55cm,
    sources/scale=0.8,
    capacitors/scale=0.9,
}
\tikzstyle{vsource}=[rmeter, t={\textsf{\tiny -- +}}] 
\tikzstyle{ammeter}=[rmeter, t={\textsf{A}}] 
\tikzstyle{vmeter}=[rmeter, t={\textsf{V}}] 
\tikzstyle{elecdot}=[circle,fill,inner sep=0.85pt]

\newcommand{\includegraffle}[1]{
  {\lower10pt\hbox{$\includegraphics[height=1cm]{tikz/#1.pdf}$}}
}

\input{./sty/circuits.tikzdefs}

\tikzstyle{none}=[anchor=center]
\tikzstyle{vertex}=[anchor=center, fill=black, draw=black, shape=circle, minimum size=2.5mm, tikzit category=hypergraph, inner sep=0]
\tikzstyle{edge subgraph}=[fill=neutral, draw=black, shape=rectangle, font={\boxsize}, tikzit category=hypergraph, minimum width=8mm, minimum height=8mm, rounded corners=3mm, very thick, dashed]
\tikzstyle{port}=[minimum size=2.5mm, fill=none, draw=none, shape=circle, tikzit draw={rgb,255: red,154; green,154; blue,154}, anchor=center, tikzit fill=neutral]
\tikzstyle{product}=[fill=neutral, draw=black, shape=circle, scale=0.66, tikzit category=string diagram]
\tikzstyle{type}=[fill=none, draw=none, shape=circle, font={\large}, tikzit fill={rgb,255: red,37; green,193; blue,141}, inner sep=0, anchor=center, tikzit category=string diagram]
\tikzstyle{tiny box white}=[{\corners}, font={\boxsize}, fill=neutral, draw=black, shape=rectangle, tikzit category=string diagram, minimum width={\tinywidth}, minimum height={\tinywidth}, anchor=center, line width={\stringwidth}, tikzit fill=neutral, inner sep={\innersep}]
\tikzstyle{tiny box comb}=[{\corners}, font={\boxsize}, fill=comb, draw=black, shape=rectangle, tikzit category=string diagram, minimum width={\tinywidth}, minimum height={\tinywidth}, anchor=center, line width={\stringwidth}, tikzit fill=comb, inner sep={\innersep}]
\tikzstyle{tiny box seq}=[{\corners}, font={\boxsize}, fill=seq, draw=black, shape=rectangle, tikzit category=string diagram, minimum width={\tinywidth}, minimum height={\tinywidth}, anchor=center, line width={\stringwidth}, inner sep={\innersep}]
\tikzstyle{tiny signal seq}=[{\corners}, font={\boxsize}, shape=signal, signal to=west, signal pointer angle=110, fill=seq, draw=black, tikzit category=string diagram, minimum width=6mm, minimum height=5mm, anchor=center, line width={\stringwidth}, inner sep={\innersep}]
\tikzstyle{small box white}=[{\corners}, font={\boxsize}, fill=neutral, draw=black, shape=rectangle, minimum height={\smallwidth}, minimum width={\tinywidth}, tikzit category=string diagram, anchor=center, line width={\stringwidth}, inner sep={\innersep}]
\tikzstyle{small square box white}=[{\corners}, font={\boxsize}, fill=neutral, draw=black, shape=rectangle, minimum height={\smallwidth}, minimum width={\smallwidth}, tikzit category=string diagram, anchor=center, line width={\stringwidth}, inner sep={\innersep}]
\tikzstyle{medium box}=[{\corners}, font={\boxsize}, fill=neutral, draw=black, shape=rectangle, tikzit category=string diagram, minimum height={\mediumwidth}, minimum width={\smallwidth}, anchor=center, line width={\stringwidth}, tikzit fill=neutral, inner sep={\innersep}]
\tikzstyle{medium box white}=[{\corners}, font={\boxsize}, fill=neutral, draw=black, shape=rectangle, tikzit category=string diagram, minimum height={\mediumwidth}, minimum width={\smallwidth}, anchor=center, line width={\stringwidth}, tikzit fill=comb, inner sep={\innersep}]
\tikzstyle{medium box comb}=[{\corners}, font={\boxsize}, fill=comb, draw=black, shape=rectangle, tikzit category=string diagram, minimum height={\mediumwidth}, minimum width={\smallwidth}, anchor=center, line width={\stringwidth}, tikzit fill=comb, inner sep={\innersep}]
\tikzstyle{medium square box white}=[{\corners}, font={\boxsize}, fill=neutral, draw=black, shape=rectangle, tikzit category=string diagram, minimum height={\mediumwidth}, minimum width={\mediumwidth}, line width={\stringwidth}, tikzit fill=neutral, inner sep={\innersep}]
\tikzstyle{medium square box comb}=[{\corners}, font={\boxsize}, fill=comb, draw=black, shape=rectangle, tikzit category=string diagram, minimum height={\mediumwidth}, minimum width={\mediumwidth}, line width={\stringwidth}, tikzit fill=comb, inner sep={\innersep}]
\tikzstyle{medium square box seq}=[{\corners}, draw=black, font={\boxsize}, fill=seq, shape=rectangle, tikzit category=string diagram, minimum height={\mediumwidth}, minimum width={\mediumwidth}, line width={\stringwidth}, tikzit fill=seq, inner sep={\innersep}]
\tikzstyle{medium square rounded box white}=[rounded corners, font={\boxsize}, fill=neutral, draw=black, shape=rectangle, tikzit category=string diagram, minimum height={\mediumwidth}, minimum width={\mediumwidth}, line width={\stringwidth}, tikzit fill=neutral, inner sep={\innersep}]
\tikzstyle{medium square rounded box comb}=[rounded corners, font={\boxsize}, fill=comb, draw=black, shape=rectangle, tikzit category=string diagram, minimum height={\mediumwidth}, minimum width={\mediumwidth}, line width={\stringwidth}, tikzit fill=comb, inner sep={\innersep}]
\tikzstyle{medium square rounded box seq}=[rounded corners, draw=black, font={\boxsize}, fill=seq, shape=rectangle, tikzit category=string diagram, minimum height={\mediumwidth}, minimum width={\mediumwidth}, line width={\stringwidth}, tikzit fill=seq, inner sep={\innersep}]
\tikzstyle{large box}=[{\corners}, font={\boxsize}, fill=neutral, draw=black, shape=rectangle, tikzit category=string diagram, minimum height=15mm, minimum width=10mm, anchor=center, line width={\stringwidth}, tikzit fill=neutral, inner sep={\innersep}]
\tikzstyle{large square box white}=[{\corners}, font={\boxsize}, fill=neutral, draw=black, shape=rectangle, tikzit category=string diagram, minimum width=15mm, minimum height=15mm, line width={\stringwidth}, tikzit fill=neutral, inner sep={\innersep}]
\tikzstyle{large square box comb}=[{\corners}, font={\boxsize}, fill=comb, draw=black, shape=rectangle, tikzit category=string diagram, minimum width=12mm, minimum height=12mm, line width={\stringwidth}, tikzit fill=comb, inner sep={\innersep}]
\tikzstyle{huge box}=[{\corners}, fill=neutral, draw=black, shape=rectangle, minimum height=40mm, minimum width=15mm, anchor=center, tikzit category=string diagram, line width={\stringwidth}, tikzit fill=neutral, inner sep={\innersep}]
\tikzstyle{output-node}=[fill=neutral, draw=black, shape=circle, anchor=west, inner sep=0.5]
\tikzstyle{state}=[fill=neutral, draw=none, shape=rectangle, minimum height=10mm, minimum width=10mm]
\tikzstyle{delay}=[fill=neutral, draw=black, font={\boxsize}, shape=signal, tikzit category=string diagram, line width={\stringwidth}, minimum height={\tinywidth}, minimum width={\tinywidth}, inner sep=0.5mm, outer sep=0mm, minimum height=1em, minimum width=1em]
\tikzstyle{register}=[fill=seq, draw=black, font={\boxsize}, shape=signal, tikzit category=string diagram, line width={\stringwidth}, minimum height={2em}, minimum width={1.5em}, align=center, anchor=center, inner xsep=1mm, inner ysep=-1mm, outer xsep=0mm]
\tikzstyle{waveform}=[fill=seq, draw=black, font={\boxsize}, shape=signal, signal from=west, signal to=east, tikzit category=string diagram, line width={\stringwidth}, minimum height={2em}, minimum width={1.5em}, align=center, anchor=center, inner xsep=1mm, inner ysep=-1mm, outer xsep=0mm]
\tikzstyle{bproduct}=[fill=black, draw=black, shape=circle, scale=0.5, tikzit category=string diagram]
\tikzstyle{wproduct}=[fill=neutral, draw=black, shape=circle, scale=0.5, line width=0.75, tikzit category=string diagram]
\tikzstyle{gproduct}=[fill=unit, draw=unit, shape=circle, scale=0.5, tikzit category=string diagram]
\tikzstyle{bport}=[minimum size=2.5mm, fill=none, draw=none, shape=circle, tikzit draw={rgb,255: red,154; green,154; blue,154}, anchor=center, tikzit fill=neutral, tikzit category=string diagram]
\tikzstyle{mux}=[fill=neutral, draw=black, shape=trapezium, tikzit category=circuits, rotate=-90, minimum height=1em, line width={\stringwidth}, scale=1.25]
\tikzstyle{fork}=[fill=black, draw=black, shape=circle, tikzit category=circuits, scale=0.25]
\tikzstyle{box}=[fill=neutral, draw=black, shape=rectangle, tikzit category=circuits]
\tikzstyle{dangling}=[fill=none, draw=none, shape=circle, anchor=east, scale=0.01, tikzit category=string diagram, tikzit fill={rgb,255: red,162; green,76; blue,77}]
\tikzstyle{label}=[fill=none, draw=none, shape=circle, align=center, inner sep=0, outer sep=0]
\tikzstyle{small label}=[scale=0.75, fill=none, draw=none, shape=circle, outer sep=0, inner sep=0, anchor=center]
\tikzstyle{wire label left}=[scale=1.25, fill=none, draw=none, shape=rectangle, outer sep=0, inner sep=0, anchor=east]
\tikzstyle{wire label right}=[scale=1.25, fill=none, draw=none, shape=rectangle, outer sep=0, inner sep=0, anchor=west]
\tikzstyle{wire label mid}=[scale=1.25, fill=\bgcolour, shape=rectangle, outer sep=0, inner sep=0, shape=circle, minimum size=0]
\tikzstyle{tile}=[fill=neutral, draw=black, shape=rectangle, tikzit category=string diagram, {\corners}, minimum height=5mm, minimum width=5mm]
\tikzstyle{commuting label}=[fill=none, draw=none, shape=circle, scale=0.5]
\tikzstyle{and}=[fill=neutral, draw=black, shape=and gate US, line width={\stringwidth}, and gate, scale=1.75, tikzit category=circuits]
\tikzstyle{or}=[fill=neutral, draw=black, or gate US, scale=1.75, line width={\stringwidth}, tikzit category=circuits]
\tikzstyle{not}=[fill=neutral, draw=black, not gate US, scale=1.5, line width={\stringwidth}, tikzit category=circuits]
\tikzstyle{nor}=[fill=neutral, draw=black, nor gate, scale=1.75, line width={\stringwidth}, tikzit category=circuits]
\tikzstyle{nand}=[fill=neutral, draw=black, nand gate, scale=1.75, line width={\stringwidth}, tikzit category=circuits]
\tikzstyle{xor}=[fill=neutral, draw=black, xor gate, scale=1.75, line width={\stringwidth}, tikzit category=circuits]
\tikzstyle{xnor}=[fill=neutral, draw=black, xnor gate, scale=1.75, line width={\stringwidth}, tikzit category=circuits]
\tikzstyle{west}=[fill=none, draw=none, shape=circle, anchor=west]
\tikzstyle{bundler}=[fill=neutral, draw=black, shape=rounded rectangle, minimum width=3em, rounded rectangle arc length=180, rotate=90, line width={\stringwidth}]
\tikzstyle{long bundler}=[fill=neutral, draw=black, shape=rounded rectangle, minimum width=5em, rounded rectangle arc length=180, rotate=90, line width={\stringwidth}]

\tikzstyle{interface}=[-, fill={rgb,255: red,238; green,238; blue,255}, dashed, draw={rgb,255: red,170; green,170; blue,225}, ultra thick]
\tikzstyle{graph}=[-, fill={rgb,255: red,238; green,238; blue,238}, draw={rgb,255: red,191; green,191; blue,191}, dashed, ultra thick, anchor=center]
\tikzstyle{tentacle}=[-, very thick]
\tikzstyle{wire}=[-, tikzit category=string diagram, line width={\stringwidth}]
\tikzstyle{empty}=[-, densely dashed, {\corners}, dash pattern=on 0pt off 1.25pt on 1.25pt, line width={\stringwidth}]
\tikzstyle{transition}=[->]
\tikzstyle{output}=[->, decorate, decoration=zigzag]
\tikzstyle{gate}=[-, fill=neutral]
\tikzstyle{thicc}=[-, line width=1]
\tikzstyle{strikethrough}=[-, decoration={markings, mark=at position 0.5 with {
        \draw [blue, thick, - ] (-0.15,-0.15);
    }}, postaction=decorate]
\tikzstyle{traced}=[-, densely dashed, draw=gray]
\tikzstyle{arrow}=[<-, line width={\stringwidth}]
\tikzstyle{arrow up}=[->, line width={\stringwidth}]
\tikzstyle{dashed arrow}=[<-, dashed]
\tikzstyle{unit wire}=[-, fill=none, draw=unit, line width={\stringwidth}]
\tikzstyle{interfacearrow}=[->, very thick, dashed]
\tikzstyle{tile none}=[-, draw=none, {\corners}, fill=none, line width={\stringwidth}]
\tikzstyle{tile white}=[-, {\corners}, fill=neutral, line width={\stringwidth}]
\tikzstyle{tile comb}=[-, {\corners}, fill=comb, line width={\stringwidth}]
\tikzstyle{tile seq}=[-, {\corners}, fill=seq, line width={\stringwidth}]
\tikzstyle{commute}=[->]
\tikzstyle{curved rectangle}=[-, rounded corners]
\tikzstyle{hasse}=[-]
\tikzstyle{wiredash}=[-, line width={\stringwidth}]
\tikzstyle{rewrite}=[-, dashed, line width={\stringwidth}]
\tikzstyle{juxtaposition}=[-, draw={rgb,255: red,128; green,128; blue,128}, densely dashed, line width={\stringwidth}]
\tikzstyle{functor box}=[-, thick, draw={rgb,255: red,83; green,83; blue,83}]
\tikzstyle{boundary box}=[-, draw=none, tikzit draw={rgb,255: red,255; green,0; blue,4}]

\input{./sty/zx.tikzdefs}

\tikzstyle{box}=[shape=rectangle, text height=1.5ex, text depth=0.25ex, yshift=0.5mm, fill=white, draw=black, minimum height=5mm, yshift=-0.5mm, minimum width=5mm, font={\small}]
\tikzstyle{Z dot}=[inner sep=0mm, minimum size=2mm, shape=circle, draw=black, fill={rgb,255: red,221; green,255; blue,221}]
\tikzstyle{Z phase dot}=[minimum size=1.2em, font={\footnotesize\boldmath}, shape=rectangle, rounded corners=0.5em, inner sep=0.2em, outer sep=-0.2em, scale=0.8, tikzit shape=circle, draw=black, fill={rgb,255: red,221; green,255; blue,221}, tikzit draw=blue]
\tikzstyle{X dot}=[Z dot, shape=circle, draw=black, fill={rgb,255: red,255; green,136; blue,136}]
\tikzstyle{X phase dot}=[Z phase dot, tikzit shape=circle, tikzit draw=blue, fill={rgb,255: red,255; green,136; blue,136}, font={\footnotesize\boldmath}]
\tikzstyle{hadamard}=[fill=yellow, draw=black, shape=rectangle, inner sep=0.6mm, minimum height=1.5mm, minimum width=1.5mm]
\tikzstyle{vertex}=[inner sep=0mm, minimum size=1mm, shape=circle, draw=black, fill=black]
\tikzstyle{vertex set}=[inner sep=0mm, minimum size=1mm, shape=circle, draw=black, fill=white, font={\footnotesize\boldmath}]
\tikzstyle{target}=[inner sep=0mm, minimum size=3mm, shape=circle, draw=black]

\tikzstyle{hadamard edge}=[-, dashed, dash pattern=on 2pt off 1.5pt, thick, draw={rgb,255: red,68; green,136; blue,255}]
\tikzstyle{brace edge}=[-, tikzit draw=blue, decorate, decoration={brace,amplitude=1mm,raise=-1mm}]
\tikzstyle{diredge}=[->]
\tikzstyle{dashed edge}=[-, dashed, dash pattern=on 2pt off 0.5pt, draw=black]

\usepackage{etoolbox}
\patchcmd{\abstract}{\titlepage}{\cleardoublepage}{}{}
\patchcmd{\endabstract}{\endtitlepage}{\clearpage}{}{}

\hypersetup{
    colorlinks=true,
    linkcolor=blue,
    filecolor=magenta,      
    urlcolor=violet,
    citecolor=red,
    pdftitle={Layered Monoidal Theories},
    }

\graphicspath{ {figures/} }

\begin{document}
\begin{titlepage}
\centering
\vspace*{4cm}
{\huge\bfseries Layered Monoidal Theories\par}
\vspace{2cm}
{\Large Leo Lobski\par}
\vfill
A thesis submitted to\par
\textsc{University College London}\par
for the degree of\par
\textsc{Doctor of Philosophy}\par
\vspace{2cm}
December 2025
\end{titlepage}

\begin{abstract}
\phantomsection
\addcontentsline{toc}{section}{Abstract}
In the first part, we develop {\em layered monoidal theories} -- a generalisation of monoidal theories combining descriptions of a system at several levels. Via their representation as {\em string diagrams}, monoidal theories provide a graphical, recursively defined syntax with a visually intuitive notion of information flow, and built-in notions of parallel and sequential composition. {\em Layered} monoidal theories allow mixing several monoidal theories (together with translations between them) within the same string diagram, while retaining mathematical precision and semantic interpretability. We define three flavours of layered monoidal theories, provide a recursively generated syntax for each one, and construct a free-forgetful adjunction with respect to three closely related semantics: {\em opfibrations}, {\em fibrations} and {\em deflations}. To illustrate the motivation and potential use cases, we provide several examples of notions appearing in the literature that are instances of layered monoidal theories.

In the second part, we develop a formal approach to {\em retrosynthesis} -- the process of backwards reaction search in synthetic chemistry. Chemical processes are treated at three levels of abstraction: (1) {\em (formal) reactions} encode all chemically feasible combinatorial rearrangements of molecules, (2) {\em reaction schemes} encode transformations applicable to ``patches'' of molecules (including the functional groups), and (3) {\em disconnection rules} encode local chemical rewrite rules applicable to a single bond or atom at a time. We show that the three levels are tightly linked: the reactions are generated by the reaction schemes via double pushout graph rewriting, while there is a functorial translation from the disconnection rules to the reactions. Moreover, the translation from the category of disconnection rules to reactions is shown to be sound, complete and universal -- allowing one to treat the disconnection rules as a formal syntax with the semantics provided by the reactions. Finally, we tie together the two parts by providing a formalisation of retrosynthesis within a certain layered monoidal theory.
\end{abstract}
\newpage

\section*{On the shoulders of variably sized creatures: Acknowledgements}\label{acknowledgements}
\addcontentsline{toc}{section}{\nameref{acknowledgements}}
Firstly, I would like to thank my supervisor Fabio Zanasi for agreeing to supervise a PhD project with such a vague idea of ``levels of abstraction'', and for his general openness towards ideas that do not necessarily follow the most mainstream trends. I'm also grateful to Fabio for encouraging me to explore concrete examples, and suggesting many of them -- the case studies in Chapter~\ref{ch:layered-examples} would certainly be much fewer if it wasn't for Fabio's efforts and suggestions.

I would like to thank Ella Gale for coauthoring a number of papers on formalising retrosynthesis -- the entirety of Part~\ref{part:chemistry} would not be possible without our discussions of (organic) chemistry and formalised retrosynthesis, as well as connections between computation, chemistry and category theory.

I thank Samson Abramsky, Mehrnoosh Sadrzadeh and Fredrik Dahlqvist for their feedback, formal and informal, during the intermediate vivas, after my talks and throughout my time at UCL. This has contributed greatly to the content and presentation of the thesis. I thank Jean Krivine and Mehrnoosh Sadrzadeh for a very thorough, encouraging and intellectually stimulating viva examination.

I thank Cole Comfort for pointing me towards the literature on pointed profunctors at an early stage of my PhD, as well as continued discussions of collages and related things.

I thank the attendees and organisers of the Categorical Logic and Type Theory~\cite{jacobs-cltt} reading group, that we ran during the academic year 2022-23: Cameron Michie, Daphne Wang, Lachlan McPheat, Luka Ilic, Tao Gu, Yll Buzoku. I am very grateful that we managed to work through almost the entire book. The fact that I've learned something is hopefully visible in parts of this thesis.

I thank Christina Vasilakopoulou and Joe Moeller for answering my questions via email, helping me understand their work on monoidal fibrations.

I thank Cheng Zhang, Kevin Batz and Wojtek Różowski for explaining the Calculus of Communicating Systems to me. Likewise, I thank Thomas Noll and Kevin Batz for answering my questions about derivable transitions in the calculus in an email exchange. I further thank Jean Krivine for explaining to me why I initially failed to capture the calculus in Section~\ref{sec:ccs} (despite the efforts of all the people mentioned above).

I thank Ralph Sarkis and Antonio Lorenzin for feedback on Section~\ref{sec:prob-channels}.

I thank Antonio Lorenzin for extremely detailed feedback on the first version of entirety of Part~\ref{part:layered} of the thesis.

I thank Cheng Zhang for (unknowingly) suggesting the title for Part~\ref{part:chemistry} of the thesis.

A very special thank you goes to Mateo Torres-Ruiz and Wojtek Różowski for sharing the two basements for the last three years or so. I am very fond of, and grateful for, the endless discussions of mathematics, politics, philosophy, literature, the State, the European Union, the lack... It has been a pleasure to live and work in such an intellectual atmosphere. I sincerely hope that a fraction of that atmosphere is discernible between the lines (or, at times, perhaps even within the lines) of this thesis.

I thank my partner Katya Piotrovskaya for support and explaining base-extension semantics to me during the breaks that I took from writing this thesis.

I would like to extend the thank yous to everyone who has worked from the (one and only) basement office of 66-72 Gower Street during these years, and thereby contributed to an amazing work collective and environment (state-mandated acknowledgements). With massive apologies to anyone I've forgotten, these people include:
Alessandro Di Giorgio,
Amin Karamlou,
Andrea Corradetti,
Antonio Lorenzin,
Carmen Constantin,
Cheng Zhang,
Fredrik Dahlqvist,
Gerco Van Heerdt,
Jas Semrl,
Jialu Bao,
Joni Puljujärvi,
Katya Piotrovskaya,
Kevin Batz,
Léo Henry,
Linpeng Zhang,
Lisa Tse,
Louis Parlant,
Luca Reggio,
Martti Karvonen,
Mateo Torres-Ruiz,
Oliver Bøving,
Rafał Stefański,
Ralph Sarkis,
Robin Piedeleu,
Sam Coward,
Stefan Zetzsche,
Thibaut Antoine,
Tiago Ferreira,
Will Smith,
Wojtek Różowski,
Yll Buzoku (honorary basement fellow via the terminal).
The third floor gang also deserves an honorary mention: Alex Gheorghiu, Tao Gu, Timo Eckhardt, Timo Lang, Yll Buzoku (third floor fellow via physical presence).

Importantly, I want to thank all the cleaning and maintenance staff responsible for the upkeep of the building, the library and information services staff, the technicians, as well as the catering staff. Due to UCL's highly fragmented hiring practises, I do not know most of their names. All the same, none of this would have been possible without the layers of labour they provide.

Lastly, I would like to thank my parents, Irina and Alexei. Nothing that could be written here could convey the scope of their support and care, so I restrict myself to saying that many parts of this thesis have been written in their house in my hometown Jyv\"askyl\"a, Finland.

\newpage

\phantomsection
\addcontentsline{toc}{section}{Contents}
\tableofcontents
\newpage

\part{Layered string diagrams}\label{part:layered}

\chapter{Introduction}\label{ch:layered-introduction}
We study {\em systems} that admit several {\em compositional} descriptions at different {\em levels} of detail.

A {\em system} here is to be interpreted very broadly: examples of what qualifies as a system include logical systems of deduction, computer architecture, cyber-physical systems, as well as natural processes such as biochemical systems or quantum systems. We shall see six extended case studies of how the general theory developed here applies to various systems in Chapter~\ref{ch:layered-examples}.

By {\em compositionality} we mean that the processes within each level of description may be combined to form composite processes. Monoidal theories provide two systematic ways to combine processes. First, any two processes can be composed in {\em parallel}, corresponding to their simultaneous occurrence or execution. Second, if the output type of a process $p:a\rightarrow b$ equals the input type of a process $q:b\rightarrow c$, the processes can be composed {\em sequentially} to obtain $p;q : a\rightarrow c$, corresponding to their occurrence or execution one after another.

An important intuition motivating this work is that the same system can be described at different {\em levels}, and that there are translations between the levels. The levels correspond to different perspectives on the system, perhaps emphasising various details or features. We usually think that some levels are {\em coarser} and some are more {\em fine-grained}, and that the translations correspond to ``zooming in'' or ``abstracting away'' the details.

While compositional view of systems is standard in categorical modelling, the novelty of this thesis lies in the considerations of the levels of description -- allowing for a more fine-grained view of compositionality.

\section{Layers of abstraction}\label{sec:layers-abstraction}

Many theories describing a given system organise naturally into a hierarchy of abstraction. For instance, different sciences operate at different levels: while chemistry has the vocabulary to talk about interactions between molecules, particle physics lacks the language to describe interactions of such complexity. Similarly, biology is able to talk about the function of a cell within an organism, while chemistry is restricted to describing reactions that are necessary to fulfil this function. Yet there is a sense in which the levels depend on each other: any chemical interaction occurs between molecules that consist of physical particles, any function of a cell results from biochemical interactions. It is therefore customary to think that the sciences can be arranged in layers of abstraction, starting from fundamental physics, each building on the previous one yet having its own distinct vocabulary, rules and laws.

Climbing up the abstraction levels, one is forced to ignore more and more details of the base level. However, new expressive power is simultaneously gained via the ability to talk about the clusters of interacting entities rather than mere elementary particles. From the point of view of logic and information processing, the syntax of scientific theories describing different levels should reflect the increased expressive power. Such differences in language also create barriers for interpretability of one science in another. Hence, choosing the right level of abstraction is paramount for successful communication between different disciplines.

The view on translations between the layers we take throughout the discussion is that objects in coarser layers uniquely name an object in the finer layer -- the difference in abstraction being captured by a richer set of rules in the finer layer, allowing for more compositions and decomposition of objects. We give an example of such a translation in Figure~\ref{fig:layers-of-abstraction}: the coarser layer contains a very restricted set of English language names for molecules, while the finer layer contains molecular graphs, together with rules to break and create chemical bonds. This view seems to be particularly suited for discrete systems with relatively few parts (or those that can be usefully modelled as such). Note that this excludes systems with dynamical qualities, such as perhaps the most well-known coarse-fine theory pair: thermodynamics and statistical mechanics in classical physics. The reason these theories do not fit the view presented here is that two different microstates (e.g.~motion of particles) in statistical mechanics (finer layer) may correspond to the same macrostate (e.g.~temperature) in thermodynamics (coarser level). Thus, in this case, the functional translation seems to be going in the reverse direction. We leave investigating usefulness of the formalism presented here for such scenarios for future work.

\begin{figure}
\centering
\input{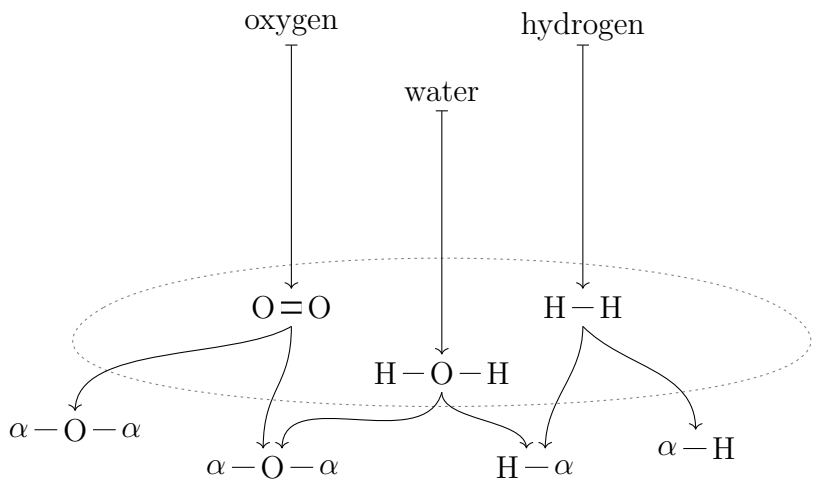}
\caption{An example translation from a coarser language (English names) to a finer language (molecular graphs).\label{fig:layers-of-abstraction}}
\end{figure}

Another question we explicitly demarcate to lie outside of the scope of the current work is any philosophical commitment to reductionism or emergentism. Very roughly, reductionism holds that any higher level theory is ultimately reducible to some lower level (usually taken as a fundamental theory of physics), whereas emergentism is the position according to which there are emergent, higher level phenomena that are lost by translating to a lower level. While the author does find such questions important for philosophy of science and a healthy scientific culture as a whole, here we focus on developing a mathematical formalism which is as free as possible from such assumptions one way or another.

\section{Layered monoidal theories}\label{sec:lay-mod-th-intro}

A {\em monoidal signature} consists of a set of {\em generators}, each drawn as a box
\begin{center}
\scalebox{1}{\input{f-morphism.tikz}}
\end{center}
with a unique {\em arity} $a$ and {\em coarity} $b$. A monoidal signature gives rise to a strict monoidal category by closing the generators under parallel and sequential composition:
\begin{center}
\scalebox{1}{\input{monoidal-sequential-composition.tikz}},
\end{center}
where the coarity of $h$ is equal to the arity of $k$, subject to certain equations\footnote{Specifically, the structural identities given in Definition~\ref{def:str-id}.}. A {\em monoidal theory} is a monoidal signature with a set of equations between parallel terms. The resulting two-dimensional syntax is known as {\em string diagrams}, and it is sound and complete with respect to monoidal categories~\cite{joyal-street88,joyal-street91,selinger,piedeleu-zanasi}. The detailed development of this is subject of Section~\ref{ch:monoidal-theories}.

One point of view on monoidal theories is that they generalise algebraic theories (in the sense of universal algebra): every algebraic theory can be presented as a monoidal theory where all wires have the same type (called {\em colour}) and all generators have coarity of size one. The wires play the role of variables, while the generators play the role of operations. For example, the theory of monoids is given by the following generators of arity $2$ and $0$ (top row) and equations (bottom row) expressing associativity and unitality:
\begin{center}
\scalebox{1}{\input{monoid.tikz}}.
\end{center}

A crucial difference between algebraic and monoidal theories is that monoidal theories are {\em resource aware}~\cite{diagrammatic-algebra19,piedeleu-zanasi}, in the sense that information cannot be copied, created or deleted at will, unlike the variables in universal algebra. This feature makes monoidal theories particularly suitable for modelling situations where keeping track of resources is crucial: we shall see several examples of such monoidal theories in Chapter~\ref{ch:layered-examples}.

Monoidal theories have found applications in quantum theory~\cite{heunen-vicary-book}, quantum computing~\cite{zx-for2020}, linear algebra and affine algebra~\cite{zanasi-thesis,graphical-affine-algebra}, signal flow theory~\cite{survey-signal-flow}, electrical circuit theory~\cite{electrical-circuits}, digital circuit theory~\cite{kaye-thesis}, linguistics~\cite{discocat}, probability theory~\cite{fritz-markov20,jacobs-spr} and causality theory~\cite{lorenz-tull-causal}.

The main contribution of this thesis is to propose a graphical formalism to study layers of abstraction. Since the formalism consists of ``glueing'' several monoidal theories together, we call the resulting mathematical structures {\em layered monoidal theories}. Intuitively, a layered monoidal theory is a diagram of monoidal categories and functors
\begin{center}
\scalebox{1}{\input{diagram-monoidal-categories.tikz}},
\end{center}
where each dot represents a monoidal category, and each arrow a monoidal functor. A typical term (morphism) in a layered theory looks like
\begin{center}
\scalebox{1}{\input{term-layered-theory.tikz}},
\end{center}
where $x$ is a morphism in the category $\omega$ (drawn as a box on the left) and $y$ is a morphism in the category $\tau$ (drawn as a box on the right). The dotted lines indicate that $x$ and $y$ sit above $\omega$ and $\tau$, and that the ``functor boundary'' $\refine_f$ sits above $f$ -- these are formally not part of the term and are included for illustrative purposes only. The morphism $x$ can be, moreover, pushed through the functor boundary as follows:
\begin{center}
\scalebox{1}{\input{term-layered-pushing.tikz}}.
\end{center}
While on its own this does not add much to the informal diagrams for categories and functors, we additionally capture the monoidal structure of each category within the diagrams. Most importantly, in addition to the functor boundaries $\refine_f$ going in the ``forward'' direction, we add functor boundaries $\coarsen_f$ in the dual (opposite) direction, resulting in diagrams of the following shape:
\begin{center}
\scalebox{1}{\input{window.tikz}},
\end{center}
which we call {\em windows} (see Section~\ref{sec:functor-coboxes}). Windows allow one to ``peek in'' at the semantics of the functor $f$ -- without performing the full translation.

The task we undertake in Chapter~\ref{ch:layered-theories} is to formalise the intuitive picture sketched here: we introduce {\em layered signatures} and recursively generate their {\em terms}. In Chapter~\ref{ch:semantics} we further show that terms organise themselves into {\em free models} upon quotienting by appropriate equations. Overall, this thesis contains four soundness and completeness results (three of which are novel), all arising by constructing a free-forgetful adjunction. We summarise the results in Table~\ref{tab:soundness-completeness-results}: the first column indicates a theory (with a reference to its definition in the current text), while the second column indicates the semantics with respect to which the theory in question is sound and complete. These semantics are tightly linked: each (split) opfibration with indexed monoids, in particular, contains models for monoidal theories as its fibres; each (split) fibration with indexed comonoids is obtained by dualising a (split) opfibration with indexed monoids; and the 1-categorical part of a (split) monoidal deflation is obtained by glueing a (split) opfibration with indexed monoids with its dual (split) fibration with indexed comonoids\footnote{What is meant by ``glueing'' is made precise by Corollary~\ref{cor:decomposition-biretrofunctor-opfibration} and Lemma~\ref{lma:monoidal-deflation-opfibration-indexed-monoids}.}.

\begin{table}
\renewcommand{\arraystretch}{1.5}
\begin{tabular}{ c c | c c }
theory & definition & semantics & definition \\
\hhline{==|==}
monoidal theory & \ref{def:monoidal-theory} & (strict) monoidal category & -- \\
\hline
opfibrational theory & \ref{def:opfib-layered-theory} & split opfibration with indexed monoids & \ref{def:opfib-indexed-mon} \\
\hline
fibrational theory & \ref{subsec:fib-theories} & split fibration with indexed comonoids & -- \\
\hline
deflational theory & \ref{def:deflational-theory} & split monoidal deflation & \ref{def:monoidal-deflation}
\end{tabular}
\renewcommand{\arraystretch}{1}
\caption{The summary of theories and the semantics they characterise.\label{tab:soundness-completeness-results}}
\end{table}

Our construction of layered monoidal theories and their free models enjoys the following desirable properties:
\begin{itemize}
\item it gives a precise specification of multiple (potentially infinitely many) monoidal theories and functorial translations between them via finite, easily manageable syntax: we give a brief example of this in Subsection~\ref{subsec:motivate-digital-circuits} below by recursively defining addition of numbers in binary on $n$-bit wires,
\item the construction of free models is completely syntactic: we build models out of recursively generated types and terms, making the theory developed here easily amenable for formalisation and automation,
\item we provide semantics in terms of well-known mathematical structures (opfibrations and profunctor collages), making potential connections with other areas of mathematics,
\item the construction is {\em modular}, in the sense that the internal terms contained between any two functor boundaries always have an interpretation within a single monoidal category; this gives the possibility of {\em partial translations} of terms, which do no affect the other parts of the diagram,
\item the case with two layers and one translation between them recovers the usual monoidal {\em functorial semantics}; this functor is faithful (corresponding to completeness) if a certain single equation in the layered theory holds (Equation~\eqref{eq:cobox-pres-id}),
\item the framework formalises examples already existing in the literature (Chapter~\ref{ch:layered-examples}): variable amount of bits in digital circuits, impedance boxes in electrical circuits, graph rewriting for quantum circuit extraction, reaction mechanisms in chemistry, derivations in concurrency theory, conditionals in probability theory,
\item the existing notation for {\em functor boxes} naturally arises within the formalism, together with its dual notion, which we dub a {\em functor cobox}, hitherto appearing in the literature only informally.
\end{itemize}

\subsection{A remark on terminology}

Throughout the thesis, the terms {\em layered monoidal theory} and {\em layered theory} are used interchangeably. We often prefer the shorter version, and since there is no notion of a layered theory that would not be monoidal, no ambiguity will arise.

In Section~\ref{sec:opfib-models}, we shall see that certain layered theories correspond to what we term {\em (op)fibrations with indexed (co)monoids}, which, in turn, correspond to (op)indexed monoidal categories (Section~\ref{sec:opfibrations-indexed-monoids}). In this light, one might argue that a better name for what we call {\em layered monoidal theories} (Definition~\ref{def:layered-theory}) would be {\em indexed monoidal theories}. However, obtaining the theories corresponding to indexed (and opindexed) monoidal categories requires restricting to particular subclasses of layered theories. The models of deflational theories (Section~\ref{sec:deflational-models}) are what we term {\em monoidal deflations} (Definition~\ref{def:monoidal-deflation}): while each monoidal deflation restricts to both an opfibration with indexed monoids and a fibration with indexed comonoids, it does, in general, carry more information than the restrictions. The definition of a layered theory (Definition~\ref{def:layered-theory}) is flexible enough to accommodate other theories that are even further removed from indexed monoidal categories, such as introducing purely profunctorial generators. For these reasons, we feel that {\em indexed monoidal theory} would be too suggestive of the interpretation as an indexed monoidal category, hence we choose the name {\em layered monoidal theory}.

\subsection{A motivating example: Digital circuits}\label{subsec:motivate-digital-circuits}

A typical example where hierarchical levels of abstraction arise is computer architecture, where higher level descriptions have to be ultimately implemented as microelectronic circuits. For example, consider an {\em arithmetic logic unit} ({ALU}) capable of the basic logic and arithmetic operations on bits. Suppose that at the level of wires carrying a single bit signal we have a binary addition operation
\begin{center}
\scalebox{.6}{\input{plus-one.tikz}},
\end{center}
which takes as input a signal (either $0$ or $1$ in each wire), and outputs its sum represented as a binary number, with the most significant bit at the bottom. Explicitly, the behaviour is defined as follows, where the most significant bit is at the right:
\begin{align*}
00 &\mapsto 00 \\
01 &\mapsto 10 \\
10 &\mapsto 10 \\
11 &\mapsto 01.
\end{align*}
One can now recursively define addition on numbers represented by a binary sequence of any length. Hence, for each $n\geq 2$, we wish to define the operations
\begin{center}
\scalebox{.7}{\input{plus.tikz}},
\end{center}
where each wire carries an $n$-bit signal, and $+$ computes the sum of the two integers. In order to avoid memory overflow, the operation discards the most significant bit (i.e.~the bottom one). Hence, we may use the above operation on $n$-bit wires in order to compute the sum of two numbers whose binary representation has strictly less than $n$ bits. This, however, is not a problem, as we shall define such an operation for every positive integer $n$ greater than or equal to $2$. Thus, for $n=2$, we define
\begin{center}
\scalebox{.7}{\input{plus-two.tikz}},
\end{center}
where the dashed box is to be interpreted as a ``boundary'' translating the wires carrying a $2$-bit signal (labelled with $2$) to two single bit wires. The $2$-bit wires live inside the box, and become a pair of single bit wires outside of the box. This gives the base case of the recursion. For the recursive case, suppose that the operation has been defined for all $i$ with $2\geq i\geq n$, and we wish to define it for $n+1$. Let us write $n=k+1$, so that $n+1=k+2$. We thus define
\begin{center}
\scalebox{.7}{\input{plus-n.tikz}},
\end{center}
where $f(k)$ denotes $k$ parallel single bit wires, and \scalebox{.4}{$\input{or-one.tikz}$} is the bitwise {OR} operation.

Note that we mix different operations and wires within the same equation, separated by the dashed boxes. Each dashed box is what is known as a {\em functor box}, and they separate the {\em layers} of the theory. In this case, there is one layer for each positive integer $n$, representing wires and operations on $n$-bit signals. Thus, as expected, each layer can be translated to the single bit layer (this need not be the case in general). We give the fully detailed definition of this layered theory in Section~\ref{sec:digital-circuits}.

\subsection{A case for 2-cells: Reasoning within the layered theories}

In many scenarios, there are transformations between processes that are non-trivial, in the sense that they are not identities, and hence change the process in question. Such a transformation is depicted as ``an arrow between morphisms (= processes)'' as follows
\begin{center}
\scalebox{1}{\input{twocell-int-xy.tikz}}.
\end{center}
Note that the domain and the codomain of the morphisms connected by such a transformation must coincide.

It is, therefore, important to incorporate such transformations as above into the notion of a layered theory. We do this via allowing a choice of generating {\em 2-cells} (Definition~\ref{def:choice-2cells}) and {\em 2-equations} between 2-terms (Definition~\ref{def:2equations}). This is a major departure from plain monoidal theories (Chapter~\ref{ch:monoidal-theories}), as we need to deal with terms arising at two levels, as well as with equations between them.

The payoff of incorporating the 2-cells is a notion of (discrete) dynamics, or {\em computational trace} capturing the change of a process within a layered theory. We may think of a sequence of 2-cells as a progressing computation, or as a chain of consecutive rule applications. The 2-cells allow for reasoning and proofs within the layered theories, hence enriching the static representation of processes.

We hope to convince the reader that the theoretical exercise of adding 2-cells is worthwhile in Chapter~\ref{ch:functor-boxes}, where the 2-cells are used to study the properties of functor boxes and coboxes, and in Section~\ref{sec:zx-extraction}, where the 2-cells are certain graph rewrites that in~\cite{thereandback} implement an algorithm which extracts a quantum circuit from a more general graph-like structure.

\section{Structure of Part~\ref{part:layered} of the thesis}

This work separates into four parts:
\begin{enumerate}
\item {\em Preliminaries} --- covered in Chapters~\ref{ch:preliminaries} and~\ref{ch:monoidal-theories},
\item {\em Syntax and general theory} --- covered in Chapters~\ref{ch:layered-theories} and~\ref{ch:functor-boxes},
\item {\em Applications} --- covered in Chapter~\ref{ch:layered-examples},
\item {\em Semantics} --- covered in Chapters~\ref{ch:indexed-monoidal}, \ref{ch:profunctor-collages} and~\ref{ch:semantics}.
\end{enumerate}
In particular, we point out that the last part (semantics) is relevant to the other parts only insofar as (op)fibrational and deflational layered theories are given a class of mathematical models, and it should thus be possible to get a good understanding of how layered theories work by just reading parts 1--3.

For a reader who wishes to see concrete examples first, it is certainly possible to follow the general ideas of the case studies in Chapter~\ref{ch:layered-examples} without reading all the preceding theory. The syntax developed in Chapter~\ref{ch:layered-theories} can then be taken as a mathematical formalisation of such ideas. Finally, the chapters on semantics provide an abstract universe of models in which the syntax is interpreted.

A more semantically inclined reader may want to read Chapters~\ref{ch:indexed-monoidal} and~\ref{ch:profunctor-collages} on indexed monoids and profunctor collages before looking at the syntax in Chapter~\ref{ch:layered-theories}, as these two chapters are independent of any preceding developments (apart from the fairly standard terminology and results recalled in the preliminaries).

In more detail, the remaining chapters are structured as follows.
\begin{itemize}
\item[Chapter~\ref{ch:preliminaries}:] We establish some terminological and notational conventions, define fibrations and indexed categories (and related notions), state the well-known equivalence between them, as well as its extension to the monoidal case.
\item[Chapter~\ref{ch:monoidal-theories}:] We define monoidal theories, and establish a free-forgetful adjunction between monoidal theories and strict monoidal categories.
\item[Chapter~\ref{ch:layered-theories}:] We introduce layered signatures and theories, as well as three classes of theories: opfibrational, fibrational and deflational theories.
\item[Chapter~\ref{ch:functor-boxes}:] We show how functor boxes and coboxes arise within deflational theories, and show how they can be used for reasoning within the theory. This provides first examples of proofs within layered theories.
\item[Chapter~\ref{ch:layered-examples}:] We present six extensive case studies of existing formalisms where layered theories provide conceptual clarity or are useful for computations. It is certainly not necessary to follow all (or, indeed, any) of the case studies in detail: they do, however, provide an important justification for development of the general theory.
\item[Chapter~\ref{ch:indexed-monoidal}:] We define categories and opfibrations with indexed monoids, and study their properties. Opfibrations with indexed monoids (Definition~\ref{def:opfib-indexed-mon}) provide the semantics of opfibrational theories (Definition~\ref{def:opfib-layered-theory}) in Section~\ref{sec:opfib-models}.
\item[Chapter~\ref{ch:profunctor-collages}:] We state the definition of profunctors, displayed categories and collages, as well as the theorem establishing an equivalence between displayed categories and functors. We introduce deflations and their monoidal versions, and show that minimal monoidal deflations are equivalent to indexed monoidal categories. Monoidal deflations (Definition~\ref{def:monoidal-deflation}) provide the semantics of deflational theories (Definition~\ref{def:deflational-theory}) in Section~\ref{sec:deflational-models}.
\item[Chapter~\ref{ch:semantics}:] We establish three free-forgetful adjunctions for the opfibrational, fibrational and deflational theories, and the corresponding categories of models.
\item[Chapter~\ref{ch:layered-discussion}:] We conclude with a discussion of related work, future applications and theoretical developments.
\end{itemize}

\section{Summary of contributions}\label{sec:summary-contributions}

The work in the first part of the thesis is a significant extension of~\cite{lobski-zanasi}. We list novel contributions per each chapter, and for the case studies (Chapter~\ref{ch:layered-examples}) per each section.
\begin{itemize}
\item[Chapter~\ref{ch:monoidal-theories}:] Nothing here is new, albeit our presentation in terms of fibrations of theories and models is somewhat nonstandard.
\item[Chapter~\ref{ch:layered-theories}:] Everything in this chapter is new. Most importantly, this includes the following definitions: layered signature (Definition~\ref{def:layered-signature}), layered theory (Definition~\ref{def:layered-theory}), opfibrational theory (Definition~\ref{def:opfib-layered-theory}) and deflational theory (Definition~\ref{def:deflational-theory}).
\item[Chapter~\ref{ch:functor-boxes}:] While functor boxes are a standard notation for functors, their decomposition into the functor boundaries (Proposition~\ref{prop:cowindow-box-equality}) is a feature of deflational theories. To the best knowledge of the author, the notion of a cobox (Section~\ref{sec:functor-coboxes}) is new.
\item[Section~\ref{sec:digital-circuits}:] Only the axiomatisation as a layered theory (Figure~\ref{fig:dig-circ-eqns}) is new. This builds on the work of Kaye~\cite{kaye-thesis} on string diagrammatic digital circuit theory.
\item[Section~\ref{sec:elec-circuits}:] Only the presentation as a layered theory and the decomposition of the impedance box (Definition~\ref{def:impedance-box}) are new. This builds on the work of Boisseau and Sobociński~\cite{electrical-circuits,boisseau-thesis} on string diagrammatic electrical circuit theory.
\item[Section~\ref{sec:zx-extraction}:] The presentation as a layered theory in new. Likewise, presentation of the gflow-preserving graph rewrite rules on abstract MBQC+LC-graphs (not interpreted in the ZX-calculus) is new. This builds on many works on the ZX-calculus, especially on circuit extraction by Backens~et~al.~\cite{thereandback}.
\item[Section~\ref{sec:ccs}:] The string diagrammatic axiomatisations of reduction and labelled transition system semantics for CCS are new. The definitions of these two semantics are due to Milner~\cite{milner}. All the proofs and constructions are new, although we, of course, recover the standard results. The diagrams are inspired by Krivine~\cite{krivine-talk}.
\item[Section~\ref{sec:glucose}:] Everything here is new, albeit there is a significant overlap with Part~\ref{part:chemistry} of this thesis. The example is inspired by Krivine~\cite{krivine-talk}.
\item[Section~\ref{sec:prob-channels}:] The presentation of conditionals as parametric channels is new, as is the presentation as a layered theory. This builds on the work of Jacobs~\cite{jacobs-spr} on representing conditionals as boxes.
\item[Chapter~\ref{ch:indexed-monoidal}:] To the best knowledge of the author, the notions of categories and opfibrations with indexed monoids (Definitions~\ref{def:indexed-monoids} and~\ref{def:opfib-indexed-mon}), as well as the subsequent theory are new.
\item[Chapter~\ref{ch:profunctor-collages}:] To the best knowledge of the author, the definition of a deflation (Definition~\ref{def:deflation}) and the subsequent theory are new, including monoidal deflations (Definition~\ref{def:monoidal-deflation}).
\item[Chapter~\ref{ch:semantics}:] Everything in this chapter is new. Most importantly, this includes the soundness and completeness results for opfibrational theories (Corollary~\ref{cor:opfibrational-completeness}) and deflational theories (Corollary~\ref{cor:deflational-completeness}).
\end{itemize}

\chapter{Preliminaries}\label{ch:preliminaries}
We assume some knowledge of category theory. Namely, we assume the reader is familiar with the definition of a {\em monoidal category} and {\em monoidal functor} between monoidal categories~\cite{joyal-street91,maclane-cwm,selinger}. Some rudimentary notions from 2-category theory~\cite{leinster98,lack10,johnson-yau21} are required to follow the discussion of profunctor collages and deflations in Chapter~\ref{ch:profunctor-collages}. Familiarity with defining (strict) monoidal categories via monoidal signatures and theories is useful, although we give detailed definitions and discussion in Chapter~\ref{ch:monoidal-theories}. Similarly, familiarity with fibred category theory and profunctors is useful, although we give the minimal necessary definitions in the current chapter, as well as in Section~\ref{sec:profunctors}.

\section{Notation, conventions and terminology}

Here we fix some notation and terminology to be used throughout the thesis.

\subsection{Categories}

Arbitrary categories are denoted by capital Latin letters written in \texttt{mathcal}: $\cat X,\cat Y,\cat Z,\dots$.

Given a small category $\cat X$, its set of objects is denoted by $\Ob(\cat X)$, while its set of morphisms is denoted by $\Mor(\cat X)$. If no confusion is likely to arise, $f\in\cat X$ shall mean that $f\in\Mor(\cat X)$. Given a pair of objects $x,y\in\Ob(\cat X)$, the homset of morphisms from $x$ to $y$ is denoted by $\cat X(x,y)$. The identity morphism of $x$ is denoted by $\id_x$.

Named categories are given a name written in \texttt{mathbf}.

$\Set$ is the category of sets and functions.

$\Set_*$ is the category of pointed sets and point-preserving functions.

$\Mon$ is the category of monoids (in $\Set$).

$\Cat$ is the 2-category of (small) categories, functors and natural transformations.

$\MonCat$ is the category of monoidal categories and strong monoidal functors.

$\MonCat_{\mathsf{st}}$ is the category of strict monoidal categories and strict monoidal functors.

\subsection{Morphisms}

Domain and codomain are part of the morphism data. That is, morphisms $f:x\rightarrow y$ and $f:z\rightarrow w$ are, {\em a priori}, distinct and unrelated, despite the same symbol being used for both. In practice, usage of the same symbol will always indicate close relationship between the two morphisms. The same applies to any other entities that have a domain and codomain, such as terms.

In Part~\ref{part:layered} of the thesis, we use the diagrammatic order of composition. That is, given two composable entities $f:x\rightarrow y$ and $g:y\rightarrow z$ (morphisms in a category, functors, 1-cells, 2-cells etc.), their composite
$$x\xrightarrow{f}y\xrightarrow{g}z$$
is denoted by $f;g$. Most of the time we use the semicolon to emphasise the order of composition is diagrammatic, but even when we write $fg$ or $f\circ g$ the order is still the same. We choose this order so as to match the composition of diagrams, which are used heavily throughout the thesis.

\subsection{Functors}

A functor is usually denoted by a lowercase Latin letter $p:\cat Y\rightarrow\cat X$, or just by an arrow $\cat Y\rightarrow\cat X$ with no symbol.

Given a functor $p:\cat Y\rightarrow\cat X$, we say that an object $y\in\Ob(\cat Y)$ is {\em above} an object $x\in\Ob(\cat X)$ if $py=x$. Similarly, a morphism $g\in\cat Y$ is above a morphism $f\in\cat X$ if $pg=f$. Given a morphism $f\in\cat X$, we denote by $\cat Y_f$ the set of all morphisms above $f$. In the special case when $f$ is an identity on some object $x\in\Ob(\cat X)$, this defines the subcategory of $\cat Y$ known as the {\em fibre} of $x$: the objects are all objects of $\cat Y$ above $x$, the morphisms are all morphisms above $\id_x$. We denote the fibre of $x$ simply by $\cat Y_x$.

Given a function $p:\Ob(\cat Y)\rightarrow\Ob(\cat X)$, a {\em liftable pair} $(f:x\rightarrow y,b)$ consists of a morphism $f\in\cat X$ and an object $b\in\Ob(\cat Y)$ with $pb=y$. Similarly, an {\em opliftable pair} $(a,f:x\rightarrow y)$ consists of a morphism $f\in\cat X$ and an object $a\in\Ob(\cat Y)$ with $pa=x$.

\subsection{2-cells}

Any 2-cells (e.g.~natural transformations) are denoted by lowercase Greek letters $\alpha,\beta,\gamma,\dots$. Given 2-cells
\begin{center}
\scalebox{1}{\input{twocells-horizontal.tikz}},
\end{center}
their horizontal composite is written as $\alpha *\beta$.

Given a 2-category $\cat Y$, its collection of 2-cells is denoted by $\cat Y_2$.

\subsection{Function application}

Given an element of a set $x\in X$ and a function $f:X\rightarrow Y$, the image of $x$ in $Y$ under $f$ is denoted by $f(x)$. We often omit the brackets if this does not introduce any ambiguities, so $fx\coloneq f(x)$. Note that function application is written in the reverse order to composition:
$$(f;g)(x) = gfx = g(f(x)).$$
We choose this order of function application over being consistent with composition as the functions (or functors) are often more important than the objects they apply to, so it is worth putting an emphasis on them by writing them first.

\section{Fibrations and opfibrations}

We spell out the definition of an {\em opfibration}, and define a {\em fibration} as its dual. This is because we primarily focus on opfibrations as the semantics for opfibrational layered theories, although fibrations will also feature, especially when organising models and theories into categories. While (op)fibrations were originally introduced by Grothendieck~\cite{grothendieck2004}, we mostly follow Jacobs~\cite{jacobs-cltt} in our presentation. Other expository texts include Borceux~\cite{borceux94-fibred}, Johnstone~\cite{johnstone02-indexed} and Loregian \& Riehl~\cite{loregian-riehl20}.

Let us fix a functor $p:\cat Y\rightarrow\cat X$.
\begin{definition}[Opcartesian map]\label{def:opcartesian-map}
A morphism $F:a\rightarrow b$ in $\cat Y$ is called {\em $p$-opcartesian} if for any morphism $G:a\rightarrow c$ in $\cat Y$ and any morphism $h:pb\rightarrow pc$ in $\cat X$ such that $pF;h=pG$ (i.e.~the lower triangle below commutes), there is a unique morphism $H:b\rightarrow c$ above $h$ such that $F;H=G$ (i.e.~the upper triangle below commutes):
\begin{center}
\scalebox{1}{\input{opcartesian-definition.tikz}}.
\end{center}
\end{definition}

\begin{remark}[Weakly opcartesian map]\label{rem:weakly-opcartesian-map}
We obtain the notion of a {\em weakly $p$-opcartesian map} upon restricting the quantification to $h$ being the identity map in the above definition: for any morphism $G:a\rightarrow c$ {\em with $pc=pb$} such that $pF=pG$, there is a unique map $H$ above $\id_{pb}$ with $F;H=G$. Clearly, each opcartesian map is, in particular, weakly opcartesian. The converse does not hold in general.
\end{remark}

\begin{definition}[Opfibration]\label{def:opfibration}
A functor $p:\cat Y\rightarrow\cat X$ is an {\em opfibration} if for every opliftable pair $(a,f:x\rightarrow y)$, there is a $p$-opcartesian morphism $F:a\rightarrow b$ above $f$:
\begin{center}
\scalebox{1}{\input{opfibration-definition.tikz}}.
\end{center}
Such morphism $F$ is called an {\em opcartesian lifting} of $(a,f:x\rightarrow y)$. Note that we drop the prefix $p$ if the functor is clear from the context.
\end{definition}
In the context of (op)fibrations, we refer to the domain category $\cat Y$ as the {\em total} category, and to the codomain category $\cat X$ as the {\em base} category.

\begin{remark}[Preopfibration]\label{rem:preopfibration}
We obtain the notion of a {\em preopfibration} if we require $F$ in the above definition to be {\em weakly} $p$-opcartesian rather than $p$-opcartesian.
\end{remark}

By the defining universal property of opcartesian maps, opcartesian liftings are unique up to an isomorphism. An opfibration which comes with a choice of an opcartesian lifting for each opliftable pair is called {\em cloven}.

If $p:\cat Y\rightarrow\cat X$ is a cloven opfibration, then every morphism $f:x\rightarrow y$ in $\cat X$ induces a functor between the fibres $f^*:\cat Y_x\rightarrow\cat Y_y$ by mapping each $a\in\cat Y_x$ to the codomain of the chosen opcartesian lifting of $(a,f:x\rightarrow y)$, and each morphism $G\in\cat Y_x$ to the unique map above $\id_y$ induced by the universal property of the opcartesian lifting of the liftable pair $(\dom(G),f)$, where $\dom(G)$ is the domain of $G$. The functor $f^*$ is referred to as the {\em reindexing functor} induced by $f$. In fact, defining a reindexing functor for each morphism in the base category of an opfibration is equivalent to providing a choice of opcartesian liftings.
\begin{proposition}\label{prop:opfibration-to-pseudofunctor}
Any cloven opfibration $p:\cat Y\rightarrow\cat X$ induces a pseudofunctor $\cat X\rightarrow\Cat$ by sending each object to its fibre and each morphism $f$ to the functor $f^*$.
\end{proposition}
We say that an opfibration is {\em split} if the induced pseudofunctor is, in fact, a functor: i.e.~if for all $x\in\Ob(\cat X)$, the induced functor $\id_x^*$ is the identity functor on $\cat Y_x$ and for all composable maps we have $(f;g)^*=f^*;g^*$.

\begin{remark}
Note that the construction of the reindexing functor $f^*:\cat Y_x\rightarrow\cat Y_y$ already works for a preopfibration. However, the resulting assignment is not compositional (not even up to an isomorphism): the composition of two weakly opcartesian maps need not be weakly opcartesian. Thus, Proposition~\ref{prop:opfibration-to-pseudofunctor} fails for a preopfibration: its proof relies on the liftings being opcartesian (rather than merely weakly opcartesian). However, a preopfibration does induce a normal lax functor $\cat X\rightarrow\Prof$ into the 2-category of profunctors. We return to this in Section~\ref{sec:collages} (see specifically Proposition~\ref{prop:displayed-factors-opfibration}).
\end{remark}

A fibration is the dual notion to opfibration, hence all the discussion here also holds for fibrations -- with the directions appropriately reversed.
\begin{definition}[Fibration]\label{def:fibration}
Given a functor $p:\cat Y\rightarrow\cat X$, a morphism $F:a\rightarrow b$ in $\cat Y$ is {\em $p$-cartesian} if it is $p$-opcartesian for $p:\cat Y^{op}\rightarrow\cat X^{op}$. The functor $p:\cat Y\rightarrow\cat X$ is a fibration if $p:\cat Y^{op}\rightarrow\cat X^{op}$ is an opfibration.
\end{definition}

The intuition behind (op)fibrations is that they capture indexing via disjoint unions of indexed sets. We demonstrate this with the following example.
\begin{example}[Family fibration]\label{ex:family-fibration}
Define the category $\Fam(\Set)$ of {\em families of sets} as follows:
\begin{itemize}
\item the objects are pairs $(I,f)$ of a set $I$ and a function $f:I\rightarrow\Set$ assigning to each element $i\in I$ a set $f(i)$,
\item a morphism $(I,f)\rightarrow (J,g)$ is a pair $\left(u,\left\{u_i\right\}_{i\in I}\right)$ of a function $u:I\rightarrow J$ and a family of functions $u_i:f(i)\rightarrow g(u(i))$,
\item the identity on $(I,f)$ is given by $(\id_I,\id_{fi})$,
\item the composition of $(u,u_i) : (I,f)\rightarrow (J,g)$ and $(v,v_j) : (J,g)\rightarrow (K,h)$ is given by $(u;v, u_i;v_{ui})$.
\end{itemize}
The projection $\Fam(\Set)\rightarrow\Set$ to the first component is a fibration:
\begin{itemize}
\item the cartesian maps $(u,\id_{fi}) : (I,f)\rightarrow (J,g)$ are given by commutative triangles
\begin{center}
\scalebox{1}{\input{families-fibration-cartesian.tikz}}
\end{center}
and the family of identity maps on each set $f(i)=g(u(i))$,
\item given a family of sets $(I,f)$ and a function $u:J\rightarrow I$, the cartesian lifting is given by
$$\left(u,\id_{(u;f)(j)}\right) : (J, u;f)\rightarrow (I,f).$$
\end{itemize}
\end{example}
Note that in Example~\ref{ex:family-fibration}, we did not use the structure of $\Set$ in the second component as the target of the indexing, apart from the fact that it is a category. This example, therefore, generalises to an arbitrary category: the category of {\em families over $\cat C$} is denoted by $\Fam(\cat C)$, and the projection to the first component $\Fam(\cat C)\rightarrow\Set$ is still a fibration.

\subsection{Categories of (op)fibrations}

Fibrations and opfibrations organise into categories. We will mostly focus on the fixed base case, however, the general case will also appear occasionally.
\begin{definition}[Morphism of (op)fibrations]
Let $p$ and $q$ in the diagram below be (op)fibrations. A {\em morphism} $(H,K):p\rightarrow q$ is a pair of functors such that the square
\begin{center}
\scalebox{1}{\input{morphism-opfibrations.tikz}}
\end{center}
commutes, and $H$ sends $p$-(op)cartesian maps to $q$-(op)cartesian maps.
\end{definition}
We denote the resulting categories of fibrations and opfibrations by $\Fib$ and $\OpFib$. For a fixed base category $\cat X$, there are also categories of (op)fibrations into $\cat X$, whose morphisms are commutative triangles which preserve the (op)cartesian maps: this is a special case of a morphism defined above when the functor $K$ between the bases is the identity. We denote the fixed base categories by $\Fib(\cat X)$ and $\OpFib(\cat X)$.

Given a morphism of (op)fibrations $(H,K):p\rightarrow q$, observe that $H$ sends $p$-(op)cartesian liftings to $q$-(op)cartesian liftings of the image of the (op)liftable pair under $K$. In the case of cloven (and hence split) opfibrations, we additionally require that the chosen liftings are preserved. The subcategories of split (op)fibrations are denoted by adding a subscript $\mathsf{sp}$ in all four cases defined above: the objects are split (op)fibrations, while morphisms are morphisms of (op)fibrations such that chosen liftings are mapped to chosen liftings.

Two important properties of (op)fibrations are closure under composition and pullbacks. We record this in the following propositions.
\begin{proposition}
If $p:\cat Y\rightarrow\cat X$ and $q:\cat X\rightarrow\cat Z$ are (op)fibrations, then $p;q : \cat Y\rightarrow\cat Z$ is also an (op)fibration.
\end{proposition}
\begin{proposition}\label{prop:opfibrations-pullbacks}
Let $p$ in the diagram below be an (op)fibration. If the square
\begin{center}
\scalebox{1}{\input{opfibrations-pullbacks.tikz}}
\end{center}
is a pullback, then $q$ is also an (op)fibration.
\end{proposition}
The above propositions imply that the categories $\Fib(\cat X)$ and $\OpFib(\cat X)$ are cartesian, with the cartesian products given by pullbacks. The terminal object is given by the identity functor $\id_{\cat X}:\cat X\rightarrow\cat X$. We denote the resulting cartesian categories by $(\Fib(\cat X),\boxtimes,\id_{\cat X})$ and $(\OpFib(X),\boxtimes,\id_{\cat X})$.

\subsection{Retrofunctors}

A split opfibration may be thought of as a way of providing chosen liftings for a functor in a way that is (1) functorial (it strictly preserves identities and composition), (2) compatible with the functor in the sense that the lifting is above the morphism that induces it, and (3) minimal in the sense made precise by the requirement of being opcartesian (Definition~\ref{def:opcartesian-map}). If we drop the third requirement (minimality), we obtain the notion of a {\em lens}, which provides {\em some} functorial assignment of liftings compatible with the functor. Further, dropping the second condition (compatibility), we obtain the notion of a {\em retrofunctor}, which is simply a functorial assignment of lifts. Note that this does not require any functor to be present in order to be defined -- a function on object suffices. We remark that retrofunctors are more commonly known in the literature as {\em cofunctors}. However, we agree with Di~Meglio~\cite{dimeglio-mscthesis} that the term cofunctor is misleading, as the concept is not dual to a functor in any meaningful way. We therefore stick to {\em retrofunctor}. We follow Clarke~\cite{clarke-thesis} in our definition of a retrofunctor.

\begin{definition}[Retrofunctor]\label{def:retrofunctor}
A {\em retrofunctor} $(p,\varphi) : \cat Y\rightarrow\cat X$ consists of a function $p:\Ob(\cat Y)\rightarrow\Ob(\cat X)$ and a function $\varphi$ from the opliftable pairs to $\Mor(\cat Y)$ satisfying the following properties:
\begin{itemize}
\item the opliftable pair $(a,f:x\rightarrow y)$ is mapped to
$$\varphi(a,f) : a\rightarrow a^{\varphi},$$
where the codomain $a^{\varphi}$ is above $y$, i.e.~$p\left(a^{\varphi}\right)=y$,
\item $\varphi(a,\id_{pa})=\id_a$,
\item $\varphi(a,f;g)=\varphi(a,f);\varphi\left(a^{\varphi},g\right)$.
\end{itemize}
\end{definition}
We depict the action of a retrofunctor as follows:
\begin{center}
\scalebox{1}{\input{retrofunctor-definition.tikz}}.
\end{center}

\section{Indexed categories}\label{sec:indexed-categories}

We define indexed and opindexed categories, and state the well-known equivalence between (op)fibrations and (op)indexed categories.
\begin{definition}[(Op)indexed category]
An {\em $\cat X$-opindexed category} is a pseudofunctor $I:\cat X\rightarrow\Cat$. Dually, an {\em $\cat X$-indexed category} is a pseudofunctor $I:\cat X^{op}\rightarrow\Cat$.
\end{definition}
We write the natural isomorphisms that witness the pseudofunctor structure as $c_{f,g}:I(f;g)\xrightarrow{\sim}I(f); I(g)$ and $u_x:I(\id_x)\xrightarrow{\sim}\id_{Ix}$.

We say that an (op)indexed category is {\em strict} when the pseudofunctor involved is an actual functor.

Given a small category $\cat X$, we denote the categories of $\cat X$-indexed and $\cat X$-opindexed categories by $\ICat(\cat X)$ and $\OpICat(\cat X)$, whose morphisms are given by pseudonatural transformations. The strict versions are denoted by adding the subscript $\mathsf{st}$, where the morphisms are ordinary natural transformations.

\begin{definition}[Grothendieck construction]
Given a pseudofunctor $I:\cat X\rightarrow\Cat$, the {\em Grothendieck construction} is the category $\Gr(\cat X)$ defined as follows:
\begin{itemize}
\item objects are pairs $(x,a)$, where $x\in\Ob(X)$ and $a\in\Ob(Iy)$,
\item morphisms $(f,F):(x,a)\rightarrow (y,b)$ are pairs of a morphism $f:x\rightarrow y$ in $\cat X$ and a morphism $F:(If)(a)\rightarrow b$ in $I(y)$,
\item the identity on $(x,a)$ is given by $(\id_x,(u_x)_a)$,
\item the composite of $(f,F):(x,a)\rightarrow (y,b)$ and $(g,G):(y,b)\rightarrow (z,c)$ is given by
$$\left(f;g, (c_{f,g})_a ; (Ig)(F); G\right).$$
\end{itemize}
\end{definition}

The following equivalence is well-known. We give a proof sketch, and suggest the reader who has not seen the details before to complete the proof, as it is instructive, and all of the subsequent developments will build on this equivalence.
\begin{theorem}\label{thm:opfibrations-opindexed-categories-equivalence}
Let $\cat X$ be a small category. There are equivalences of categories
\begin{align*}
\Fib(\cat X) &\simeq \ICat(\cat X) \\
\OpFib(\cat X) &\simeq \OpICat(\cat X).
\end{align*}
Moreover, the equivalences restrict to split (op)fibrations and strict (op)indexed categories:
\begin{align*}
\Fib_{\mathsf{sp}}(\cat X) &\simeq \ICat_{\mathsf{st}}(\cat X) \\
\OpFib_{\mathsf{sp}}(\cat X) &\simeq \OpICat_{\mathsf{st}}(\cat X).
\end{align*}
\end{theorem}
\begin{proof}[Proof sketch]
The left-to-right functor sends each (op)fibration to the pseudofunctor defined in Proposition~\ref{prop:opfibration-to-pseudofunctor}. The right-to-left functor sends each $\cat X$-(op)indexed category to the projection $\Gr(\cat X)\rightarrow\cat X$ to the first component from the Grothendieck construction.
\end{proof}
In fact, the above equivalences extend to 2-equivalences of 2-categories by adding the appropriate 2-cells to all the categories: we will, however, focus on the 1-categorical case here.

\subsection{(Op)fibrations with pseudomonoids and indexed monoidal categories}\label{subsec:opfibrations-indexed-monoidal}

Moeller and Vasilakopoulou~\cite{monoidal-gro} have extended the equivalence between (op)fibrations and (op)indexed monoidal categories in Theorem~\ref{thm:opfibrations-opindexed-categories-equivalence} to include monoidal structure. In fact, they have extended it in two ways: in the first case, each fibre is equipped with monoidal structure, which is preserved by the morphisms between the fibres; in the second case, there is a monoidal structure on the total and the base categories, which is preserved by the (op)fibration. Following Shulman~\cite{shulman2008}, we refer to the first (`fiberwise') case as {\em internal} monoidal structure, while to the second (`global') case as {\em external} monoidal structure. The two cases are, in general, distinct: we refer the reader to~\cite{monoidal-gro} for the details. Here we briefly sketch the situation of the first case, as our theory will build on a special case thereof.

On the side of (op)fibrations, the monoidal structure is captured by a pseudomonoid in the cartesian categories $(\Fib(\cat X),\boxtimes,\id_{\cat X})$ and $(\OpFib(X),\boxtimes,\id_{\cat X})$. A pseudomonoid in one of these categories is given by a triple $(\cat Y,\otimes,\one)$, where $\cat Y\rightarrow\cat X$ is an object (i.e.~an (op)fibration), while $\otimes:\cat Y\boxtimes\cat Y\rightarrow\cat Y$ and $\one :\cat X\rightarrow\cat Y$ are morphisms, i.e.~the diagrams
\begin{center}
\scalebox{1}{\input{gr-pseudomonoid.tikz}}
\end{center}
commute, and each functor preserves (op)cartesian maps, such that the usual unitality and associativity equations hold up to natural isomorphisms. A morphism of pseudomonoids is given by a morphism of (op)fibrations which commutes with the multiplication and the unit, again up to a natural isomorphism. We denote the resulting categories of pseudomonoids by $\PsMonFib(\cat X)$ and $\PsMonOpFib(\cat X)$. If the pseudomonoids are required to be actual monoids strictly preserved by the morphisms, the resulting categories of monoids are denoted by $\MonFib(\cat X)$ and $\MonOpFib(\cat X)$.

On the side of (op)indexed categories, the monoidal structure is captured by requiring the image of the pseudofunctor into $\Cat$ to be, in fact, contained in $\MonCat$. We make the following definitions.
\begin{definition}[(Op)indexed monoidal category]\label{def:opindexed-monoidal-category}
An {\em $\cat X$-opindexed monoidal category} is a pseudofunctor $I:\cat X\rightarrow\MonCat$. Dually, an {\em $\cat X$-indexed monoidal category} is a pseudofunctor $I:\cat X^{op}\rightarrow\MonCat$.
\end{definition}
We say that an (op)indexed monoidal category is {\em strict} if the pseudofunctor is an actual functor and its image is contained in $\MonCat_{\mathsf{st}}$, that is, each monoidal category and functor are strict.

Similarly to the plain indexed case, we denote the resulting categories by $\IMonCat(\cat X)$ and $\OpIMonCat(\cat X)$, and add the subscript $\mathsf{st}$ in the strict case.

\begin{warning}
The terms ``indexed monoidal category'' and ``monoidal indexed category'' are not used consistently in the literature. Some authors use it to refer to the external case when the indexing category $\cat X$ itself has a monoidal structure, and the pseudofunctor is required to preserve it in some sense. The terminology of Definition~\ref{def:opindexed-monoidal-category} is consistent with Hofstra \& De~Marchi~\cite{hofstra-demarchi2006} and Ponto \& Shulman~\cite{ponto-shulman12}.
\end{warning}

The following extends Theorem~\ref{thm:opfibrations-opindexed-categories-equivalence} to (op)fibrations with pseudomonoids and (op)indexed monoidal categories.
\begin{theorem}[Moeller and Vasilakopoulou~\cite{monoidal-gro}, Theorem 3.14]\label{thm:moeller-and-vasilakopoulou}
Let $\cat X$ be a small category. There are equivalences of categories
\begin{align*}
\PsMonFib(\cat X) &\simeq \IMonCat(\cat X) \\
\PsMonOpFib(\cat X) &\simeq \OpIMonCat(\cat X).
\end{align*}
\end{theorem}
As for Theorem~\ref{thm:opfibrations-opindexed-categories-equivalence}, this is, in fact, part of a 2-equivalence between 2-categories upon adding the appropriate 2-cells to all the categories.

While it would be interesting to study graphical languages for the general (non-strict) case as presented above, this would take us too far from the currently existing string diagrams for monoidal categories, almost all of which are only understood well in the strict case. Here, we therefore focus on the strict case stated below, and leave studying the non-strict case using the methods developed here for future work.
\begin{theorem}\label{thm:monoids-opfibrations-strict-imoncat-equivalence}
Let $\cat X$ be a small category. There are equivalences of categories
\begin{align*}
\MonFib_{\mathsf{sp}}(\cat X) &\simeq \IMonCat_{\mathsf{st}}(\cat X) \\
\MonOpFib_{\mathsf{sp}}(\cat X) &\simeq \OpIMonCat_{\mathsf{st}}(\cat X).
\end{align*}
\end{theorem}
\begin{proof}
Let $(\cat Y,\otimes,\one)$ be a monoid on a split opfibration $\cat Y\rightarrow\cat X$. For each $x\in\Ob(\cat X)$, this defines a functor
\begin{align*}
\otimes_x : \cat Y_x\times\cat Y_x &\rightarrow\cat Y_x \\
(a,b) &\mapsto a\otimes b,
\end{align*}
while the unit is given by $\one(x)\in\Ob(\cat Y_x)$. Associativity and unitality follow from those of $\otimes$ and $\one$.

Conversely, given a functor $\cat X\rightarrow\MonCat_{\mathsf{st}}$ (i.e.~a strict opindexed monoidal category), we define the monoid on the Grothendieck construction as follows:
\begin{center}
\scalebox{1}{\input{gr-monoid-induces-by-opindexed.tikz}},
\end{center}
where $(\otimes_x,I_x)$ denotes the monoidal structure in the category indexed by $x\in\Ob(\cat X)$. Associativity and unitality now follow from the corresponding properties of strict monoidal categories. Moreover, the triangles above commute by definition, and opcartesian maps (i.e.~pairs whose second component is the identity) are preserved by strictness of the monoidal structure.
\end{proof}

\chapter{Monoidal theories}\label{ch:monoidal-theories}
Here we define monoidal signatures, theories and models. We recover the standard result that monoidal theories are sound and complete with respect to strict monoidal categories (Corollary~\ref{cor:monoidal-free-forgetful}). Everything in this chapter is well-known, perhaps except for a couple of examples in Section~\ref{sec:examples-monoidal-theories}. However, we will mimic the developments here when constructing layered theories in Chapter~\ref{ch:layered-theories}, so we advice the reader to at least skim through. We give several examples of monoidal theories in Section~\ref{sec:examples-monoidal-theories}, many of which will play a role in subsequent developments.

\section{Monoidal signatures}

A monoidal signature consists of a set of generating objects called {\em colours}, and for each pair of lists of colours, a set of generating morphisms called {\em monoidal generators}.

Let us denote by $(-)^*:\Set\rightarrow\Mon$ the free monoid functor.
\begin{definition}[Monoidal signature]
A {\em monoidal signature} is a tuple $(C,\Sigma)$, where $C$ is a set, and $\Sigma : C^*\times C^*\rightarrow\Set$ is a function.
\end{definition}
Given a monoidal signature $(C,\Sigma)$, we call $C$ the set of {\em colours}, and for $a,b\in C^*$, the elements in $\Sigma(a,b)$ the {\em monoidal generators} with arity $a$ and coarity $b$.

\begin{definition}[Morphism of monoidal signatures]
A {\em morphism} between two monoidal signatures
$$\left(f,f_{a,b}\right) : (C,\Sigma)\rightarrow (D,\Gamma)$$
is given by a function $f:C\rightarrow D$, and for each pair $(a,b)\in C^*\times C^*$, a function
$$f_{a,b}:\Sigma(a,b)\rightarrow\Gamma\left(f^*a,f^*b\right).$$
\end{definition}
We denote the resulting category of monoidal signatures by $\MSgn$. We often denote a morphism in $\MSgn$ simply by $f:(C,\Sigma)\rightarrow (D,\Gamma)$ ranging over the whole family of functions: $f:C\rightarrow D$ (without subscripts) is the function on the colours, and $f_{a,b}:\Sigma(a,b)\rightarrow\Gamma(f^*a,f^*b)$ are the functions on the monoidal generators.

Note that there are two ways of ``forgetting'' the structure of $\MSgn$. There is the projection to the first component in the tuple defining the monoidal signature $\MSgn\rightarrow\Set$, and there is the functor into the category of families of sets\footnote{See Example~\ref{ex:family-fibration} for the definition of the category of families of sets.} $\MSgn\rightarrow\Fam(\Set)$ defined by $(C,\Sigma)\mapsto (C^*\times C^*,\Sigma)$, which ``forgets'' that morphisms need to preserve (co)arities. These functors arise in the following situation:
\begin{proposition}\label{prop:msgn-pullback}
The square below is a pullback, where the right vertical map is the family fibration:
\begin{center}
\input{msgn-pullback.tikz}.
\end{center}
\end{proposition}
\begin{proof}
A functor into the category of families of sets such that the indexing sets and functions are of the form $C^*\times C^*$ and $f^*\times f^*$ is the same as a functor into the category of monoidal signatures.
\end{proof}
\begin{corollary}\label{cor:forgetful-msgn-fibration}
The forgetful functor $\MSgn\rightarrow\Set$ is a fibration.
\end{corollary}
\begin{proof}
We have seen in Example~\ref{ex:family-fibration} that the projection $\Fam(\Set)\rightarrow\Set$ is a fibration. The claim then follows by stability of fibrations under pullback (Proposition~\ref{prop:opfibrations-pullbacks}).
\end{proof}
We note that, due to uniqueness of pullback, we could take Proposition~\ref{prop:msgn-pullback} as the definition of the category of monoidal signatures.

\section{Terms and theories}\label{sec:monoidal-theories}

In order to equip monoidal signatures with non-trivial equations, we need a way to build more complicated expressions from the generators. The following defines the {\em terms} of a monoidal signature, which are used to define monoidal theories as well as free models. Upon quotienting by the structural identities (Definition~\ref{def:str-id}), the terms become what is known as {\em string diagrams}.

\begin{definition}[Terms of a monoidal signature]
Given a monoidal signature $(C,\Sigma)$, the set of {\em sorts} is given by $C^*\times C^*$. The {\em terms} are generated by the recursive sorting procedure below:
\begin{center}
  \centering\small\noindent
  \begin{prooftree}
    \AxiomC{$\sigma\in\Sigma(a,b)$}
    \RightLabel{\;\;}
    \UnaryInfC{$\scalebox{.9}{\begin{tikzpicture}
	\begin{pgfonlayer}{nodelayer}
		\node [style=none] (4) at (-2, 0) {};
		\node [style=none] (5) at (2, 0) {};
		\node [style=bdytextbox] (6) at (0, 0) {$\sigma$};
	\end{pgfonlayer}
	\begin{pgfonlayer}{edgelayer}
		\draw [style=edge] (4.center) to (6);
		\draw [style=edge] (6) to (5.center);
	\end{pgfonlayer}
\end{tikzpicture}
} : (a,b)$}
    \DisplayProof
    \AxiomC{$a\in C$}
    \RightLabel{\;\;}
    \UnaryInfC{$\scalebox{1}{\input{iddiag-dashed.tikz}} : (a,a)$}
    \DisplayProof
    \AxiomC{}
    \UnaryInfC{$\scalebox{.6}{\input{emptydiag.tikz}} : (\varepsilon,\varepsilon)$}
  \end{prooftree}
  \begin{prooftree}
    \AxiomC{$\scalebox{.9}{\begin{tikzpicture}
	\begin{pgfonlayer}{nodelayer}
		\node [style=none] (0) at (-2, 0) {};
		\node [style=none] (1) at (2, 0) {};
		\node [style=bdytextbox] (2) at (0, 0) {$t$};
	\end{pgfonlayer}
	\begin{pgfonlayer}{edgelayer}
		\draw [style=edge] (0.center) to (2);
		\draw [style=edge] (2) to (1.center);
	\end{pgfonlayer}
\end{tikzpicture}
} : (a,b)$}
    \AxiomC{$\scalebox{.9}{\begin{tikzpicture}
	\begin{pgfonlayer}{nodelayer}
		\node [style=none] (0) at (-2, 0) {};
		\node [style=none] (1) at (2, 0) {};
		\node [style=bdytextbox] (2) at (0, 0) {$s$};
	\end{pgfonlayer}
	\begin{pgfonlayer}{edgelayer}
		\draw [style=edge] (0.center) to (2);
		\draw [style=edge] (2) to (1.center);
	\end{pgfonlayer}
\end{tikzpicture}
} : (b,c)$}
    \RightLabel{\;\;}
    \BinaryInfC{$\scalebox{.9}{\input{t-s-compose.tikz}} : (a,c)$}
    \DisplayProof
    \AxiomC{$\scalebox{.9}{\begin{tikzpicture}
	\begin{pgfonlayer}{nodelayer}
		\node [style=none] (0) at (-2, 0) {};
		\node [style=none] (1) at (2, 0) {};
		\node [style=bdytextbox] (2) at (0, 0) {$t$};
	\end{pgfonlayer}
	\begin{pgfonlayer}{edgelayer}
		\draw [style=edge] (0.center) to (2);
		\draw [style=edge] (2) to (1.center);
	\end{pgfonlayer}
\end{tikzpicture}
} : (a,b)$}
    \AxiomC{$\scalebox{.9}{\begin{tikzpicture}
	\begin{pgfonlayer}{nodelayer}
		\node [style=none] (0) at (-2, 0) {};
		\node [style=none] (1) at (2, 0) {};
		\node [style=bdytextbox] (2) at (0, 0) {$s$};
	\end{pgfonlayer}
	\begin{pgfonlayer}{edgelayer}
		\draw [style=edge] (0.center) to (2);
		\draw [style=edge] (2) to (1.center);
	\end{pgfonlayer}
\end{tikzpicture}
} : (c,d)$}
    \RightLabel{.}
    \BinaryInfC{$\scalebox{.9}{\input{t-s-tensor.tikz}} : (ac,bd)$}
  \end{prooftree}
  \end{center}
When using the linear notation, we denote the terms generated by the above rules by $\sigma$, $\id_a$, $\id_{\varepsilon}$, $(t;s)$ and $(t\otimes s)$, respectively. We denote the set of terms by $\Term_{C,\Sigma}$, and by $S:\Term_{C,\Sigma}\rightarrow C^*\times C^*$ the ``underlying sort function" mapping $t:(a,b)\mapsto (a,b)$.
\end{definition}

A pair of terms $(t,s)\in\Term_{C,\Sigma}\times\Term_{C,\Sigma}$ is {\em parallel} if $S(t)=S(s)$. Let us denote by $P_{C,\Sigma}$ the set of parallel terms for the monoidal signature $(C,\Sigma)$.

Given a morphism $f:(C,\Sigma)\rightarrow (D,\Gamma)$ of monoidal signatures, it immediately extends to a function on terms, also denoted by $f:\Term_{C,\Sigma}\rightarrow\Term_{D,\Gamma}$, by recursively defining:\label{p:extend-morphism-signatures-terms}
\begin{align*}
\sigma : (a,b) &\mapsto f_{a,b}(\sigma) : (f^*a,f^*b), \\
\id_a : (a,a) &\mapsto \id_{fa} : (fa,fa), \\
\id_{\varepsilon} : (\varepsilon,\varepsilon) &\mapsto \id_{\varepsilon} : (\varepsilon,\varepsilon), \\
(t;s) : (a,c) &\mapsto \left(f(t);f(s)\right) : (f^*a,f^*c), \\
(t\otimes s) : (ac,bd) &\mapsto \left(f(t)\otimes f(s)\right) : (f^*(ac),f^*(bd)).
\end{align*}
Observe that $f:\Term_{C,\Sigma}\rightarrow\Term_{D,\Gamma}$ preserves parallel terms: if $(t,s)\in P_{C,\Sigma}$, then $(ft,fs)\in P_{D,\Gamma}$.

\begin{definition}[Monoidal theory]\label{def:monoidal-theory}
A {\em monoidal theory} $\mathcal T$ is a triple $(C,\Sigma,E)$, where $(C,\Sigma)$ is a monoidal signature, and $E\sse P_{C,\Sigma}$ is a set of parallel terms. We refer to $E$ as the {\em equations} of $\mathcal T$.
\end{definition}

A morphism of monoidal theories
$$f:(C,\Sigma,E)\rightarrow (D,\Gamma,F)$$
is given by a morphism of monoidal signatures $f:(C,\Sigma)\rightarrow (D,\Gamma)$ such that for all $(s,t)\in E$, we have $(fs,ft)\in F$. We denote the resulting category of monoidal theories by $\MTh$.

In the next section, we will define {\em free models} by quotienting the set of terms by the following equations, which one may think of as the ``minimal'' monoidal theory.
\begin{definition}[Structural identities]\label{def:str-id}
Given a monoidal signature $(C,\Sigma)$, the {\em structural identities} are the following equations, where $s$, $s_i$ and $t_i$ range over the terms of the appropriate type:
\begin{center}
\scalebox{1}{\input{structural-identities.tikz}}.
\end{center}
We denote the structural identities by $S$.
\end{definition}

\begin{definition}[Term congruence]
Given a monoidal theory $(C,\Sigma,E)$, the {\em term congruence} $\Eeq$ is the smallest equivalence relation on $\Term_{C,\Sigma}$ generated by $E\cup S$, which is a congruence with respect to $;$ and $\otimes$, i.e.~if $t_1\Eeq t_2$ and $s_1\Eeq s_2$, then $t_1;s_1\Eeq t_2;s_2$
and $t_1\otimes s_1\Eeq t_2\otimes s_2$, whenever the composition is defined.
\end{definition}

Note that the structural equations justify dropping all the dashed boxes when drawing the terms up to a term congruence: each diagram uniquely determines an equivalence class of terms. Such equivalence classes of terms are known as {\em string diagrams}. We utilise this in the following convention.
\begin{remark}\label{rem:id-terms}
If $w\in C^*$ with $w=a_1\cdots a_n$, we often abbreviate the identity term on $w$ to
\begin{center}
\input{const-id-terms.tikz}.
\end{center}
The structural identities guarantee that the above expression unambiguously identifies an equivalence class of terms, no matter how the diagram is parsed.
\end{remark}

\section{Models of monoidal theories}\label{sec:models-monoidal-theories}

A model interprets a monoidal signature (or theory) in a strict monoidal category.

\begin{definition}\label{def:model-monoidal-signature}
A {\em model} of a monoidal signature $(C,\Sigma,E)$ is a strict monoidal category $(\cat C,\otimes,I)$ together with a function $i:C\rightarrow\Ob(\cat C)$, and for each $a,b\in C^*$, a function $i_{a,b}:\Sigma(a,b)\rightarrow\cat C(i^*a,i^*b)$.
\end{definition}
We often denote a model of a monoidal signature $(C,\Sigma)$ simply by $(\cat C,i)$, where $\cat C$ is a strict monoidal category, and $i$ a family of functions as in the definition -- referred to as the {\em interpretation functions}.

\begin{definition}
The category of {\em monoidal models} $\MMod$ has as objects quadruples $(C,\Sigma,\cat C,i)$, where $(C,\Sigma)$ is a monoidal signature and $(\cat C,i)$ its model. A morphism
$$(f,F):(C,\Sigma,\cat C,i)\rightarrow (D,\Gamma,\cat D,j)$$
is given by a morphism of monoidal signatures $f:(C,\Sigma)\rightarrow (D,\Gamma)$ and a strict monoidal functor $F:\cat C\rightarrow\cat D$ such that the diagram on the left commutes (where $F_0$ is the action of $F$ on objects), and the diagram on the right commutes for all $a,b\in C^*$,
\ctikzfig{mor-mon-models}
where we denote by $F$ the appropriate restriction of the functor to the hom-set. Note that the restriction is indeed well-defined, as commutativity of the left diagram together with the fact that $F$ is strict monoidal imply that $i^* ; F_0 = f^* ; j^*$.
\end{definition}

Given a model $(\cat C,i)$ of a monoidal signature $(C,\Sigma)$, it extends to a function from terms to the morphisms of $\cat C$, also denoted by $i:\Term_{C,\Sigma}\rightarrow\Mor(\cat C)$, as follows:\label{p:extend-model-terms}
\begin{align*}
\sigma : (a,b) &\mapsto i_{a,b}(\sigma) : (i^*a,i^*b), \\
\id_a : (a,a) &\mapsto \id_{ia} : (ia,ia), \\
\id_{\varepsilon} : (\varepsilon,\varepsilon) &\mapsto \id_I : (I,I), \\
(t;s) : (a,c) &\mapsto i(t);i(s) : (i^*a,i^*c), \\
(t\otimes s) : (ac,bd) &\mapsto i(t)\otimes i(s) : (i^*(ac),i^*(bd)),
\end{align*}
where $I$ is the unit of $\cat C$, while $;$ and $\otimes$ on the right-hand side are the composition and the monoidal product of $\cat C$.

\begin{definition}\label{def:model-monoidal-theory}
A {\em model} of a monoidal theory $(C,\Sigma,E)$ is a model $(\cat C,i)$ of the monoidal signature $(C,\Sigma)$ such that for all $(s,t)\in E$, we have $i(s)=i(t)$.
\end{definition}

We denote by $\MThMod$ the category of {\em models of monoidal theories}, whose objects are pairs of a monoidal theory and its model, and whose morphisms are pairs $(f,F)$ such that $f$ is a morphism of monoidal theories (hence, in particular, a morphism of monoidal signatures) and $(f,F)$ is a morphism in $\MMod$.

We note that every model of a monoidal signature $(C,\Sigma)$ can be viewed as a model of the monoidal theory $(C,\Sigma,\eset)$, so that $\MMod$ can be identified with the full subcategory of $\MThMod$ of theories with no equations.

We summarise the relationship between monoidal signatures, theories and their models in the following proposition.
\begin{proposition}\label{prop:monoidal-signatures-theories-models}
The vertical forgetful functors in the diagram below are fibrations, while the horizontal forgetful functors form a morphism of fibrations:
\begin{center}
\input{mthmod-pullback.tikz}.
\end{center}
Moreover, the objects in the fibre $\MMod(C,\Sigma)$ are precisely the models of the monoidal signature $(C,\Sigma)$, and the objects in the fibre $\MThMod(\mathcal T)$ are precisely the models of the monoidal theory $\mathcal T$.
\end{proposition}
\begin{proof}
The cartesian maps are those pairs whose monoidal functor part is (naturally isomorphic to) the identity functor.

Given an object $(C,\Sigma,E,\cat C,i)\in\MThMod_0$ and a morphism of monoidal theories $f:(D,\Gamma,F)\rightarrow (C,\Sigma,E)$, the cartesian lifting is given by
$$(f,\id_{\cat C}) : (D,\Gamma,F,\cat C,f;i)\rightarrow (C,\Sigma,E,\cat C,i),$$
where the interpretation functions of the domain model are given by $f;i : D\rightarrow\cat C$ and
$$f_{a,b};i_{f^*a,f^*b} : \Gamma(a,b)\rightarrow\cat C(i^*f^*a,i^*f^*b).$$

For $\MMod\rightarrow\MSgn$, the situation is the same, except that all the theories are empty. It is then clear that the horizontal forgetful functors make the diagram commute, and that the top functor preserves cartesian liftings.
\end{proof}
Note that, while $\MTh\rightarrow\MSgn$ is also a fibration, the functor $\MThMod\rightarrow\MMod$ is {\em not} itself a fibration.

\subsection{Free models}

The usefulness of string diagrams lies in the fact that they are the {\em free models} of monoidal theories: an equation is derivable in the string diagrams of a monoidal theory if and only if it holds in all models of that theory.

\begin{definition}[Term model]\label{def:term-model}
Given a monoidal theory $\mathcal T = (C,\Sigma,E)$, the {\em term model} (or the {\em free model}) $F(\mathcal T)$ is the strict monoidal category defined by the following:
\begin{itemize}
\item the objects are $C^*$,
\item the morphisms $a\rightarrow b$ are the equivalence classes of terms with sort $(a,b)$ under the term congruence $\Eeq$,
\item the composition of $[s]$ and $[t]$ is given by $[s;t]$,
\item the monoidal unit is the empty word $\varepsilon\in C^*$,
\item the monoidal product is given by concatenation on objects and by $\otimes$ on morphisms.
\end{itemize}
We note that the operations $;$ and $\otimes$ above are well-defined since $\Eeq$ is a congruence. The interpretation function $i:C\rightarrow C^*$ sends the colour $a$ to the one-element word $a$, and
$$i_{a,b}:\Sigma(a,b)\rightarrow F(\mathcal T)(a,b)$$
sends each generator to its equivalence class.
\end{definition}
For a monoidal signature $(C,\Sigma)$, we denote by $F(C,\Sigma)$ the free model on the theory $(C,\Sigma,\eset)$.

\begin{proposition}
The term model construction extends to functors $F:\MTh\rightarrow\MThMod$ and $F:\MSgn\rightarrow\MMod$, which are moreover the left adjoints to the respective forgetful functors.
\end{proposition}
\begin{proof}
The action of the functor on morphisms is given by the recursive extension of the morphism of monoidal signatures to terms, defined on page~\pageref{p:extend-morphism-signatures-terms}.

The unit $\eta_{\mathcal T}:\mathcal T\rightarrow UF(\mathcal T) =\mathcal T$ is given by the identity $\id_{\mathcal T}$. The counit $\varepsilon_{\mathcal T,\mathcal M} : (\mathcal T, F(\mathcal T)) = FU(\mathcal T,\mathcal M)\rightarrow (\mathcal T,\mathcal M)$, where we denote $\mathcal M\coloneq (\mathcal C,i)$, is given by $(\id_{\mathcal T},i)$, where $i:F(\mathcal T)\rightarrow\mathcal C$ is given by $i^*$ on objects, and on morphisms by the recursive extension of $i$ to terms defined on page~\pageref{p:extend-model-terms}.
\end{proof}
\begin{corollary}\label{cor:monoidal-free-forgetful}
For every monoidal theory $\mathcal T$ and monoidal signature $(C,\Sigma)$, there are the following equivalences with coslice categories on the right:
\begin{align*}
\MThMod(\mathcal T) &\simeq \quot{F(\mathcal T)}{\MonCat_{\mathsf{st}}}, \\
\MMod(C,\Sigma) &\simeq \quot{F(C,\Sigma)}{\MonCat_{\mathsf{st}}}.
\end{align*}
\end{corollary}

We conclude our discussion of models by pointing out that the notion of a model of a monoidal theory can be extended to non-strict monoidal categories upon interpreting the term congruence as isomorphism rather than equality. However, obtaining a free model via string diagrams is more involved in this case, as one needs to explicitly keep track of the isomorphisms capturing non-strict associativity and unitality. We refer the reader to Wilson, Ghica and Zanasi~\cite{non-strict-monoidal,non-strict-monoidal-journal} for the details.

\section{Examples of monoidal theories}\label{sec:examples-monoidal-theories}

We give examples of monoidal theories, with not much comment. Most of these are very well-known and commonly used. The exceptions to this are the theories in Definitions~\ref{def:thy-11-nat-monoids} and~\ref{def:thy-indexed-monoids}, which are included as they will play a prominent role in the theory of layered monoidal theories.

Monoidal theories with more involved generators and equations appear in Chapter~\ref{ch:layered-examples}, where we discuss several domain specific theories. An example of obtaining a monoidal theory from an existing one appears in Appandix~\ref{ch:para-copara}, where we present a construction that allows morphisms to be {\em parameterised}.

\begin{definition}[Symmetric monoidal theory]\label{def:symm-mon-thy}
A {\em symmetric monoidal theory} is a monoidal theory $(C,\Sigma,\mathcal S)$ such that
\begin{itemize}
\item for all $a,b\in C$, the set $\Sigma(ab,ba)$ contains the special generator, denoted by $\scalebox{.4}{\input{symmetry.tikz}}$, called the {\em symmetry},
\item we extend the symmetry by constructing for all $v,w\in C^*$ the term $\sw^v_w:(vw,wv)$ by the following recursion, where $a,b\in C$ and $v,w\in C^*$:
\begin{center}
\scalebox{1}{\input{swap.tikz}},
\end{center}
\item for $v,w\in C^*$, we denote $\sw^v_w$ by the same symbol as the symmetry, making the above recursive definition more intuitive:
\begin{center}
\scalebox{1}{\input{swap-redefine.tikz}},
\end{center}
\item the set $\mathcal S$ contains the following equations, where $s$ ranges over all terms of the appropriate type:
\begin{center}
\scalebox{1}{\input{symmetry-eqns.tikz}}.
\end{center}
\end{itemize}
\end{definition}
We remark that most of the concrete examples of monoidal theories we will see in this thesis will be symmetric.

\begin{definition}[Theory of monoids]\label{def:thy-mon}
The {\em theory of monoids} is the monoidal theory $(\{\bullet\},\Sigma,\mathcal M)$ with the following generators and equations:
\begin{center}
\scalebox{.7}{\input{monoid.tikz}}.
\end{center}
\end{definition}

\begin{definition}[Theory of comonoids]\label{def:thy-comon}
The {\em theory of comonoids} is the monoidal theory $(\{\bullet\},\Sigma,\mathcal C)$ with the following generators and equations:
\begin{center}
\scalebox{.7}{\input{comonoid.tikz}}.
\end{center}
\end{definition}

\begin{definition}[Theory with uniform comonoids]\label{def:thy-univ-comonoids}
We say that a symmetric monoidal theory $(C,\Sigma,\mathcal U)$ has {\em uniform comonoids} if
\begin{itemize}
\item every $a\in C$ has a comonoid structure (i.e.~the sets $\Sigma(a,aa)$ and $\Sigma(a,\varepsilon)$ contain the comonoid generators and $\mathcal U$ contains the comonoid equations from Definition~\ref{def:thy-comon}),
\item upon extending the comonoid structure to all $w\in C^*$ by the following recursion:
\begin{center}
\scalebox{1}{\input{comon-extend.tikz}},
\end{center}
the set $\mathcal U$ contains the following equations, where $s$ ranges over all terms with appropriate type:
\begin{center}
\scalebox{1}{\input{comon-nat-eqns.tikz}}.
\end{center}
\end{itemize}
\end{definition}

\begin{proposition}\label{prop:universal-commutative}
Let $(C,\Sigma,\mathcal U)$ be a symmetric monoidal theory with uniform comonoids. Then for every $w\in C^*$, the construction $(w,d_w,e_w)$ is a cocommutative comonoid, i.e.~the following equations hold:
\begin{center}
\scalebox{0.8}{\input{univ-cocomm-comon.tikz}}.
\end{center}
\end{proposition}
\begin{proof}
Commutativity on the colours is a version of the {\em Eckmann-Hilton argument}. By naturality of the comonoids, we have the equation below left, whence we obtain the equation on the right by composing with the counits:
\begin{center}
\scalebox{0.8}{\input{naturality-cocomm.tikz}}.
\end{center}
We then see that the left-hand side simplifies to the comonoid, while the right-hand side simplifies to the comonoid followed by a symmetry, giving the commutativity equation. This establishes the claim for all $a\in C$. The full claim then follows by a straightforward induction.
\end{proof}

\begin{definition}[Theory with $1\mdash 1$-natural monoids]\label{def:thy-11-nat-monoids}
We say that a symmetric monoidal theory $(C,\Sigma,\mathcal N)$ has {\em $1\mdash 1$-natural monoids} if every $a\in C$ has a monoid structure (i.e.~the sets $\Sigma(aa,a)$ and $\Sigma(\varepsilon,a)$ contain the monoid generators and $\mathcal N$ contains the monoid equations from Definition~\ref{def:thy-mon}), and for every $\sigma\in\Sigma(a,b)$ with $a,b\in C$, we have the following equations:
\begin{center}
\scalebox{1}{\input{locally-nat-monoids.tikz}}.
\end{center}
\end{definition}

\begin{definition}[Theory with indexed monoids]\label{def:thy-indexed-monoids}
We say that a symmetric monoidal theory has {\em indexed monoids} if it has both $1\mdash 1$-natural monoids (Definition~\ref{def:thy-11-nat-monoids}) and uniform comonoids (Definition~\ref{def:thy-univ-comonoids}).
\end{definition}

\begin{proposition}\label{prop:indexed-bialgebra}
In any theory with indexed monoids, the bialgebra equations are derivable from naturality of the comonoids:
\begin{center}
\scalebox{1}{\input{bialgebra-eqns.tikz}}.
\end{center}
\end{proposition}

\chapter{Layered theories}\label{ch:layered-theories}
This chapter contains the main syntactic constructions of the thesis. We define {\em layered signatures} in Section~\ref{sec:layered-signatures}, their types and terms in Section~\ref{sec:types-terms} and finally give the full general definition of a layered theory in Section~\ref{sec:equations-theories}. We then define {\em opfibrational}, {\em fibrational} and {\em deflational} theories in Section~\ref{sec:opfib-defl-theories}.

We try to mimic as close as possible the structure of the presentation in Chapter~\ref{ch:monoidal-theories} on monoidal theories, with additions where necessary, such as the 2-terms.

We remind the reader that the terms {\em layered theory} and {\em layered monoidal theory} are used interchangeably; ditto the terms {\em layered signature} and {\em layered monoidal signature}.

\section{Layered signatures}\label{sec:layered-signatures}

A layered monoidal signature can be thought of as an indexing of monoidal signatures.

\begin{definition}[Layered signature]\label{def:layered-signature}
A {\em layered signature} is a tuple $\left(\Omega,\mathcal F,\left\{\M_{\omega}\right\}_{\omega\in\Omega}\right)$,
where $\Omega$ is a set, $\mathcal F:\Omega\times\Omega\rightarrow\Set$ is a function, and $\M_{\omega}$ is a monoidal signature for each $\omega\in\Omega$.
\end{definition}
Given a layered signature with monoidal signatures $\M_{\omega}$, we write $\M_{\omega}=(C_{\omega},\Sigma_{\omega})$. We often abbreviate a layered signature to $(\Omega,\mathcal F,\M_{\omega})$, where the index $\omega$ is implicitly assumed to range over $\Omega$. We refer to the elements of $\Omega$ as {\em layers} and to the elements of $\mathcal F(\omega,\tau)$ as {\em generators} with domain layer $\omega$ and codomain layer $\tau$.

\begin{definition}[Morphism of layered signatures]
A morphism of layered signatures
$$\left(F,F_{\tau,\omega},F^{\omega}\right):\left(\Omega,\mathcal F,\M_{\omega}\right)\rightarrow\left(\Psi,\mathcal G,\M_{\psi}\right)$$
is given by:
\begin{itemize}
\item a function $F:\Omega\rightarrow\Psi$,
\item for every pair $(\omega,\tau)\in\Omega\times\Omega$, a function $F_{\omega,\tau}:\mathcal F(\omega,\tau)\rightarrow\mathcal G(F(\omega), F(\tau))$,
\item for every $\omega\in\Omega$, a morphism of monoidal signatures $F^{\omega}:\M_{\omega}\rightarrow\M_{F(\omega)}$.
\end{itemize}
\end{definition}
We denote the resulting category of layered signatures by $\LSgn$. As for the category of monoidal signatures $\MSgn$, we denote a morphism in $\LSgn$ simply by $F:(\Omega,\mathcal F,\M_{\omega})\rightarrow (\Psi,\mathcal G,\M_{\psi})$ as follows: $F:\Omega\rightarrow\Psi$ (without subscripts or superscripts) denotes the function on layers, $F_{\omega,\tau}:\mathcal F(\omega,\tau)\rightarrow\mathcal G(F(\omega),F(\tau))$ the functions on the generators, and $F^{\omega}:\M_{\omega}\rightarrow\M_{F(\omega)}$ the morphisms of monoidal signatures.

As for monoidal signatures, there are two ways of ``forgetting'' the structure of $\LSgn$. First, one can ``forget'' the monoidal signatures, obtaining a functor $\LSgn\rightarrow\mathtt{Fam}(\Set)$ given by $(\Omega,\mathcal F,\M_{\omega})\mapsto (\Omega\times\Omega,\mathcal F)$. Second, one can ``forget'' the generators between the layers, obtaining a functor $\LSgn\rightarrow\mathtt{Fam}(\MSgn)$ simply given by $(\Omega,\mathcal F,\M_{\omega})\mapsto (\Omega,\M_{\omega})$. These functors arise as part of a pullback of categories (cf.~Proposition~\ref{prop:msgn-pullback}).
\begin{proposition}\label{prop:lsgn-pullback}
The square below is a pullback, where the right vertical map is the family fibration, and the bottom left horizontal map is the forgetful functor from families of monoidal signatures to sets:
\begin{center}
\input{lsgn-pullback.tikz}.
\end{center}
\end{proposition}
As for Proposition~\ref{prop:msgn-pullback}, we observe that Proposition~\ref{prop:lsgn-pullback} could be taken as the definition of layered signatures, using the uniqueness of pullbacks.
\begin{corollary}
The forgetful functor $\LSgn\rightarrow\Fam(\MSgn)$ is a fibration.
\end{corollary}

\subsection{Shape of a layered signature}

When defining a layered signature (especially one that has finitely many layers and generators), it will be useful to simply draw the graph with the layers and the generators. We formalise this as the notion of {\em shape}.
\begin{definition}[Shape of a layered signature]\label{def:shape-layered-signature}
Given a layered signature $(\Omega,\mathcal F,\M_{\omega})$, its {\em shape} is the multigraph with vertices $\Omega$ and the edges from $\omega$ to $\tau$ given by $\mathcal F(\omega,\tau)$. Conversely, given a multigraph $M$, a layered signature with shape $M$ has the vertices of $M$ as the layers and the edges of $M$ as the generators.
\end{definition}

\begin{example}\label{ex:layered-sgn-one}
A layered signature with the shape $\bullet$ (a multigraph with one object and no vertices) is an ordinary monoidal signature.
\end{example}

\begin{example}\label{ex:layered-sgn-two}
A layered signature with the shape
$$\bullet\longrightarrow\bullet$$
consists of two monoidal signatures and a single generator between them.
\end{example}

\section{Types and terms}\label{sec:types-terms}

We define types, terms and 2-terms generated by a layered signature, which will be used to generate, respectively, objects, morphisms and 2-cells in syntactic models (Chapter~\ref{ch:semantics}). All three levels may come with equations, ultimately giving rise the the notion of a {\em layered theory} (Definition~\ref{def:layered-theory}). We summarise the three levels of generated objects in the following table:
\begin{center}
\renewcommand{\arraystretch}{1.5}
\begin{tabular}{ c | c | c | c }
level & name & meaning & equations \\
\hhline{=|=|=|=}
$\Type_{\mathcal L}$ & types & objects / 0-cells & 0-equations (Definition~\ref{def:0equations}) \\
\hline
$\Term^1_{\mathcal L}$ & terms & morphisms / 1-cells & 1-equations (Definition~\ref{def:1equations}) \\
\hline
$\Term^2_{\mathcal L}$ & 2-terms & 2-morphisms / 2-cells & 2-equations (Definition~\ref{def:2equations})
\end{tabular}
\renewcommand{\arraystretch}{1}
\end{center}

\begin{definition}[Types of a layered signature]\label{def:types-layered-signature}
Given a layered signature $(\Omega,\mathcal F,\M_{\omega})$ with $\M_{\omega}=(C_{\omega},\Sigma_{\omega})$, the set of {\em types} is recursively generated as follows:
\begin{itemize}
\item $\varepsilon : \varepsilon$ is a type,
\item for every $\omega\in\Omega$, the expression $\varepsilon : \omega$ is a type,
\item for every $\omega\in\Omega$ and $a\in C_{\omega}$, the expression $a : \omega$ is a type,
\item if $A : \omega$ is a type and $f\in\mathcal F(\omega,\tau)$, then the expression $f(A) : \tau$ is a type,
\item if $A : \omega$ and $B : \omega$ are types, then so is $AB : \omega$,
\item if $T$ and $S$ are types, then so is $T,S$,
\end{itemize}
such that the types formed with the last rule define a monoid with unit $\varepsilon : \varepsilon$:
\begin{align*}
(T,S),K &= T,(S,K) \\
\varepsilon : \varepsilon, T &= T = T, \varepsilon : \varepsilon.
\end{align*}
\end{definition}
Given a layered signature $\mathcal L$, we denote its set of types by $\Type_{\mathcal L}$. Given a morphism of layered signatures $F:\mathcal L\rightarrow\mathcal K$ (with $\mathcal L = (\Omega,\mathcal F,\mathcal M_{\omega})$), it induces a function $F:\Type_{\mathcal L}\rightarrow\Type_{\mathcal K}$ recursively defined as follows:
\begin{align*}
\varepsilon : \varepsilon &\mapsto \varepsilon : \varepsilon \\
\varepsilon : \omega &\mapsto \varepsilon : F(\omega) \\
a : \omega &\mapsto F^{\omega}(a) : F(\omega) \\
f(A) : \tau &\mapsto F_{\omega,\tau}(f)(F(A)) : F(\tau) \\
AB : \omega &\mapsto F(A)F(B) : F(\omega) \\
T,S &\mapsto F(T),F(S).
\end{align*}

The types which are generated without applying the first and the last rules in Definition~\ref{def:types-layered-signature} are called {\em internal}. Thus, internal types are of the form $A:\omega$ for some $\omega\in\Omega$. We denote the subset of internal types by $\IntType_{\mathcal L}$. Let us define the binary relation $P^0_{\mathcal L}\sse\IntType_{\mathcal L}\times\IntType_{\mathcal L}$ on the internal types as containing those terms which are in the same layer: $(A:\omega,B:\tau)\in P^0_{\mathcal L}$ if and only if $\omega = \tau$. We note that this relation is preserved by any function $F:\Type_{\mathcal L}\rightarrow\Type_{\mathcal K}$ induced by a morphism of layered signatures: if $(T,S)\in P^0_{\mathcal L}$, then $(F(T),F(S))\in P^0_{\mathcal K}$.

\begin{definition}[0-equations]\label{def:0equations}
Given a layered signature $\mathcal L$, a set of {\em 0-equations} is a subset $E^0\sse P^0_{\mathcal L}$.
\end{definition}
Given a set of 0-equations $E^0$, we extend it to the {\em type congruence} $\Zeq$, i.e.~to the smallest equivalence relation on $\Type_{\mathcal L}$ satisfying:
\begin{itemize}
\item $E^0\sse {\Zeq}$,
\item if $A:\omega\Zeq B:\omega$ and $f\in\mathcal F(\omega,\tau)$, then $f(A):\tau\Zeq f(B):\tau$,
\item if $A:\omega\Zeq C:\omega$ and $B:\omega\Zeq D:\omega$, then $AB:\omega\Zeq CD:\omega$,
\item if $T\Zeq S$ and $U\Zeq K$, then $T,U\Zeq S,K$.
\end{itemize}

We call the pairs of types $\Type_{\mathcal L}\times\Type_{\mathcal L}$ {\em sorts}, and denote a sort by $(T \mid S)$.
\begin{definition}[Basic terms]\label{def:terms-layered-signature}
Let $(\Omega,\mathcal F,\M_{\omega})$ be a layered signature. The {\em basic terms} are generated by the recursive sorting procedure in Figure~\ref{fig:layered-terms}, with the side condition that the rules~\ref{term:int-box} and \ref{term:int-tensor} only apply to {\em internal} terms, defined as follows:
\begin{itemize}[align=parleft, leftmargin=*, labelsep=3.5cm]
\item[\bf Base~case:] the terms generated by the rules~\ref{term:int-unit},~\ref{term:int-id},~\ref{term:int-gen},~\ref{term:int-box} or~\ref{term:int-tensor} are internal,
\item[\bf Recursive~case:] if the terms $x : (T \mid S)$ and $y : (S \mid U)$ are internal, then so is the term $x;y : (T \mid U)$ obtained by~\ref{term:comp}.
\end{itemize}
\end{definition}

\begin{figure}
  \centering\small\noindent
  \begin{prooftree}
    \AxiomC{$\omega\in\Omega$}
    \RightLabel{\customlabel{term:int-unit}{(int-unit)}\;\;}
    \UnaryInfC{$\scalebox{.9}{\input{emptydiag-sheet.tikz}} : (\varepsilon : \omega\mid\varepsilon : \omega)$}
    \DisplayProof
    \AxiomC{$A:\omega$}
    \AxiomC{$A\neq\varepsilon$}
    \RightLabel{\customlabel{term:int-id}{(int-id)}\;\;}
    \BinaryInfC{$\scalebox{.9}{\input{iddiag-sheet.tikz}} : (A:\omega \mid A:\omega)$}
  \end{prooftree}
  \begin{prooftree}
    \AxiomC{$\sigma\in\Sigma_{\omega}(a,b)$}
    \RightLabel{\customlabel{term:int-gen}{(int-gen)}\;\;}
    \UnaryInfC{$\scalebox{.9}{\input{internalsigmadiag.tikz}} : (a:\omega\mid b:\omega)$}
  \end{prooftree}
  \begin{prooftree}
    \AxiomC{$\scalebox{.9}{\input{internalxdiag.tikz}} : (A:\omega\mid B:\omega)$}
    \AxiomC{$f\in\mathcal F(\omega,\tau)$}
    \RightLabel{\customlabel{term:int-box}{(int-box)}\;\;}
    \BinaryInfC{$\scalebox{.9}{\input{f-box.tikz}} : (f(A):\tau \mid f(B):\tau)$}
  \end{prooftree}
  \begin{prooftree}
    \AxiomC{$\scalebox{.9}{\input{internalxdiag.tikz}} : (A:\omega\mid B:\omega)$}
    \AxiomC{$\scalebox{.9}{\input{internalydiag.tikz}} : (C:\omega\mid D:\omega)$}
    \RightLabel{\customlabel{term:int-tensor}{(int-tensor)}\;\;}
    \BinaryInfC{$\scalebox{.9}{\input{internalxydiag.tikz}} : (AC:\omega \mid BD:\omega)$}
  \end{prooftree}
  \begin{prooftree}
    \AxiomC{}
    \RightLabel{\customlabel{term:ext-unit}{(ext-unit)}\;\;}
    \UnaryInfC{$\scalebox{.9}{\input{emptydiag.tikz}} : (\varepsilon : \varepsilon\mid\varepsilon : \varepsilon)$}
  \end{prooftree}
  \begin{prooftree}
    \AxiomC{$x : (T\mid S)$}
    \AxiomC{$y : (S\mid U)$}
    \RightLabel{\customlabel{term:comp}{(comp)}\;\;}
    \BinaryInfC{$x;y : (T\mid U)$}
    \DisplayProof
    \AxiomC{$x : (T\mid S)$}
    \AxiomC{$y : (U \mid W)$}
    \RightLabel{\customlabel{term:ext-tensor}{(ext-tensor)}\;\;}
    \BinaryInfC{$x\otimes y : (T,U \mid S,W)$}
  \end{prooftree}
  \caption{Rules for generating the basic terms of a layered signature\label{fig:layered-terms}}
\end{figure}

In the definition of a layered signature (Definition~\ref{def:layered-signature}), the generators between the layers are all of the type 1-1, i.e.~the arity and coarity both consist of a single layer (this is in contrast to the generators within the layers, whose arity and coarity are lists of colours). This restriction is imposed as there is no uniform way to draw an arbitrary generator between a list (product) of layers. Since we do want the layers to interact, we introduce the mixed (co)arities at the level of terms. To this end, we need the notion of a {\em recursive sorting procedure}, which imposes some restrictions on what kinds of terms can be introduced. Each layered theory comes with a fixed recursive sorting procedure, within which the terms of the theory are generated.

\begin{definition}[Recursive sorting procedure]\label{def:sorting-procedure}
A {\em recursive sorting procedure} is a set of recursive rules such that
\begin{itemize}
\item given a layered signature $\mathcal L$, the rules generate the set of expressions $\Term^1_{\mathcal L}$ called {\em terms},
\item each term has a unique sort: we denote a term with its sort by $t:(T\mid S)$,
\item the set contains the rules for the {\em basic terms} of Figure~\ref{fig:layered-terms},
\item any additional rules do not change the set of internal terms, as defined in Definition~\ref{def:terms-layered-signature},
\item any morphism of layered signatures $F:\mathcal L\rightarrow\mathcal K$ uniquely extends to a function
$$F:\Term^1_{\mathcal L}\rightarrow\Term^1_{\mathcal K}$$
by translating the types and terms in the assumption of each recursive rule, and then applying the corresponding rule in $\Term^1_{\mathcal K}$, which is moreover consistent with the induced function on types: $t:(T\mid S)\mapsto F(t) : (F(T)\mid F(S))$.
\end{itemize}
\end{definition}
The additional rules we shall consider are given in Definition~\ref{def:symmetry-terms} as well as Figures~\ref{fig:opfibrational-terms} and~\ref{fig:fibrational-terms}. In the next chapter, we will also add the rule in Definition~\ref{def:cobox}, although the terms generated by this rule can always be eliminated.

We say that two terms are {\em parallel} when they have the same sort, up to 0-equations, if there are any. Formally, given a set of 0-equations $E^0$, we define the binary relation $P^1_{\mathcal L}\sse\Term^1_{\mathcal L}\times\Term^1_{\mathcal L}$ by $\left(t:(T_t\mid S_t),s:(T_s\mid S_s)\right)\in P^1_{\mathcal L}$ if and only if both $T_t\Zeq T_s$ and $S_t\Zeq S_s$, and call the pairs of terms in $P^1_{\mathcal L}$ {\em parallel} with respect to $E^0$. Note that if $E^0$ is empty (i.e.~there are no equations), then terms are parallel if and only if they have the same sort.

\begin{proposition}\label{prop:parallel-terms-preserved}
Let $F:(\mathcal L,E^0_{\mathcal L})\rightarrow (\mathcal K,E^0_{\mathcal K})$ be a morphism of layered signatures such that the induced function on types $F:\Type_{\mathcal L}\rightarrow\Type_{\mathcal K}$ preserves the chosen 0-equations. Then the induced function on terms $F:\Term^1_{\mathcal L}\rightarrow\Term^1_{\mathcal K}$ preserves parallel terms: if $(s,t)\in P^1_{\mathcal L}$, then $(F(s),F(t))\in P^1_{\mathcal K}$.
\end{proposition}
\begin{proof}
By the fact that terms with the same sort are mapped to the terms with the same sort, and by induction on the construction of the type congruence.
\end{proof}

\begin{definition}[1-equations]\label{def:1equations}
Given a layered signature $\mathcal L$ together with a recursive sorting procedure and a set of 0-equations $E^0$, a set of {\em 1-equations} with respect to $E^0$ is a subset $E^1\sse P^1_{\mathcal L}$.
\end{definition}
Given a set of 1-equations $E^1$, we extend it to the {\em term congruence} $\Oeq$, i.e.~the smallest equivalence relation on $P^1_{\mathcal L}$ containing $E^1$ that is preserved by the recursive rules in Figure~\ref{fig:layered-terms}.

In many cases, equality between terms is too strong a notion. Often, two 1-cells are not equal but merely adjoint, as in deflational models of layered signatures, or there is a clear direction to a transformation that preserves semantics, as in the graph rewrites for MBQC-graphs (Section~\ref{sec:zx-extraction}). To capture this, we need to extend the layered signature to include generating 2-cells (in addition to 1-equations).
\begin{definition}[Choice of 2-cells]\label{def:choice-2cells}
Given a layered signature $\mathcal L$ together with a recursive sorting procedure and a set of 0-equations $E^0$, a {\em choice of 2-cells} with respect to $E^0$ is a function $\eta : P^1_{\mathcal L}\rightarrow\Set$, assigning to each parallel pair of terms a set of {\em generating 2-cells}.
\end{definition}
Note that any layered signature in the sense of Definition~\ref{def:layered-signature} can be seen as having a choice of 2-cells by setting $\eta$ to be the constant function returning the empty set. A morphism between layered signatures with a choice of 2-cells
$$(F,F^2) : (\mathcal L, E^0_{\mathcal L}, \eta_{\mathcal L})\rightarrow (\mathcal K, E^0_{\mathcal K}, \eta_{\mathcal K})$$
is given by a morphism of layered signatures $F:\mathcal L\rightarrow\mathcal K$ such that the induced function $F:\Type_{\mathcal L}\rightarrow\Type_{\mathcal K}$ preserves the 0-equations, and for each pair $(t,s)\in P^1_{\mathcal L}$, a function
$$F^2_{t,s}:\eta_{\mathcal L}(t,s)\rightarrow\eta_{\mathcal K}(F(t),F(s)).$$

\begin{definition}[2-terms]\label{def:2terms}
Given a layered signature $\mathcal L$ with a choice of 2-cells $\eta$, a {\em 2-term} is an expression of the form $\alpha : (s,t)$, where $(s,t)\in P^1_{\mathcal L}$ is its {\em sort}, generated by the following recursive procedure:
\begin{center}
  \small
  \begin{prooftree}
    \AxiomC{$a\in\eta(t,s)$}
    \RightLabel{\;\;}
    \UnaryInfC{$a : (t,s)$}
    \DisplayProof
    \AxiomC{$\alpha : (t,s)$}
    \AxiomC{$\beta : (s,k)$}
    \RightLabel{\;\;}
    \BinaryInfC{$\alpha;\beta : (t,k)$}
    \DisplayProof
    \AxiomC{$\alpha : (t_1,s_1)$}
    \AxiomC{$\beta : (t_2,s_2)$}
    \RightLabel{\;\;}
    \BinaryInfC{$\alpha\otimes\beta : (t_1\otimes t_2,s_1\otimes s_2)$}
  \end{prooftree}
  \begin{prooftree}
    \AxiomC{$t\in\Term^1_{\mathcal L}$}
    \RightLabel{\;\;}
    \UnaryInfC{$\id_t : (t,t)$}
    \DisplayProof
    \AxiomC{$\alpha : (t_1,s_1)$}
    \AxiomC{$\beta : (t_2,s_2)$}
    \AxiomC{$t_1 : (T\mid S)$}
    \AxiomC{$t_2 : (S\mid K)$}
    \RightLabel{.}
    \QuaternaryInfC{$\alpha *\beta : (t_1;t_2, s_1;s_2)$}
  \end{prooftree}
  \end{center}
\end{definition}
\begin{remark}
We have chosen not to extend the 2-terms to the recursively defined internal terms (i.e.~the ones generated by the rules~\ref{term:int-box} and~\ref{term:int-tensor}). There is no inherent reason for omitting these terms: upon adding them, one would just need to add the corresponding structural 2-equations (see Definition~\ref{def:str-2eqns}). However, in the models and the examples we will consider the 2-categorical structure only arises at the level of external terms, or does not need to be propagated between layers. Thus, for the sake of slightly reducing the complexity of the current presentation, we omit these terms. Moreover, as we shall observe in Corollary~\ref{cor:defl-twocells-extend}, in deflational theories, such terms are automatically induced by the external ones.
\end{remark}
Akin to terms, we denote the set of 2-terms by $\Term^2_{\mathcal L}$. Let
$$(F,F^2) : (\mathcal L, E^0_{\mathcal L}, \eta_{\mathcal L})\rightarrow (\mathcal K, E^0_{\mathcal K}, \eta_{\mathcal K})$$
be a morphism between layered signatures with a choice of 2-cells. The functions between the generating 2-cells $F^2_{t,s}$ then yield a function $F^2:\Term^2_{\mathcal L}\rightarrow\Term^2_{\mathcal K}$ recursively defined as follows:
\begin{align*}
a : (t,s) &\mapsto F^2_{t,s}(a) : (F(t),F(s)) \\
\id_t : (t,t) &\mapsto \id_{F(t)} : (F(t),F(t)) \\
\alpha;\beta : (t,k) &\mapsto F^2(\alpha); F^2(\beta) : (F(t),F(k)) \\
\alpha\otimes\beta : (t_1\otimes t_2,s_1\otimes s_2) &\mapsto F^2(\alpha)\otimes F^2(\beta) : (F(t_1)\otimes F(t_2),F(s_1)\otimes F(s_2)) \\
\alpha *\beta : (t_1;t_2, s_1;s_2) &\mapsto F^2(\alpha) * F^2(\beta) : (F(t_1);F(t_2), F(s_1);F(s_2)).
\end{align*}

As for the terms, given a set of 1-equations $E^1$, a pair of 2-terms is in the {\em parallel 2-term relation} $P^2_{\mathcal L}$ with respect to $E^1$ if and only if both 2-terms have the same sort up to the 1-equations: $\left(\alpha:(t_{\alpha},s_{\alpha}), \beta:(t_{\beta},s_{\beta})\right)\in P^2_{\mathcal L}$ if and only if both $t_{\alpha}\Oeq t_{\beta}$ and $s_{\alpha}\Oeq s_{\beta}$. We draw a 2-term $\alpha : (t,s)$ either as
\begin{center}
\input{2term.tikz},
\end{center}
or simply as an arrow $t\xrightarrow{\alpha} s$ (see e.g.~Figure~\ref{fig:structural-twocells-adjoints} on page~\pageref{fig:structural-twocells-adjoints}).

\begin{proposition}\label{prop:parallel-2terms-preserved}
Let
$$(F,F^2) : (\mathcal L, E^0_{\mathcal L}, \eta_{\mathcal L}, E^1_{\mathcal L})\rightarrow (\mathcal K, E^0_{\mathcal K}, \eta_{\mathcal K}, E^1_{\mathcal K})$$
be a morphism between layered signatures with a choice of 2-cells such that the induced function on terms $F:\Term^1_{\mathcal L}\rightarrow\Term^1_{\mathcal K}$ preserves the specified 1-equations. Then the induced function on 2-terms $F^2:\Term^2_{\mathcal L}\rightarrow\Term^2_{\mathcal K}$ preserves parallel terms: if $(\alpha,\beta)\in P^2_{\mathcal L}$, then $(F^2(\alpha),F^2(\beta))\in P^2_{\mathcal K}$.
\end{proposition}
\begin{proof}
By the fact that 2-terms with the same sort are mapped to 2-terms with the same sort, and by induction on the construction of the term congruence.
\end{proof}

\begin{definition}[2-equations]\label{def:2equations}
Let $\mathcal L$ be layered signature with a set of 0-equations $E^0$. Given a set of 1-equations $E^1$ and a choice of 2-cells $\eta$ with respect to $E^0$, a set of {\em 2-equations} with respect to $E^1$ and $\eta$ is a subset $E^2\sse P^2_{\mathcal L}$.
\end{definition}
Given a set of 2-equations $E^2$, we extend it to the {\em 2-term congruence} $\Teq$, i.e.~the smallest equivalence relation on $P^2_{\mathcal L}$ containing $E^2$ that is preserved by the recursive rules in Definition~\ref{def:2terms}.

\section{Equations and theories}\label{sec:equations-theories}

We now have all the ingredients to define a layered theory.

\begin{definition}[Layered theory]\label{def:layered-theory}
Let us fix a recursive sorting procedure (Definition~\ref{def:sorting-procedure}). A {\em layered theory} $(\mathcal L,E^0,E^1,\eta,E^2)$ consists of the following:
\begin{itemize}
\item a layered signature $\mathcal L$ (Definition~\ref{def:layered-signature}),
\item a set of 0-equations $E^0$ (Definition~\ref{def:0equations}),
\item a set of 1-equations $E^1$ with respect to $E^0$ (Definition~\ref{def:1equations}),
\item a choice of 2-cells $\eta$ with respect to $E^0$ (Definition~\ref{def:choice-2cells}),
\item a set of 2-equations $E^2$ with respect to $E^1$ and $\eta$ (Definition~\ref{def:2equations}).
\end{itemize}
\end{definition}
For a fixed recursive sorting procedure, a {\em morphism of layered theories}
$$(F,F^2) : (\mathcal L,E^0_{\mathcal L},E^1_{\mathcal L},\eta_{\mathcal L},E^2_{\mathcal L})\rightarrow (\mathcal K,E^0_{\mathcal K},E^1_{\mathcal K},\eta_{\mathcal K},E^2_{\mathcal K})$$
is given by a morphism of layered signatures with a choice of 2-cells $(F,F^2)$ (thus, in particular, the induced function $F:\Type_{\mathcal L}\rightarrow\Type_{\mathcal K}$ preserves the 0-equations), such that the induced functions $F:\Term^1_{\mathcal L}\rightarrow\Term^1_{\mathcal K}$ and $F^2:\Term^2_{\mathcal L}\rightarrow\Term^2_{\mathcal K}$ preserve the 1-equations and the 2-equations, respectively. We denote any category of layered theories by $\LTh$: note that each recursive sorting procedure results in a different category.
\begin{proposition}\label{prop:lth-lsgn-fibration}
Let us fix a recursive sorting procedure. The forgetful functor $\LTh\rightarrow\LSgn$ is a fibration. Moreover, the objects in the fibre $\LTh(\mathcal L)$ are precisely the layered theories with the signature $\mathcal L$.
\end{proposition}

Next, we define the structural equations that will hold in nearly all the models we will encounter, with the exception the ``conditional box'' notation defined in Section~\ref{sec:prob-channels}, where we have to drop the functoriality requirement.

\begin{definition}[Structural 0-equations]\label{def:str-0eqns}
Given a layered signature, the {\em structural 0-equations} are given by:
\begin{align*}
(AB)C : \omega &= A(BC) : \omega, \\
\varepsilon A : \omega &= A : \omega = A\varepsilon : \omega, \\
f(AB):\tau &= f(A)f(B):\tau, \\
f(\varepsilon) : \tau &= \varepsilon : \tau.
\end{align*}
We denote the structural 0-equations by $S^0$.
\end{definition}

\begin{definition}[Structural 1-equations]\label{def:str-1eqns}
Given a layered signature with the structural 0-equations (Definition~\ref{def:str-0eqns}), the {\em structural 1-equations} are given by the following:
\begin{itemize}
\item the structural identities for monoidal theories (Definition~\ref{def:str-id}), where sequential composition is given by the rule~\ref{term:comp}, the identities are given by the terms~\ref{term:int-id}, parallel composition is given by the rule~\ref{term:ext-tensor}, and the monoidal unit is given by~\ref{term:ext-unit},
\item the structural identities for the monoidal product (the bottom half of Definition~\ref{def:str-id}) hold for the internal terms, where parallel composition is given by the rule~\ref{term:int-tensor}, the monoidal unit is given by~\ref{term:int-unit}, and the sequential composition is given by~\ref{term:comp}\footnote{We note that the structural identities for internal terms which only involve sequential composition and identities are already implied by requiring them for all terms.},
\item the following identity for the identity terms with sort $(AB:\omega\mid AB:\omega)$ generated by the rules~\ref{term:int-id} and~\ref{term:int-tensor}:
\begin{equation*}
\scalebox{.9}{\input{int-id-equals-int-tensor.tikz}} : (AB:\omega\mid AB:\omega),
\end{equation*}
where the term on the left-hand side is generated by the rule~\ref{term:int-id}, while the term on the right-hand side is obtained by applying~\ref{term:int-tensor} to the identity terms with sorts $(A:\omega\mid A:\omega)$ and $(B:\omega\mid B:\omega)$, which are, in turn, obtained from~\ref{term:int-id},
\item the identities in Figure~\ref{fig:structural-twocells-functors-int} involving the~\ref{term:int-box} terms, making each $f$ a monoidal functor.
\end{itemize}
We denote the structural 1-equations by $S^1$.
\end{definition}

\begin{figure}
  \centering\small\noindent
  \scalebox{1}{\input{structural-twocells-functors-int.tikz}}
  \caption{1-equations defining monoidal functors inside the fibres\label{fig:structural-twocells-functors-int}}
\end{figure}

\begin{definition}[Structural 2-equations]\label{def:str-2eqns}
Given a layered signature and the structural 1-equations of Definition~\ref{def:str-1eqns}, the {\em structural 2-equations} are given in Figure~\ref{fig:str-2eqns} by the expressions on the left, as long as both sides are defined (on the right, we draw the situation in which each equation applies). We denote the structural 2-equations by $S^2$.
\end{definition}

\begin{figure}
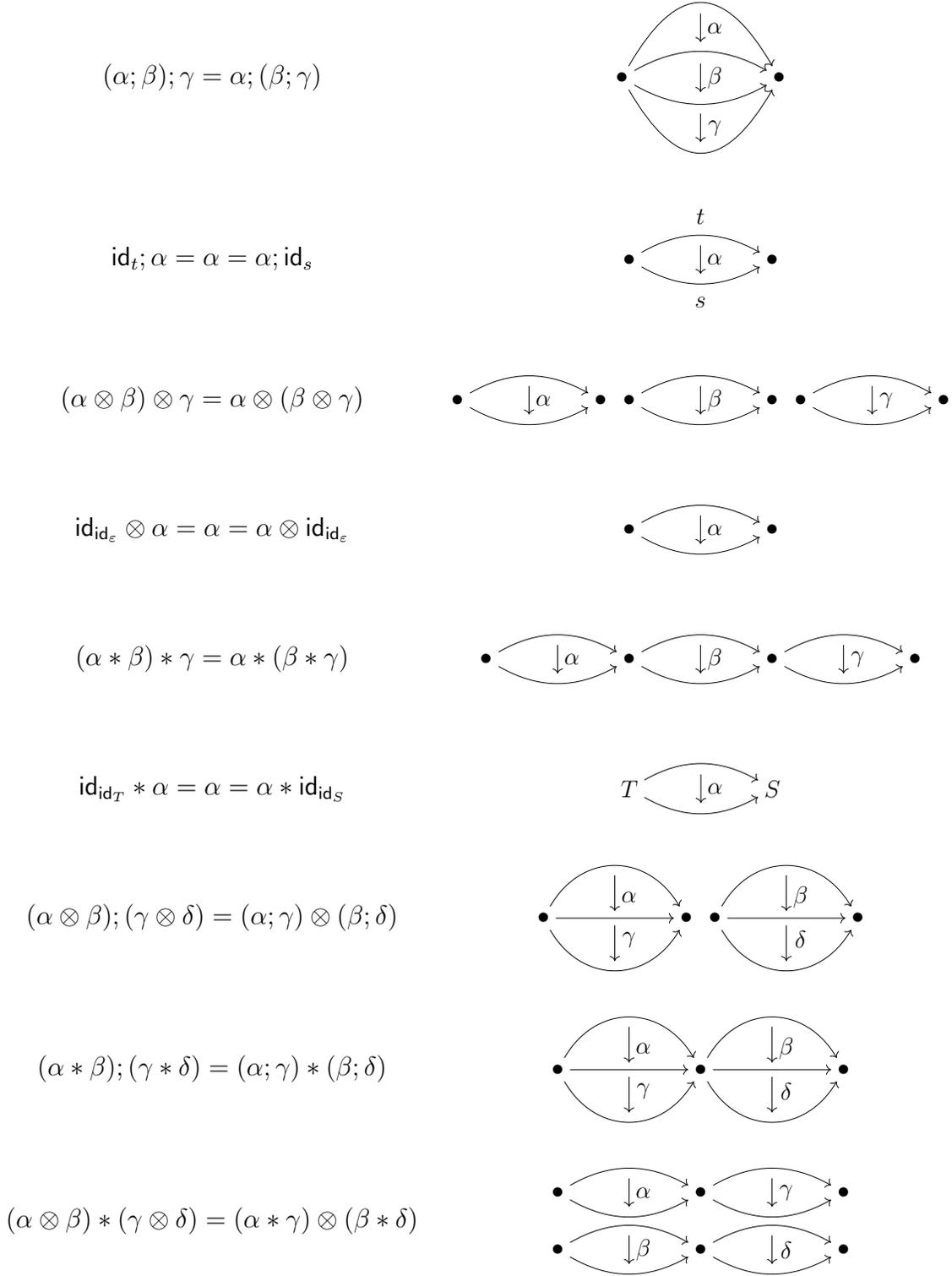

\centering
\renewcommand{\arraystretch}{4}
\begin{tabular}{c c}
$(\alpha;\beta);\gamma = \alpha;(\beta;\gamma)$ & \scalebox{.9}{\input{2cells-compose-assoc.tikz}} \\
$\id_t;\alpha = \alpha = \alpha;\id_s$ & \scalebox{.9}{\input{2cells-compose-id.tikz}} \\
$(\alpha\otimes\beta)\otimes\gamma = \alpha\otimes(\beta\otimes\gamma)$ & \scalebox{.9}{\input{2cells-tensor-assoc.tikz}} \\
$\id_{\id_{\varepsilon}}\otimes\alpha = \alpha = \alpha\otimes\id_{\id_{\varepsilon}}$ & \scalebox{.9}{\input{2cells-tensor-id.tikz}} \\
$(\alpha *\beta) *\gamma = \alpha *(\beta *\gamma)$ & \scalebox{.9}{\input{2cells-horizontal-assoc.tikz}} \\
$\id_{\id_T} *\alpha = \alpha = \alpha *\id_{\id_S}$ & \scalebox{.9}{\input{2cells-horizontal-id.tikz}} \\
$(\alpha\otimes\beta);(\gamma\otimes\delta) = (\alpha;\gamma)\otimes (\beta;\delta)$ & \scalebox{.9}{\input{2cells-tensor-compose.tikz}} \\
$(\alpha *\beta);(\gamma *\delta) = (\alpha;\gamma) * (\beta;\delta)$ & \scalebox{.9}{\input{2cells-horizontal-compose.tikz}} \\
$(\alpha\otimes\beta) * (\gamma\otimes\delta) = (\alpha *\gamma)\otimes (\beta *\delta)$ & \scalebox{.9}{\input{2cells-tensor-horizontal.tikz}}
\end{tabular}
\renewcommand{\arraystretch}{1}
\caption{The structural 2-equations. Here $\id_T$ and $\id_S$ are the identity terms on the types $T$ and $S$, while $\id_{\varepsilon}$ is the identity term on $\varepsilon : \varepsilon$ obtained by the rule~\ref{term:ext-unit}.\label{fig:str-2eqns}}
\end{figure}

We say that a layered theory {\em has structural equations} when its sets of equations contain all the structural equations of Definitions~\ref{def:str-0eqns}, \ref{def:str-1eqns} and~\ref{def:str-2eqns}.

\begin{definition}[Symmetry terms]\label{def:symmetry-terms}
Given a layered signature, the {\em symmetry terms} are generated by the following rule:
\begin{center}
  \begin{prooftree}
    \AxiomC{$A:\omega$}
    \AxiomC{$B:\tau$}
    \RightLabel{\customlabel{term:swap}{(swap)}.\;\;}
    \BinaryInfC{$\scalebox{.9}{\input{symdiag-sheet1.tikz}} : (A:\omega, B:\tau \mid B:\tau, A:\omega)$}
  \end{prooftree}
\end{center}
\end{definition}

\begin{definition}[Externally symmetric layered theory]\label{def:symmetry-1equations}
A layered theory is {\em externally symmetric} if it has the symmetry terms, and the structural 0-equations (Definition~\ref{def:str-0eqns}) and 1-equations (Definition~\ref{def:str-1eqns}) hold, and the 1-equations contain the equations for symmetric monoidal theories in Definition~\ref{def:symm-mon-thy} for the symmetry terms.
\end{definition}

\begin{example}
Consider a layered theory $(\bullet,S^0,E^1,\eset,\eset)$ with no generating 2-cells or 2-equations, where $\bullet$ is a layered signature from Example~\ref{ex:layered-sgn-one} (one layer and no generators). The internal terms contain the terms of monoidal theories as defined in Section~\ref{sec:monoidal-theories}. If further $S^1\sse E^1$, then there is a one-to-one correspondence between the equivalence classes of internal terms and the equivalence classes of monoidal terms.
\end{example}

\begin{example}
Consider a layered theory $(\omega\rightarrow\tau,S^0,S^1,\eset,\eset)$ with no generating 2-cells or 2-equations, where $\omega\rightarrow\tau$ is a layered signature from Example~\ref{ex:layered-sgn-two} (two layers with one generator between them). Then the internal terms with type $\omega$ give the free monoidal category generated by the signature indexed by $\omega$, while the internal terms with type $\tau$ give the free monoidal category generated by the disjoint union of the signature indexed by $\tau$ and the generators obtained by the rule~\ref{term:int-box}. The 1-equations of Figure~\ref{fig:structural-twocells-functors-int} then make $f:\omega\rightarrow\tau$ a strict monoidal functor between the two free models. The functor is free in the sense that all the terms obtained by applying $f$ to the internal terms with type $\omega$ (the \ref{term:int-box} rule) remain uninterpreted, and are simply added as generators to the monoidal category defined by the internal terms with type $\tau$.
\end{example}

\section{(Op)fibrational and deflational theories}\label{sec:opfib-defl-theories}

Here we define three classes of layered theories that we will focus on: opfibrational, fibrational and deflational theories. In each case, we specify the recursive sorting procedure as well as the corresponding structural equations. In the case of deflational theories, we also define the structural 2-cells. The models of these theories are studied in Chapter~\ref{ch:semantics}, building on the theory of Chapters~\ref{ch:indexed-monoidal} and~\ref{ch:profunctor-collages}.

\subsection{Opfibrational theories}\label{subsec:opfib-theories}

Opfibrational theories are defined for the recursive procedure specified in Definition~\ref{def:opfib-layered-theory} and have no generating 2-cells.

\begin{figure}
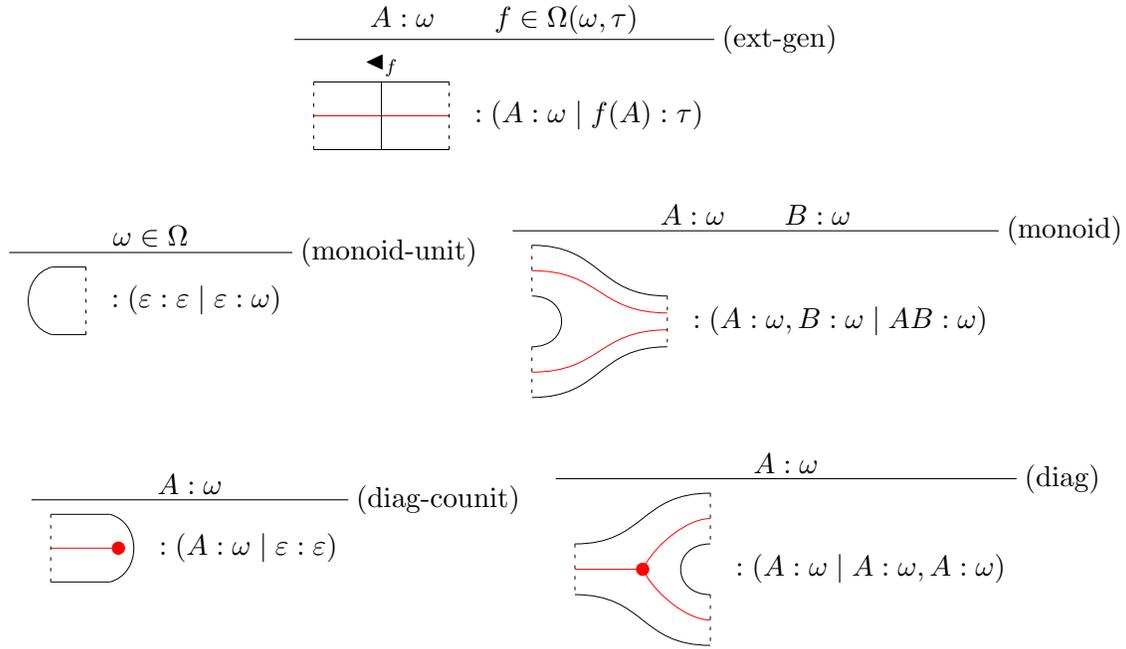

  \centering\small\noindent
  \begin{prooftree}
    \AxiomC{$A:\omega$}
    \AxiomC{$f\in\Omega(\omega,\tau)$}
    \RightLabel{\customlabel{term:ext-gen}{(ext-gen)}\;\;}
    \BinaryInfC{$\scalebox{.9}{\input{refine-sheet.tikz}} : (A:\omega \mid f(A):\tau)$}
  \end{prooftree}
  \begin{prooftree}
    \AxiomC{$\omega\in\Omega$}
    \RightLabel{\customlabel{term:monoid-unit}{(monoid-unit)}\;\;}
    \UnaryInfC{$\scalebox{.9}{\input{cup.tikz}} : (\varepsilon : \varepsilon\mid\varepsilon : \omega)$}
    \DisplayProof
    \AxiomC{$A:\omega$}
    \AxiomC{$B:\omega$}
    \RightLabel{\customlabel{term:monoid}{(monoid)}\;\;}
    \BinaryInfC{$\scalebox{.9}{\input{pants.tikz}} : (A:\omega, B:\omega \mid AB:\omega)$}
  \end{prooftree}
  \begin{prooftree}
    \AxiomC{$A:\omega$}
    \RightLabel{\customlabel{term:diag-counit}{(diag-counit)}\;\;}
    \UnaryInfC{$\scalebox{.9}{\input{a-cap.tikz}} : (A:\omega \mid \varepsilon : \varepsilon)$}
    \DisplayProof
    \AxiomC{$A:\omega$}
    \RightLabel{\customlabel{term:diag}{(diag)}\;\;}
    \UnaryInfC{$\scalebox{.9}{\input{copants-copy.tikz}} : (A:\omega \mid A:\omega, A:\omega)$}
  \end{prooftree}
  \caption{Rules for generating the opfibrational terms\label{fig:opfibrational-terms}}
\end{figure}

\begin{definition}[Opfibrational layered theory]\label{def:opfib-layered-theory}
We say that a layered theory is {\em opfibrational} if it has no generating 2-cells, and its recursive sorting procedure consists of the rules in Figure~\ref{fig:opfibrational-terms} and the symmetry terms of Definition~\ref{def:symmetry-terms} (in addition to the basic terms in Figure~\ref{fig:layered-terms}).
\end{definition}

The intuition behind the opfibrational terms is that we take an ``external'' view of the monoidal categories and monoidal functors between them. In more detail:
\begin{itemize}
\item the terms~\ref{term:monoid} and~\ref{term:monoid-unit} capture the monoidal product and unit in the layer $\omega$: e.g.~sliding the internal terms through the term~\ref{term:monoid} defines the monoidal product on morphisms (see Figure~\ref{fig:structural-twocells-monoidal}),
\item the terms~\ref{term:ext-gen} capture the monoidal functors $f:\omega\rightarrow\tau$: any internal term appearing on the left-hand side of the boundary is in the domain $\omega$, and can be pushed through the boundary into the codomain $\tau$ (see Figure~\ref{fig:structural-twocells-functors-ext}),
\item the terms~\ref{term:diag-counit} and~\ref{term:diag} capture the cartesian monoidal structure of the category of monoidal categories; note that they do not correspond to any deleting or copying {\em inside} a layer, rather, it might be instructive to think of them as {\em branching} or {\em nondeterminism}: \ref{term:diag} represents branching into two potential histories of a system, while \ref{term:diag-counit} corresponds to discarding one of the branches or histories.
\end{itemize}

\begin{remark}
It is somewhat inadequate to call the layered theories of Definition~\ref{def:opfib-layered-theory} {\em opfibrational}. First, exclusion of non-trivial 2-cells does not allow for any 2-categorical structure, forcing all the equations to be strict. Second, the recursive sorting procedure already includes the terms that define indexed monoids on the opfibration. Thus, a more accurate name for what we have defined would be {\em strict layered theory for opfibrations with indexed monoids}. Since this is somewhat cumbersome, and we shall not consider any other layered theories featuring an opfibration so that no confusion can arise, we stick with simply calling these theories {\em opfibrational}.
\end{remark}

We denote the resulting category of opfibrational layered theories by $\OpFTh$. Note that in this case the morphisms are simply equation preserving morphisms of signatures rather than pairs of morphisms, as we assume there are no generating 2-cells.

\begin{definition}[Structural opfibrational equations]\label{def:str-opfib-eqns}
By the {\em structural opfibrational equations} we mean the structural equations, the external symmetry equations, as well as the following sets of 1-equations:
\begin{itemize}
\item the equations of uniform comonoids (Definition~\ref{def:thy-univ-comonoids}) for the terms~\ref{term:diag} and~\ref{term:diag-counit},
\item the defining identities of monoidal categories in Figure~\ref{fig:structural-twocells-monoidal},
\item the defining identities of monoidal functors in Figure~\ref{fig:structural-twocells-functors-ext}.
\end{itemize}
We denote the structural opfibrational 1-equations by $S_{\mathsf{opf}}^1$.
\end{definition}

\begin{figure}
  \centering\small\noindent
  \scalebox{.8}{\input{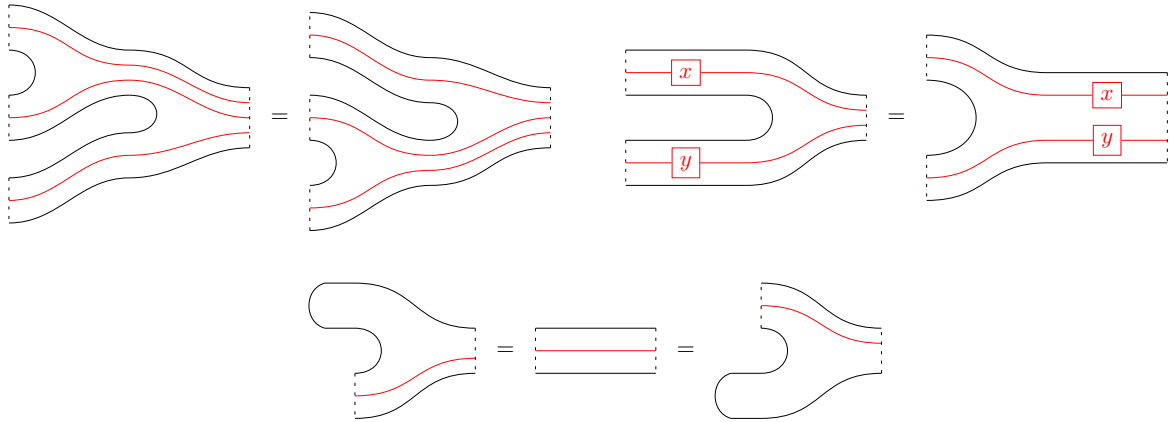}}
  \caption{1-equations defining monoidal categories\label{fig:structural-twocells-monoidal}}
\end{figure}

\begin{figure}
  \centering\small\noindent
  \scalebox{1}{\input{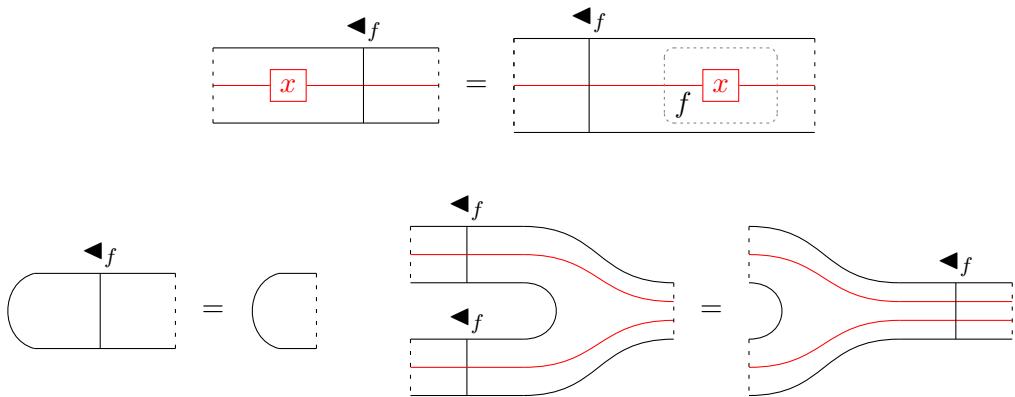}}
  \caption{1-equations defining monoidal functors\label{fig:structural-twocells-functors-ext}}
\end{figure}

\subsection{Fibrational theories}\label{subsec:fib-theories}

Dualising the construction of opfibrational theories, we obtain {\em fibrational theories}, where the composition of internal and external terms go in opposite directions. Formally, we obtain opfibrational theories by replacing the opfibrational terms (Figure~\ref{fig:opfibrational-terms}) by the {\em fibrational terms} (Figure~\ref{fig:fibrational-terms}) in Definition~\ref{def:opfib-layered-theory}. Note that each fibrational term is a horizontal reflection of an opfibrational term, and that in the rule~\ref{term:ext-gen-op} we use of an unfilled triangle for emphasis -- the lack of filling plays no formal role.

\begin{figure}
  \centering\small\noindent
  \begin{prooftree}
    \AxiomC{$A:\omega$}
    \AxiomC{$f\in\Omega(\omega,\tau)$}
    \RightLabel{\customlabel{term:ext-gen-op}{(ext-gen-op)}\;\;}
    \BinaryInfC{$\scalebox{.9}{\input{coarsen-sheet.tikz}} : (f(A):\tau \mid A:\omega)$}
  \end{prooftree}
  \begin{prooftree}
    \AxiomC{$\omega\in\Omega$}
    \RightLabel{\customlabel{term:counit}{(counit)}\;\;}
    \UnaryInfC{$\scalebox{.9}{\input{cap.tikz}} : (\varepsilon : \omega \mid \varepsilon : \varepsilon)$}
    \DisplayProof
    \AxiomC{$A:\omega$}
    \AxiomC{$B:\omega$}
    \RightLabel{\customlabel{term:comonoid}{(comonoid)}\;\;}
    \BinaryInfC{$\scalebox{.9}{\input{copants.tikz}} : (AB:\omega \mid A:\omega, B:\omega)$}
  \end{prooftree}
  \begin{prooftree}
    \AxiomC{$A:\omega$}
    \RightLabel{\customlabel{term:codiag-unit}{(codiag-unit)}\;\;}
    \UnaryInfC{$\scalebox{.9}{\input{a-cup.tikz}} : (\varepsilon : \varepsilon \mid A:\omega)$}
    \DisplayProof
    \AxiomC{$A:\omega$}
    \RightLabel{\customlabel{term:codiag}{(codiag)}\;\;}
    \UnaryInfC{$\scalebox{.9}{\input{pants-copy.tikz}} : (A:\omega, A:\omega \mid A:\omega)$}
  \end{prooftree}
  \caption{Rules for generating the fibrational terms\label{fig:fibrational-terms}}
\end{figure}

The structural fibrational equations are likewise obtained by dualising the structural opfibrational equations (Definition~\ref{def:str-opfib-eqns}). For the sake of completeness, and since they will be used for defining the structural deflational equations, we state the definition explicitly.
\begin{definition}[Structural fibrational equations]\label{def:str-fib-eqns}
By the {\em structural fibrational equations} we mean the structural equations, the external symmetry equations, as well as the following sets of 1-equations, where ``dual'' and ``dualising'' means reflecting the terms horizontally:
\begin{itemize}
\item the equations of uniform monoids (the dual of Definition~\ref{def:thy-univ-comonoids}) for the terms~\ref{term:codiag} and~\ref{term:codiag-unit},
\item the defining identities of monoidal categories obtained by dualising Figure~\ref{fig:structural-twocells-monoidal},
\item the defining identities of monoidal functors obtained by dualising Figure~\ref{fig:structural-twocells-functors-ext}.
\end{itemize}
\end{definition}
With these modifications, any properties of fibrational theories can be obtained by dualising the corresponding properties of opfibrational theories: all the diagrams are simply written in reverse direction. Their study does not, therefore, provide any novelty to the study of opfibrational theories. We, nonetheless, chose to briefly discuss them for two reasons. First, fibrational theories whose models have $\Fim(\cat X^{op})^{op}$ as the base provide string diagrams for monoids in $(\Fib_{\mathsf{sp}}(\cat X),\boxtimes,\one)$, and therefore, for functors $\cat X^{op}\rightarrow\MonCat_{\mathsf{st}}$, making an explicit connection with the fibrational version of Theorem~\ref{thm:moeller-and-vasilakopoulou} of Moeller and Vasilakopoulou. Second, next we shall define deflational theories and models as kinds of glueings of an opfibrational and a fibrational theories along the internal terms.

\subsection{Deflational theories}\label{subsec:deflational-theories}

We now arrive to our last -- and most important -- class of theories. Deflational theories can be thought of as glueing a pair of opfibrational and fibrational theories together along the internal terms.

\begin{definition}[Deflational layered theory]\label{def:deflational-theory}
A layered theory is {\em deflational} if its recursive sorting procedure consists of the rules for the symmetry terms of Definition~\ref{def:symmetry-terms}, as well as of the rules for opfibrational and fibrational terms in Figures~\ref{fig:opfibrational-terms} and~\ref{fig:fibrational-terms} (in addition to the basic terms in Figure~\ref{fig:layered-terms}).
\end{definition}
We denote the resulting category of deflational layered theories by $\DeflTh$.

In a deflational theory, every opfibrational term $x:(T\mid S)$ generated by a rule in Figure~\ref{fig:opfibrational-terms} has a corresponding fibrational term $\bar x:(S\mid T)$ generated by the corresponding rule in Figure~\ref{fig:fibrational-terms}. We use this to define the {\em structural 2-cells}, expressing $\bar x$ as the right adjoint to $x$.
\begin{definition}[Structural 2-cells]\label{def:str-twocells}
Given a deflational layered theory with layered signature $\mathcal L$, the {\em structural 2-cells} are given by the choice of 2-cells $\eta_{\mathsf{str}}$ defined as follows: for every opfibrational term $x:(T\mid S)$ generated by a rule in Figure~\ref{fig:opfibrational-terms} one has
\begin{align*}
\eta_{\mathsf{str}}\left(\id_{T},x;\bar x\right) &\coloneq\{\eta_x\} \\
\eta_{\mathsf{str}}\left(\bar x;x,\id_{S}\right) &\coloneq\{\varepsilon_x\},
\end{align*}
and $\eta_{\mathsf{str}}$ returns the empty set otherwise.
\end{definition}
We display the 10 structural 2-cells explicitly in Figure~\ref{fig:structural-twocells-adjoints}: the left column contains the unit 2-cells ($\eta_x$) and the right column contains the counit 2-cells ($\varepsilon_x$). We consistently omit the labels of the structural 2-cells, as the domain and the codomain determine them uniquely.

\begin{figure}
  \centering\small\noindent
  \scalebox{.9}{\input{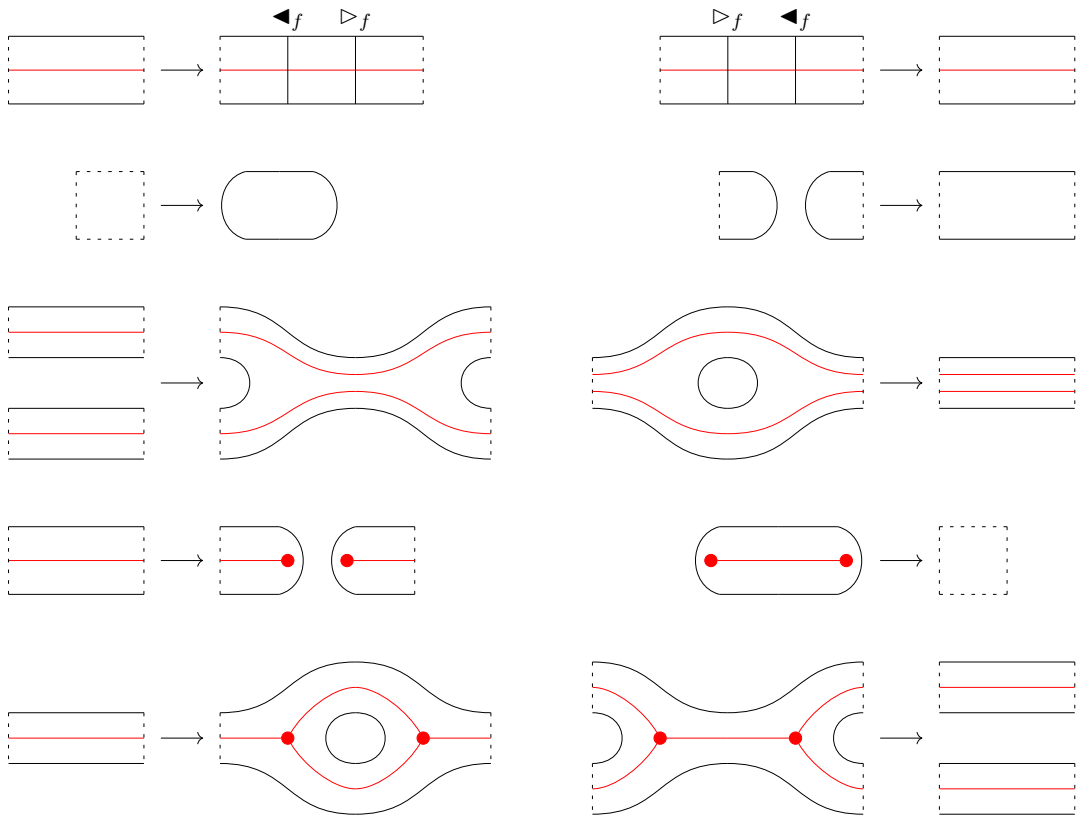}}
  \caption{Structural 2-cells for deflational theories\label{fig:structural-twocells-adjoints}}
\end{figure}

\begin{definition}[Structural deflational equations]\label{def:str-defl-eqns}
By the {\em structural deflational equations} we mean: the structural opfibrational equations (Definition~\ref{def:str-opfib-eqns}), the structural fibrational equations (Definition~\ref{def:str-fib-eqns}) and the following 2-equations for each opfibrational term $x:(T\mid S)$ and all terms $y:(T\mid T)$ and $z:(S\mid S)$ that are either internal or a product (obtained by the~\ref{term:ext-tensor}) of internal terms:
\begin{align*}
\left(\eta_x*\id_x\right);\left(\id_x*\varepsilon_x\right) &=\id_x & \left(\id_{\bar x}*\eta_x\right);\left(\varepsilon_x*\id_{\bar x}\right) &=\id_{\bar x} \\
\id_y*\eta_x &= \eta_x*\id_y & \id_z*\varepsilon_x &= \varepsilon_x*\id_z.
\end{align*}
We denote the structural deflational {1-} and 2-equations by $S_{\mathsf{defl}}^1$ and $S_{\mathsf{defl}}^2$.
\end{definition}
Note that the structural deflational equations contain, in particular, the structural {0-}, {1-} and 2-equations of Definitions~\ref{def:str-0eqns}, \ref{def:str-1eqns} and~\ref{def:str-2eqns}. The additional 2-equations that are required to hold are the usual zigzag equations defining an adjunction (see Definition~\ref{def:zigzag-category} for their diagrammatic depiction), as well as the equations stating the compatibility of the structural 2-cells with the ``sliding through'' 1-equations for internal terms in Figures~\ref{fig:structural-twocells-monoidal} and~\ref{fig:structural-twocells-functors-ext}.

Given a deflational theory $\mathcal T$, let us write $\mathcal T_{\varepsilon}$ for the theory obtained from $\mathcal T$ by adding a section to each counit 2-cell in Figure~\ref{fig:structural-twocells-adjoints}, i.e.~$\kappa$ such that $\kappa;\varepsilon=\id$. While $\mathcal T$ and $\mathcal T_{\varepsilon}$ are, in general, not equivalent as deflational theories, in~\ref{subsec:defl-th-opin-moncat} we will show that they have the same (op)indexed monoidal categories as models. Since these are the main models of interest, in Chapters~\ref{ch:functor-boxes} and~\ref{ch:layered-examples} we will, therefore, make the assumption that the counits in Figure~\ref{fig:structural-twocells-adjoints} have sections $\kappa$, as this will allow us to derive more properties within the formalism of deflational theories.

\chapter{Functor boxes and coboxes}\label{ch:functor-boxes}
{\em Functor boxes} (or {\em functorial boxes}) extend the graphical syntax for monoidal categories (i.e.~string diagrams) to monoidal functors. Given a monoidal functor $F:\cat C\rightarrow\cat D$, the idea is to draw a morphism $f:A\rightarrow B$ in $\cat C$ inside a box:
\begin{center}
\input{functor-box.tikz}
\end{center}
which turns it into a morphism of type $FA\rightarrow FB$ in $\cat D$. Thus, one may treat such boxes as terms composable with morphisms in $\cat D$. This idea extends to (co)lax monoidal functors. We refer the reader to the tutorial by Melli\` es~\cite{functorial-boxes} for the details.

In this chapter, we show that functor boxes arise naturally (and formally) within deflational theories. We will further observe that a box decomposes into the functor boundaries. Within our formalism, a functor box is, therefore, a derived notion rather than a primitive one. This flexibility allows for defining the dual notion, which we dub a {\em cobox} (Definition~\ref{def:cobox}), by composing the refinement and coarsening maps in the opposite order.

While a box can be thought of as acting on its interior, a cobox can be thought of as acting on its exterior. Another intuition is to think of a cobox as a {\em window}, allowing one to peak at the semantics (or at a finer level of granularity) of a part of a diagram.

This chapter also provides the first examples of reasoning {\em within} layered theories. This is in contrast with the previous chapter as well as Chapter~\ref{ch:semantics}, which focus on proving properties {\em of} layered theories.

\section{Functor boxes}\label{sec:functor-boxes}

In any (op)fibrational or deflational theory, boxes are readily provided by the basic terms generated by the rule~\ref{term:int-box}. In deflational theories, boxes decompose into an internal term sandwiched between a coarsening and a refinement map, so that the term can slide out through either side of the box. We make this observation precise in Proposition~\ref{prop:cowindow-box-equality}. We call a decomposed box a {\em cowindow}.

\begin{definition}[Cowindow]\label{def:cowindow}
In a deflational theory, we call the 1-cell of the following form a {\em cowindow}:
\begin{center}
\input{cowindow.tikz},
\end{center}
where $x$ is an internal term.
\end{definition}

The following proposition makes explicit the connection between cowindows and the usual notation for functor boxes.
\begin{proposition}\label{prop:cowindow-box-equality}
In a deflational theory, the following bidirectional 2-cells exist:
\begin{center}
\input{cowindow-box-equality.tikz}.
\end{center}
\end{proposition}
\begin{proof}
The internal term $x$ slides through $\refine_f$ (or $\coarsen_f$) by the first equation in Figure~\ref{fig:structural-twocells-functors-ext}, after which we remove the fibrational-opfibrational generator pair by applying the counit $\varepsilon$ (Figure~\ref{fig:structural-twocells-adjoints}), which has a section $\kappa$.
\end{proof}
Note that in Proposition~\ref{prop:cowindow-box-equality}, the term on the right-hand side is obtained by a rule for basic terms (Figure~\ref{fig:layered-terms}), while the left-hand side is a composition of three terms in a deflational theory. Thus, functor boxes correspond directly to the terms on the right-hand side, while the proposition demonstrates that we can take one step further and decompose a functor box into three constituent parts. This implies that the standard functor equalities (Figure~\ref{fig:structural-twocells-functors-int}) that had to be imposed in fibrational and opfibrational models are, in fact, provable in deflational theories:
\begin{corollary}
In a deflational theory, the internal identities for monoidal functors (Figure~\ref{fig:structural-twocells-functors-int}) are derivable from the external identities for monoidal functors (Figure~\ref{fig:structural-twocells-functors-ext}), up to a bidirectional 2-cell.
\end{corollary}
\begin{proof}
We give the derivation of preservation of monoidal products below; other identities follow similarly:
\begin{center}
\scalebox{.9}{\input{ext-to-int-monoidal-product.tikz}}.
\end{center}
\end{proof}

\begin{corollary}\label{cor:defl-twocells-extend}
In a deflational theory, the 2-cells recursively extend to all internal 1-terms.
\end{corollary}
\begin{proof}
A 2-cell between internal terms $\alpha:x\rightarrow y$ with sort $(A:\omega\mid B:\omega)$ extends to
\begin{center}
\scalebox{.9}{\input{twocells-extend-internal-terms.tikz}}
\end{center}
for any $f\in\mathcal F(\omega,\tau)$. The extension to internal monoidal terms is similar, and uses the counit for the monoid (and its section).
\end{proof}

\section{Functor coboxes}\label{sec:functor-coboxes}

In the definition of a cowindow (Definition~\ref{def:cowindow}), there was no inherent reason to start with coarsening and end with refinement. The reason we chose this order of composition is to obtain the usual functor box: the internal term $x$ is in the domain of $f$, while the whole term is in the codomain. Reversing the order of composition leads to the dual notion, a {\em cobox}, which, to the best knowledge of the author, has hitherto not been studied abstractly, and has only appeared as a notational shorthand in string diagrammatic electrical circuit theory~\cite{electrical-circuits,boisseau-thesis}, which we make explicit as one of our case studies in Section~\ref{sec:elec-circuits}.

Unlike a box, a cobox does not correspond to any internal term, that is, there is no analogue of Proposition~\ref{prop:cowindow-box-equality} for coboxes. Nonetheless, in many respects a cobox behaves as if it was an internal term, as we demonstrate in Proposition~\ref{prop:cobox-slides}. To do that, we need to separate the notion of a window (Definition~\ref{def:window}) from that of a cobox (Definition~\ref{def:cobox}), which allows some wires to bypass the window.

\begin{definition}[Window]\label{def:window}
In a deflational model, we call the 1-cell of the following form a {\em window}:
\begin{center}
\input{window.tikz},
\end{center}
where $x$ is an internal term.
\end{definition}

A the level of syntax, a cobox is obtained by dualising the rule~\ref{term:int-box} for generating boxes. We then impose a 1-equation expressing the cobox as an existing term.
\begin{definition}[Cobox]\label{def:cobox}
Given a deflational theory, a {\em cobox} is a term generated by the rule
\begin{center}
  \begin{prooftree}
    \AxiomC{$f\in\mathcal F(\omega,\tau)$}
    \AxiomC{$A,B,C,D:\omega$}
    \AxiomC{$\scalebox{.9}{\input{internalxdiag.tikz}} : (f(A):\tau\mid f(B):\tau)$}
    \RightLabel{\customlabel{term:cobox}{(cobox)}\;\;}
    \TrinaryInfC{$\scalebox{.9}{\input{f-cobox.tikz}} : (CAD:\omega \mid CBD:\omega)$}
  \end{prooftree}
\end{center}
subject to the following 1-equations:
\begin{center}
\input{cobox-def.tikz}.
\end{center}
\end{definition}
\begin{definition}[Symmetric deflational theory]
A deflational theory is {\em symmetric} if it contains the 1-equations for a symmetric monoidal theory (Definition~\ref{def:symm-mon-thy}) for all internal terms for all layers, as well as the following isomorphic generating 2-cells:
\begin{center}
\input{symmetric-layered-theory.tikz},
\end{center}
together with the dual 2-cells obtained by horizontally reflecting the terms above.
\end{definition}
\begin{proposition}\label{prop:cobox-slides}
In a symmetric deflational theory, the following equations hold up to bidirectional 2-cells:
\begin{center}
\scalebox{.9}{\input{cobox-naturality.tikz}}.
\end{center}
\end{proposition}
Note that a cobox is {\em not} an internal term. However, we may take ``fibrewise monoidal products'' with the cobox, as well as ``apply the box to the cobox'' via the following identifications:
\begin{center}
\scalebox{1}{\input{cobox-products-box.tikz}}.
\end{center}
With Proposition~\ref{prop:cobox-slides} and the above identifications, the cobox may, in most cases, be treated as an internal term. One should, however, be aware that in reality it is not: for example, the cobox is, in general, not in the fibre of its domain category. One should also beware that while some identities that hold for boxes also hold for coboxes, in general, the box identities will not hold. We give some examples of this below.
\begin{proposition}
The following 2-cells are derivable in any deflational theory:
\begin{center}
\scalebox{1}{\input{cobox-derivable-twocells.tikz}}.
\end{center}
There is, in general, no 2-cell in the other direction.
\end{proposition}
\begin{proposition}
The following bidirectional 2-cell is derivable in any deflational theory:
\begin{center}
\scalebox{1}{\input{cobox-composition.tikz}}.
\end{center}
\end{proposition}
Note that since the cobox is not an internal term, this equality does not imply the general composition of coboxes (with the contextual wires present).

The following is the most important property of coboxes we will use: it allows to detect when the functor the cobox is representing is faithful (up to bidirectional 2-cells).
\begin{proposition}\label{prop:cobox-iff-faithful}
There are bidirectional 2-cells
\begin{equation}\label{eq:cobox-pres-id}
\scalebox{1}{\input{cobox-pres-id.tikz}}
\end{equation}
in a deflational theory if and only if for all terms $x$ and $y$ with the sort $(A:\omega\mid B:\omega)$, the existence of 2-cells on the right implies the existence of 2-cells on the left:
\begin{equation}\label{eq:box-term-faithfulness}
\scalebox{1}{\input{box-term-faithfulness.tikz}}.
\end{equation}
Moreover, if~\eqref{eq:cobox-pres-id} is an equality, then the implication~\eqref{eq:box-term-faithfulness} holds also when the bidirectional 2-cells on both sides are replaced with equalities (the converse does not hold in general).
\end{proposition}
\begin{proof}
Suppose the 2-cells~\eqref{eq:cobox-pres-id} exist, and let $x$ and $y$ be terms such that $fx\rightleftarrows fy$. We then compute as follows:
\begin{center}
\scalebox{1}{\input{box-term-faithfulness-proof-1.tikz}}.
\end{center}
It is then clear that the same argument shows the equational case, upon replacing the bidirectional 2-cells with equalities in~\eqref{eq:cobox-pres-id} and in the antecedent of~\eqref{eq:box-term-faithfulness}.

Conversely, suppose that the implication (as stated with the bidirectional 2-cells) holds for all terms. We observe that
\begin{center}
\scalebox{1}{\input{box-term-faithfulness-proof-2.tikz}},
\end{center}
whence existence of 2-cells~\eqref{eq:cobox-pres-id} follows.
\end{proof}

\chapter{Case studies}\label{ch:layered-examples}
Here we give extended examples of layered monoidal theories, boxes and coboxes. The examples with digital circuits (Section~\ref{sec:digital-circuits}), electrical circuits (Section~\ref{sec:elec-circuits}), ZX-calculus (Section~\ref{sec:zx-extraction}), and probabilistic channels (Section~\ref{sec:prob-channels}) all build on monoidal theories from existing literature. We discuss how several levels of monoidal structure are used within existing work, and show how layered monoidal theories explicate and formalise such levels. In the remaining two examples, the calculus of communicating systems (Section~\ref{sec:ccs}) and chemical reactions (Section~\ref{sec:glucose}), the monoidal theories are constructed here for the first time in order to allow for a layered perspective. We note that in Chapter~\ref{ch:retrosynthesis} in Part~\ref{part:chemistry} of this thesis, we give a formalisation of chemical synthesis as a layered monoidal theory, which is closely related to the example in Section~\ref{sec:glucose}.

The case studies are intended to be self-contained and independent of each other. At the beginning of each case study, we give a minimal introduction required to construct the layered theory at hand, and refer the reader to the literature for more details. It is certainly not necessary to read the case studies in order, nor are the details in any one of them required to understand the others.

\section{Digital circuits}\label{sec:digital-circuits}
We give a very simple example of an {\em arithmetic logic unit} ({ALU}) that is able to perform two operations on tuples of bits, the choice between which is controlled by one bit. To this end, we define the layered monoidal theory of {\em simple arithmetic circuits} with one layer for each $n\in\N_{+}$: the interpretation of logic and arithmetic gates within a layer is that they operate on $n$-bit wires. One benefit of our approach is that the monoidal signatures within each layer are nearly identical (a total of eight generators with arities and coarities in $\{0,1,2\}$), resulting in a rather compact definition. While the example {ALU} is very simple, the theory we define is expressive enough to represent an {ALU} for any computing circuit with arbitrary bitwise logic and arithmetic operations, such as a central processing unit ({CPU}).

In Kaye~\cite[Example~3.20]{kaye-thesis}, the simple arithmetic circuits are introduced as the {\em generalised circuit signature for simple arithmetic circuits}, which is defined as an ordinary monoidal signature with an infinite number of colours (one for each number of bits $n\in\N_{+}$), and an infinite number of generators. The generators represent the same (finite) operations on bits as we define below in~\eqref{eq:sac-generators}: an infinite number of generators is needed as their arities and coarities keep track of the number of incoming and outgoing wires, as well as the number of bits within each wire. In contrast, the layered theory presented here has a finite number of generators (eight) and colours (just a single one) in each layer, but there are now infinitely many layers (again, one for each number of bits).

Consider the following layered signature: $(\N_{+}, \mathcal F, \Sigma_n)$, where for each $n\in\N_{+}$, the set $\mathcal F(n,1)$ has exactly one element, and $\mathcal F$ returns the empty set otherwise. For $n\neq 1$, the monoidal signature $\Sigma_n$ is defined to have a single colour $n$, and the following monoidal generators:
\begin{equation}\label{eq:sac-generators}
\scalebox{.6}{\begin{tikzpicture}
	\begin{pgfonlayer}{nodelayer}
		\node [style=none] (1) at (1, 0) {};
		\node [style=none] (4) at (0.75, 0.5) {$n$};
		\node [style=vertex] (5) at (-1, 0) {};
	\end{pgfonlayer}
	\begin{pgfonlayer}{edgelayer}
		\draw [style=wire] (1.center) to (5);
	\end{pgfonlayer}
\end{tikzpicture}
}, \scalebox{.6}{\input{cofork.tikz}}, \scalebox{.6}{\begin{tikzpicture}
	\begin{pgfonlayer}{nodelayer}
		\node [style=none] (1) at (-1, 0) {};
		\node [style=none] (4) at (-0.75, 0.5) {$n$};
		\node [style=vertex] (5) at (1, 0) {};
	\end{pgfonlayer}
	\begin{pgfonlayer}{edgelayer}
		\draw [style=wire] (1.center) to (5);
	\end{pgfonlayer}
\end{tikzpicture}
}, \scalebox{.6}{\input{fork.tikz}}, \scalebox{.6}{\input{and.tikz}}, \scalebox{.6}{\input{or.tikz}}, \scalebox{.6}{\input{plus.tikz}}, \scalebox{.6}{\input{not.tikz}},
\end{equation}
while for $n=1$, the monoidal signature $\Sigma_1$ is as above, except that we replace the adder (generator labelled with `+') with
\begin{center}
\scalebox{.6}{\input{plus-one.tikz}}.
\end{center}
Note that the label $n$ is added for emphasis: there is just one colour within each layer, and hence the generators are the same for all layers, except for when $n=1$. We refer to this layered signature as {\em simple arithmetic circuits}. The intended interpretation is that the wires in layer $n$ carry an $n$-bit signal, while the above generators modify the signal as follows: introduce a wire with no signal, combine the signal from two wires into a single wire, delete a wire, copy the signal, bitwise {AND}, bitwise {OR}, binary sum, bitwise {NOT}.

Next, we define the deflational theory for the simple arithmetic circuits. For the bitwise operators, the usual expected equalities hold in layer $1$. Additionally, we define the binary sum of two bits as follows:
\begin{center}
\scalebox{.6}{\input{plusdef.tikz}}.
\end{center}
The remaining equations of the layered theory are introduced in Figure~\ref{fig:dig-circ-eqns}, defined by recursion on $n$, so that each $f\in\mathcal F(n,1)$ can be thought of as a functor expressing the logical and arithmetic operations on $n$-bit wires in terms of $1$-bit wires.
\begin{figure}
    \centering
    \scalebox{0.6}{
        \input{digital-thick-wire-eqns.tikz}
    }
    \caption{Layered theory defining the simple arithmetic circuits.\label{fig:dig-circ-eqns}}
\end{figure}

We demonstrate the organisational power of layered theories by recasting Example~3.22 from Kaye~\cite{kaye-thesis} as a term for the theory of simple arithmetic circuits. It represents a simple {ALU} (arithmetic logic unit) operating on four bit wires with a one bit control wire: when the control bit is false, it performs bitwise {AND}, while when the control bit is true, it performs bitwise addition:
\begin{center}
\scalebox{.6}{\input{ALU.tikz}}.
\end{center}
In the above term, the labels $1$ and $4$ refer to the whole enclosed region, corresponding to the layers $n=1$ and $n=4$. The layers are separated by a functor boundary labelled with $\coarsen$, corresponding to what is called a {\em bundler} in Kaye~\cite{kaye-thesis}. The translation from $4$-bit wires to $1$-bit wires can, therefore, be thought of as a bracketing, demarcating the two layers from each other. The formalism, however, allows for more than just separating different layers. Using the equations within this layered theory, we can see that the above term is equal to:
\begin{center}
\scalebox{.6}{\input{ALU-rewritten.tikz}}.
\end{center}
Note that while the functor boundary $\coarsen$ ``acts'' from right to left, the overall logical flow of the diagram is from left to right.

\section{Electrical circuits}\label{sec:elec-circuits}
We formalise the notion of an {\em impedance box}, introduced in the study of compositional electrical circuit theory by Boisseau \& Sobociński~\cite{electrical-circuits} and Boisseau~\cite{boisseau-thesis}. We find that it corresponds to a composition of a box with a cobox in a layered monoidal theory generated by translating electrical circuits to graphical affine algebra (Definition~\ref{def:impedance-box}). Within this deflational theory, we are able to replicate what is called the {\em impedance calculus} in~\cite{electrical-circuits}. To illustrate this, we derive the rule governing the sequential composition of resistors using our formalism. We summarise the layers we define in the table below:
\begin{center}
\begin{tabular}{ c | c }
$\ECirc$ & The electrical circuits \\
\hline
$\GAA$ & Graphical affine algebra: axiomatises affine relations \\
\hline
$\Bip$ & Bipole electrical circuits \\
\hline
$\Imp$ & All terms of $\GAA$ with single input and output
\end{tabular}
\end{center}

Compositional electrical circuit theory~\cite{electrical-circuits,boisseau-thesis,compositional-networks,graphical-affine-algebra} treats the components of electrical circuits as generators in certain monoidal theory. The terms (string diagrams) in the layered theory bear very close resemblance to the electrical circuit diagrams in classical electrical circuit theory, and hence to the physical circuit. As formal mathematics, the diagrams allow for equational reasoning and a semantic interpretation as affine relations, bridging the gap between the physical wiring of the diagrams and proving their properties by the means of calculations. Our starting point is the work of Boisseau and Sobociński~\cite{electrical-circuits}, which introduced impedance boxes as a notational device for simplifying and clarifying proofs: the idea is to allow parts of the electrical circuit diagram to be translated to the semantics (graphical affine algrbra).

Let $\R(x)$ denote the field of fractions of the polynomial ring over the real numbers. We define the layered monoidal theory with the following shape:
\begin{equation}\label{ecirc-layers}
\scalebox{1}{\input{ecirc-layers.tikz}},
\end{equation}
where $\Bip$ is the layer of {\em bipoles}, $\ECirc$ the layer of {\em electrical circuits}, $\GAA_{\R(x)}$ the layer of {\em graphical affine algebra} $\R(x)$, and $\Imp$ is the {\em impedance layer}, consisting of all terms that may appear inside an impedance box (Definition~\ref{def:impedance-box}). We drop the subscript $\R(x)$ for brevity, implicitly assuming that all the parameters in $\GAA$ come from the field $\R(x)$.

Each of the four layers has exactly one colour. Following~\cite{electrical-circuits}, the wires and the generators in $\ECirc$ and $\Bip$ are coloured blue, while in $\GAA$ and $\Imp$ they are coloured black, however, this plays no formal role. The generators in $\GAA$ are given by the generators of {\em graphical affine algebra}:
\begin{equation}\label{gaa-generators}
\scalebox{.8}{\input{gaa-generators.tikz}},
\end{equation}
where $k\in\R(x)$. The generators in $\Bip$ are given by the {\em bipoles}:
\begin{equation}\label{ecirc-bipoles}
\scalebox{.8}{\input{ecirc-bipoles.tikz}},
\end{equation}
where $R,L,C\in\R_{+}$ and $V,I\in\R$. The generators in $\ECirc$ are given by the bipoles as above~\eqref{ecirc-bipoles}, together with the following:
\begin{equation}\label{ecirc-non-bipoles}
\scalebox{.8}{\input{ecirc-non-bipoles.tikz}}.
\end{equation}
Finally, the generators in $\Imp$ are given by all the terms in $\GAA$ with sort $(1,1)$.

The equations of the layered monoidal theory are given by the following:
\begin{itemize}
\item in $\GAA$, the equations of the graphical affine algebra hold: for the generators on the first two lines in~\eqref{gaa-generators}, the equations of {\em interacting Hopf algebras}~\cite{survey-signal-flow,zanasi-thesis} hold, together with three additional equations that make the interpretation of the last generator the relation relating the unique element in the zero dimensional vector space to the unit vector in the one dimensional vector space~\cite{graphical-affine-algebra,electrical-circuits},
\item in $\Imp$, the equations of $\GAA$ hold inside the generators, together with the following additional equations, where the identity on the left-hand side of the right equation is the identity term in $\Imp$, {\em not} the identity in $\GAA$:
\begin{center}
\scalebox{.8}{\input{imp-equations.tikz}},
\end{center}
\item on the colours, the functors $\mathcal I:\ECirc\rightarrow\GAA$ and $W:\Imp\rightarrow\GAA$ are defined by $\bullet\mapsto\bullet\bullet$, so that the resulting map on objects is $n\mapsto 2n$,
\item the functor $B:\Bip\rightarrow\Imp$ is identity on colours,
\item the functor $\Bip\hookrightarrow\ECirc$ is identity on both colours and the generators, so that the resulting map is the inclusion,
\item the functor $\mathcal I:\ECirc\rightarrow\GAA$ is defined by equations in Figure~\ref{fig:ecirc-layered-theory},
\item the functor $B:\Bip\rightarrow\Imp$ is defined by equations in Figure~\ref{fig:ecirc-layered-theory-2},
\item the functor $W:\Imp\rightarrow\GAA$ is defined by the following equation:
\begin{center}
\scalebox{.8}{\input{wrapping-equation.tikz}}.
\end{center}
\end{itemize}

\begin{figure}
    \centering
    \scalebox{0.8}{
        \input{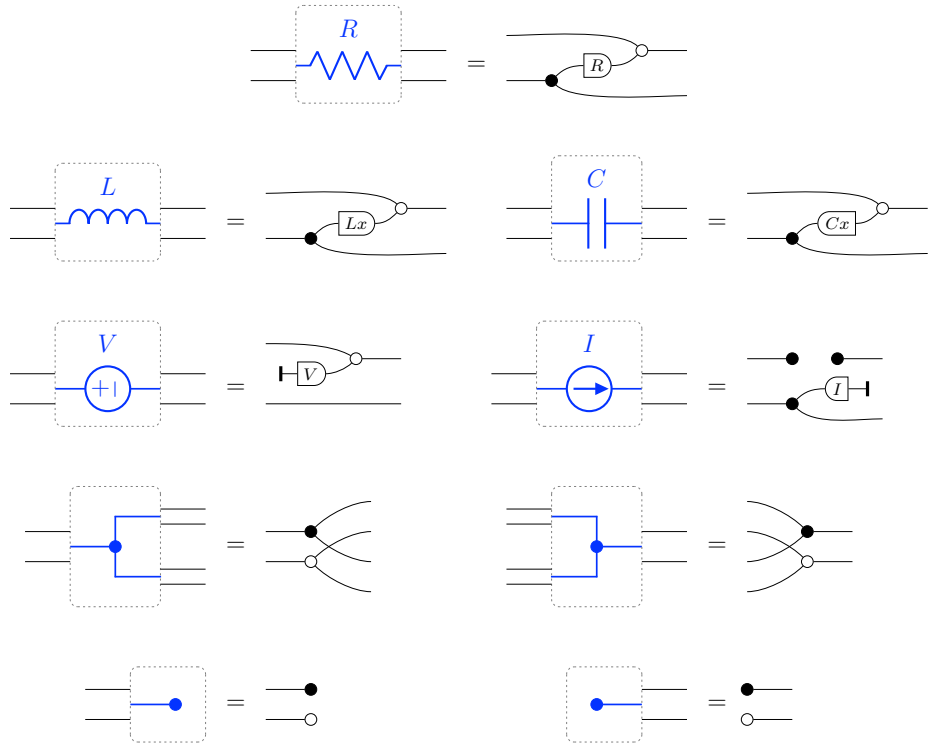}
    }
    \caption{Equations between terms defining the functor $\mathcal I:\ECirc\rightarrow\GAA$\label{fig:ecirc-layered-theory}}
\end{figure}

\begin{figure}
    \centering
    \scalebox{0.8}{
        \input{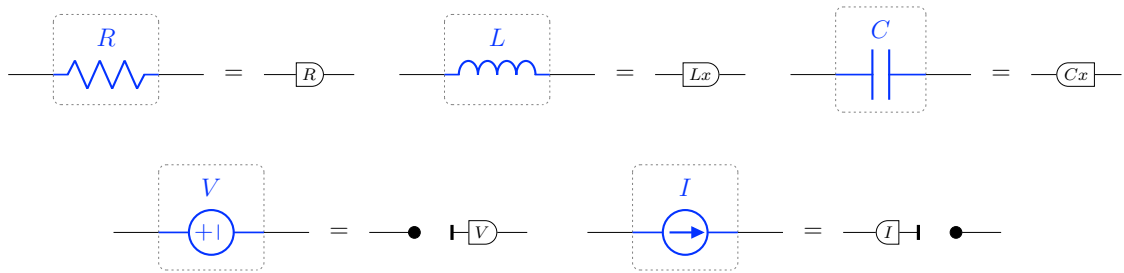}
    }
    \caption{Equations between terms defining the functor $B:\Bip\rightarrow\Imp$\label{fig:ecirc-layered-theory-2}}
\end{figure}

\begin{remark}
The definition of the layer $\Imp$ might appear somewhat strange: it seems that all the terms disconnect due to the identity term disconnecting. This is, however, not the case, as only the identity generator of $\Imp$ disconnects: the wires inside the terms of $\GAA$ remain connected. Note that the composition is defined by the monoidal product of terms that is then connected on both sides to obtain a term with a single input and output. The composition and the identity are well-defined due to the black and white generators of $\GAA$ (and hence ultimately those of interacting Hopf algebras) being (co)associative and (co)unital. For example, the identity indeed works as expected:
\begin{center}
\scalebox{0.8}{\input{imp-identity-example.tikz}}.
\end{center}
\end{remark}

Note that the equations of the layered monoidal theories imply that~\eqref{ecirc-layers} is a commutative square. We use this to define the {\em impedance box} of~\cite{electrical-circuits}.
\begin{definition}[Impedance box]\label{def:impedance-box}
Let $C$ be a generator of $\Imp$ (i.e.~a term with exactly one input and output in $\GAA$). Define the {\em impedance box} as the following box-cobox combination:
\begin{center}
\scalebox{0.8}{\input{impedance-box-definition.tikz}}.
\end{center}
\end{definition}

The impedance box can now be treated like in~\cite{electrical-circuits}, where it is introduced as an additional set of generators for $\ECirc$. However, viewing the impedance box as an emergent feature of the layered monoidal theory reveals some subtleties, as we will see in the following proposition.
\begin{proposition}[Lemma~1 in~\cite{electrical-circuits}]\label{prop:impedance-box-derived}
The following bidirectional 2-cells and equations are derivable:
\begin{center}
\begin{tabularx}{\textwidth}{X X}
\begin{equation}\label{eq:imp-box-i}
\scalebox{.65}{\input{imp-box-compose-seq.tikz}}\tag{i}
\end{equation}%
&
\begin{equation}\label{eq:imp-box-ii}
\scalebox{.65}{\input{imp-box-compose-par.tikz}}\tag{ii}
\end{equation} \\
\begin{equation}\label{eq:imp-box-iii}
\scalebox{.65}{\input{imp-box-snake.tikz}}\tag{iii}
\end{equation}%
&
\begin{equation}\label{eq:imp-box-iv}
\scalebox{.65}{\input{imp-box-units.tikz}}\tag{iv},
\end{equation}
\end{tabularx}
\end{center}
where the superscript $\mathcal I$ in~\eqref{eq:imp-box-i}, \eqref{eq:imp-box-ii} and~\eqref{eq:imp-box-iii} means we need to quotient the terms in $\ECirc$ by equality under the functor $\mathcal I$ in order to derive these equations (the bidirectional 2-cells in~\eqref{eq:imp-box-i} are derivable without any additional equations). In other words, by Proposition~\ref{prop:cobox-iff-faithful}, we need to add the following equation to the layered monoidal theory:
\begin{center}
\scalebox{.8}{\input{i-cobox-id-eqn.tikz}}.
\end{center}
\end{proposition}
The first difference between Proposition~\ref{prop:impedance-box-derived} and Lemma~1 in~\cite{electrical-circuits} is that we are able to derive equation~\eqref{eq:imp-box-iv} and the bidirectional 2-cells in~\eqref{eq:imp-box-i} without quotienting $\ECirc$ by the equality under the translation functor to $\GAA$. Instead, they hold by virtue of the cobox being applied to all the wires on both sides of the equation: compare this to~\eqref{eq:imp-box-iii}, where we have to assume additional equations in order to apply the cobox to the wires that bypass the impedance box on the top and on the bottom on the left-hand side of the equation. The interpretation is that equations~\eqref{eq:imp-box-i}, \eqref{eq:imp-box-ii} and~\eqref{eq:imp-box-iii} require more than existence of a functorial translation $\ECirc\rightarrow\GAA$: they transfer information from the lower level ($\GAA$) to the higher one ($\ECirc$). The second difference is the explicit appearance of a cobox in equation~\eqref{eq:imp-box-iv}, which can be viewed as ensuring that the types on both side of the equality remain in $\ECirc$.

We conclude our discussion of electrical circuits by showing how the impedance box can be used to derive the rule for composing two resistors in $\ECirc$:
\begin{center}
\scalebox{.8}{\input{resistors-high.tikz}},
\end{center}
replicating part~(i) of Proposition~3 of~\cite{electrical-circuits}. As indicated above, this derivation is only possible once we quotient by the equality under the translation functor. The derivation is given below:
\begin{center}
\scalebox{.8}{\input{resistors-add.tikz}}.
\end{center}

\section{ZX-calculus and quantum circuit extraction}\label{sec:zx-extraction}
We show how using rewriting of graphs representing measurement based quantum computations (MBQC) to extract a quantum circuit can be seen as a procedure inside a layered monoidal theory, whose layers are quantum circuits, graphs representing MBQC-computations and the ZX-calculus. We demonstrate the advantage of the layered approach by clearly separating the layers of graph rewriting -- i.e.~formal manipulations of labelled graphs -- from the semantics in terms of the ZX-calculus. The procedure of circuit extraction can then be seen as an interaction between these layers. We summarise the four layers at play in the following table:
\begin{center}
\begin{tabular}{ c | c }
$\ZX$ & The ZX-calculus: describes all linear maps between qubits \\
\hline
$\QCirc$ & Quantum circuits defined as compositions of gates \\
\hline
$\MBQC$ & Graphs representing a measurement based quantum computation \\
\hline
$\MBQCLC$ & An extension of MBQC-graphs convenient for technical reasons
\end{tabular}
\end{center}

The {\em ZX-calculus} is a prop (monoidal theory with a single colour) with the following generators:
\begin{equation}\label{eq:zx-generators}
\scalebox{1}{\input{zx-generators.tikz}},
\end{equation}
called the {\em Z-spider}, the {\em X-spider} and the {\em Hadamard gate}, subject to the following monoidal theory:
\begin{equation}\label{eq:zx-rules}
\scalebox{1}{\input{ZX-rules.tikz}},
\end{equation}
where $\alpha, \beta \in [0, 2 \pi)$ is called the {\em phase} and the addition is modulo $2\pi$. Note that the ellipsis is to be read as `zero or more wires', hence e.g.~on the left-hand side of the spider fusion equation \SpiderRule the spiders are connected by one or more wires. In addition to these generators and equations, the Hadamard gate is often abbreviated to dashed blue edge to avoid clutter:
\begin{center}
\scalebox{1}{\input{hadamard-shorthand.tikz}}.
\end{center}

The ZX-calculus is able to describe any linear map between any finite number of qubits: we discuss this in more detail in Appendix~\ref{ch:zx-mbqc}. While this makes the ZX-calculus highly expressive, in practice one has to restrict available linear maps to those that can be realistically implemented as physical procedures. There are two dominant paradigms to obtain such an operational restriction of quantum maps: {\em circuit based quantum computation} and {\em measurement based quantum computation} (MBQC for short). In the former, the computation is performed by applying {\em gates}, i.e.~certain subset of unitary linear maps, to qubits, and the resulting circuit diagrams look very similar to electrical and digital circuits as encountered in engineering and computer science literature (as well as in Sections~\ref{sec:digital-circuits} and~\ref{sec:elec-circuits} of this work). In the latter, the computation proceeds by applying single qubit destructive measurements to a state consisting of a finite number of qubits whose every pair may be entangled. Both models of quantum computation can be represented within the ZX-calculus. For circuits, this is straightforward by restricting the ZX-diagrams to those that are circuit-like:
\begin{definition}[Quantum circuits]\label{def:quantum-circ}
The prop $\QCirc$ of {\em quantum circuits} is the sub-prop of the ZX-calculus generated by the following morphisms:
\begin{equation}\label{eq:circuit-generators-1}
\scalebox{1}{\input{circuit-generators.tikz}}.
\end{equation}
\end{definition}
Other gates commonly used in circuit based quantum computation are then definable using the generating gates~\eqref{eq:circuit-generators-1}. In particular, we have:
\begin{equation}\label{eq:circuit-generators-2}
\scalebox{1}{\input{circuit-derivable.tikz}}.
\end{equation}

The translation from MBQC protocols (known as {\em measurement patterns}, Definition~\ref{def:measurement-pattern}) to the ZX-calculus is more complicated: we give the details in Appendix~\ref{ch:zx-mbqc}. However, the translation of {\em deterministic} patterns (i.e.~those which result in the same linear map at each run) is given by the composition of functors $D\iota:\MBQC\rightarrow\ZX$ that we define as part of the layered theory presented here.

In MBQC, the computation starts with a {\em graph state}: a finite number of qubits, which may be pairwise entangled by applying the CZ gate. This is represented by an {\em undirected graph}, whose vertices represent qubits (Z-spiders with the zero phase) and whose edges represent CZ-gates (i.e. a Hadamard edge between two qubits). Once we add the information about how the qubits are measured, we obtain the notion of an MBQC-graph (formalised in Definition~\ref{def:mbqc-graph}). It turns out that certain graph-theoretic operations, whose definitions we now state, preserve the semantics of measurement patterns.

We denote a graph by $G=(V,E)$, where $V$ is a set of {\em vertices} and $E$ is a binary relation on $V$ that specifies the {\em edges}. A graph is {\em undirected} when $E$ is symmetric. Given a vertex $u\in V$, we denote the set of {\em neighbours} of $u$ by
$$N_G(u)\coloneq \{v\in V : (u,v)\in E\}.$$
The graphs we consider are undirected, in which case there is no distinction between forward and backward neighbours.
\begin{definition}[Local complementation]\label{def:local-comp}
Let $G=(V,E)$ be an undirected irreflexive graph, and let $u\in V$ be a vertex. The {\em local complementation} about $u$ is the graph $G\star u$ defined by
$$G\star u\coloneq (V, E \Delta \{(a,b) : a,b\in N_G(u) \text{ and } a\neq b\}),$$
where $\Delta$ denotes the symmetric difference.
\end{definition}
In other words, $G\star u$ has the same vertices as $G$, and any two neighbours $a$ and $b$ of $u$ are connected in $G\star u$ if and only if they are not connected in $G$. All other edges are the same as in $G$.
\begin{definition}[Pivot]\label{def:pivot}
Let $G=(V,E)$ be an undirected irreflexive graph, and let $(u,v)\in E$. The {\em pivot} about $(u,v)$ is defined as $G\wedge uv\coloneq ((G\star u)\star v)\star u$.
\end{definition}
A pivot about $(u,v)$ results in interchanging $u$ and $v$, and complementing the edges between the three sets $N_G(u)\cap N_G(v)$, $(N_G(u)\cap N_G(v)^c)\setminus\{v\}$ and $(N_G(v)\cap N_G(u)^c)\setminus\{u\}$: any vertices $a$ and $b$ from any two distinct sets are connected in $G\wedge uv$ if and only if they are not connected in $G$. We illustrate the effect of a pivot in the following picture, where $A\coloneq N_G(u)\cap N_G(v)$, $B\coloneq (N_G(u)\cap N_G(v)^c)\setminus\{v\}$, $C\coloneq (N_G(v)\cap N_G(u)^c)\setminus\{u\}$, and crossing lines between two sets indicate complementing the edges:
\begin{center}
\scalebox{1}{\input{pivot.tikz}}.
\end{center}

In order to define an MBQC-graph, we define the following sets and terminology. Let $\mathcal P\coloneq\{\XYplane,\XZplane,\YZplane\}$ denote the set of {\em measurement planes}. Let $\VS$ be a fixed countable set of {\em vertex names}. By an {\em ordered set} we mean a finite list with no repeated elements.
\begin{definition}[MBQC-graph]\label{def:mbqc-graph}
An {\em MBQC-graph} is a tuple $(V,E,I,O,\lambda)$ such that
\begin{itemize}
\item $V\sse\VS$ is a finite set,
\item $(V,E)$ is an undirected irreflexive graph,
\item $I\sse V$ is an ordered set of {\em inputs},
\item $O\sse V$ is an ordered set of {\em outputs},
\item $\lambda : V\setminus O\rightarrow\mathcal P\times [0,2\pi)$ is a {\em measurement labelling function}, assigning a measurement plane and angle to each non-output vertex.
\end{itemize}
\end{definition}
An MBQC-graph represents the ``desired'' execution of an MBQC-computation, where all the measurements have yielded the expected outcome, so that no corrections were needed. Formally, an MBQC-graph carries the same information as the branch of a measurement pattern where all signals are evaluated to $0$, vanishing all the correction commands (see the discussion in Appendix~\ref{ch:zx-mbqc}). Thus, in this scenario, the measurement angles specified by $\lambda$ translate directly to measurements in the ZX-diagram, as depicted in Table~\ref{tab:measurement-planes}.

 \begin{table}
  \centering
  \renewcommand{\arraystretch}{3}
  \begin{tabular}{ c || c | c | c}
   measurement plane & $\XYplane$ & $\XZplane$ & $\YZplane$ \\ \hline
   ZX-diagram & \input{XY-effect-uncorrected.tikz} & \input{XZ-effect-uncorrected.tikz} & \input{YZ-effect-uncorrected.tikz}
  \end{tabular}
  \renewcommand{\arraystretch}{1}
  \caption{Correspondence between the measurement planes and the ZX-diagrams.\label{tab:measurement-planes}}
 \end{table}

In general, different executions of a measurement pattern will yield different linear maps. In practise, one is interested in those patterns where every execution gives the same outcome linear map (up to a global scalar). Such patterns are called {\em deterministic} (Definition~\ref{def:determinism}). Thus, for a deterministic pattern, every execution is equal to the desired one with no correction commands, i.e.~the one represented by an MBQC-graph. It turns out that there is a condition on MBQC-graphs -- called {\em generalised flow}, or {\em gflow} -- that characterises deterministic measurement patterns~\cite{browne-gflow}. We omit the definition here, as it is somewhat lengthy and requires defining additional notation, and we will not need the details. We refer the reader to~\cite{extended-meas-calculus,thereandback} for the definition.

\begin{theorem}[Browne et al.~\cite{browne-gflow}]\label{thm:gflow}
A measurement pattern with no correction commands can be completed to an equivalent strongly, uniformly and stepwise deterministic pattern (by adding signals and corrections) if and only if the MBQC-graph corresponding to the pattern has a gflow.
\end{theorem}
Thus, instead of working with measurement patterns (Definition~\ref{def:measurement-pattern}) with intricate determinism conditions (Definition~\ref{def:determinism}), this characterisation allows us to work with MBQC-graphs that have a gflow.

For the purposes of rewriting, it is convenient to work with a slight generalisation of MBQC-graphs that allow local Clifford operators at each input and output.

Let us denote by $\LC$ the free monoid with the following generators:
$$\left\{\frac{r}{2},r,-\frac{r}{2},\frac{g}{2},g,-\frac{g}{2}\right\}.$$
The monoid $\LC$ captures syntactically the {\em local Clifford operators}, i.e.~X- and Z-spiders with a single input and output whose phase is an integer multiple of $\frac{\pi}{2}$. The translation functor $\LC\rightarrow\ZX$ is defined by the action on the generators in Table~\ref{tab:local-cliffords}.

 \begin{table}
  \centering
  \renewcommand{\arraystretch}{2}
  \begin{tabular}{ c || c | c | c | c | c | c }
   $\LC$ generator & $\frac{r}{2}$ & $r$ & $-\frac{r}{2}$ & $\frac{g}{2}$ & $g$ & $-\frac{g}{2}$ \\ \hline
   ZX-diagram & \input{red-half-pi.tikz} & \input{red-pi.tikz} & \input{minus-red-half-pi.tikz} & \input{green-half-pi.tikz} & \input{green-pi.tikz} & \input{minus-green-half-pi.tikz}
  \end{tabular}
  \renewcommand{\arraystretch}{1}
  \caption{Translation from $\LC$ to ZX.\label{tab:local-cliffords}}
 \end{table}

\begin{definition}[MBQC+LC-graph]\label{def:mbqc-lc-graph}
An {\em MBQC+LC-graph} is a tuple $(V,E,I,O,\lambda,\ell^i,\ell^o)$ where $(V,E,I,O,\lambda)$ is an MBQC-graph, while $\ell^i:I\rightarrow\LC$ and $\ell^o:O\rightarrow\LC$ are {\em input} and {\em output labelling functions}.
\end{definition}
Note that every MBQC-graph can be viewed as an MBQC+LC-graph by choosing all the additional labels (i.e.~values returned by the functions $\ell^i$ and $\ell^o$) to be the empty word. While it seems that we are capturing a larger class of computations by allowing the local Clifford operators at the ends, Proposition~\ref{prop:mbqc-is-mbqc-lc} will show that, in fact, the class of linear maps that is captured by the MBQC+LC-graphs is exactly the same as the one of MBQC-graphs. We say that an MBQC+LC-graph has a gflow if the underlying MBQC-graph has a gflow.

The following definitions combine local complementation (Definition~\ref{def:local-comp}) and pivoting (Definition~\ref{def:pivot}) with the additional data of an MBQC+LC-graph. While the precise details are somewhat technical, the importance lies in the facts that all the operations (1) preserve the semantics of the graph when interpreted as a ZX-diagram (Theorem~\ref{thm:zx-soundness}), and (2) preserve the existence of gflow (Theorem~\ref{thm:pres-gfow}). We point out that the operations add local Clifford operators to inputs and outputs, which is the reason we need to work with MBQC+LC-graphs rather than MBQC-graphs.

\begin{definition}[Local complementation on MBQC+LC-graphs]\label{def:local-comp-mbqc}
Given a vertex name $u\in\VS$, we define the partial function $\star u$ on the set of MBQC+LC-graphs as follows. A graph $G$ is in the domain if $u\in V_G$, and the output graph $G_{\star}$ is defined by the following cases:
\begin{enumerate}
\item if $u\notin I_G$, then, on the vertices and edges we apply the local complementation about $u$: $G_{\star}\coloneq G\star u$, while $I_{\star}\coloneq I_G$, $O_{\star}\coloneq O_G$, and the labelling functions are defined as follows:
\begin{itemize}
\item if $u\in O_G$, let $\ell^o_{\star}(u)\coloneq \frac{r}{2}\cdot\ell^o_G(u)$,
\item if $u\notin O_G$ and $\lambda_G(u)=(\XYplane,\alpha)$, let $\lambda_{\star}(u)\coloneq \left(\XZplane,\frac{\pi}{2}-\alpha\right),$
\item if $u\notin O_G$ and $\lambda_G(u)=(\XZplane,\alpha)$, let $\lambda_{\star}(u)\coloneq \left(\XYplane,\alpha -\frac{\pi}{2}\right),$
\item if $u\notin O_G$ and $\lambda_G(u)=(\YZplane,\alpha)$, let $\lambda_{\star}(u)\coloneq \left(\YZplane,\alpha +\frac{\pi}{2}\right),$
\item if $v\in N_G(u)\cap O_G$, let $\ell^o_{\star}(v)\coloneq -\frac{g}{2}\cdot\ell^o_G(v),$
\item if $v\in N_G(u)\cap O_G^c$ and $\lambda_G(v)=(\XYplane,\alpha)$, let $\lambda_{\star}(v)\coloneq \left(\XYplane, \alpha - \frac{\pi}{2}\right),$
\item if $v\in N_G(u)\cap O_G^c$ and $\lambda_G(v)=(\XZplane,\alpha)$, let $\lambda_{\star}(v)\coloneq \left(\YZplane, \alpha \right),$
\item if $v\in N_G(u)\cap O_G^c$ and $\lambda_G(v)=(\YZplane,\alpha)$, let $\lambda_{\star}(v)\coloneq \left(\XZplane, -\alpha \right),$
\item all other labels are the same as in $G$,
\end{itemize}
\item if $u\in I_G$, we first choose a vertex name $u'\in\VS$ not appearing in $G$, and define the graph $G'$ as follows:
\begin{itemize}
\item $V'\coloneq V\cup\{u'\}$,
\item $E'\coloneq E\cup\{(u,u')\}$,
\item $I'\coloneq I_G[u'/u]$,
\item $O'\coloneq O_G$,
\item $\lambda'(u')\coloneq (\XYplane, 0)$,
\item $(\ell^i)'(u')\coloneq \ell^i_G(u)\cdot\frac{g}{2}\cdot\frac{r}{2}\cdot\frac{g}{2}$,
\item all other labels are the same as in $G$,
\end{itemize}
now, by construction, $u\notin I'$, so that Case~1 applies, whence we define the output graph $G'_{\star}$ using the construction of Case~1.
\end{enumerate}
\end{definition}
Given an MBQC+LC graph $G$ in the domain of $\star u$, we denote the resulting graph by $G\star u$ -- we will point out whether we mean a simple graph or an MBQC+LC-graph, unless the context is clear enough to disambiguate between the two.

\begin{definition}[Pivot on MBQC+LC-graphs]\label{def:pivot-mbqc}
Given two vertex names $u,v\in\VS$, we define the partial function $\wedge uv$ on the set of MBQC+LC-graphs as follows. A graph $G$ is in the domain if $u,v\in V_G$ and $(u,v)\in E_G$, in which case we define $G\wedge uv\coloneq ((G\star u)\star v)\star u$ via three consecutive applications of the partial functions implementing local complementation.
\end{definition}

Certain vertices can also be removed altogether without altering the semantics of the computation. This serves as a starting point for many simplification methods for circuits and MBQC-patterns that use the ZX-calculus.
\begin{definition}[Vertex removal on MBQC+LC-graphs]\label{def:vertex-removal}
Given a vertex name $u\in\VS$, we define the partial function $\setminus u$ as follows. A graph $G$ is in the domain if the following conditions hold:
\begin{itemize}
\item $u\in V_G$ and $u\notin I_G\cup O_G$,
\item $\lambda_G(u)=(P,a\pi)$ with $P\in\{\YZplane,\XZplane\}$ and $a\in\{0,1\}$,
\end{itemize}
in which case we define the output graph $G\setminus u = (V_u,E_u,I_u,O_u,\lambda_u,\ell^i_u,\ell^o_u)$ by letting $V_u\coloneq V_G\setminus\{u\}$, $E_u\coloneq E\setminus\{(u,x) : x\in V_G\}$, $I_u\coloneq I_G$, $O_u\coloneq O_G$, while the labels are defined by the following cases:
\begin{itemize}
\item if $v\in N_G(u)\cap O_G$, let $\ell^o_u(v)\coloneq g^a\cdot\ell^o_G(v)$,
\item if $v\in N_G(u)\cap O_G^c$ and $\lambda_G(v)=(Q,\alpha)$, define $\lambda_u(v)\coloneq (Q, (-1)^a\alpha)$ if $Q\in\{\YZplane,\XZplane\}$, and $\lambda_u(v)\coloneq (\XYplane, \alpha + a\pi)$ if $Q=\XYplane$,
\item if $v\notin N_G(u)$, the labels are the same as in $G$.
\end{itemize}
\end{definition}

Finally, we need to be able to rename vertices.
\begin{definition}[Renaming]\label{def:zx-renaming}
Let $G$ be an MBQC+LC-graph, let $U\sse V_G$ be an ordered subset of vertices of $G$, and let $W\sse\VS$ be an ordered set of vertices such that $|U|=|W|$ and $W\cap V_G=\eset$. We define the MBQC+LC-graph $G[W/U]$ by replacing each vertex in $U$ with a vertex in $W$ in the specified order, keeping the rest of the data as in $G$.
\end{definition}

We now have all the ingredients to define the layered monoidal theory with the following shape:
\begin{center}
\scalebox{1}{\input{zx-layers.tikz}}.
\end{center}
All the layers are props, i.e.~generated by a single colour, while the generators and the equations within the layers are defined as follows:
\begin{itemize}
\item $\ZX$ is the ZX-calculus, with generators~\eqref{eq:zx-generators} and equations~\eqref{eq:zx-rules},
\item $\QCirc$ is generated by the ZX-diagrams corresponding to quantum circuits~\eqref{eq:circuit-generators-1} and~\eqref{eq:circuit-generators-2}, with the equations of the ZX-calculus,
\item a generator in $\MBQCLC$ with type $n\rightarrow m$ is given by an MBQC+LC-graph $G$ such that $|I_G|=n$ and $|O_G|=m$; the 2-cells and the 2-equations are defined in Figure~\ref{fig:twocells-eqns-mbqclc},
\item the layer $\MBQC$ is defined as $\MBQCLC$, with the graphs restricted to MBQC-graphs.
\end{itemize}

\begin{figure}
    \centering
    \scalebox{1}{%
        \input{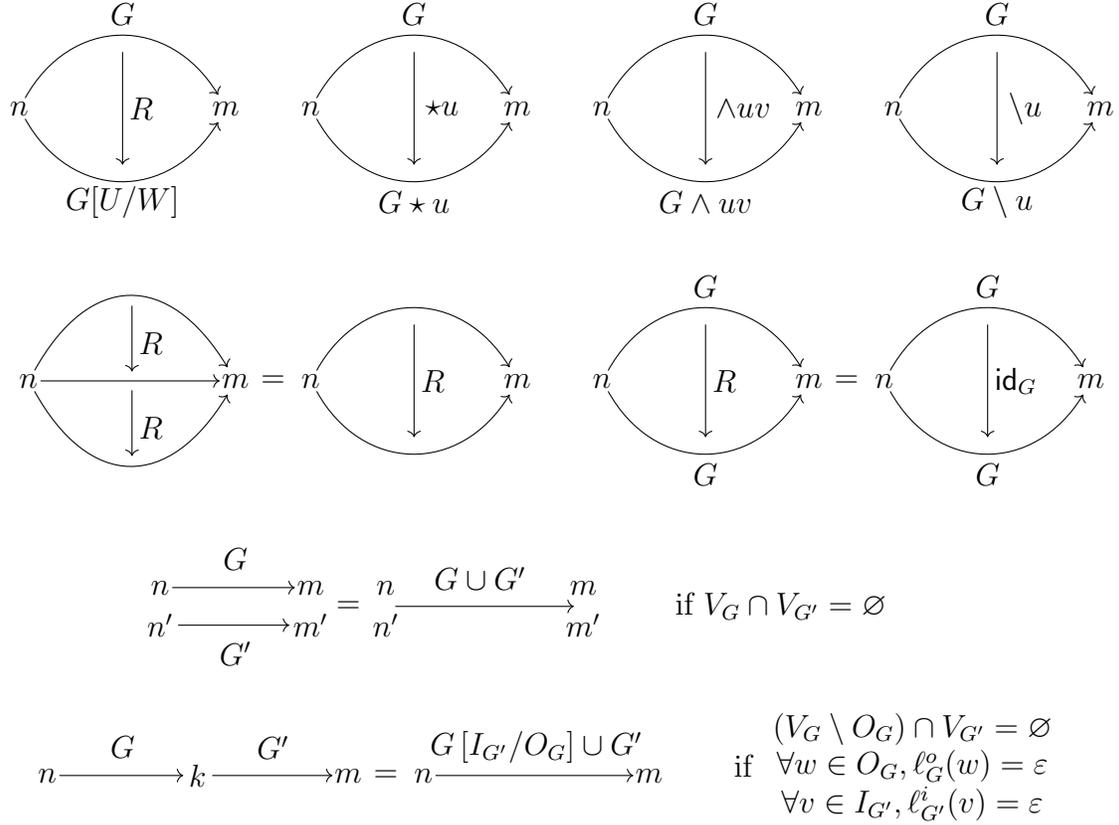}
    }
    \caption{Generating 2-cells and equations in $\MBQCLC$.\label{fig:twocells-eqns-mbqclc}}
\end{figure}

The equations defining the functors make $\QCirc\hookrightarrow\ZX$ and $\iota:\MBQC\hookrightarrow\MBQCLC$ into inclusions, while the functor $D:\MBQCLC\rightarrow\ZX$ is identity on objects and is defined on a generator $G:n\rightarrow m$ as follows:
\begin{enumerate}
\item for each $u\in V_G$, draw a Z-spider with zero phase,
\item for each $(u,v)\in E_G$, draw a Hadamard gate connecting $u$ and $v$,
\item connect all the vertices in $I_G$ to the $n$ input wires in the specified order,
\item connect all the vertices in $O_G$ to the $m$ output wires in the specified order,
\item to each vertex $u\in V_G\setminus O_G$ with $\lambda_G(P,\alpha)$, attach the measurement angle $\alpha$ in the plane $P$, as specified in Table~\ref{tab:measurement-planes},
\item precompose each input $v\in I_G$ with the translation of the monoid element $\ell^i_G(v)$, as specified in Table~\ref{tab:local-cliffords},
\item postcompose each output $w\in O_G$ with the translation of the monoid element $\ell^o_G(w)$, as specified in Table~\ref{tab:local-cliffords}.
\end{enumerate}
Graphically, the action of $D$ on a generator $G:n\rightarrow m$ is depicted by the following equation between terms (note that the inputs $I$ and the outputs $O$ need not be disjoint):
\begin{center}
\scalebox{.8}{\input{mbqclc-to-zx.tikz}}.
\end{center}

The layered point of view allows for a rather succinct statement of the following result.
\begin{proposition}[Lemma~4.1 in~\cite{thereandback}]\label{prop:mbqc-is-mbqc-lc}
We have $D(\MBQCLC) = D(\iota(\MBQC))$. In other words, every ZX-diagram that is a translation of an MBQC+LC-graph is equal to a ZX-diagram that is a translation of an MBQC-graph.
\end{proposition}

The main observations about the translation $D:\MBQCLC\rightarrow\ZX$ are that all the 2-cells preserve the semantics of the translation as well as the existence of gflow, which we record in the following two theorems. They can be seen as high-level summaries of the results presented in Sections~3.1, 4.2 and~4.3 of~\cite{thereandback}.

\begin{theorem}[Soundness]\label{thm:zx-soundness}
Let $G,G':n\rightarrow m$ be two terms in $\MBQCLC$ such that there is a 2-cell $\eta:G\rightarrow G'$. Then $D(G)=D(G')$ in $\ZX$.
\end{theorem}

\begin{theorem}[Preservation of gflow]\label{thm:pres-gfow}
Let $G,G':n\rightarrow m$ be two terms in $\MBQCLC$ such that there is a 2-cell $\eta:G\rightarrow G'$. Then, if $G$ has a gflow, then so does $G'$.
\end{theorem}

Compositions of the generating 2-cells already allow for some simplification of MBQC+LC-graphs (and hence of the corresponding measurement patterns). For example, by combining the vertex removal 2-cells with the local complementation and pivoting 2-cells, one can show that any non-input Clifford vertices\footnote{Vertices whose measurement angle is an integer multiple of $\frac{\pi}{2}$} can be removed (\cite[Theorem~4.12]{thereandback}). One could, of course, add more 2-cells. For example, we could add 2-cells that transform local Clifford labels (given by the functions $\ell^i$ and $\ell^o$) into measurements, hence transforming an MBQC+LC-graph into an MBQC-graph one step at a time. This, in turn, could be used to prove Proposition~\ref{prop:mbqc-is-mbqc-lc} within the layered theory.

Theorems~\ref{thm:zx-soundness} and~\ref{thm:pres-gfow} give a general pattern on how to add new 2-cells to $\MBQCLC$. Namely, one needs to choose the 2-cells in such a way that the two theorems remain valid. In other words, one has to check that the graph is changed in such a way that both preserves the semantics under the functor $D$ and the existence of gflow.

Likewise, we could add the remaining 2-cells needed for {\em circuit extraction} -- the process of transforming any MBQC-graph with a gflow into a semantically equivalent circuit. If an MBQC-graph has gflow, there is an efficient algorithm for rewriting it step-by-step into an equivalent circuit form. This shows that the image of graphs (MBQC or MBQC+LC) with a gflow is, in fact, contained in $\QCirc$. In fact, the whole algorithm can be presented within a layered theory. For example, the following equation appears as Step~1 of the extraction procedure (\cite[p.~42]{thereandback}):
\begin{center}
\scalebox{1}{\input{example-unfuse-gates.tikz}}.
\end{center}
We conclude this case study by noting that the informal boundary between the ``unextracted'' and ``extracted'' parts of the diagram now becomes a formal term within the layered theory:
\begin{center}
\scalebox{1}{\input{circuit-extraction-in-layers.tikz}},
\end{center}
and each step of the algorithm consists of applying a 2-cell on the left-hand side ($\MBQCLC$) and moving a part of the graph to the right-hand side ($\QCirc$).

\section{Calculus of communicating systems}\label{sec:ccs}
We construct a layered monoidal theory capturing two kinds of semantics for a fragment of the calculus of communicating systems (CCS)~\cite{milner} -- {\em reduction semantics} and {\em labelled transition system} (LTS) semantics -- as well as a functorial translation from a subcategory of the former to the latter. Intuitively, the LTS semantics may be seen as a lower level implementation of the concurrent processes described more coarsely by the reduction semantics. We use the layered monoidal theory to show soundness and completeness of the reduction semantics with respect to a fragment of the LTS semantics (Corollary~\ref{cor:ccs-soundness-completeness}). All the results we obtain here are standard (cf.~Corollaries~\ref{cor:lts-simulation} and~\ref{cor:ccs-soundness-completeness}); however, the treament via monoidal theories and string diagrams is novel.

The calculus of communicating systems~\cite{milner} is widely used to reason about concurrent programs and processes. One of its main features is {\em synchronisation} of two processes whose input and output match, modelled as a {\em reduction} or a {\em handshake}. Here we consider a restricted version of CCS with only action prefixing and parallel composition, and two ways to give semantics to the CCS expressions: the reduction semantics (Definition~\ref{def:reduction-semantics}) operates at a higher level, while the LTS semantics (Definition~\ref{def:lts-semantics}) is somewhat more fine-grained. We loosely follow the diagrams from the talk by Krivine~\cite{krivine-talk}.

Let us fix a set $A$, whose elements are called the {\em action names}. We define the sets $\bar A\coloneqq\{\bar a : a\in A\}$ and $\Act\coloneqq A\cup\bar A\cup\{\tau\}$, the latter called the set of {\em actions}, and $\tau$ called the {\em silent action}. The set of {\em processes} is defined recursively as follows, where $x$ ranges over $\Act$:
\begin{center}
\begin{tabular}{ c c | c | c}
  $P\Coloneqq$ & $0$ & $x.P$ & $P\parallel P$.
\end{tabular}
\end{center}

\begin{definition}[Congruence]
  Define the {\em congruence} as the equivalence relation $\sim$ on the set of processes generated by:
  \begin{center}
  \begin{tabular}{ c c c }
    $P\parallel Q\sim Q\parallel P$, & \quad & $(P\parallel Q)\parallel R\sim P\parallel (Q\parallel R)$, \\ \\
    $0\parallel P\sim P$, & \quad & if $P\sim P'$ and $Q\sim Q'$, then $P\parallel Q\sim P'\parallel Q'$.
  \end{tabular}
\end{center}
\end{definition}

\begin{definition}[Reduction semantics]\label{def:reduction-semantics}
A {\em rewrite rule} in {\em reduction semantics} is an ordered pair of processes, which we write as $P\rightarrow Q$, generated by the following three production rules, where $a\in A$:
\begin{center}
  \begin{prooftree}
    \AxiomC{\phantom{$P\rightarrow Q$}}
    \RightLabel{\customlabel{red-sem:tau}{\texttt s}}
    \UnaryInfC{$\tau.P\rightarrow P$}
    \DisplayProof
    \AxiomC{$P\rightarrow Q$}
    \RightLabel{\customlabel{red-sem:p}{\texttt p}}
    \UnaryInfC{$P\parallel R\rightarrow Q\parallel R$}
    \DisplayProof
    \AxiomC{\phantom{$P\rightarrow Q$}}
    \RightLabel{\customlabel{red-sem:r}{\texttt r}}
    \UnaryInfC{$a.P\parallel\bar a.Q\rightarrow P\parallel Q$}
  \end{prooftree}
  \begin{prooftree}
    \AxiomC{$P\rightarrow Q$}
    \AxiomC{$P\sim P'$}
    \AxiomC{$Q\sim Q'$}
    \RightLabel{\customlabel{red-sem:c}{\texttt c}}
    \TrinaryInfC{$P'\rightarrow Q'$}
  \end{prooftree}
\end{center}
\end{definition}
In other words, the rewrite rules are parallel compositions~\ref{red-sem:p} of either the silent action~\ref{red-sem:tau} or the {\em reduction}~\ref{red-sem:r}, up to the congruence~\ref{red-sem:c}. For instance, we can derive the following rewrite rule, where $a\in A$ and $x\in\Act$:
\begin{equation}\label{eqn:CCS-derivation}
a.0\parallel (x.0\parallel\bar a.0)\rightarrow 0\parallel (x.0\parallel 0).
\end{equation}

\begin{definition}[LTS semantics]\label{def:lts-semantics}
A {\em labelled transition} is a triple $(P,x,Q)$, written as $P\xrightarrow x Q$, where $P$ and $Q$ are processes and $x\in\Act$, generated by the production rules below, where $a\in A$:
\begin{center}
  \begin{prooftree}
    \AxiomC{\phantom{$P\xrightarrow x Q$}}
    \RightLabel{\customlabel{lts-sem:a}{\texttt{a}}}
    \UnaryInfC{$x.P\xrightarrow x P$}
    \DisplayProof
    \AxiomC{$P\xrightarrow x Q$}
    \RightLabel{\customlabel{lts-sem:pl}{\texttt{pl}}}
    \UnaryInfC{$R\parallel P\xrightarrow x R\parallel Q$}
    \DisplayProof
    \AxiomC{$P\xrightarrow x Q$}
    \RightLabel{\customlabel{lts-sem:pr}{\texttt{pr}}}
    \UnaryInfC{$P\parallel R\xrightarrow x Q\parallel R$}
  \end{prooftree}
  \begin{prooftree}
    \AxiomC{$P\xrightarrow a P'$}
    \AxiomC{$Q\xrightarrow{\bar a} Q'$}
    \RightLabel{\customlabel{lts-sem:r1}{\texttt{r1}}}
    \BinaryInfC{$P\parallel Q\xrightarrow{\tau} P'\parallel Q'$}
    \DisplayProof
    \AxiomC{$P\xrightarrow{\bar a} P'$}
    \AxiomC{$Q\xrightarrow a Q'$}
    \RightLabel{\customlabel{lts-sem:r2}{\texttt{r2}}}
    \BinaryInfC{$P\parallel Q\xrightarrow{\tau} P'\parallel Q'$}
  \end{prooftree}
\end{center}
\end{definition}

We are now ready to define the layered monoidal theory for the two semantics of CCS defined above, which we refer to as CCS. CCS has three layers: $\Red$ capturing reduction semantics, $\LTS$ capturing LTS semantics, and $\Comp(\Red)$ capturing those reduction semantics rules where some {\em computation} occurs, that is, the rules without the congruence (i.e.~only the production rules~\ref{red-sem:tau}, \ref{red-sem:r} and~\ref{red-sem:p}). CCS has the following shape:
\begin{center}
\scalebox{1}{\begin{tikzpicture}
	\begin{pgfonlayer}{nodelayer}
		\node [style=objectbox] (6) at (-4, 2) {$\Red$};
		\node [style=objectbox] (7) at (0, -1) {$\Comp(\Red)$};
		\node [style=objectbox] (8) at (4, 2) {$\LTS$};
	\end{pgfonlayer}
	\begin{pgfonlayer}{edgelayer}
		\draw [style=morphism] (7) to (8);
		\draw [style=hookarrowmirror] (7) to (6);
	\end{pgfonlayer}
\end{tikzpicture}
}.
\end{center}

The colours of the monoidal signature $\Red$ are given by the set containing three symbols for every process: $P$, $P\uparrow$ and $P\quest$. We refer to the plain process symbols as the {\em active fragment}, to the process symbols with $\uparrow$ as the {\em silent fragment}, and to the processes with $\quest$  as the {\em subsilent fragment}. The intuition behind the fragments is as follows:
\begin{itemize}
\item for the processes $P$ in the active fragment, no computation has yet occurred -- the only transformations that keep a process inside the active fragment correspond to the congruence,
\item for the processes $P\uparrow$ in the silent fragment, a computation (silent action, reduction or parallel composition with a process in the silent fragment) has occurred, and no further reductions or silent actions are possible
\item the processes $P\quest$ in the subsilent fragment are needed in order to capture the congruence in the silent fragment -- each process in the subsilent fragment needs to be eventually merged with a process in the silent fragment in order to obtain a valid rewrite rule (see Proposition~\ref{prop:red-string-corresp}).
\end{itemize}

The generators of $\Red$ are given by: (1) the following generators capturing the production rules~\ref{red-sem:tau}, \ref{red-sem:r} and~\ref{red-sem:p} of Definition~\ref{def:reduction-semantics}, where $a\in A$:
\begin{equation}\label{generators-red}
\scalebox{.9}{\input{generators-CSS.tikz}},
\end{equation}
and (2) the {\em structural generators}, consisting of the symmetry and the generators in Figure~\ref{fig:ccs-structural-generators}.

\begin{figure}
    \centering
    \scalebox{0.9}{
        \input{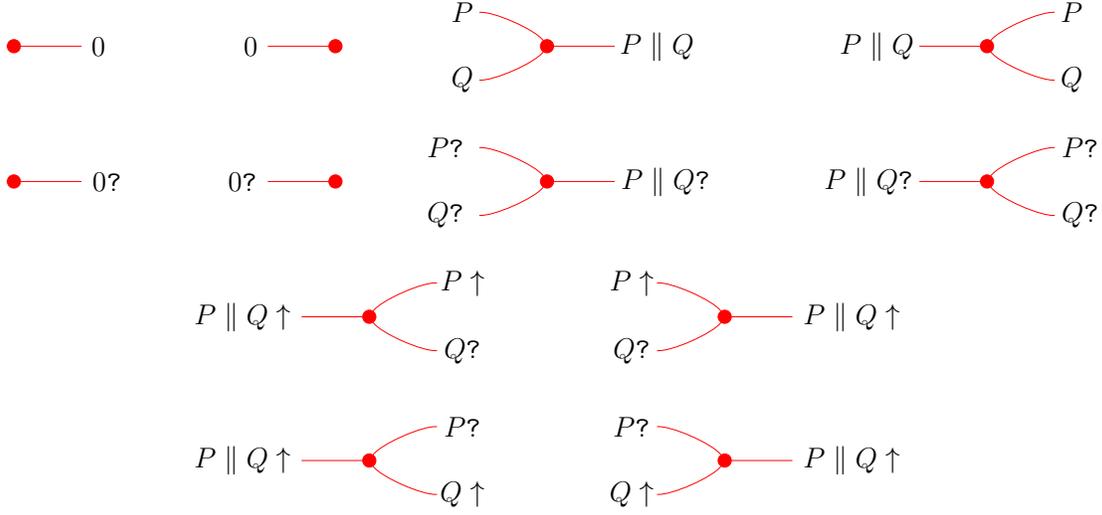}
    }
    \caption{The structural generators of $\Red$.\label{fig:ccs-structural-generators}}
\end{figure}

The monoidal signature $\Comp(\Red)$ is given by restricting the signature $\Red$ to the active and the silent fragments on colours, and the only generators are given by~\eqref{generators-red} and the ``splitting'' structural generator:
\begin{equation}\label{generator-ccs-splitting}
\scalebox{.9}{\input{generator-ccs-splitting.tikz}}.
\end{equation}
Note that, in particular, $\Comp(\Red)$ does not have the symmetry generators.

The monoidal signature $\LTS$ has as colours the symbols $P$ and $P\uparrow x$ for every process $P$ and an action $x\in\Act$. Similarly to $\Red$, we refer to the plain process symbols as the {\em active fragment}. The process symbols with $\uparrow x$ are referred to as the {\em pending fragment}: one may think of $P\uparrow x$ as a process with a ``pending'' action $x$. In this case, we call the {\em silent fragment} the colours of the form $P\uparrow\tau$, i.e.~a pair of a process and the silent action. The generators of $\LTS$ consist of the ``splitting'' structural generator~\eqref{generator-ccs-splitting}, together with the following generators capturing the production rules in Definition~\ref{def:lts-semantics}, where $a\in A$ and $x\in\Act$:
\begin{equation}\label{generators-LTS}
\scalebox{.9}{\input{generators-LTS.tikz}}.
\end{equation}

The equations of the layered theory for CCS are given in Figure~\ref{fig:ccs-layered-theory}, in addition to which the usual identities of a symmetric monoidal theory (Definition~\ref{def:symm-mon-thy}) hold in $\Red$. The reason for calling the generators in Figure~\ref{fig:ccs-structural-generators} {\em structural} is that they capture the congruence by casting it as (non-strict) symmetric monoidal structure. The first set of equations in Figure~\ref{fig:ccs-layered-theory} says that the structural generators in $\Red$ are isomorphisms. One then obtains the congruence relation by considering isomorphisms which have exactly one input and one output wire such that both are in the same fragment (active, silent or subsilent). We make this precise in the following lemma.

\begin{figure}
    \centering
    \scalebox{0.9}{
        \input{ccs-layered-theory.tikz}
    }
    \caption{Layered theory for CCS. The equations in the first four rows hold for any structural generators in $\Red$ for which the composition is defined. The remaining equations define the functor $\Comp(\Red)\rightarrow\LTS$, where the arrows $\mapsto$ denote the 0-equations. The functor $\Comp(\Red)\hookrightarrow\Red$ is defined as the inclusion, whose defining equations we have, therefore, omitted from the figure.\label{fig:ccs-layered-theory}}
\end{figure}

\begin{lemma}\label{lma:iso-congruence}
Let $P$ and $Q$ be processes. The following are equivalent:
\begin{enumerate}[label={(\arabic*)}]
\item $P\sim Q$,
\item there is an isomorphism $P\xrightarrow{\sim}Q$,
\item there is an isomorphism ${P\quest} \xrightarrow{\sim}Q\quest$,
\item there is an isomorphism ${P\uparrow} \xrightarrow{\sim}Q\uparrow$,
\end{enumerate}
where all the isomorphisms are in $\Red$.
\end{lemma}
\begin{proof}
One observes that the terms consisting of only the structural generators whore arity and coarity are both a single colour change the process only up to adding or removing $0$, reassociating the brackets and reordering, i.e.~precisely up to the congruence. For example, $(P\parallel Q)\parallel R \sim P\parallel (Q\parallel R)$ is witnessed by the following isomorphism:
\begin{center}
\scalebox{.9}{\input{ccs-associativity-iso.tikz}}.
\end{center}
Note that there is an apparent choice of two other terms that achieve the same rebracketing, corresponding to choosing which of the processes $P$ and $Q$ remains in the silent fragment (i.e.~whether $P\uparrow$ or $Q\uparrow$ appears in the term). All three terms are, however, equal under the theory CCS in Figure~\ref{fig:ccs-layered-theory}.
\end{proof}

Certain terms in $\Red$ and $\LTS$ correspond to the derivable rewrite rules and transitions in reduction semantics (Definition~\ref{def:reduction-semantics}) and LTS semantics (Definition~\ref{def:lts-semantics}). We begin with the correspondence with reduction semantics.
\begin{proposition}\label{prop:red-string-corresp}
Let $P$ and $Q$ be processes. There is a term in $\Red(P,{Q\uparrow})$ if and only if the rewrite rule $P\rightarrow Q$ is derivable using the production rules of reduction semantics in Definition~\ref{def:reduction-semantics}.
\end{proposition}
\begin{proof}
To show that we have a term whenever a rewrite rule is derivable (the `if' direction), we show that every deduction rule holds as a recursion scheme for terms. First, the silent action rule~\ref{red-sem:tau} holds as we have the generators $\Sgen : {\tau.P}\rightarrow {P\uparrow}$, and the reduction rule~\ref{red-sem:r} holds since we have the terms
\begin{center}
\scalebox{.9}{\input{red-sem-r-string.tikz}}.
\end{center}
For the parallel composition rule~\ref{red-sem:p}, suppose that we have a term $\mathtt t:P\rightarrow {Q\uparrow}$. One then obtains the term as required by the rule:
\begin{center}
\scalebox{.9}{\input{red-sem-p-string.tikz}}.
\end{center}
The congruence rule~\ref{red-sem:c} follows from Lemma~\ref{lma:iso-congruence} by replacing the assumed congruences with isomorphisms in the active and silent fragments.

To show the `only if' direction, i.e.~that every term corresponds to a derivable rewrite rule, we show that every term in $\Red(P,{Q\uparrow})$ is equal to a diagram that can be obtained by the recursion scheme above. Hence suppose that there is a term $\mathtt t:P\rightarrow {Q\uparrow}$. Since there is no way to combine the objects in the silent fragment, we have two cases: (1) $\mathtt t$ contains exactly one instance of the generator $\Sgen : {\tau.R}\rightarrow {R\uparrow}$ and no instances of the generator $\Rgen$, and (2) $\mathtt t$ contains exactly one instance of the generator
\begin{center}
\scalebox{.9}{\input{red-sem-rgen.tikz}},
\end{center}
and no instances of the generator $\Sgen$.

In Case~(1) we can rewrite the term $\mathtt t$ to have the following form:
\begin{equation}\label{red-sem-generic-form-1}
\scalebox{.9}{\input{red-sem-generic-form-1.tikz}},
\end{equation}
where the vertical dashed lines indicate any number of isomorphisms, that is, symmetries and the structural generators (Figure~\ref{fig:ccs-structural-generators}). The ellipsis stands for a sequence of parallel composition generators, of which there are assumed to be $n$: at step $i$, the process $R_i\uparrow$ in the silent fragment is composed with the $i$th identity wire $T_i$, resulting in ${R_i\parallel T_i}\uparrow$, which is isomorphic to ${R_{i+1}\uparrow}$ (taking ${R_{n+1}\uparrow}={Q\uparrow}$). We conclude Case~(1) by observing that the segments between the vertical dashed lines corresponds to an application of a reduction semantics production rule (either~\ref{red-sem:tau} or~\ref{red-sem:p}), and the dashed lines to an application of the congruence rule~\ref{red-sem:c}.

In Case~(2), the term may likewise be rewritten to have the following form:
\begin{equation}\label{red-sem-generic-form-2}
\scalebox{.9}{\input{red-sem-generic-form-2.tikz}},
\end{equation}
where the structure of the isomorphisms and parallel composition is depicted as in Case~(1). We then read off a derivation in reduction semantics exactly as in Case~(1), except that we replace the silent action rule~\ref{red-sem:tau} with the reduction rule~\ref{red-sem:r}.
\end{proof}

For LTS semantics, the connection between derivable labelled transitions and string diagrammatic terms requires a further restriction on the latter, as in this case there are no isomorphisms. For example, one can construct the term
\begin{equation}\label{lts-non-derivable-term}
\scalebox{.9}{\input{lts-non-derivable-term.tikz}},
\end{equation}
but the corresponding labelled transition
$$x.P\parallel (Q\parallel R)\xrightarrow x (P\parallel Q)\parallel R$$
is not derivable in LTS semantics. This is because string diagrammatic representation imports some associativity equations into the terms. The correct term corresponding to the labelled transition
$$x.P\parallel (Q\parallel R)\xrightarrow x P\parallel (Q\parallel R)$$
would be
\begin{equation}\label{lts-correct-derivable-term}
\scalebox{.9}{\input{lts-correct-derivable-term.tikz}}.
\end{equation}
We, therefore, need to restrict to the terms of the following form.
\begin{definition}[Standard terms]
The set of {\em standard} terms in $\LTS$ is recursively generated as follows:
\noindent\paragraph{Base case:} for every process $P$, the following terms are standard:
\begin{center}
\scalebox{.9}{\input{ccs-normal-form-base-case.tikz}},
\end{center}
\noindent\paragraph{Recursive case:} if the terms $\mathtt t$ and $\mathtt s$ are standard, then so is
\begin{center}
\scalebox{.9}{\input{ccs-normal-form-recursive-case.tikz}},
\end{center}
where $\mathtt G\in\Gen\coloneq\{\Pgenl,\Pgenr,\Rgenone,\Rgentwo\}$, whenever the composition is defined.
\end{definition}
As expected, the term~\eqref{lts-correct-derivable-term} is standard, whereas the term~\eqref{lts-non-derivable-term} is not.

\begin{proposition}\label{prop:lts-string-corresp}
Let $P$ and $Q$ be processes, and let $x\in\Act$ be an action. There is a standard term in $\LTS(P,{Q\uparrow x})$ if and only if the labelled transition $P\xrightarrow x Q$ is derivable using the production rules of the LTS semantics in Definition~\ref{def:lts-semantics}.
\end{proposition}
\begin{proof}
By straightforward induction on labelled transitions / standard terms, after observing that there is a one-to-one correspondence between the production rules in Definition~\ref{def:lts-semantics} and the generators~\ref{generators-LTS}.
\end{proof}

A term in $\Comp(\Red)$ is {\em standard} if its translation to $\LTS$ is standard. The following proposition shows the connection between the layer $\Red$ and the standard terms in $\Comp(\Red)$.
\begin{lemma}\label{lma:red-term-to-comp}
Let $\mathtt t\in\Red(P,{Q\uparrow})$ be a term. Then there is a term in the layered theory CCS
\begin{center}
\scalebox{.9}{\input{red-iso-standard.tikz}},
\end{center}
such that the term $t'\in\Comp(\Red)(P',{Q'\uparrow})$ is standard, and $\mathtt i$ and $\mathtt j$ are isomorphisms.
\end{lemma}
\begin{proof}
By the proof of Proposition~\ref{prop:red-string-corresp}, there are two cases: (1) $\mathtt t$ contains exactly one instance of the generator $\Sgen$ and no instances of the generator $\Rgen$, and (2) $\mathtt t$ contains exactly one instance of the generator $\Rgen$ and no instances of the generator $\Sgen$. In Case~(1), existence of the term~\eqref{red-sem-generic-form-1} implies existence of the following term:
\begin{center}
\scalebox{.9}{\input{red-sem-generic-form-1-standard.tikz}},
\end{center}
where, as before, the vertical dashed lines stand for some isomorphisms, and we have $n$ parallel composition generators $\Pgen$. This time, the equations $R_{i+1}=R_i\parallel T_i$ hold strictly for all $i=1,\dots,n-1$, and we still have an isomorphism ${R_n\parallel T_n\uparrow}\simeq {Q\uparrow}$. We observe that the term between the vertical dashed lines is in $\Comp(\Red)$ and is standard, so that it is indeed of the required form. This term will not, in general, be equal to the original term, as they correspond to different derivation trees, but they do derive the same rewrite rule.

Case~(2) is nearly identical, except that we replace the generator $\Sgen$ with $\Rgen$.
\end{proof}

The following proposition shows that the congruence is a simulation relation for the standard terms in $\LTS$.
\begin{proposition}\label{prop:lts-bisimulation}
Let $\mathtt t\in\LTS(P,{Q\uparrow x})$ be a standard term. If there is a process $P'$ with $P\sim P'$, then there is a standard term $\mathtt t'\in\LTS(P',{Q'\uparrow x})$ such that $Q\sim Q'$.
\end{proposition}
\begin{proof}
We argue by induction on construction of the congruence $P\sim P'$.

\noindent\paragraph{Base case~1:} $P={0\parallel R}$ and $P'=R$. Since the term $\mathtt t:{0\parallel R}\rightarrow {Q\uparrow x}$ is standard, it must, in fact, be of the form
\begin{center}
\scalebox{.9}{\input{lts-bisimulation-case1.tikz}},
\end{center}
so that the sought-after term is given by $\mathtt r:R\rightarrow {R'\uparrow x}$. The argument for the case when $P=R$ and $P'={0\parallel R}$ is symmetric.

\noindent\paragraph{Base case~2:} $P=(R\parallel S)\parallel T$ and $P'=R\parallel (S\parallel T)$. There are two cases for the form of the standard term $\mathtt t:(R\parallel S)\parallel T\rightarrow {Q\uparrow x}$. Either it is of the form
\begin{center}
\scalebox{.9}{\input{lts-bisimulation-case2-1.tikz}},
\end{center}
or it is of the form
\begin{center}
\scalebox{.9}{\input{lts-bisimulation-case2-2.tikz}},
\end{center}
where $\mathtt G,\mathtt H\in\Gen$. In the first case, the sought-after term $\mathtt t'$ is simply given by
\begin{center}
\scalebox{.9}{\input{lts-bisimulation-case2-1-new.tikz}}.
\end{center}
The second case has two subcases: first, when $\mathtt G=\mathtt H=\Pgenr$ (and consequently $\mathtt s$ and $\mathtt u$ are identities), and second, when either $\mathtt G\neq\Pgenr$ or $\mathtt H\neq\Pgenr$. In the first subcase, the sought-after term $\mathtt t'$ is given by
\begin{center}
\scalebox{.9}{\input{lts-bisimulation-case2-2-1-new.tikz}}.
\end{center}
For the second subcase, we first define the function $(-',-'):\Gen\times\Gen\rightarrow\Gen\times\Gen$ as follows: if $\mathtt G\neq\Pgenr$, then $(\mathtt G',\mathtt H')\coloneq (\mathtt H,\mathtt G)$, and $(\Pgenr',\mathtt H')\coloneq (\Pgenl,\mathtt H)$. With this notation, taking $(\mathtt G,\mathtt H)$ that appear in the term as the input for the function $(-',-')$, the sought-after term $\mathtt t'$ is given by
\begin{center}
\scalebox{.9}{\input{lts-bisimulation-case2-new.tikz}}.
\end{center}
The argument for the case when $P$ and $P'$ are interchanged is symmetric.

\noindent\paragraph{Base case~3:} $P=R\parallel S$ and $P'=S\parallel R$. In this case, the standard term $\mathtt t:R\parallel S\rightarrow {Q\uparrow x}$ is of the generic form
\begin{center}
\scalebox{.9}{\input{lts-bisimulation-case3.tikz}}.
\end{center}
Let us define $\Pgenl^*\coloneq\Pgenr$, $\Pgenr^*\coloneq\Pgenl$, $\Rgenone^*\coloneq\Rgentwo$ and $\Rgentwo^*\coloneq\Rgenone$. With this notation, the sought-after term $\mathtt t'$ is given by
\begin{center}
\scalebox{.9}{\input{lts-bisimulation-case3-new.tikz}}.
\end{center}

\noindent\paragraph{Inductive case:} $P=R\parallel S$ and $P'=R'\parallel S'$, such that $R\sim R'$ and $S\sim S'$, and the statement holds for these instances of the congruence. As in the previous case, the standard term $\mathtt t:R\parallel S\rightarrow {Q\uparrow x}$ is of the generic form
\begin{center}
\scalebox{.9}{\input{lts-bisimulation-case4.tikz}}.
\end{center}
By the induction hypothesis, there are standard terms $\mathtt r':R'\rightarrow T'$ and $\mathtt s':S'\rightarrow U'$ such that $T'\sim T$ and $U'\sim U$ (note that if the codomain of either $\mathtt r$ or $\mathtt s$ remains in the active fragment, then the term must be the identity). The sought-after term $\mathtt t'$ is thus given by
\begin{center}
\scalebox{.9}{\input{lts-bisimulation-case4-new.tikz}},
\end{center}
thereby completing the induction.
\end{proof}
Translating the above proposition into standard notation (using Proposition~\ref{prop:lts-string-corresp}), we obtain the following.
\begin{corollary}\label{cor:lts-simulation}
The following rule is admissible in the LTS semantics:
\begin{center}
  \begin{prooftree}
    \AxiomC{$P'\sim P$}
    \AxiomC{$P\xrightarrow x Q$}
    \BinaryInfC{$\exists Q'.~{Q'\sim Q}.~{P'\xrightarrow x Q'}$}
  \end{prooftree}
\end{center}
\end{corollary}

The following lemma connects the silent fragments of $\LTS$ and $\Red$. Note that the conclusion of the lemma is exactly the same as the conclusion of Lemma~\ref{lma:red-term-to-comp}, but this time we start from a term in $\LTS$.
\begin{lemma}\label{lma:lts-term-to-comp}
Let $\mathtt t\in\LTS(P,{Q\uparrow\tau})$ be a standard term. Then there is a term in the layered theory CCS
\begin{center}
\scalebox{.9}{\input{red-iso-standard.tikz}},
\end{center}
such that the term $t'\in\Comp(\Red)(P',{Q'\uparrow})$ is standard, and $\mathtt i$ and $\mathtt j$ are isomorphisms.
\end{lemma}
\begin{proof}
There are two cases: either (1) $\mathtt t$ contains the generator $\Agen : \tau.R\rightarrow R\uparrow\tau$ for some process $R$ and all other generators in $\mathtt t$ are instances of the parallel composition generators $\Pgenr$ and $\Pgenl$, or (2) $\mathtt t$ contains exactly one of the generators
\begin{center}
\scalebox{.9}{\input{lts-reduction-gens.tikz}}
\end{center}
for some processes $R$ and $S$ and some $a\in A$, the generators $\Agen : a.R\rightarrow R\uparrow a$ and $\Agen :\bar a.S\rightarrow S\uparrow\bar a$, and all other generators in $\mathtt t$ are instances of the parallel composition generators $\Pgenr$ and $\Pgenl$.

In Case~(1), let us rearrange the processes separated by $\parallel$ in $P$ in such a way that $\tau.R$ is the leftmost process, and the parentheses are associated to the left; let the resulting process be called $P'$:
$$P'\coloneq (\dots((\tau.R\parallel T_1)\parallel T_2)\dots)\parallel T_n.$$
Since $P'\sim P$, Proposition~\ref{prop:lts-bisimulation} produces a standard term $\mathtt t'\in\LTS(P', {Q'\uparrow\tau})$ such that $Q'\sim Q$. Due to the choice of bracketing, only the right parallel composition generators $\Pgenr$ can appear in $\mathtt t'$ (in addition to the generator $\Agen : \tau.R\rightarrow R\uparrow\tau$). Thus, $\mathtt t'$ is in the image of the translation $\Comp(\Red)\rightarrow\LTS$, producing the required standard term in $\Comp(\Red)$ upon replacing $\Agen$ with $\Sgen$ and each $\Pgenr$ with $\Pgen$.

In Case~(2), we similarly rearrange the processes in $P$ so as to obtain the following $P'$, with $P'\sim P$:
$$P'\coloneq (\dots(a.R\parallel\bar a.S)\parallel T_1)\dots)\parallel T_n.$$
The same argument as in Case~(1) then produces the required term.
\end{proof}

\begin{theorem}\label{thm:red-lts-term-iff}
There is a term in $\Red(P,{Q\uparrow})$ if and only if there is a process $Q'$ with $Q\sim Q'$ and a standard term $\LTS(P,{Q'\uparrow\tau})$.
\end{theorem}
\begin{proof}
First, let $\mathtt t\in\Red(P,{Q\uparrow})$ be a term. By Lemma~\ref{lma:red-term-to-comp}, we obtain a standard term $\mathtt t':P'\rightarrow {Q'\uparrow}$ in $\Comp(\Red)$, with $P\sim P'$ and $Q\sim Q'$. Therefore, the translation of $\mathtt t'$ to $\LTS$ gives the standard term $P'\rightarrow {Q'\uparrow\tau}$. By Proposition~\ref{prop:lts-bisimulation}, there is a standard term in $\LTS(P,{Q''\uparrow\tau})$ such that $Q''\sim Q'$, whence also $Q''\sim Q$, as required.

Conversely, let $\mathtt t\in\LTS(P,{Q'\uparrow\tau})$ be a standard term and let $Q\sim Q'$. By Lemma~\ref{lma:lts-term-to-comp}, there is a term
\begin{center}
\scalebox{.9}{\input{red-iso-standard-equiv.tikz}},
\end{center}
so that we obtain the required term in $\Red(P,{Q\uparrow})$ by composing with the isomorphism ${Q'\uparrow}\simeq {Q\uparrow}$.
\end{proof}

In light of Theorem~\ref{thm:red-lts-term-iff}, one can think of the layered theory CCS as extending the (rather restricted) derivations in $\Comp(\Red)$ in two different but equivalent ways: $\Red$ closes all the terms under isomorphism (congruence), while $\LTS$ adds the symmetric counterparts of the parallel composition and reduction generators, as well as permits pending actions. The theorem then shows that $\Red$ and $\LTS$ have the same transitions from the active to the silent fragment. Using Propositions~\ref{prop:red-string-corresp} and~\ref{prop:lts-string-corresp}, we translate this to the usual terminology as follows.
\begin{corollary}[Soundness and completeness]\label{cor:ccs-soundness-completeness}
The reduction semantics is sound and complete with respect to the silent transitions in the LTS semantics: a rewrite rule $P\rightarrow Q$ is derivable if and only if there is a process $Q'$ such that $Q'\sim Q$ and the labelled transition $P\xrightarrow{\tau}Q'$ is derivable.
\end{corollary}

The terms in $\Red(P,{Q\uparrow})$ and $\LTS(P,{Q\uparrow x})$ correspond to the {\em derivations} of $P\rightarrow Q$ and $P\xrightarrow x Q$. Note, however, that the string diagrammatic terms carry slightly more information than the derivation trees, as they record the precise way in which the congruence is constructed.

We conclude this section by using the layered monoidal theory to study different derivations of the rewrite rule~\eqref{eqn:CCS-derivation}: we give its derivation in $\Red$, and two derivations in $\LTS$. The first derivation in $\LTS$ is obtained by first decomposing the term as in Lemma~\ref{lma:red-term-to-comp} and then translating from $\Comp(\Red)$, while the second one is constructed directly and does not corresponding to any derivation in $\Red$. We begin with a derivation in $\Red$ (term on the top):
\begin{center}
\scalebox{.8}{\input{CSS-red.tikz}}.
\end{center}
In the term on the bottom, we have applied the translation $\Comp(\Red)\rightarrow\LTS$, drawn as the cobox. The isomorphisms at each side can be seen as ``preparations'' that ensure that the computations in $\Comp(\Red)$ can occur.

Next, we give the direct derivation as a standard term in $\LTS$:
\begin{center}
\scalebox{.9}{\input{example-term-lts.tikz}}.
\end{center}
Note that this derivation is not in in the image of the translation $\Comp(\Red)\rightarrow\LTS$. However, if we apply to this term the procedure of Lemma~\ref{lma:lts-term-to-comp}, we obtain precisely the term in $\Red$ we started with.

From the point of view of communication, one can think of the LTS semantics as a {\em refinement} of the reduction semantics: while the same synchronisations can be achieved in both, the LTS semantics has more redundancy in their implementation, in the sense that there are multiple places where the passage from the active to the pending fragment can occur.

\section{Chemical reactions}\label{sec:glucose}
Here we construct an example that illustrates how in a layered theory of the shape $A\rightarrow C\leftarrow B$ the terms internal to $B$ and $C$ can be used to construct terms whose types are internal to $A$. We think of $A$ as a {\em high-level} language, of $C$ as a {\em low-level} description with few constraints, and of $B$ as an {\em explanatory} level describing those low-level processes which will actually occur. The fact that we can construct a term between the types in $A$ using only the terms from $B$ and $C$ can be seen, in a sense, as an instance of an emergent phenomenon: there is seemingly a process between the entities in $A$, while in reality all the processes live in $B$ and $C$.

We work inside a model for organic chemistry that has just enough structure to talk about an important biochemical process known as {\em phosphorylation of glucose}. The discussion here is inspired by a talk by Krivine~\cite{krivine-talk}, which, in turn, is motivated by the problem of systematising a vast amount of experimental data in systems biology in a way that is easy for humans to both understand and use~\cite{krivine-siglog}. Our strategy is to define a layered monoidal theory with three layers, each corresponding to viewing chemical change at a different abstraction level:
\begin{center}
\begin{tabular}{ c | c }
$\Name$ & English names of the relevant molecules \\
\hline
$\Scheme$ & Reaction schemes that represent reaction mechanisms \\
\hline
$\Disc$ & Local rewrite rules that capture all possible chemical change
\end{tabular}
\end{center}

Let us fix the following notions needed for Definition~\ref{def:chemlabgraph-layered}:
\begin{itemize}
\item a countable set of {\em vertex names} $\VS$,
\item a finite set of {\em vertex labels} $\Atset$, which contains the special symbol $\alpha$,
\item the set of {\em edge labels} $\Lab\coloneqq\{0,1,2,3,4,\ib\}$.
\end{itemize}
We denote the vertex names by either positive integers or lowercase Latin letters, as appropriate to the situation. While formally the set of vertex labels can contain arbitrary symbols, in what follows we shall assume that $\Atset$ contains a symbol for each main-group element of the periodic table: $\{H,C,O,P,\dots\}\sse\Atset$. For this reason, we will also refer to $\Atset$ as the {\em atom labels}. The special symbol $\alpha$ may be thought of as representing an unpaired electron or a free charge. The integers $\{0,1,2,3,4\}$ in the set of edge labels stand for covalent bonds, while $\ib$ stands for an ionic bond.

\begin{definition}[Chemically labelled graph]\label{def:chemlabgraph-layered}
A {\em chemically labelled graph} is a triple $(V,\tau,m)$, where $V\sse\VS$ is a finite set of {\em vertices}, $\tau:V\rightarrow\Atset\times\Z$ is a {\em vertex labelling function}, and $m:V\times V\rightarrow\Lab$ is an {\em edge labelling function} satisfying $m(v,v)=0$ and $m(v,w)=m(w,v)$ for all $v,w\in V$.
\end{definition}
Thus, a chemically labelled graph is irreflexive (we interpet the edge label $0$ as no edge) and symmetric, and each of its vertices is labelled with an element of $\Atset$, together with an integer indicating the charge.

When drawing chemically labelled graphs, we adopt the following conventions:
\begin{enumerate}[label=(\arabic*)]
\item the vertex label from $\Atset$ is drawn at the centre of a vertex,
\item the vertex name is drawn as a superscript on the left (so within a single graph, no left superscript appears twice),
\item a non-zero charge is drawn as a superscript on the right (hence the lack of a right superscript indicates zero charge)
\item the charge $-1$ is abbreviated as $-$, and similarly the charge $1$ as $+$,
\item $n$-ary covalent bonds are drawn as $n$ parallel lines,
\item ionic bonds are drawn as dashed lines.
\end{enumerate}
\begin{example}\label{ex:phosphate}
The phosphate functional group with one unpaired electron is drawn as a chemically labelled graph below:
\begin{center}
\scalebox{1}{\input{example-phosphate.tikz}}.
\end{center}
\end{example}

A {\em disconnection rule} is a partial endofunction on the set of chemical graphs that only changes one atom or bond (and the atoms it connects). We define four classes of disconnection rules, all of which have a clear chemical significance: two versions of {\em electron detachment}, {\em ionic bond breaking} and {\em covalent bond breaking}. We give an informal definition in Figure~\ref{fig:disc-rules-layered}: the left-hand side of the arrow with the conditions written above define the application conditions, and the output of the function is obtained by replacing the left-hand side with the right-hand side. The full formal definition can be found in Chapter~\ref{ch:disc-rules}.

\begin{figure}
\centering
\scalebox{1}{\input{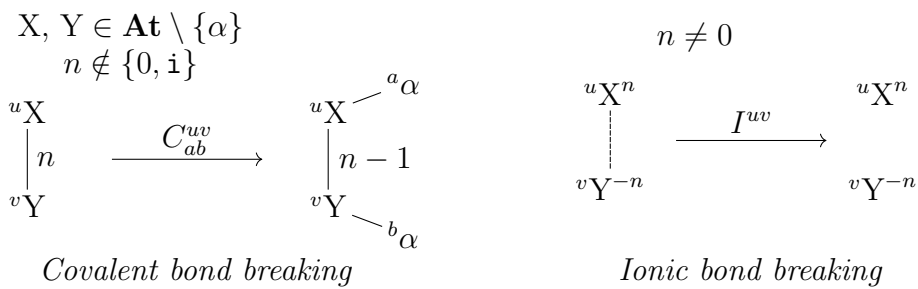}}
\caption{The four disconnection rules.\label{fig:disc-rules-layered}}
\end{figure}

We observe that each disconnection rule is injective (as a partial function), and hence has an inverse partial function. The disconnection rules or their inverses only add or remove the $\alpha$-vertices, hence they are ``matter preserving''. Moreover, valence is preserved locally at each atom, in the sense that the sum of the absolute value of the charge and the number of incident bonds remains unchanged; while the charge is preserved globally, in the sense that the net charge is the same before and after a rule is applied. In Chapter~\ref{ch:disc-rules}, we show that the disconnection rules (together with their inverses) capture all the charge and matter preserving combinatorial rearrangements of chemical graphs, which can be thought of as formal reactions.

We define a layered monoidal theory of the shape
\begin{center}
\scalebox{1}{\begin{tikzpicture}
	\begin{pgfonlayer}{nodelayer}
		\node [style=textbox] (3) at (-4, 2) {$\Name$};
		\node [style=textbox] (4) at (4, 2) {$\Scheme$};
		\node [style=textbox] (5) at (0, -1) {$\Disc$};
	\end{pgfonlayer}
	\begin{pgfonlayer}{edgelayer}
		\draw [style=diredge] (3) to (5);
		\draw [style=diredge] (4) to (5);
	\end{pgfonlayer}
\end{tikzpicture}
}.
\end{center}
There are no generators in $\Name$, whose colours are given by
$$\{\glucose, \ATP, \glucosesix, \ADP, \hydrogenion\}.$$
Here $\ATP$ and $\ADP$ stand for {\em adenosine triphosphate} and {\em adenosine diphosphate}.

For both $\Disc$ and $\Scheme$, the colours are the chemically labelled graphs. The {\em structural generators} of $\Disc$ are the symmetry together with the following:
\begin{center}
\scalebox{.9}{\input{glucose-structural-generators.tikz}},
\end{center}
where $\eset$ is the unique chemically labelled graph on the empty set, $A$ and $B$ have disjoint sets of vertex names and $u(A,B)$ is the union graph, while $r(A)$ is any graph obtained from $A$ by only renaming some vertex names. We assume that the union of two graphs is a symmetric operation: $u(A,B)=u(B,A)$. The structural generators capture the up-to-an-isomorphism notion of the disjoint union of chemically labelled graphs as a non-strict monoidal structure. The ``union'' generators can also be given a chemical interpretation: either the molecules $A$ and $B$ have come close enough ($u(A,B)$) for a reaction to occur, or, dually, a group of molecules $u(A,B)$ splits into spatially separated groups $A$ and $B$, so that reactions can no longer occur between them.

The other generators in $\Disc$ are given by:
\begin{center}
\scalebox{.9}{\input{disc-generators.tikz}},
\end{center}
where $d$ is a disconnection rule (defined as four classes of partial function in Figure~\ref{fig:disc-rules-layered}), $A$ is in the domain of $d$, and $d(A)$ is the result of applying $d$ to $A$.

The only generators of $\Scheme$ are given by the family
\begin{center}
\scalebox{.9}{\input{scheme-generator-phosph.tikz}},
\end{center}
which is universally quantified over all vertex names: that is, the vertices can be renamed arbitrarily, as long as renaming is done simultaneously in the domain and the codomain graphs.

The layered monoidal theory is defined as follows. The structural generators are isomorphisms in $\Disc$: the first four equations in Figure~\ref{fig:ccs-layered-theory} hold whenever the composition is defined and the types match, and similarly, we add the following equations\footnote{Note that the first equation is well-typed since the union is symmetric.}:
\begin{center}
\scalebox{.9}{\input{glucose-structural-eqns.tikz}}.
\end{center}
The functor $\Name\rightarrow\Disc$ is defined in Figure~\ref{fig:translation-name-disc}, while the functor $\Scheme\rightarrow\Disc$ is identity on colours, and is defined on the generators by the following equation:
\begin{center}
\scalebox{.9}{\input{phosph-scheme-as-disc-rules.tikz}}.
\end{center}
While it might seem that the functor $\Scheme\rightarrow\Disc$ is not doing much, as it simply encodes a sequence of disconnection rules in $\Disc$, it will allow us to define where the non-trivial chemical change is happening: namely, we shall require that the only such change happens if a morphism in $\Disc$ is in the image of this functor.

\begin{figure}
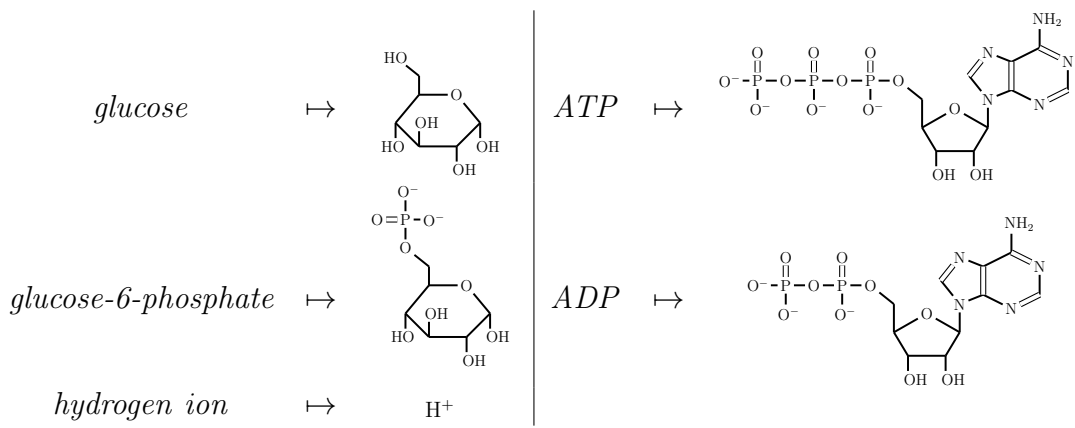

\centering
\begin{tabular}{ c c c | c c c }
$\glucose$ & $\mapsto$ & \scalebox{.5}{\input{glucose.tikz}} & $\ATP$ & $\mapsto$ & \scalebox{.5}{\input{ATP.tikz}} \\
$\glucosesix$ & $\mapsto$ & \scalebox{.5}{\input{glucosesix.tikz}} & $\ADP$ & $\mapsto$ & \scalebox{.5}{\input{ADP.tikz}} \\
$\hydrogenion$ & $\mapsto$ & \scalebox{.7}{\input{H.tikz}} & & &
\end{tabular}
\caption{The functor $\Name\rightarrow\Disc$. We adopt the usual organic chemistry convention that an unlabelled vertex stands for a carbon atom with an appropriate number of hydrogen atoms attached to make the valence add up to $4$. For clarity, we omit the vertex names, as their precise choice is immaterial: the only requirement we impose is that the vertex sets of all five graphs are pairwise disjoint.\label{fig:translation-name-disc}}
\end{figure}

We are now ready to show that this mathematically fairly simple setup allows us to derive the reaction for glucose phosphorylation, which in $\Name$ would be expressed as:
\begin{equation}\label{eqn:phosphorylation}
\glucose+\ATP \longrightarrow \glucosesix+\ADP+\hydrogenion,
\end{equation}
where $+$ denotes the monoidal product. Of course, such reaction cannot exist as a morphism in $\Name$, as it only contains the identity morphisms. However, it is derivable within the layered theory. Note that reaction~\eqref{eqn:phosphorylation} is very close to what one might see as a high-level explanation in a chemistry textbook. We give the derivation below, where we omit the parts of the larger molecules that remain unchanged:
\begin{center}
\scalebox{.8}{\input{phosphorylation-layers.tikz}}.
\end{center}
Note that the above term can be constructed such that the following properties hold: (1) inside the cobox, no other generators of $\Disc$ occur after the renaming generators, (2) the vertex names in the subscript introduced by a single disconnection rule ($ab$ or $cd$) must also be removed by a single inverse rule (the generators with a bar). Importantly, restriction (2) does not apply to the generators appearing inside a box in the translation of $\mathtt{phosph}$. The interpretation is that part of a sequence of disconnection rules is recognised as a translation of a reaction scheme in $\Scheme$. Properties (1) and (2) guarantee that the only non-trivial chemical change occurs within the reaction scheme: the other disconnections must be patched back to restore the original configuration.

Without restrictions (1) and (2), one could always derive a sequence of disconnection rules in $\Disc$ resulting in reaction~\ref{eqn:phosphorylation}, as both sides of the reaction have the same atoms and charge as the ``building material''. The layers $\Disc$ and $\Scheme$ can, therefore, be thought to carry information at distinct epistemic levels: the terms in $\Disc$ are {\em possibilistic}, expressing which chemical change is possible as far as conservation of atoms and charge is concerned, while the terms in $\Scheme$ contain {\em reaction mechanisms}, which can represent empirical knowledge or hypothetical assumptions.

\section{Probabilistic channels}\label{sec:prob-channels}
Here our aim is to formalise the `shaded box' notation introduced by Jacobs~\cite{jacobs-spr} in order to represent {\em conditionals} in probability theory. In particular, we define the conditional of a channel as a certain term (in fact, a box) in a three-layered monoidal theory. In our formalisation of the shaded box notation, we upgrade it in several significant ways, which are discussed in detail in Remark~\ref{rem:shaded-box}. As the first step, we use {\em (co)parametric channels}, which allow part of the output (or input) of a channel as a parameter to be conditioned over. We then define the conditional of a coparametric channel as a certain parametric channel. The original shaded box notation for conditionals fails to be functorial for two reasons: its action on objects depends on whether an object appears as a domain or a codomain of a channel, while on the level of channels, the normalisation results in non-preservation of composition. Our construction remedies the first issue by passing to (co)parametric categories. We further show that the construction is functorial on an important fragment: it preserves composition with the `trivially' coparameterised channels (Lemma~\ref{lma:disintegration-functorial}). We use this ``restricted functoriality'' to prove Proposition~\ref{prop:disintegration-depara}, which gives a diagrammatic proof of uniqueness of conditional channels with marginally full support (Definition~\ref{def:marg-full-support}).

By a {\em distribution} on a set $X$ we mean a function $\omega:X\rightarrow [0,1]$ from $X$ to the unit interval. The {\em support} of a distribution $\omega:X\rightarrow [0,1]$ is the subset $\supp(\omega)\coloneq\{x\in X : \omega(x)\neq 0\}$. We say that a distribution $\omega:X\rightarrow [0,1]$ has {\em full support} if $\supp(\omega) = X$.

Let $\mathcal D:\Set\rightarrow\Set$ be the finite distribution monad, i.e.
$$\mathcal D(X)\coloneq\left\{ \omega : X\rightarrow [0,1] : \supp(\omega) \text{ is finite and } \sum_{x\in X} \omega(x) = 1\right\},$$
and given a function $f:X\rightarrow Y$, the map $\mathcal D(f):\mathcal D(X)\rightarrow\mathcal D(Y)$ is defined by
$$\mathcal D(f)(\omega)(y)\coloneq\sum_{x\in f^{-1}(y)}\omega(x).$$

\begin{definition}[Probabilistic channels]
The category of {\em finite probabilistic channels} $\Chan$\footnote{The morphisms in this category are often called {\em finite stochastic maps / processes / matrices}. We stick with probabilistic channels, often simply saying `channels'.} has the finite sets as objects, while a morphism $X\circarrow Y$ is given by a function $X\rightarrow\mathcal D(Y)$. The identity $X\circarrow X$ is given by the delta distribution: $x\mapsto\delta_x$, while the composition of $f:X\circarrow Y$ and $g:Y\rightarrow Z$ is given by the formula
$$gf(x)(z) = \sum_{y\in Y} f(x)(y)\cdot g(y)(z).$$
\end{definition}
Note that $\Chan$ is the full subcategory on the finite sets of the Kleisli category of the finite distribution modad $\mathcal D$.

We say that a channel $f:X\circarrow Y$ has {\em full support} if for every $x\in X$, the distribution $f(x)$ has full support.

The parametric channels arise via the {\em (co)para construction}, whose construction as a monoidal theory we cover in Appendix~\ref{ch:para-copara}. The construction has appeared several times in the (applied) category theory literature, explicitly under this name in~\cite{fong2019backprop,gavranovic19,cruttwell21}, and it allows morphisms to depend (or be indexed by) objects in the monoidal category at hand.

We refer to the morphisms in $\Para(\Chan)$ as {\em parametric channels}, and to the morphisms in $\Copara(\Chan)$ as {\em coparametric channels}. Explicitly, the parametric channel $(X,f):Z\rightarrow Y$ and the coparametric channel $(V,g):Z\rightarrow Y$ look as follows:
\begin{center}
\scalebox{1}{\input{para-copara-channels.tikz}}.
\end{center}
In the above situation, we call $X$ and $V$ the {\em parameter} and the {\em coparameter}, and say that $f$ is parameterised by $X$ and that $g$ is coparameterised by $V$.

\begin{definition}[Jacobs~\cite{jacobs-spr}, Definition~7.3.1]\label{def:disintegration}
A parametric channel $f':X\times Z\rightarrow Y$ is the {\em conditional} for a coparametric channel $(X,f):Z\rightarrow Y$ if
\begin{equation}\label{eq:disintegration}
\scalebox{1}{\input{disintegration-eqn.tikz}},
\end{equation}
and for any channel $h:Z\rightarrow X$ with full support and any parametric channel $(X,g):Z\rightarrow Y$, the equality on the left implies the equalities on the right:
\begin{center}
\scalebox{1}{\input{disintegration-unique.tikz}}.
\end{center}
\end{definition}

We will solely focus on coparametric channels with marginally full support, as for them the conditional channels exist in the sense of Definition~\ref{def:disintegration} (Proposition~\ref{prop:disintegration-depara}).
\begin{definition}\label{def:marg-full-support}
We say that a morphism $(X,f):Z\rightarrow Y$ in $\Copara(\Chan)$ has {\em marginally full support} if the channel
\begin{center}
\scalebox{1}{\input{marg-full-support.tikz}},
\end{center}
has full support.
\end{definition}
Note that a coparametric channel $(X,f):Z\rightarrow Y$ has marginally full support if and only if for all $z\in Z$ and $x\in X$ we have that the sum $\sum_{y\in Y}f(z)(x,y)$ is non-zero. Let us denote by $\Mfs$ those coparametric channels in $\Copara(\Chan)$ that have marginally full support.

Define the {\em conditional box} $B:\Mfs\rightarrow\ParaSwap(\Chan)$ by mapping $(X,f):Z\rightarrow Y$ to $(X,B(f))$, where $B(f):X\times Z\circarrow Y$ is the channel defined by
$$B(f)(x,z)(y)\coloneq\frac{f(z)(x,y)}{\mathcal N_f(z,x)},$$
where $\mathcal N_f(z,x)\coloneq\sum_{y\in Y}f(z)(x,y)$.

We thus have the following diagram, where the unlabelled arrows are the faithful embeddings mapping each morphism to itself (co)parameterised by the monoidal unit:
\begin{center}
\scalebox{1}{\input{b-embeddings.tikz}}.
\end{center}
Note that $B$ is {\em not} a functor -- hence we cannot assume the functoriality equations of a layered theory (Figures~\ref{fig:structural-twocells-functors-int} and~\ref{fig:structural-twocells-functors-ext}). However, since $B$ preserves the domain and the codomain, we may still use the terms generated by treating the above diagram as a layered signature to reason about the properties of conditionals. Moreover, Lemma~\ref{lma:disintegration-functorial} shows that $B$ is functorial on (the superset of) the embedding of $\Chan$, and that the above diagram commutes.

Graphically, the action of $B$ is denoted as follows:
\begin{center}
\scalebox{1}{\input{disintegration-functor.tikz}}.
\end{center}
Note that we omit the subscript (label) $B$ as there is only one translation involved.

\begin{remark}[A note on deparameterisation]\label{rem:deparameterisation}
Since applying the conditional box to a coparametric channel yields a parametric channel, the result admits deparameterisation, as defined in Appendix~\ref{ch:para-copara}, denoted by
\begin{center}
\scalebox{1}{\input{disintegration-channel.tikz}}.
\end{center}
The above notation makes explicit the computational steps taken to obtain this channel: one starts with a {\em bona fide} channel $f:Z\rightarrow X\times Y$, one then specifies which part of the codomain is to be regarded as the coparameter, thus passing to the coparametric channel $(X,f):Z\rightarrow Y$, after which one applies the conditional box, obtaining the parametric channel $(X,Bf):Z\rightarrow Y$, finally, one views the parametric channel as the honest channel $Bf:X\times Z\rightarrow Y$. Proposition~\ref{prop:disintegration-depara} that we are about to prove shows that the deparameterisation of a conditional box gives the conditional of the original channel (Definition~\ref{def:disintegration}).
\end{remark}

\begin{remark}\label{rem:shaded-box}
The graphical depiction of the conditional box has been introduced as an informal notation by Jacobs~\cite[Section~7.3]{jacobs-spr}, where it is called the `shaded box'. However, our notation differs from the shaded box notation in three crucial ways. First, it is defined formally as a certain term in a layered monoidal theory; consequently, inside our shaded box, all the usual equalities of channels hold. Second, we keep a track of the order of the parameters. For example, the second equality in~\cite[Exercise~7.3.7]{jacobs-spr} becomes
\begin{center}
\scalebox{1}{\input{disintegration-nested.tikz}},
\end{center}
making it apparent that there is a swap of order of the parameters on the right-hand side. This equality can then be used to justify ignoring the order. Third, thinking of the coparameter wire as being attached to the right side of the box gives a good visual intuition for the equalities that hold for the shaded box: any morphism that is not attached (so that there are no obstructions to dragging it out) is able to slide out of the box -- see Lemma~\ref{lma:disintegration-functorial} and~\cite[Exercise~7.3.7]{jacobs-spr}. Since the parameter on the left has to have the same type as the coparameter attached to the right side of the box, the intuition of ``bending" the wire backwards as in~\cite{jacobs-spr} is preserved by thinking of the boundary of the box as carrying the type information.
\end{remark}

\begin{lemma}\label{lma:disintegration-functorial}
Let $(1,g):Y\rightarrow W$ be a channel coparameterised by the terminal set $1$. Then $B$ acts as identity on $(1,g)$, in the sense that $B(1,g)=(1,g)$ is the channel parameterised by $1$ whose action is determined by $g$:
\begin{center}
\scalebox{1}{\input{disintegration-id.tikz}}.
\end{center}
Moreover, if $(X,f):Z\rightarrow Y$ is any coparametric channel in $\Mfs$ composable with $(1,g)$, then $B$ preserves the composition $(X,f);(1,g)$:
\begin{center}
\scalebox{1}{\input{disintegration-functorial-wobox.tikz}}.
\end{center}
\end{lemma}

\begin{lemma}\label{lma:fs-id}
Let $h:Z\rightarrow X$ be any channel with full support. We then have
\begin{center}
\scalebox{1}{\input{disintegration-fs-id.tikz}}.
\end{center}
\end{lemma}

\begin{proposition}\label{prop:disintegration-depara}
Every channel in $\Mfs$ admits has the conditional channel, and the unique conditional channel of $(X,f):Z\rightarrow Y$ is given by the deparameterisation of $B(X,f)$:
\begin{center}
\scalebox{1}{\input{disintegration-channel.tikz}}.
\end{center}
\end{proposition}
\begin{proof}
Verifying that equation~\ref{eq:disintegration} holds when we take $f'$ to be the above box is a matter of computing the composite channel. For uniqueness, we give a diagrammatic argument using Lemmas~\ref{lma:disintegration-functorial} and~\ref{lma:fs-id}. Hence suppose $h$ and $g$ are as in Definition~\ref{def:disintegration}, such that their composite is equal to $f$. The first equality then follows straightforwardly:
\begin{center}
\scalebox{1}{\input{disintegration-uniqueness-eq1.tikz}}.
\end{center}
The second equality relies on the aforementioned Lemmas:
\begin{center}
\scalebox{1}{\input{disintegration-uniqueness-eq2.tikz}}.
\end{center}
\end{proof}

We have chosen to present the construction of the conditional box concretely in $\Chan$ for the sake of demonstrating how one can capture existing constructions using layered theories. However, it should be possible to express the conditional box more generally (and more synthetically) in the setting of copy-discard categories. Such a construction would follow existing axiomatisations of normalisation by Lorenz and Tull~\cite{lorenz-tull-causal} and its use for axiomatising conditionals in the discrete case by~\cite{lorenz-tull-causal} and Jacobs, Sz\'eles and Stein~\cite{jacobs-szeles-stein25}.

\chapter{Indexed monoids}\label{ch:indexed-monoidal}
In this chapter, we study categories with indexed monoids (Definition~\ref{def:indexed-monoids}). We ultimately extend this notion to opfibrations with indexed monoids (Definition~\ref{def:opfib-indexed-mon}), a special case of which will be used as a semantics for opfibrational theories in Section~\ref{sec:opfib-models}.

A symmetric monoidal category is denoted by $(\cat C,\otimes,I,\alpha,\lambda,\rho,\sigma)$, where $I$ is the monoidal unit, while the Greek letters denote the natural isomorphisms: $\alpha$ is the associator, $\lambda$ and $\rho$ are, respectively, the left and right unitors, while $\sigma$ is the symmetry. For brevity, we shall leave the natural isomorphisms implicit, and simply write $(\cat C,\otimes,I)$ for a monoidal category. We take the liberty to use string diagrams even when reasoning about non-strict monoidal categories: we trust that the reader will be able to infer the missing natural isomorphisms where necessary.

\section{Categories with indexed monoids}

Here we define categories with indexed monoids, and show that free models of monoidal theories with indexed monoids give the left adjoint to the forgetful functor from the categories with indexed monoids to categories.

The following definition should be compared to Definition~\ref{def:thy-univ-comonoids}, where we defined the monoidal theory with uniform comonoids.
\begin{definition}[Uniform comonoids]
A symmetric monoidal category $(\cat C,\otimes,I)$ has {\em uniform comonoids} if the following obtain:
\begin{itemize}
\item for every object $x\in\Ob(\cat C)$, there are morphisms $d_x: x\rightarrow x\otimes x$ and $e_x:x\rightarrow I$,
\item for every object $x$, the morphism $d_x$ is a comonoid with the counit $e_x$ (up to the natural isomorphisms),
\item $d_I=\rho_I^{-1}$ and $e_I=\id_1$,
\item $d$ and $e$ are natural, i.e.~for every morphism $f:x\rightarrow y$, we have $d_x;(f\otimes f)=f;d_y$ and $f;e_y = e_x$,
\item for all objects $x,y\in\Ob(\cat C)$, we have (up to associators)
$$(d_x\otimes d_y) ; (\id_x\otimes\sigma_{xy}\otimes\id_y)= d_{x\otimes y}.$$
\end{itemize}
\end{definition}

Note that the naturality conditions imply that the uniform comonoids are unique, in the sense that if both $(d,e)$ and $(d',e')$ are uniform comonoids on the same symmetric monoidal category, then $(d,e)=(d',e')$. The following result is known as {\em Fox's theorem}. We refer the reader to Melli\` es~\cite[Section~6.4]{mellies-cat-semantics-linear-logic} for a more detailed discussion, including more refined versions of the equivalence.
\begin{proposition}\label{prop:fox}
A symmetric monoidal category $(\cat C,\otimes,I)$ is cartesian monoidal (i.e.~$\otimes$ is the cartesian product and $I$ is the terminal object) if and only if it has uniform comonoids.
\end{proposition}
\begin{proof}
Suppose $(\cat C,\otimes,I)$ has uniform comonoids $(d,e)$. Given maps $f$ and $g$ as below
\begin{center}
\scalebox{1}{\input{product-diagram.tikz}},
\end{center}
the map $(f,g) : x\rightarrow a\otimes b$ is given by $d_x;(f\otimes g)$, while the projections are given by $\pi_a\coloneq(\id_a\otimes e_b);\rho_a$ and $\pi_b\coloneq(e_a\otimes\id_b);\lambda_b$. It is then immediate that composing $(f,g)$ with the projections recovers the original maps $f$ and $g$. For uniqueness, suppose there is some map $h:x\rightarrow a\otimes b$ such that
\begin{center}
\scalebox{1}{\input{fox-proof-uniqueness-equations.tikz}}.
\end{center}
We then have that $h=(f,g)$ as follows:
\begin{center}
\scalebox{1}{\input{fox-proof-uniqueness.tikz}}.
\end{center}
Finally, observe that $I$ is the terminal object by naturality of $e$.

Conversely, if $(\cat C,\otimes,I)$ is cartesian monoidal, then $d_x$ is the diagonal and $e_x$ is the unique map into the terminal object.
\end{proof}

\begin{definition}[Consistent monoids]
{\em Consistent monoids} in a symmetric monoidal category $(\cat C,\otimes,I)$ consist of the following data:
\begin{itemize}
\item for every object $x\in\Ob(\cat C)$, there are morphisms $m_x: x\otimes x\rightarrow x$ and $u_x:I\rightarrow x$,
\item for every object $x$, the morphism $m_x$ is a monoid with the unit $u_x$ (up to the natural isomorphisms),
\item $m_I=\rho_I$ and $u_I=\id_I$,
\item $u$ is natural, i.e.~for every morphism $f:x\rightarrow y$, we have $u_x ; f = u_y$,
\item for all objects $x,y\in\Ob(\cat C)$, we have (up to associators)
$$(\id_x\otimes\sigma_{yx}\otimes\id_y) ; (m_x\otimes m_y) = m_{x\otimes y}.$$
\end{itemize}
\end{definition}
Similarly to uniform comonoids, the units of consistent monoids are unique: in fact, naturality and $u_I=\id_I$ imply that $I$ is the initial object. Contrariwise, since naturality is not asked of monoid multiplication, they need not be unique, and therefore correspond to a choice of structure rather than being a property of a category. This motivates the following definition.
\begin{definition}[Indexed monoids]\label{def:indexed-monoids}
We call a choice of consistent monoids in a symmetric monoidal category with uniform comonoids {\em indexed monoids}. We denote a {\em category with indexed monoids} by $(\cat C,\otimes,I,m)$.
\end{definition}
Using Fox's theorem (Proposition~\ref{prop:fox}), we restate the definition of a category with indexed monoids as follows.
\begin{proposition}
A symmetric monoidal category $(\cat C,\otimes,I)$ has indexed monoids if
\begin{itemize}
\item it is cartesian monoidal,
\item $I$ is also an initial object (thus a zero object),
\item for every object $x\in\Ob(\cat C)$, there is a monoid $m_x: x\otimes x\rightarrow x$ with unit $u_x:I\rightarrow x$,
\item $m_I=\rho_I$,
\item for all objects $x,y\in\Ob(\cat C)$, we have (up to associators)
$$(\id_x\otimes\sigma_{yx}\otimes\id_y) ; (m_x\otimes m_y) = m_{x\otimes y}.$$
\end{itemize}
\end{proposition}
\begin{example}
Let $(\Vect_k,\oplus,\{0\})$ be the category of vector spaces over a fixed field $k$, with the direct sum as the monoidal product. Then, for every vector space $V$, addition of vectors $+:V\oplus V\rightarrow V$ gives an indexed monoid structure. This example generalises to any semiadditive category.
\end{example}
\begin{example}
Consider the forgetful functor $U:\Mon\rightarrow\Set_*$ from the category of monoids to the category of pointed sets. Define the category $U\mdash\Mon$ to have the same objects as $\Mon$, and the hom-sets as $U\mdash\Mon(M,N)\coloneqq\Set_*(UM,UN)$, i.e.~the maps are required to preserve the monoidal unit, but not the binary operation. The category $U\mdash\Mon$ has a cartesian monoidal structure given by the product of monoids, the zero object is the monoid with one element, and for every monoid $M=(X,\cdot,1)$, we define the morphism $m_M:M\times M\rightarrow M$ by the multiplication of the monoid: $(x,y)\mapsto x\cdot y$. This gives $U\mdash\Mon$ an indexed monoid structure.
\end{example}
Let us denote by $\IMon$ the category whose objects are symmetric monoidal categories with indexed monoids, and whose morphisms are strong monoidal functors $(F,\mu):(\cat C,\otimes,I,m)\rightarrow (\cat D,\otimes,I,m)$ such that for all $x\in\Ob(\cat C)$, we have $\mu_{xx};F(m_x)=m_{Fx}$.
\begin{definition}
Let $\cat I=(\cat C,\otimes,I,m)$ be a category with indexed monoids. A morphism $f:x\rightarrow y$ in $\cat C$ is a {\em monoid homomorphism} if $m_x;f = (f\otimes f) ; m_y$. We denote the set of monoid homomorphisms by $\hom(\cat I)$.
\end{definition}
\begin{proposition}\label{prop:imon-subcat-preserved}
Given a category with indexed monoids, its monoid homomorphisms form a wide monoidal subcategory. Moreover, the monoid homomorphisms are preserved by the morphisms in $\IMon$.
\end{proposition}
\begin{proof}
The identity morphisms are evidently monoid homomorphisms. Closure under composition and monoidal products follows from the defining equation of monoid homomorphisms and the consistency requirement. Preservation of a monoid homomorphism $f:x\rightarrow y$ under a functor $(F,\mu)$ in $\IMon$ amounts to the equation
$$m_{Fx};Ff = (Ff\otimes Ff);m_{Fy},$$
which we obtain as the commutativity of the following diagram:
\begin{center}
\scalebox{1}{\input{monoid-homomorphism-preserved.tikz}},
\end{center}
where the bottom square commutes since $f$ is a monoid homomorphism, and the top square commutes by naturality of the strongator $\mu$.
\end{proof}

We define the functor $\hom:\IMon\rightarrow\Cat$ by mapping each $\cat I$ to $\hom(\cat I)$, and each functor to its restriction to monoid homomorphisms. We conclude this section by showing that the monoid homomorphism functor so defined has a left adjoint, obtained by freely adding indexed monoids.

\begin{definition}[Theory with indexed monoids generated by a category]\label{def:indexed-by-category}
Let $\cat X$ be a small category. The {\em theory with indexed monoids generated by $\cat X$} is the theory with indexed monoids (Definition~\ref{def:thy-indexed-monoids}) $\mathsf{im}(\cat X)=(\Ob(\cat X),\Sigma,\mathcal I)$ whose set of colours is given by the objects of $\cat X$, and for each $a,b\in\Ob(\cat X)$ we have $\Sigma(a,b)\coloneqq\cat X(a,b)$, together with the following additional equations for all $a\in\Ob(\cat X)$ and all $f\in\cat X(a,b)$, $g\in\cat X(b,c)$ and $h\in\cat X(a,c)$ such that $f;g=h$:
\begin{center}
\scalebox{1}{\input{id-eqn.tikz}}.
\end{center}
\end{definition}

\begin{definition}[Free category with indexed monoids]\label{def:free-indexed-cat}
Let $\cat X$ be a small category. The {\em free category with indexed monoids generated by $\cat X$} is the term model (Definition~\ref{def:term-model}) $\Fim(\cat X)$ of the theory with indexed monoids generated by $\cat X$ (Definition~\ref{def:indexed-by-category}).
\end{definition}

\begin{proposition}
In the situation of Definition~\ref{def:free-indexed-cat}, the interpretation functions $i:\Ob(\cat X)\rightarrow\Ob(\cat X)^*$ and $i_{a,b}:\cat X(a,b)\rightarrow \Fim(\cat X)(a,b)$ define a functor
$$i:\cat X\rightarrow \hom\left(\Fim(\cat X)\right).$$
Moreover, the functor $i$ satisfies the following universal property: for every category with indexed monoids $\cat I$ and a functor $G:\cat X\rightarrow\hom\left(\cat I\right)$ in $\Cat$, there is a unique (up to a monoidal natural isomorphism) morphism $\bar G : \Fim(\cat X)\rightarrow\cat I$ in $\IMon$ such that the diagram
\begin{center}
\scalebox{1}{\input{indexed-universal.tikz}}
\end{center}
commutes; thus exhibiting $\Fim$ as the left adjoint to $\hom:\IMon\rightarrow\Cat$.
\end{proposition}
\begin{proof}
Functoriality of $i$ follows from the equations in Definition~\ref{def:indexed-by-category}. The fact that $i_{a,b}(f)=f$ is a monoid homomorphism is part of the definition of a theory with $1\mdash 1$-natural monoids (Definition~\ref{def:thy-11-nat-monoids}).

When defining a morphism $\bar G : \Fim(\cat X)\rightarrow\cat I$ in $\IMon$, the only freedoms we have is choosing where to map the colours (objects of $\cat X$) and the generators that are not part of the indexed monoid structure (the morphisms in $\hom\left(\Fim(\cat X)\right)$), both of which are uniquely (up to the structural isomorphisms of $\cat I$) determined by commutativity of the diagram.
\end{proof}

\section{Opfibrations with indexed monoids}\label{sec:opfibrations-indexed-monoids}

Here we extend the notion of indexed monoids from categories to opfibrations. The central notion is that of an {\em opfibration with indexed monoids}, or {\em im-opfibration}, which is shown to be very closely related to monoids in the category of split opfibrations $\MonOpFib_{\mathsf{sp}}(\cat X)$ (Theorem~\ref{thm:monopfib-imopfib}). We choose to work with opfibrations rather than fibrations, however, all the developments dualise to fibrations, in which the appropriate notion is {\em indexed comonoids} rather than indexed monoids.

\begin{definition}[Im-opfibration]\label{def:opfib-indexed-mon}
We say that an opfibration $\cat Y\rightarrow\cat X$ {\em has indexed monoids} if the base category $\cat X$ has indexed monoids, $\cat Y$ is cartesian monoidal, the functor preserves and reflects the cartesian products, and the cartesian product of opcartesian maps in $\cat Y$ is opcartesian. For the sake of brevity, we refer to an opfibration with indexed monoids as an {\em im-opfibration}.
\end{definition}
Note that in the above definition, $\cat X$ having indexed monoids implies that $\cat X$ is cartesian, so that it makes sense to ask the functor to preserve and reflect the cartesian products.

An im-opfibration is {\em split} if it is a split opfibration and the chosen cartesian products are preserved and reflected strictly.

\begin{definition}[Morphism of im-opfibrations]\label{def:morphism-im-opfib}
A {\em morphism of im-opfibrations} is a morphism of opfibrations $(H,K)$ such that the functor between the base categories $K$ is a morphism in $\IMon$, and the functor between the total categories $H$ preserves the cartesian monoidal structure.
\end{definition}
A morphism of split im-opfibrations is a morphism of split opfibrations such that both functors are strict monoidal. Let us denote the category of im-opfibrations by $\imOpFib$, and the fixed base case by $\imOpFib(\cat X)$. As for (op)fibrations, the subcategory of split im-opfibrations is denoted by adding the subscript $\mathsf{sp}$ in each case.

Given an im-opfibration $\cat Y\rightarrow\cat X$, define its {\em restriction} as the functor $\cat Y^*\rightarrow\hom(\cat X)$, where $\hom(\cat X)$ is the subcategory of monoid homomorphisms, and $\cat Y^*$ is the wide subcategory of $\cat Y^*$ defined by $F\in\cat Y^*$ if and only if the image of $F$ is in $\hom(\cat X)$, while the action of the functor is given by that of the original functor.

\begin{proposition}
For any im-opfibration $\cat Y\rightarrow\cat X$, its restriction $\cat Y^*\rightarrow\hom(\cat X)$ is an opfibration. Moreover, the restriction extends to a functor $\hom:\imOpFib\rightarrow\OpFib$.
\end{proposition}
\begin{proof}
The opcartesian maps are inherited by restriction: $F\in\cat Y^*$ is opcartesian if and only if $F\in\cat Y$ is opcartesian. The opcartesian liftings are likewise given by opcartesian liftings in $\cat Y$. The morphisms restrict to morphisms of opfibrations by Proposition~\ref{prop:imon-subcat-preserved}.
\end{proof}

\begin{theorem}\label{thm:monopfib-imopfib}
Let $\cat X$ be a small category. The category of split opfibrations with monoids with base $\cat X$ is equivalent to the category of split im-opfibrations with base $\Fim(\cat X)$:
$$\MonOpFib_{\mathsf{sp}}(\cat X)\simeq\imOpFib_{\mathsf{sp}}(\Fim(\cat X)).$$
\end{theorem}
\begin{proof}
Given a split opfibration $\cat Y\rightarrow\cat X$ with a monoid
\begin{center}
\scalebox{1}{\input{gr-pseudomonoid.tikz}},
\end{center}
define the split im-opfibration $\cat Y_{\otimes}\rightarrow\Fim(\cat X)$ as follows. The total category $\cat Y_{\otimes}$ is the term model of the following monoidal theory:
\begin{itemize}
\item the theory contains the monoidal theory with uniform comonoids generated by $\cat Y$,
\item for all $x\in\Ob(\cat X)$ and $a,b\in\cat Y_x$, we have the following additional generators:
\begin{center}
\scalebox{1}{\input{im-extension-generators.tikz}},
\end{center}
\item for all $a,b,c\in\Ob(\cat Y)$ in the same fibre, $(\sigma,\tau)\in\cat Y\boxtimes\cat Y$, and $f:x\rightarrow w$ in $\cat X$, the theory contains the following equations:
\begin{center}
\scalebox{1}{\input{im-extension-equations.tikz}}.
\end{center}
\end{itemize}
The action of the functor $\cat Y_{\otimes}\rightarrow \Fim(\cat X)$ is defined by extending $\cat Y\rightarrow\cat X$ to cartesian products, and by mapping the newly added generators to the freely added monoid multiplication and unit in the base. The opcartesian morphisms in $\cat Y_{\otimes}$ are generated as follows:
\begin{itemize}
\item if $f\in\cat Y$ is opcartesian, then so is $f\in\cat Y_{\otimes}$,
\item the newly added generators are opcartesian,
\item if $f$ and $g$ are opcartesian, then so are the composition $f;g$ (whenever defined) and the cartesian product $f\times g$.
\end{itemize}
It is then straightforward to construct the opcartesian liftings. Since the cartesian structure is preserved and reflected by construction, we indeed obtain an im-opfibration.

Conversely, given a split im-opfibration $\cat Y\rightarrow\Fim(\cat X)$, we proceed to define a monoid
\begin{center}
\scalebox{1}{\input{im-opfibration-to-monoid.tikz}}
\end{center}
on the split opfibration $\cat Y^*\rightarrow\cat X$, where $\cat Y^*$ is the wide subcategory of $\cat Y$ obtained by restriction, i.e.~$F\in\cat Y^*$ if and only if the image of $F$ is in $\cat X$. First, observe that, since the cartesian products are strictly preserved and reflected, for all morphisms $f,g\in\cat X$, we have $\cat Y_f\times\cat Y_g\simeq\cat Y_{f\times g}$, and the isomorphism is given (from left to right) by the action of the product functor on $\mathcal Y$. Further, each $f:x\rightarrow y$ in $\cat X$ is a monoid homomorphism by construction of $\Fim(\cat X)$, which induces a function $\cat Y_{f\times f}\rightarrow\cat Y_f$ by the universal property of the opcartesian liftings of the monoid multiplication on $x$. Hence, for every morphism $f\in\cat X$, we define the mapping $\otimes_f:\cat Y_f\times\cat Y_f\rightarrow\cat Y_f$ by composing the isomorphism $\cat Y_f\times\cat Y_f\simeq \cat Y_{f\times f}$ with the function $\cat Y_{f\times f}\rightarrow\cat Y_f$. We now define the functor $\otimes_{\cat Y}:\cat Y^*\boxtimes\cat Y^*\rightarrow\cat Y^*$ as the disjoint union of these mappings:
$$\otimes_{\cat Y}\coloneq\coprod_{f\in\cat X}\otimes_f.$$
Next, once again by strict preservation and reflection of cartesian products, we have that $\cat Y_{I}=\{I\}$, so that the reindexing functor $\one_x:\cat Y_{I}\rightarrow\cat Y_x$ induced by the monoidal unit on $x$ constitutes a choice of an object in $\cat Y_x$, which we denote by $\one(x)$. Moreover, the reindexing functor $f^*:\cat Y_x\rightarrow\cat Y_w$ induced by a morphism $f:x\rightarrow w$ in $\cat X$ preserves these chosen object in the sense that $f^*(\one(x))=\one(w)$. We thus define the unit functor $\one_{\cat Y}:\cat X\rightarrow\cat Y^*$ by $x\mapsto\one(x)$ on objects and by the opcartesian liftings on morphisms. Associativity and unitality are now induced by associativity and unitality of the monoids in $\Fim(\cat X)$, so that $(\cat Y^*,\otimes_{\cat Y},\one_{\cat Y})$ indeed defines a monoid in $(\OpFib_{\mathsf{sp}}(\cat X),\boxtimes,\id_{\cat X})$.

Applying the monoid construction to the im-opfibration obtained from a monoid recovers the original monoid: $\otimes_{\cat Y_{\otimes}}=\otimes$ and $\one_{\cat Y_{\otimes}}=\one$, since the restriction of $\cat Y_{\otimes}$ is $\cat Y$. In the other direction, we recover the original im-opfibration up to an isomorphism: $(\cat Y^*)_{\otimes}\simeq\cat Y$, as the objects and morphisms in $\cat Y$ that are not in $\cat Y^*$ are first removed and then added with potentially new labels.
\end{proof}

We observe that Theorem~\ref{thm:monopfib-imopfib} can be thought of as a kind of internalisation of the monoid on the opfibration. On the left-hand side of the equivalence, the monoid is an externally imposed structure, whereas on the right-hand side, the monoids are absorbed into the internal structure of the base and the total categories. In particular, given an im-opfibration $\cat Y\rightarrow\cat X$ and an object $x\in\Ob(\cat X)$, we denote by $\otimes_x : \cat Y_x\times\cat Y_x\rightarrow\cat Y_x$ the monoidal functor on the fibre $\cat Y_x$ obtained from $\otimes_f$ by setting $f=\id_x$ in the above proof. We refer to this monoidal structure as {\em fibrewise}. It is precisely this internalisation that will make opfibrations with indexed monoids convenient to work with when defining the semantics of layered theories in Section~\ref{sec:opfib-models}.

Combining Theorem~\ref{thm:monopfib-imopfib} with Theorem~\ref{thm:monoids-opfibrations-strict-imoncat-equivalence}, we also have the following:
\begin{corollary}\label{cor:opfibrations-indexed-monoids-opindexed-monoidal}
Let $\cat X$ be a small category. The category of split im-opfibrations with base $\Fim(\cat X)$ is equivalent to the category of strict $\cat X$-opindexed monoidal categories:
$$\imOpFib_{\mathsf{sp}}\left(\Fim(\cat X)\right)\simeq\OpIMonCat_{\mathsf{st}}(\cat X).$$
\end{corollary}

The reader may wonder what happens in the general case, when the base of an im-opfibration is not freely generated. In such a case, each monoid in the base still induces some monoidal structure on the fibres, but now we are able to detect (or impose) properties of monoidal categories at the level of the base. For example, requiring the commutativity equation
\begin{center}
\scalebox{.7}{\input{monoid-commutative.tikz}}
\end{center}
to strictly hold in the split case would enforce the equation $a\otimes b=b\otimes a$ in the corresponding fibre. In the general (non-split) case, this defines a natural isomorphism, i.e. a symmetric monoidal category. Thus, in order to properly exploit the additional equations in the base, one needs to move beyond the split (strict) setting. It is unclear whether an analogue of Theorem~\ref{thm:monopfib-imopfib} holds in this case, as it is not obvious what the equations between the (pseudo)monoids on opfibrations should be. We leave exploring these developments to future work.

\chapter{Profunctor collages}\label{ch:profunctor-collages}
Our aim here is to obtain a structure that is able to interpret both the forward $\refine_f$ and the reverse $\coarsen_f$ functor boundaries. We seek a characterisation similar to opfibrations with indexed monoids whose base category is $\Fim(\cat X)$: the goal is to obtain a structure that is equivalent to an (op)indexed monoidal category, and has the bidirectional functor boundaries. We achieve this in the notion of a {\em monoidal deflation} (Definition~\ref{def:monoidal-deflation}). To motivate this and to connect it with well-known structures, we first discuss profunctors (Section~\ref{sec:profunctors}) and their collages (Section~\ref{sec:collages}).

We present and state existing work and results in Sections~\ref{sec:profunctors} and~\ref{sec:collages}, while Section~\ref{sec:deflations} contains results that are, to the best knowledge of the author, novel.

\section{Profunctors}\label{sec:profunctors}

We follow Loregian~\cite{loregian} in our discussion of profunctors. We fix the convention that profunctors are contravariant in the first variable, and covariant in the second one. Profunctors are also known under the names of {\em distributors}, {\em relators}, {\em correspondences} and {\em bimodules}.

\begin{definition}[Profunctor]
Let $\cat A$ and $\cat B$ be small categories. A {\em profunctor} from $\cat A$ to $\cat B$ is a functor $\cat A^{op}\times\cat B\rightarrow\Set$.
\end{definition}
We denote a profunctor $P$ from $\cat A$ to $\cat B$ by $P:\cat A\srarrow\cat B$. Thus, a profunctor assigns to each objects $a\in\Ob(\cat A)$ and $b\in\Ob(\cat B)$ a set $P(a,b)$, which we think of as ``morphisms'' $a\rightsquigarrow b$ between objects in different categories\footnote{Such arrows are sometimes called {\em proarrows} or {\em heteromorphisms}.}. The action of $P$ on morphisms is contravariant in the first variable (morphisms of $\cat A$) and covariant in the second variable (morphisms of $\cat B$):
\begin{center}
\scalebox{1}{\input{profunctor-def.tikz}}.
\end{center}
The assignment is functorial: $P(\id_a,\id_b)$ is the identity function on the set $P(a,b)$, and $P(f';f,g;g') = P(f,g);P(f',g')$. These assumption ensure that the arrows $a\rightsquigarrow b$ behave essentially like morphisms within a category, which leads one to ask how to compose two such arrows. This leads to defining composition of two profunctors.

Given profunctors $P:\cat A\srarrow\cat B$ and $Q:\cat B\srarrow\cat C$, their composite profunctor $P;Q:\cat A\srarrow\cat C$ is defined on objects by the following coend formula
\begin{equation}\label{eqn:profunctor-composition}
P;Q(a,c)\coloneq\int^{b\in\cat B}P(a,b)\times Q(b,c)\coloneq\quot{\coprod\limits_{b\in\cat B}P(a,b)\times Q(b,c)}{\sim},
\end{equation}
where $\sim$ is the equivalence relation on the set $\coprod_{b\in\cat B}P(a,b)\times Q(b,c)$ generated by requiring that for all $f\in P(a,b)$, $g\in\cat B(b,b')$ and $h\in Q(b',c)$ one has
$$\left(f,Q(g,\id_c)(h)\right)\sim\left(P(\id_a,g)(f),h\right).$$
On morphisms, we define
\begin{equation}\label{eqn:profunctor-composition-morphisms}
\scalebox{1}{\input{profunctor-composition-def.tikz}}.
\end{equation}

Note that the equivalence relation defined above can be thought of as a kind of associativity condition, ensuring that the situation
\begin{center}
\scalebox{1}{\input{profunctor-composition.tikz}}
\end{center}
unambiguously defines an element in $P;Q(a,c)$. In the next section, we shall see a construction making this formal: the above arrows will become actual morphisms within the same category, and the equivalence relation will guarantee associativity. We do not discuss the general theory of coends here, and simply treat the expression on the right-hand side of~\eqref{eqn:profunctor-composition} as the definition of the coend on the left. For a detailed exposition, we refer the reader to Loregian~\cite{loregian}.

\begin{example}\label{ex:identity-profunctor}
Let $\cat A$ be a category. Then the assignment
\begin{align*}
\cat A^{op}\times\cat A &\rightarrow\Set \\
(a,b) &\mapsto\cat A(a,b) \\
\left(f : a'\rightarrow a,g : b\rightarrow b'\right) &\mapsto f;{-};g : A(a,b)\rightarrow A(a',b')\ ::\ h\mapsto f;h;g
\end{align*}
defines a profunctor $\cat A\srarrow\cat A$.
\end{example}

\begin{definition}\label{def:bicat-profunctors}
The {\em bicategory of profunctors} $\Prof$ has
\begin{itemize}
\item all small categories as the 0-cells,
\item all profunctors as the 1-cells,
\item given profunctors $P,Q:\cat A\srarrow\cat B$, the 2-cells $P\rightarrow Q$ are given by the natural transformations,
\item the composition $\Prof(\cat A,\cat B)\times\Prof(\cat B,\cat C)\rightarrow\Prof(\cat A,\cat C)$ is given by~\eqref{eqn:profunctor-composition} and~\eqref{eqn:profunctor-composition-morphisms},
\item the identity on a 0-cell $\cat A$ is given by the profunctor in Example~\ref{ex:identity-profunctor},
\item composition and identities for 2-cells are those for the natural transformations.
\end{itemize}
\end{definition}

There are two ways to embed $\Cat$ into $\Prof$: both embedding functors are identity on objects, one is contravariant on the 2-cells, while the other is contravariant on the 1-cells. Hence, let us denote by $\Cat^{co}$ the category with the same 0- and 1-cells as in $\Cat$, and whose 2-cells are those of $\Cat$ with the direction reversed. Similarly, let us denote by $\Cat^{op}$ the category with the same 0- and 2-cells as in $\Cat$, but whose 1-cells are reversed.

First, let us define
$$-^{\refine}:\Cat^{co}\rightarrow\Prof$$
by mapping the functor $F:\cat A\rightarrow\cat B$ to the following profunctor:
\begin{align*}
F^{\refine} : \cat A^{op}\times\cat B &\rightarrow\Set \\
(a,b) &\mapsto\cat B(Fa,b) \\
\left(f : a'\rightarrow a,g : b\rightarrow b'\right) &\mapsto Ff;{-};g : B(Fa,b)\rightarrow B(Fa',b')\ ::\ h\mapsto Ff;h;g,
\end{align*}
and by mapping a natural transformation $\eta:F\rightarrow G$ to the natural transformation $\eta^{\refine}:G^{\refine}\rightarrow F^{\refine}$ whose $(a,b)$-component is given by $\eta_a;{-}$.

Likewise, we define
$$-^{\coarsen}:\Cat^{op}\rightarrow\Prof$$
by mapping the functor $F:\cat A\rightarrow\cat B$ to the following profunctor:
\begin{align*}
F^{\coarsen} : \cat B^{op}\times\cat A &\rightarrow\Set \\
(b,a) &\mapsto\cat B(b,Fa) \\
\left(f : b'\rightarrow b,g : a\rightarrow a'\right) &\mapsto f;{-};Fg : B(b,Fa)\rightarrow B(b',Fa')\ ::\ h\mapsto f;h;Fg,
\end{align*}
and by mapping a natural transformation $\eta:F\rightarrow G$ to the natural transformation $\eta^{\coarsen}:F^{\coarsen}\rightarrow G^{\coarsen}$ whose $(b,a)$-component is given by ${-};\eta_a$.

\begin{proposition}\label{prop:prof-embedding-adjoints}
Let $F:\cat A\rightarrow\cat B$ be a functor (i.e.~a 1-cell in $\Cat$). Then the 1-cells $F^{\refine}$ and $F^{\coarsen}$ are adjoint with
$$F^{\refine}\dashv F^{\coarsen}$$
in the bicategory $\Prof$.
\end{proposition}
\begin{proof}
The unit 2-cell $\eta : \id_{\cat A}\rightarrow F^{\refine};F^{\coarsen}$ is defined as
\begin{align*}
\eta_{a,a'} : \cat A(a,a') &\rightarrow\int^{b\in\cat B}\cat B(Fa,b)\times\cat B(b,Fa') \\
f &\mapsto (Ff,\id_{Fa'}),
\end{align*}
while the counit 2-cell $\varepsilon : F^{\coarsen};F^{\refine}\rightarrow\id_{\cat B}$ is defined as
\begin{align*}
\varepsilon_{b,b'} : \int^{a\in\cat A}\cat B(b,Fa)\times\cat B(Fa,b') &\rightarrow\cat B(b,b') \\
(g,h) &\mapsto g;h.
\end{align*}
Naturality and the triangle equations now follow by observing that for the composite $F^{\refine};F^{\coarsen}$ we have $(f,g)\sim (\id_{Fa},f;g)$ for all $f\in\cat B(Fa,b)$ and $g\in\cat B(b,Fa')$, and similarly, for the composite $F^{\coarsen};F^{\refine}$ we have $(Ff,g)\sim (\id_{Fa},Ff;g)$ for all $f\in\cat A(a,a')$ and $g\in\cat B(Fa',b)$.
\end{proof}

\begin{remark}\label{rem:counit-surjective-twocells}
We observe that the counit 2-cell defined in the proof of Proposition~\ref{prop:prof-embedding-adjoints} has surjective components on the image of the functor $F$. More precisely, for all $a,a'\in\cat A$, the function
$$\varepsilon_{Fa,Fa'} : \int^{a''\in\cat A}\cat B(Fa,Fa'')\times\cat B(Fa'',Fa') \rightarrow\cat B(Fa,Fa')$$
is surjective: a section is given by mapping $h:Fa\rightarrow Fa'$ to $(\id_{Fa},h)$. The surjectivity of the components is what allows us to introduce the sections for the counit 2-cells for deflational theories in~\ref{subsec:deflational-theories}.
\end{remark}

\section{Collages}\label{sec:collages}

The notion of a {\em collage} generalises the Grothendieck construction (see Section~\ref{sec:indexed-categories}) from (op)fibrations to arbitrary functors. What one gets on the other side of the equivalence are {\em displayed categories}, and a functor is obtained from a displayed category by taking its collage.

The term {\em displayed category} was introduced by Ahrens and Lumsdaine~\cite{displayed-categories} for the purposes of studying fibrations in the type theoretic setting, where equivalence of their presentation with the notion used in Definition~\ref{def:displayed-category} is pointed out. The definition used here, as well as Theorem~\ref{thm:collage-functors} and Proposition~\ref{prop:displayed-pseudo-decompose} are discussed in the notes taken by Streicher based on B\' enabou's course~\cite{benabou}.

Let $\cat X$ and $\cat Y$ be bicategories. Recall that a {\em lax functor} $(D,\varphi):\cat X\rightarrow\cat Y$ consists of the following data:
\begin{itemize}
\item a function $D:\Ob(\cat X)\rightarrow\Ob(\cat Y)$,
\item for all $x,y\in\cat X$, a functor $D:\cat X(x,y)\rightarrow\cat Y(Dx,Dy)$,
\item for all $x\in\cat X$, a 2-cell $\varphi^x:\id_{Dx}\rightarrow D(\id_x)$ in $\cat Y$,
\item for all composable morphisms $f:x\rightarrow y$ and $g:y\rightarrow z$ in $\cat X$, a 2-cell $\varphi^{f,g}:D(f);D(g)\rightarrow D(f;g)$ in $\cat Y$ which is natural in $f$ and $g$
\end{itemize}
such that lax associativity and unitality diagrams commute (see e.g.~\cite{leinster98,johnson-yau21}). We refer to the natural collections of 2-cells $\varphi$ as the {\em laxators}. We say that a lax functor is {\em normal} if for each $x\in\Ob(\cat X)$, the laxator $\varphi^x$ is the identity, i.e.~$\id_{Dx}=D(\id_x)$.
\begin{definition}[Displayed category]\label{def:displayed-category}
Let $\cat X$ be a small 1-category. A {\em displayed category} is a normal lax functor $(D,\varphi):\cat X\rightarrow\Prof$.
\end{definition}

Every functor $p:\cat Y\rightarrow\cat X$ induces a displayed category $D_p:\cat X\rightarrow\Prof$ as follows:
\begin{itemize}
\item for every $x\in\Ob(\cat X)$, let $D_p(x)\coloneq\cat Y_x$,
\item for every morphism $f:x\rightarrow y$ in $\cat X$, define
\begin{center}
\scalebox{1}{\input{functor-to-displayed.tikz}},
\end{center}
\item given $x\xrightarrow{f}y\xrightarrow{g}z$, the laxator $\varphi_p^{f,g} : D_p(f);D_p(g)\rightarrow D_p(f;g)$ is defined componentwise for each $a\in\cat Y_x$ and $c\in\cat Y_z$ by
\begin{align*}
\left(\varphi_p^{f,g}\right)_{a,c} : \int^{b\in\cat Y_y} D_p(f)(a,b)\times D_p(g)(b,c) &\rightarrow D_p(f;g)(a,c) \\
(F,G) &\mapsto (F;G).
\end{align*}
\end{itemize}
Note that the lax functor so defined is indeed normal: $D_p(\id_x)=\cat Y_x(-,-)=\id_{D_p(x)}$.

\begin{definition}[Collage]\label{def:collage}
Given a displayed category $(D,\varphi):\cat X\rightarrow\Prof$, its {\em collage} is the category $\coprod D$ defined as follows:
\begin{itemize}
\item the objects are pairs $(x,a)$, where $x\in\cat X$ and $a\in D(x)$,
\item a morphism $(f,F) : (x,a)\rightarrow (y,b)$ consists of a morphism $f:x\rightarrow y$ in $\cat X$ and an element $F\in D(f)(a,b)$,
\item the identity on $(x,a)$ is given by $(\id_x,\id_a)$,
\item the composition of $(f,F) : (x,a)\rightarrow (y,b)$ and $(g,G) : (y,b)\rightarrow (z,c)$ is given by $\left(f;g,\varphi^{f,g}_{a,c}(F,G)\right)$.
\end{itemize}
\end{definition}
Note that there is a forgetful functor $\coprod D\rightarrow\cat X$ projecting the objects and the morphisms to their first component. The key observation by B\' enabou is that {\em any} functor into $\cat X$ arises as the collage of some displayed category. In order to state this result, let us define the category $\Disp(\cat X)$ as having all displayed categories $\cat X\rightarrow\Prof$ as objects, and the morphisms are all lax transformations whose components are in the image of the embedding $-^{\refine}:\Cat^{co}\rightarrow\Prof$.
\begin{theorem}[B\' enabou]\label{thm:collage-functors}
Let $\cat X$ be a small category. The slice category over $\cat X$ is equivalent to the category of displayed categories on $\cat X$:
$$\quot{\Cat}{\cat X}\simeq\Disp(\cat X).$$
\end{theorem}
\begin{proof}[Proof sketch]
Given a functor $p$ into $\cat X$, the corresponding displayed category is given by $D_p$. Conversely, given a displayed category $D$ on $\cat X$, the corresponding functor is given by the projection from the collage $\coprod D$ into $\cat X$
\end{proof}
We refer the reader to~\cite{benabou}, \cite{street-powerful-functors} and~\cite[Theorem 5.4.5]{loregian} for the detailed discussion and proofs.

Collage of a displayed category, together with Theorem~\ref{thm:collage-functors} allow to detect properties of displayed categories on $\cat X$ via properties of functors into $\cat X$. We shall need two such properties, namely, we characterise when $D:\cat X\rightarrow\Prof$ is a pseudofunctor (rather than merely a lax functor) and when it factors through $\Cat$. In the latter case, each morphism in $\cat X$ is, in fact, indexing a functor, so that we have an indexed category $\cat X\rightarrow\Cat$, showing that displayed categories generalise indexed categories, while the collage is indeed a generalised Grothendieck construction.

Definition~\ref{def:decomposition-lifting} characterises those functors which correspond to pseudofunctorial displayed categories. In order to state the definition, we need the notion of {\em path components} above a pair of composable maps, which we now proceed to describe.

Given a functor $\cat Y\rightarrow\cat X$, a pair of composable maps $f:x\rightarrow y$ and $g:y\rightarrow z$ in $\cat X$ and objects $a\in\cat Y_x$ and $c\in\cat Y_z$ in the fibres over $x$ and $z$, let us define the category $\cat Y(f,g)(a,c)$ as follows:
\begin{itemize}
\item the objects $(F,b,G)$ are composable pairs $F:a\rightarrow b$ and $G:b\rightarrow c$ such that $b\in\cat Y_y$, $F$ is above $f$ and $G$ is above $g$,
\item a morphism $H:(F,b,G)\rightarrow (F',b',G')$ is a map $H:b\rightarrow b'$ in the fibre $\cat Y_y$ such that the diagram below commutes:
\begin{center}
\scalebox{1}{\input{composable-pairs-morphism.tikz}}.
\end{center}
\end{itemize}
We say that two objects in $\cat Y(f,g)(a,c)$ are in the same {\em path component} if they are related by the equivalence relation generated by stipulating that $(F,b,G)\sim (F',b',G')$ if there is a morphism $(F,b,G)\rightarrow (F',b',G')$. In other words, two composable pairs are in the same path component if there is a zigzag of morphisms between them.

\begin{definition}[Factorisation lifting]\label{def:decomposition-lifting}
We say that a functor $\cat Y\rightarrow\cat X$ is a {\em factorisation lifting} if for any morphism $F:a\rightarrow c$ in $\cat Y$ above $f:x\rightarrow z$, whenever there are morphisms $g:x\rightarrow y$ and $g':y\rightarrow z$ such that $g;g'=F$, there are induced morphisms $G:a\rightarrow b$ and $G':b\rightarrow c$ above $G$ and $G'$, respectively, such that $G;G'=F$:
\begin{center}
\scalebox{1}{\input{decomposition-pair-induced.tikz}},
\end{center}
and, moreover, any such induced pairs of morphisms $(G,b,G')$ and $(H,b',H')$ are in the same path component of $\cat Y(g,g')(a,c)$.
\end{definition}
\begin{remark}
What we call a {\em factorisation lifting} is usually called a {\em Conduch\' e fibration}, a {\em Conduch\' e functor} or an {\em exponentiable functor}. The author considers the term {\em factorisation lifting} more suggestive of the property such functor is required to satisfy. The term {\em exponentiable} should be reserved to the property of the last bullet point of Proposition~\ref{prop:displayed-pseudo-decompose}. Said proposition then establishes that a functor is a factorisation lifting if and only if it is exponentiable.
\end{remark}
By analogy with (op)fibrations, a factorisation lifting is {\em cloven} if it comes with chosen lifts for each factorisation. Further, a cloven factorisation lifting is {\em split} if the composition of any chosen lifts is the chosen lift of the composition, and the chosen lift of the trivial factorisation (either $f;\id_z$ or $\id_x;f$) is the corresponding trivial factorisation ($F;\id_c$ or $\id_a;F$).

The following proposition is due to Giraud~\cite{giraud64} and Conduch\'e~\cite{conduche72}. A detailed discussion can be found in Street~\cite{street-powerful-functors}.
\begin{proposition}\label{prop:displayed-pseudo-decompose}
For any functor $p:\cat Y\rightarrow\cat X$, the following are equivalent:
\begin{itemize}
\item $p$ is a factorisation lifting,
\item the displayed category $D_p:\cat X\rightarrow\Prof$ is a pseudofunctor,
\item the functor $p^*:\quot{\Cat}{\cat X}\rightarrow\quot{\Cat}{\cat Y}$ defined by taking pullbacks has a right adjoint.
\end{itemize}
\end{proposition}
We only prove the equivalence of first two bullet points, as this suffices for developing the subsequent theory. The third one is stated for the sake of completeness of the presentation: the reader is referred to Street~\cite{street-powerful-functors} for the proof.
\begin{proof}
Observe that the laxator of $D_p$
$$\left(\varphi_p^{f,g}\right)_{a,c} : \int^{b\in\cat Y_y} D_p(f)(a,b)\times D_p(g)(b,c) \rightarrow D_p(f;g)(a,c)$$
is an isomorphism if and only if for all $H:a\rightarrow c$ above $f;g$ there are $F:a\rightarrow b$ and $G:b\rightarrow c$ above $f$ and $g$ such that $F;G=H$, and such pair is unique up to the equivalence defining the coend, i.e.~precisely when $p$ is a factorisation lifting.
\end{proof}
In light of Proposition~\ref{prop:displayed-pseudo-decompose}, a factorisation lifting being cloven corresponds to a choice of an inverse for each component of the laxator of $D_p$. If such choice of inverses is, moreover, functorial in both $f$ and $g$, the factorisation lifting is split.

The following proposition and the ensuing corollary connect our discussion of displayed categories and collages to opfibrations.
\begin{proposition}\label{prop:displayed-factors-opfibration}
A functor $p:\cat Y\rightarrow\cat X$ is an preopfibration if and only if the displayed category $D_p:\cat X\rightarrow\Prof$ factors through the embedding $-^{\refine}:\Cat^{co}\rightarrow\Prof$.
\end{proposition}
\begin{proof}
Observe that we have a natural isomorphism
\begin{center}
\scalebox{1}{\input{preopfib-factors.tikz}}
\end{center}
for some functor $F:\cat Y_x\rightarrow\cat Y_y$ if and only if the opliftable pair $(a,f:x\rightarrow y)$ has a weak opcartesian lifting $f_a:a\rightarrow F(a)$ for each $a\in\cat Y_x$. To prove this, we argue as follows. If such liftings exist, then $F\coloneq f^*$ is the usual reindexing functor between the fibres, and the isomorphism is given by the universal property of weak opcartesian maps: any $F':a\rightarrow b$ in $D_p(f)(a,b)$ uniquely factorises as
\begin{center}
\scalebox{1}{\input{preopfib-factors-factorises.tikz}}.
\end{center}
Conversely, given such a functor and a natural isomorphism, define the lifting of the opliftable pair $(a,f:x\rightarrow y)$ as $\alpha^{-1}_{a,Fa}\left(\id_{Fa}\right)$.
\end{proof}
\begin{corollary}\label{cor:preopfib-opfib-factlift}
For a preopfibration $\cat Y\rightarrow\cat X$, the following are equivalent:
\begin{itemize}
\item it is an opfibration,
\item it is a factorisation lifting,
\item the composition of weakly opcartesian maps is weakly opcartesian.
\end{itemize}
\end{corollary}
Proposition~\ref{prop:displayed-factors-opfibration} and Corollary~\ref{cor:preopfib-opfib-factlift}, of course, dualise to (pre)fibrations. Namely, a functor $p:\cat Y\rightarrow\cat X$ is an prefibration if and only if the displayed category $D_p:\cat X\rightarrow\Prof$ factors through the embedding $-^{\coarsen}:\Cat^{op}\rightarrow\Prof$, and Corollary~\ref{cor:preopfib-opfib-factlift} holds verbatim upon removing `op' throughout.

\section{Deflations}\label{sec:deflations}

We now wish to impose additional structure on a functor $\cat Y\rightarrow\cat X$ so that it can be used as a model for layered theories with bidirectional functor boundaries. Proposition~\ref{prop:displayed-pseudo-decompose} suggests it should be a factorisation lifting, while Proposition~\ref{prop:displayed-factors-opfibration} suggests it should somehow restrict to an opfibration, which would give the forward functor boundaries. We ensure that it also restricts to a fibration, giving the reverse functor boundary. We do this by freely adding right adjoints to the base category $\cat X$ (Definition~\ref{def:zigzag-category}), and by requiring that the adjoints lift to the total category (Definition~\ref{def:local-retrofunctor}).

Let $\cat X$ be a category. The collection of {\em zigzags} on $\cat X$ is recursively defined as follows:
\begin{itemize}
\item for every morphism $f:x\rightarrow y$ in $\cat X$, the expressions $f:x\rightarrow y$ and $\overline f:y\rightarrow x$ are zigzags,
\item if $\varphi:x\rightarrow y$ and $\gamma:y\rightarrow z$ are zigzags, then so is $\varphi ;\gamma : x\rightarrow z$.
\end{itemize}
In other words, a zigzag is a (bracketed) list of morphisms with possibly alternating directions; for example, the following is a zigzag $x\rightarrow s$:
$$x\rightarrow y\leftarrow z\leftarrow w\rightarrow s.$$

\begin{definition}[Zigzag 2-category]\label{def:zigzag-category}
Given a small 1-category $\cat X$, the {\em zigzag 2-category} $\Zg(\cat X)$ of $\cat X$ is the strict 2-category whose 0-cells are the objects of $\cat X$ and the 1-cells are given by the zigzags, subject to the following equations:
\begin{align*}
\overline{\id}_x &= \id_x & f;g &= fg \\
(\varphi;\gamma);\psi &= \varphi;(\gamma;\psi) & \overline f;\overline g &= \overline{gf},
\end{align*}
where $;$ is used for composition of zigzags and concatenation is used for composition in $\cat X$. Given a morphism $f:x\rightarrow y$ in $\cat X$, we have the following generating 2-cells:
\begin{align*}
\eta_f &: \id_x\rightarrow f;\overline f \\
\varepsilon_f &: \overline f;f\rightarrow\id_y,
\end{align*}
subject to the zigzag equations defining an adjunction:
\begin{center}
\scalebox{.9}{\input{zigzag-equations-twocells.tikz}},
\end{center}
as well as the following coherence equations:
\begin{align*}
\eta_{\id_x} &= \id_{\id_x}  & \eta_{f;h} &= \eta_f;(\id_f*\eta_h*\id_{\overline f}) \\
\varepsilon_{\id_x} &= \id_{\id_x} & \varepsilon_{f;h} &= (\id_{\overline h}*\varepsilon_f*\id_h);\varepsilon_h.
\end{align*}
\end{definition}
For any category $\cat X$, there are the following identity-on-objects inclusions into its zigzag 2-category:
\begin{alignat*}{2}
\cat X &\xhookrightarrow{\ \iota^*\ }\Zg(\cat X) {} &&\xhookleftarrow{\ \iota^{\circ}\ }\cat X^{op} \\
f &\longmapsto f\quad\ \overline f {} &&\longmapsfrom f.
\end{alignat*}

\begin{definition}[Local retrofunctor]\label{def:local-retrofunctor}
Let $\cat X$ and $\cat Y$ be strict 2-categories. A {\em local retrofunctor} $(p,\varphi):\cat Y\rightarrow\cat X$ consists of a 1-functor $p:\cat Y\rightarrow\cat X$, and for all pairs of objects $a,b\in\Ob(\cat Y)$, a retrofunctor $(p,\varphi_{a,b}) : \cat Y(a,b)\rightarrow\cat X(pa,pb)$ whose function part is given by the action of $p$ on morphisms, satisfying the following coherence equation for all $F:a\rightarrow b$, $F':b\rightarrow c$, $\alpha:pF\rightarrow g$, $\alpha':pF'\rightarrow g'$:
\begin{equation*}
\varphi_{a,b}(F,\alpha)*\varphi_{b,c}(F',\alpha') = \varphi_{a,c}(F;F',\alpha *\alpha').
\end{equation*}
\end{definition}

We will now focus on local retrofunctors whose codomain is a zigzag 2-category, whose functor part is a factorisation lifting, together with a certain condition on the liftings of counits. This situation will play such a prominent role in the subsequent development that we give it a special name: {\em deflation}. We shall further introduce monoidal deflations, which will be used to give semantics to layered theories in Section~\ref{sec:deflational-models}.

Let $(p,\varphi):\cat Y\rightarrow\Zg(\cat X)$ be a local retrofunctor such that $p:\cat Y\rightarrow\Zg(\cat X)$ is a factorisation lifting. A {\em lifting} of the opliftable pair $(a,f:x\rightarrow y)$ such that $f\in\cat X$ is any morphism $F:a\rightarrow b$ obtained by decomposing the codomain of the 2-cell $\varphi_{a,a}\left(\id_a,\eta_f\right)$ into $F:a\rightarrow b$ above $f:x\rightarrow y$ and $\overline F:b\rightarrow a$ above $\overline f:y\rightarrow x$:
\begin{center}
\scalebox{1}{\input{decomposition-opfibration-lifting.tikz}}.
\end{center}
Dually, we say that $\overline F:b\rightarrow a$ is a {\em lifting} of the liftable pair $(\overline f:y\rightarrow x,a)$. In the context of deflations (which we about to define), the morphism of an opliftable pair will always be assumed to lie in $\iota^*\cat X$, and similarly, the morphism of a liftable pair will always be assumed to lie in $\iota^{\circ}\cat X$. While {\em a priori} the same pair can have many liftings, as we shall see imminently (Lemma~\ref{lma:deflation-unique-liftings}), in our situation of interest the liftings are essentially unique.
\begin{definition}[Deflation]\label{def:deflation}
Let $\cat X$ be a 1-category and $\cat Y$ be a strict 2-category. A {\em deflation} from $\cat Y$ to $\cat X$ is a local retrofunctor $(p,\varphi):\cat Y\rightarrow\Zg(\cat X)$ such that $p:\cat Y\rightarrow\Zg(\cat X)$ is a factorisation lifting, and for every lifting $F:a\rightarrow b$ of an opliftable pair $(a,f:x\rightarrow y)$, the codomain of the 2-cell $\varphi_{b,b}\left(\overline F;F,\varepsilon_f\right)$ is $\id_b$:
\begin{center}
\scalebox{1}{\input{deflation-def-counit-iso.tikz}}.
\end{center}
\end{definition}
We say that a deflation is {\em cloven} or {\em split} if the underlying factorisation lifting is.

\begin{lemma}\label{lma:deflation-unique-liftings}
Let $(p,\varphi):\cat Y\rightarrow\Zg(\cat X)$ be a deflation, and let $F:a\rightarrow b$ be a lifting of an opliftable pair $(a,f:x\rightarrow y)$. Then for any morphism $F':a\rightarrow b'$ above $f$, there is a unique morphism $\alpha:b\rightarrow b'$ in the fibre $\cat Y_y$ such that $F;\alpha =F'$:
\begin{center}
\scalebox{1}{\input{deflation-unique-liftings.tikz}}.
\end{center}
\end{lemma}
\begin{proof}
Define $\alpha:b\rightarrow b'$ as the codomain of the 2-cell $\varphi_{b,b'}\left(\overline F;F',\varepsilon_f\right)$:
\begin{center}
\scalebox{1}{\input{unique-lifting-mediator.tikz}}.
\end{center}
We claim that this 1-cell has the desired property. To see this, consider the following diagram of 2-cells in $\cat Y$:
\begin{center}
\scalebox{1}{\input{unique-lifting-mediator-commutes.tikz}},
\end{center}
which simplifies to:
\begin{align*}
&\phantom{=} \left(\varphi_{a,a}\left(\id_a,\eta_f\right)*\id_{F'}\right);\left(\id_F*\varphi_{b,b'}\left(\overline F;F',\varepsilon_f\right)\right) \\
&= \left(\varphi_{a,a}\left(\id_a,\eta_f\right)*\varphi_{a,b'}\left(F',\id_f\right)\right);\left(\varphi_{a,b}\left(F,\id_f\right)*\varphi_{b,b'}\left(\overline F;F',\varepsilon_f\right)\right) \\
&= \varphi_{a,b'}\left(F',\eta_f*\id_f\right);\varphi_{a,b'}\left(F;\overline F;F',\id_f*\varepsilon_f\right) \\
&= \varphi_{a,b'}\left(F',\left(\eta_f*\id_f\right);\left(\id_f*\varepsilon_f\right)\right) \\
&= \varphi_{a,b'}\left(F',\id_f\right) \\
&= \id_{F'},
\end{align*}
which, in particular, implies that $F;\alpha=F'$, as required.

For uniqueness, suppose that $\alpha':b\rightarrow b'$ is some 1-cell such that $F;\alpha'=F'$. We then compute
\begin{align*}
\varphi_{b,b'}\left(\overline F;F',\varepsilon_f\right) &= \varphi_{b,b'}\left(\overline F;F;\alpha',\varepsilon_f*\id_{\id_y}\right) \\
&= \varphi_{b,b}\left(\overline F;F,\varepsilon_f\right)*\varphi_{b,b'}\left(\alpha',\id_{\id_y}\right) \\
&= \varphi_{b,b}\left(\overline F;F,\varepsilon_f\right)*\id_{\alpha'}.
\end{align*}
Since the codomain of $\varphi_{b,b}\left(\overline F;F,\varepsilon_f\right)$ is $\id_b$, this, in particular, implies that $\alpha'=\alpha$.
\end{proof}

\begin{corollary}\label{cor:decomposition-biretrofunctor-opfibration}
Let $(p,\varphi):\cat Y\rightarrow\Zg(\cat X)$ be a (split) deflation. Consider the restrictions of $p$
\begin{center}
\scalebox{1}{\input{decomposition-biretrofunctor-opfibration.tikz}}
\end{center}
to the wide subcategories $\cat Y^*$ and $\cat Y^{\circ}$ of $\cat Y$ defined by $F\in\cat Y^*$ if and only if $pF\in\iota^*\cat X$, and $F\in\cat Y^{\circ}$ if and only if $pF\in\iota^{\circ}\cat X^{op}$. Then $p^*$ is a (split) opfibration, and $p^{\circ}$ is a (split) fibration.
\end{corollary}
\begin{proof}
Lemma~\ref{lma:deflation-unique-liftings} establishes that $p^*$ is a preopfibration. It is also a factorisation lifting by restricting the liftings of factorisations under $p$. Thus, by Corollary~\ref{cor:preopfib-opfib-factlift}, $p^*$ is an opfibration. The argument establishing that $p^{\circ}$ is a fibration is dual.
\end{proof}

\begin{definition}[Morphism of deflations]\label{def:morphism-deflations}
Let $(p,\varphi):\cat Y\rightarrow\Zg(\cat X)$ and $(q,\varphi):\cat Y'\rightarrow\Zg(\cat X')$ be deflations. A {\em morphism} $(H,K):p\rightarrow q$ is given by a 2-functor $H:\cat Y\rightarrow\cat Y'$ and a 1-functor $K:\cat X\rightarrow\cat X'$ such that the diagram
\begin{center}
\scalebox{1}{\input{morphism-deflations.tikz}},
\end{center}
where $K$ is freely extended to the zigzag 2-categories, commutes, and for every $F:a\rightarrow b$ in $\cat Y$ above $f:x\rightarrow y$ in $\Zg(\cat X)$ and every 2-cell $\alpha:f\rightarrow g$ we have
$$H\left(\varphi_{a,b}(F,\alpha)\right)=\varphi_{Ha,Hb}(HF,K\alpha).$$
\end{definition}
A morphism of split deflations is additionally required to strictly preserve the chosen liftings of factorisations. We denote the category of deflations by $\Defl$, and the fixed base case by $\Defl(\cat X)$. As before, we add the subscript $\mathsf{sp}$ for the subcategories of split deflations.

\begin{definition}[Minimal deflation]\label{def:minimal-deflation}
We say that a deflation $(p,\varphi):\cat Y\rightarrow\Zg(\cat X)$ is {\em minimal} if the only non-identity 2-cells of $\cat Y$ are the ones in the image of the retrofunctors $(p,\varphi_{a,b})$.
\end{definition}
An equivalent way to define a minimal deflation would be to require $p$ to be a 2-functor, and each retrofunctor $(p,\varphi_{a,b}):\cat Y(a,b)\rightarrow\cat X(pa,pb)$ to be a {\em lens}, i.e.~the lifted 2-cells should also be compatible with the action of $p$ on 2-cells, not just the 1-cells. Minimal deflations are equivalent to (op)indexed categories (Theorem~\ref{thm:minimal-deflation-indexed-category}).

\begin{remark}\label{rem:minimal-deflation}
The reason we wish to consider non-minimal deflations are many scenarios in which there are transformations between morphisms which are non-trivial, in the sense that they do not arise as liftings of the 2-cells in the base. For example, in Section~\ref{sec:zx-extraction} we shall consider 2-cells that are semantics-preserving graph rewrites. Moreover, the syntax developed in Chapter~\ref{ch:layered-theories} is sound and complete with respect to all deflations, not just the minimal ones. The general deflations, therefore, seem like a more natural semantic domain for interpreting layered theories (Chapter~\ref{ch:layered-theories}) than merely the minimal ones. Notwithstanding, at the moment of writing, it is unclear to the author what the general 2-cells of a deflation correspond to on the profunctor side. It seems that one would have to expand the allowed 2-cells from natural transformations to some kind of ``natural relations''.
\end{remark}
Let us denote by $\MinDefl(\cat X)$ the full subcategory of $\Defl(\cat X)$ on minimal deflations.

\begin{theorem}\label{thm:minimal-deflation-indexed-category}
Let $\cat X$ be a small category. There is an equivalence of categories between minimal deflations into $\cat X$ and opindexed categories on $\cat X$:
$$\MinDefl(\cat X)\simeq\OpICat(\cat X),$$
which restricts to split minimal deflations and strict functors $\cat X\rightarrow\Cat$.
\end{theorem}
\begin{proof}
By Corollary~\ref{cor:decomposition-biretrofunctor-opfibration}, each deflation restricts to an opfibration into $\cat X$. By Proposition~\ref{prop:displayed-factors-opfibration}, the corresponding displayed category $\cat X\rightarrow\Prof$ factors as
$$\cat X\rightarrow\Cat\xrightarrow{\refine}\Prof,$$
thus producing the desired indexed category. It is indeed a pseudofunctor by Proposition~\ref{prop:displayed-pseudo-decompose}, since the original deflation is a factorisation lifting.

Conversely, any opindexed category $D:\cat X\rightarrow\Cat$ extends to a normal pseudofunctor $\hat D:\Zg(\cat X)\rightarrow\Prof$ given by the action of the original functor $D$ on the 0-cells, and by the following recursive definition on the 1-cells:
\begin{align*}
f &\mapsto (Df)^{\refine} \\
\overline f &\mapsto (Df)^{\coarsen} \\
\varphi;\gamma &\mapsto\hat D(\varphi);\hat D(\gamma),
\end{align*}
where $;$ on the right-hand side denotes the composition of profunctors. The generating 2-cells $\eta_f$ and $\varepsilon_f$ are, respectively, sent to the unit $\id_{Dx}\rightarrow (Df)^{\refine};(Df)^{\coarsen}$ and the counit $(Df)^{\coarsen};(Df)^{\refine}\rightarrow\id_{Dy}$ constructed in Proposition~\ref{prop:prof-embedding-adjoints}. The collage of $\hat D$ then gives the sought-after deflation $\coprod\hat D\rightarrow\Zg(\cat X)$, with the 2-cells between two parallel morphisms $(f,F)$ and $(g,G)$ of type $(x,a)\rightarrow (y,b)$ given by a 2-cell $\alpha:f\rightarrow g$ in $\Zg(\cat X)$ such that $\hat D(\alpha)_{a,b}(F)=G$.

The equivalence on the 1-categorical part then follows by restricting the equivalence of Theorem~\ref{thm:collage-functors} to deflations and indexed categories. Moreover, the 2-cells of the collage precisely recover the 2-cells of a minimal deflation, hence concluding the proof.
\end{proof}

\subsection{Monoidal deflations}

Next, we wish to incorporate monoidal structure into deflations. As for opfibrations, we do this via the means of indexed monoids. Recall that $\Fim(\cat X)$ denotes the free category with indexed monoids generated by a small category $\cat X$ (Definition~\ref{def:free-indexed-cat}). We denote the cartesian monoidal structure of $\Fim(\cat X)$ by $(\times,1)$. First, we observe that $\Zg(\Fim(\cat X))$ can be given a monoidal structure as follows.
\begin{proposition}\label{prop:zg-monoidal-structure}
Let $\cat X$ be a small category with a cartesian monoidal structure $(\times,1)$. The zigzag 2-category $\Zg(\cat X)$ becomes a strict monoidal 2-category upon imposing the following additional equation for all $f:x\rightarrow y$ and $g:z\rightarrow w$ in $\cat X$:
$$(f\times\id_w);\overline{(\id_y\times g)} = \overline{(g\times\id_x)};(\id_z\times f).$$
\end{proposition}
\begin{proof}
We define the monoidal functor $\times:\Zg(\cat X)\times\Zg(\cat X)\rightarrow\Zg(\cat X)$ on objects by the action of the product $\times$ on $\cat X$, while for all $f:x\rightarrow y$ and $g:z\rightarrow w$ in $\cat X$, we define the product on the generating zigzags as follows, where on the right-hand side $\times$ denotes the cartesian product of $\cat X$:
\begin{align*}
(f,g) &\mapsto f\times g, \\
(\overline f,\overline g) &\mapsto \overline{f\times g}, \\
(f,\overline g) &\mapsto (f\times\id_w);\overline{(\id_y\times g)}, \\
(\overline f,g) &\mapsto \overline{(f\times\id_z)};(\id_x\times g).
\end{align*}
Finally, on the generating 2-cells we stipulate as follows, where, as before, $f:x\rightarrow y$ and $g:z\rightarrow w$ are in $\cat X$ and $\times$ on the right-hand side denotes the cartesian product of $\cat X$:
\begin{align*}
\eta_f\times\eta_g &\coloneq\eta_{f\times g}, \\
\varepsilon_f\times\varepsilon_g &\coloneq\varepsilon_{f\times g}, \\
\eta_f\times\varepsilon_g &\coloneq \varepsilon_{\id_x\times g};\eta_{f\times\id_w}, \\
\varepsilon_f\times\eta_g &\coloneq \varepsilon_{f\times\id_z};\eta_{\id_y\times g}.
\end{align*}
The monoidal unit is given by that of $\cat X$, i.e.~by $1$.
\end{proof}
Whenever we write $\Zg(\Fim(\cat X))$, we assume that it has the monoidal structure as described in Proposition~\ref{prop:zg-monoidal-structure}, which we also denote by $(\times,1)$. Note that, despite notation, this monoidal structure is {\em not} cartesian monoidal.

\begin{definition}[Monoidal deflation]\label{def:monoidal-deflation}
A {\em monoidal deflation} is a deflation $(p,\varphi):\cat Y\rightarrow\Zg(\Fim(\cat X))$ such that $\cat Y$ is a strict monoidal 2-category with the monoidal structure $(\otimes,I)$ such that $p$ strictly preserves and reflects the monoidal structure:
\begin{itemize}
\item for all $a\in\Ob(\cat Y)$, we have $p(a)=1$ if and only if $a=I$,
\item for all $F,G\in\cat Y$, we have $p(F\otimes G)=pF\times pG$,
\item for all $H\in\cat Y$ and $f,g\in\Zg(\Fim(\cat X))$, if $pH=f\times g$, then there are $F,G\in\cat Y$ with $pF=f$, $pG=g$ and $F\otimes G=H$,
\end{itemize}
and the following additional equation holds for all 1-cells $F:a\rightarrow b$ and $G:c\rightarrow d$ in $\cat Y$ that, respectively, are above the domains of the 2-cells $\alpha$ and $\beta$ in $\Zg(\Fim(\cat X))$:
$$\varphi_{a\otimes c,b\otimes d}(F\otimes G,\alpha\times\beta)=\varphi_{a,b}(F,\alpha)\otimes\varphi_{c,d}(G,\beta).$$
\end{definition}
Note that, unlike in the definition of im-opfibrations (Definition~\ref{def:opfib-indexed-mon}), we do not require that the monoidal product on the total category of a monoidal deflation sends liftings to liftings in any sense. This, however, follows from other requirements, as we record in the following proposition.
\begin{proposition}\label{prop:monoidal-deflation-liftings-preserved}
Let $(p,\varphi):\cat Y\rightarrow\Zg(\Fim(\cat X))$ be a monoidal deflation. Let $F:a\rightarrow b$ and $G:c\rightarrow d$, respectively, be liftings of opliftable pairs $(a,f:x\rightarrow y)$ and $(c,g:z\rightarrow w)$. Then $F\otimes G$ is a lifting of $(a\otimes c,f\times g)$.
\end{proposition}
\begin{proof}
The lifting of $(a\otimes c,f\times g)$ is defined by factorising the codomain of the following 2-cell:
\begin{align*}
\varphi_{a\otimes c,a\otimes c}(\id_{a\otimes c},\eta_{f\times g}) &= \varphi_{a\otimes c,a\otimes c}(\id_{a}\otimes\id_{c},\eta_{f}\times\eta_{g}) \\
&= \varphi_{a,a}(\id_{a},\eta_{f})\otimes\varphi_{c,c}(\id_{c},\eta_{g}),
\end{align*}
which is therefore equal to
$$(F;\overline F)\otimes (G;\overline G) = (F\otimes G);(\overline F;\overline G).$$
Thus, $F\otimes G$ is indeed a lifting.
\end{proof}
Since liftings in deflations are unique up to an isomorphism, Proposition~\ref{prop:monoidal-deflation-liftings-preserved} implies that in monoidal deflations liftings of products are given by products of liftings. We say that a monoidal deflation is {\em split} if the underlying deflation is split, and the chosen liftings of products are the products of chosen liftings.

A morphism of monoidal deflations is given by a morphism of deflations $(H,K)$ such that $H$ is additionally a strict monoidal functor. We denote the various categories of monoidal deflations by $\MonDefl$.

The following lemma establishes a tight link between monoidal deflations and opfibrations with indexed monoids (Definition~\ref{def:opfib-indexed-mon}).
\begin{lemma}\label{lma:monoidal-deflation-opfibration-indexed-monoids}
Let $(p,\varphi):\cat Y\rightarrow\Zg(\Fim(\cat X))$ be a split deflation. If $p$ is monoidal, then the opfibration $p^*$ obtained by restriction
\begin{center}
\scalebox{1}{\input{monoidal-deflation-restriction-indexed-monoids.tikz}}
\end{center}
has indexed monoids. Moreover, if $p$ is minimal, then the converse also holds: if the restriction $p^*$ has indexed monoids, then the deflation $p$ is monoidal.
\end{lemma}
\begin{proof}
First, suppose that the (not necessarily minimal) deflation $p$ is monoidal. Given objects $a,b\in\Ob(\cat Y)$ with $pa=x$, define the counit $C_a:a\rightarrow 1$ as the opcartesian lifting of the counit $(a,x\rightarrow 1)$, and the maps
\begin{center}
\scalebox{1}{\input{monoidal-deflation-restriction-indexed-monoids-lifted-maps.tikz}}
\end{center}
as the opcartesian liftings of the diagonal and the symmetry, respectively. Compositionality of the opcartesian liftings, together with preservation and reflection of the monoidal structure, then implies that $a'=a''=\hat a$ and $b=\hat b$, so that we may define the above maps as the diagonal and the symmetry. The equations for symmetries and uniform comonoids then follow, once again by compositionality of the opcartesian liftings. Thus, $\cat Y^*$ is cartesian monoidal by Proposition~\ref{prop:fox}, and the cartesian monoidal structure is preserved and reflected by assumption, establishing that $p^*$ has indexed monoids.

For the (partial) converse, suppose that $p$ is minimal and that $p^*$ has indexed monoids. Let us denote by $\times^*:\cat Y^*\times\cat Y^*\rightarrow\cat Y^*$ the cartesian product functor on $\cat Y^*$. We define the monoidal functor $\otimes:\cat Y\times\cat Y\rightarrow\cat Y$ on objects by the action of $\times^*$. For defining $\otimes$ on morphisms, observe that Lemma~\ref{lma:deflation-unique-liftings} (and its dual) implies that the morphisms in $\cat Y$ are generated by morphisms of the form
$$\alpha;F;\beta;\overline G,$$
where $F$ is a lifting of some $f\in\iota^*(\cat X)$ and $\overline G$ is a lifting of some $\overline g\in\iota^{\circ}(\cat X^{op})$, while $\alpha$ and $\beta$ are in the fibres of the domain and codomain of $f$. Absorbing $\alpha$ and $\beta$ into $F$, we conclude that the morphisms of $\cat Y$ are generated by morphisms of the form $F;\overline G$, where $F\in\cat Y^*$ and $G$ is, as before, a lifting of some $\overline g\in\iota^{\circ}(\cat X^{op})$. Hence, it suffices to define the monoidal product on these two classes of morphisms. Hence let $F,F'\in\cat Y^*$ and let $\overline G,\overline{G'}$ lifting of liftable pairs, with $F:a\rightarrow b$ and $\overline G:c\rightarrow d$. We define
\begin{align*}
F\otimes F' &\coloneq F\times^* F' \\
\overline G\otimes\overline{G'} &\coloneq \overline{G\times^* G'} \\
F\otimes\overline G &\coloneq (F\times^*\id_c);\overline{(\id_b\times^* G)} \\
\overline G\otimes F &\coloneq \overline{(G\times^*\id_a)};(\id_d\times^* F).
\end{align*}
The monoidal unit is given by that of $\cat Y^*$, and this monoidal structure is preserved and reflected by $p$. It remains to define $\otimes$ on the 2-cells. Since the deflation $p$ is minimal, it suffices to define $\otimes$ on the liftings of the 2-cells in $\Zg(\Fim(\cat X))$, which we do as follows, where $\alpha$ and $\beta$ are any 2-cells in $\Zg(\Fim(\cat X))$, while $F:a\rightarrow b$ and $G:c\rightarrow d$ are above their domains:
$$\varphi_{a,b}(F,\alpha)\otimes\varphi_{c,d}(G,\beta) \coloneq \varphi_{a\times^*c,b\times^*d}(F\otimes G,\alpha\times\beta).$$
\end{proof}
Note that Lemma~\ref{lma:monoidal-deflation-opfibration-indexed-monoids} implies that every split monoidal deflation has a fibrewise monoidal structure induced by the fibrewise monoidal structure of the split opfibration with indexed monoids obtained by restriction. Moreover, this fibrewise monoidal structure coincides with the fibrewise monoidal structure induced by the split fibration with indexed comonoids obtained by restriction (Corollary~\ref{cor:decomposition-biretrofunctor-opfibration}). As for im-opfibrations, we denote the fibrewise monoidal structure in the fibre of $x$ by $\otimes_x$.

The construction in the proof of Lemma~\ref{lma:monoidal-deflation-opfibration-indexed-monoids} establishes the following theorem, linking monoidal deflations and opfibrations with indexed monoids. By results of Section~\ref{sec:opfibrations-indexed-monoids} (specifically~Theorem~\ref{thm:monopfib-imopfib} and Corollary~\ref{cor:opfibrations-indexed-monoids-opindexed-monoidal}), minimal monoidal deflations are equivalent to opfibrations with a monoid and strict opindexed monoidal categories (i.e.~functors) $\cat X\rightarrow\MonCat_{\mathsf{st}}$.
\begin{theorem}\label{thm:monoidal-deflations-indexed-monoidal}
Let $\cat X$ be a small category. The categories of minimal split monoidal deflations into $\cat X$ and split im-opfibrations into $\Fim(\cat X)$ are equivalent:
$$\MinMonDefl_{\mathsf{sp}}(\cat X)\simeq\imOpFib_{\mathsf{sp}}(\Fim(\cat X)).$$
\end{theorem}
\begin{proof}
Given a minimal split monoidal deflation $(p,\varphi):\cat Y\rightarrow\cat \Zg(\Fim(X))$, the corresponding split im-opfibration is given by the restriction $p^*:\cat Y^*\rightarrow\Fim(\cat X)$.

Conversely, given a split im-opfibration $p:\cat Y\rightarrow\Fim(\cat X)$, we extend it to a split deflation $(\hat p,\varphi):\hat{\cat Y}\rightarrow\cat \Zg(\Fim(X))$ as follows. The category $\hat{\cat Y}$ has the same objects as $\cat Y$, while the morphisms are defined as:
\begin{itemize}
\item each morphism of $\cat Y$ is a morphism of $\hat{\cat Y}$, so that $\cat Y$ is a subcategory of $\hat{\cat Y}$,
\item for each opcartesian $F\in\cat Y$, we add the morphism $\overline F$ in the opposite direction and let $p\overline F\coloneq\overline{pF}$,
\item we impose the following equations on the morphisms:
\begin{align*}
\overline{\id_a} &= \id_a & \overline F;\overline G &= \overline{G;F}.
\end{align*}
\end{itemize}
The 2-cells are generated by adding exactly those 2-cells required by the local retrofunctor, i.e.~for every 2-cell $\alpha$ in $\Fim(\cat X)$ and every $F:a\rightarrow b$ in $\hat{\cat Y}$ above the domain of $\alpha$ we add a generating 2-cell
$$\varphi_{a,b}(F,\alpha),$$
subject to the equations of a deflation. Now, $(\hat p,\varphi)$ is a minimal split deflation, whose restriction is the split im-opfibration $p$ we started with. Lemma~\ref{lma:monoidal-deflation-opfibration-indexed-monoids} guarantees that $(\hat p,\varphi)$ is monoidal.
\end{proof}

In light of Theorem~\ref{thm:monoidal-deflations-indexed-monoidal}, one might ask why go through all the developments of the current chapter just to obtain something we already defined in the previous chapter? Why not directly work with im-opfibrations and be content? The answer is twofold. First, presenting an (op)indexed monoidal category (equivalently, an im-opfibration) as a monoidal deflation explicitly introduces the functor boundaries in both directions, displaying the categories and the functors involved in a highly modular way. This makes deflations ideal for a purely syntactic treatment, which is the subject of the next chapter. Second, dropping the requirement of minimality, we do obtain a genuine generalisation of (op)indexed monoidal categories. The axiomatisation we develop in the next chapter is sound and complete with respect to all monoidal deflations, not merely the minimal ones.

\chapter{Semantics of layered theories}\label{ch:semantics}
In this final chapter, we finally define the models of opfibrational, fibrational and deflation theories in full generality. In each case, we construct a free-forgetful adjunction with the category of models: opfibrations with indexed monoids, fibrations with indexed monoids and monoidal deflations.

\section{(Op)fibrational models}\label{sec:opfib-models}

Here we consider models of layered theories consisting of a split im-opfibration $\cat Y\rightarrow\Fim(\cat X)$ (see section~\ref{sec:opfibrations-indexed-monoids}) together with an interpretation of all types, terms and 2-terms of the theory.

\begin{definition}[Opfibrational model of a layered signature]\label{def:model-layered-signature}
An {\em opfibrational model} of a layered signature $(\Omega,\mathcal F,\M_{\omega})$ is a split im-opfibration $p:\cat Y\rightarrow\Fim(\cat X)$, together with the following data:
\begin{itemize}
\item a function $i:\Omega\rightarrow\Ob(\cat X)$,
\item for every $\omega,\tau\in\Omega$, a function $i_{\omega,\tau}:\mathcal F(\omega,\tau)\rightarrow\cat X(i\omega,i\tau)$,
\item for every $\omega\in\Omega$, a function $i^{\omega}:C_{\omega}\rightarrow\Ob(\cat Y_{i\omega})$,
\item for every $\omega\in\Omega$ and all $a,b\in C_{\omega}^*$, a function $i^{\omega}_{a,b}:\Sigma_{\omega}(a,b)\rightarrow\cat Y_{i\omega}\left(\left(i^{\omega}\right)^*a,\left(i^{\omega}\right)^*b\right)$,
\end{itemize}
where $(i^{\omega})^*:C_{\omega}^*\rightarrow\Ob(\cat Y_{i\omega})$ is defined by extending $i^{\omega}$ to the fibrewise monoidal structure of $p$.
\end{definition}
Note that the last two bullet points imply that for each $\omega\in\Omega$, the fibre $\cat Y_{i\omega}$ is a model of the monoidal signature $\M_{\omega}$ in the sense of Definition~\ref{def:model-monoidal-signature}. We denote a model of $(\Omega,\mathcal F,\M_{\omega})$ by $(p,i)$.

Given an opfibrational model $(p:\cat Y\rightarrow\Fim(\cat X),i)$ of $\mathcal L=(\Omega,\mathcal F,\M_{\omega})$, it induces a function $i:\Type_{\mathcal L}\rightarrow\Ob(\cat Y)$ recursively defined as follows:
\begin{align*}
\varepsilon : \varepsilon &\mapsto 1, \\
\varepsilon : \omega &\mapsto I_{i\omega} \\
a : \omega &\mapsto i^{\omega}(a) \\
f(A) : \tau &\mapsto i_{\omega,\tau}(f)^*\left(i(A:\omega)\right) \\
AB : \omega &\mapsto i(A:\omega)\otimes_{i\omega} i(B:\omega), \\
T,S &\mapsto i(T)\times i(S),
\end{align*}
where $i_{\omega,\tau}(f)^*:\cat Y_{i\omega}\rightarrow\cat Y_{i\tau}$ is the reindexing functor induced by $i_{\omega,\tau}(f):i\omega\rightarrow i\tau$.

Likewise, the same opfibrational model further induces the function $i:\Term^1_{\mathcal L}\rightarrow\cat Y$ form the basic (Figure~\ref{fig:layered-terms}), symmetry (Definition~\ref{def:symmetry-terms}) and opfibrational (Figure~\ref{fig:opfibrational-terms}) terms as follows:
\begingroup
\addtolength{\jot}{1em}
\begin{align*}
\scalebox{.9}{\input{emptydiag-sheet.tikz}} : (\varepsilon :\omega\mid\varepsilon :\omega) &\mapsto \id_{I_{i\omega}} & \scalebox{.9}{\input{emptydiag.tikz}} &\mapsto \id_1 \\
\scalebox{.9}{\input{iddiag-sheet.tikz}} : (A:\omega\mid A:\omega) &\mapsto \id_{i(A:\omega)} & \scalebox{.9}{\input{symdiag-sheet1.tikz}} &\mapsto \scalebox{.9}{\input{symmetry.tikz}} \\
\scalebox{.9}{\input{internalsigmadiag.tikz}} : (a:\omega\mid b:\omega) &\mapsto i^{\omega}_{a,b}(\sigma) & \scalebox{.9}{\input{cup.tikz}} &\mapsto \scalebox{.9}{\begin{tikzpicture}
	\begin{pgfonlayer}{nodelayer}
		\node [style=wnode] (0) at (-1, 0) {};
		\node [style=none] (1) at (1, 0) {};
	\end{pgfonlayer}
	\begin{pgfonlayer}{edgelayer}
		\draw [style=edge] (0) to (1.center);
	\end{pgfonlayer}
\end{tikzpicture}
} \\
\scalebox{.9}{\input{f-box.tikz}} : (f(A):\tau\mid f(B):\tau) &\mapsto i_{\omega,\tau}(f)^*(i(x)) & \scalebox{.9}{\input{pants.tikz}} &\mapsto \scalebox{.9}{\input{monoid-multiplication.tikz}} \\
\scalebox{.9}{\input{internalxydiag.tikz}} : (AC:\omega\mid BD:\omega) &\mapsto i(x)\otimes_{\omega} i(y) & \scalebox{.9}{\input{a-cap.tikz}} &\mapsto \scalebox{.9}{\begin{tikzpicture}
	\begin{pgfonlayer}{nodelayer}
		\node [style=node] (0) at (1, 0) {};
		\node [style=none] (1) at (-1, 0) {};
	\end{pgfonlayer}
	\begin{pgfonlayer}{edgelayer}
		\draw [style=edge] (0) to (1.center);
	\end{pgfonlayer}
\end{tikzpicture}
} \\
\scalebox{.9}{\input{refine-sheet.tikz}} : (A:\omega\mid f(A):\tau) &\mapsto i_{\omega,\tau}(f)_{i(A:\omega)} & \scalebox{.9}{\input{copants-copy.tikz}} &\mapsto \scalebox{.9}{\input{comonoid-comultiplication.tikz}} \\
x;y &\mapsto i(x);i(y) & x\otimes y &\mapsto i(x)\times i(y), \\
\end{align*}
\endgroup
where $i_{\omega,\tau}(f)^*:\cat Y_{i\omega}\rightarrow\cat Y_{i\tau}$ is the reindexing functor induced by $i_{\omega,\tau}(f):i\omega\rightarrow i\tau$, and
$$i_{\omega,\tau}(f)_{i(A:\omega)}:i(A:\omega)\rightarrow i(f(A):\tau)$$
is the opcartesian lifting of the opliftable pair
$$\left(i(A:\omega), i_{\omega,\tau}(f):i\omega\rightarrow i\tau\right).$$

\begin{definition}\label{def:model-opfib-thy}
A {\em model of an opfibrational layered theory} with signature $\mathcal L$ is an opfibrational model $(p,i)$ of $\mathcal L$ such that the induced functions $i:\Type_{\mathcal L}\rightarrow\Ob(\cat Y)$ and $i:\Term^1_{\mathcal L}\rightarrow\cat Y$ preserve the 0- and the 1-equations, respectively.
\end{definition}

\begin{proposition}\label{prop:opfib-equations-preserved}
Any opfibrational model of a layered signature preserves the structural opfibrational equations (Definition~\ref{def:str-opfib-eqns}).
\end{proposition}
\begin{proof}
The structural 0-equations (Definition~\ref{def:str-0eqns}) are preserved by associativity and unitality of the monoids, and by the fact that the reindexing functors preserve the monoidal multiplication and units.

The structural 1-equations (Definition~\ref{def:str-1eqns}) are preserved since $\cat Y$ is (cartesian) monoidal, each fibre $\cat Y_{i\omega}$ is a monoidal category, and $\otimes_x:\cat Y_x\times\cat Y_x\rightarrow\cat Y_x$ is a functor.

The structural 1-equations in Figure~\ref{fig:structural-twocells-functors-int} are preserved since the reindexing functors between the fibres are strictly functorial and monoidal.

The symmetry equations are preserved since $\cat Y$ is cartesian monoidal, hence in particular symmetric monoidal. Likewise, the equations of uniform comonoids are preserved since $\cat Y$ is cartesian monoidal.

The equations in Figure~\ref{fig:structural-twocells-monoidal} are preserved by lifting the corresponding equations from the base category.

Preservation of the equations in Figure~\ref{fig:structural-twocells-functors-ext} is somewhat less immediate. For the first equation, assuming the sorts $x:(A:\omega\mid B:\omega)$ and $f\in\mathcal F(\omega,\tau)$ we have to show that
$$i(x);i_{\omega,\tau}(f)_{i(B:\omega)} = i_{\omega,\tau}(f)_{i(A:\omega)};i_{\omega,\tau}(f)^*(i(x)),$$
which holds by the definition of $i_{\omega,\tau}(f)^*(i(x))$: it is the unique map induced by the lifting property of the opcartesian map $i_{\omega,\tau}(f)_{i(A:\omega)}$ that makes the square
\begin{center}
\scalebox{1}{\input{interpretation-lifting-x.tikz}}
\end{center}
commute. The remaining two equations in Figure~\ref{fig:structural-twocells-monoidal} are preserved by Theorem~\ref{thm:monopfib-imopfib}.
\end{proof}

\begin{definition}\label{def:category-models-opfibrational-theories}
The category of {\em models of opfibrational theories} $\OpFThMod$ has as objects tuples $(\mathcal T,p,i)$ of an opfibrational theory $\mathcal T$ and its model $p:\cat Y\rightarrow\Fim(\cat X)$ (with the corresponding family of interpretation functions $i$). A morphism
$$(F,P,Q) : (\mathcal T,p:\cat Y\rightarrow\cat X,i)\rightarrow (\mathcal K,q:\cat Y'\rightarrow\cat X',j)$$
is given by a morphism of opfibrational theories $F:\mathcal T\rightarrow\mathcal K$ together with a morphism of im-opfibrations $(P,Q):p\rightarrow q$ such that the diagrams below commute, where we assume that the layered signature of $\mathcal T$ is $(\Omega,\mathcal F,\M_{\omega})$ and the layered signature of $\mathcal K$ is $(\Psi,\mathcal G,\M_{\psi})$:
\begin{center}
\input{mor-lay-mod.tikz}.
\end{center}
Note that the bottom diagrams say that for every $\omega\in\Omega$, the restriction of $P:\cat Y\rightarrow\cat Y'$ to the fibre $P_{i\omega}:\mathcal Y_{i\omega}\rightarrow\mathcal Y'_{jF(\omega)}$ makes $(F^{\omega},P_{i\omega})$ into a morphism in $\MMod$.
\end{definition}

We denote the full subcategory of $\OpFThMod$ defined by the theories whose sets of equations are empty by $\OpFMod$, and call it the category of {\em opfibrational models of layered signatures}, since models of empty theories are in one-to-one correspondence with models of signatures. We summarise the relationship between opfibrational theories, models and layered signatures in the following proposition (cf.~Proposition~\ref{prop:monoidal-signatures-theories-models}).

\begin{proposition}\label{prop:opfibrational-signatures-theories-models}
The vertical forgetful functors in the diagram below are fibrations, while the horizontal forgetful functors form a morphism of fibrations:
\begin{center}
\scalebox{1}{\input{opfibrational-signatures-theories-models.tikz}}.
\end{center}
Moreover, the objects in the fibre $\OpFMod(\mathcal L)$ are precisely the opfibrational models of the layered signature $\mathcal L$, and the objects in the fibre $\OpFThMod(\mathcal T)$ are precisely the models of the opfibrational theory $\mathcal T$.
\end{proposition}
Similarly to Proposition~\ref{prop:monoidal-signatures-theories-models}, we note that $\OpFTh\rightarrow\LSgn$ is also a fibration, while the functor $\OpFThMod\rightarrow\OpFMod$ fails to be a fibration.

We conclude the discussion of opfibrational models by constructing free models generated by opfibrational theories, i.e.~by constructing the left adjoint to the vertical forgetful functors of Proposition~\ref{prop:opfibrational-signatures-theories-models}. We first define the total and the base categories generated by a theory and construct the im-opfibration between them.
\begin{definition}[Total category generated by an opfibrational theory]
Let $\mathcal T=(\mathcal L,E^0,E^1)$ be an opfibrational theory. The {\em total category} $T(\mathcal T)$ generated by $\mathcal T$ has the equivalence classes of types $\Type_{\mathcal L}$ under the type congruence generated by $E^0\cup S^0$ as objects, and the morphisms are the equivalence classes of terms $\Term^1_{\mathcal L}$ under the term congruence generated by $E^1\cup S_{\mathsf{opf}}^1$.
\end{definition}

\begin{definition}[Base category generated by a layered signature]
Let $\mathcal L=(\Omega,\mathcal F,\M_{\omega})$ be a layered signature. The {\em base category} $B(\mathcal L)$ generated by $\mathcal L$ is the free category with indexed monoids generated by $(\Omega,\mathcal F)$ viewed as a monoidal signature.
\end{definition}

Given an opfibrational theory $\mathcal T=(\mathcal L,E^0,E^1)$ with $\mathcal L=(\Omega,\mathcal F,\M_{\omega})$, define the functor $p_{\mathcal T}:T(\mathcal T)\rightarrow B(\mathcal L)$ on internal types by $A:\omega\mapsto\omega$, on all types by extension to products, and by action on the generating terms as follows:
\begin{center}\label{page:total-to-base-opf}
\scalebox{.87}{\input{total-to-base.tikz}},
\end{center}
where $x:(A:\omega\mid B:\omega)$ is any internal term.

\begin{proposition}\label{prop:syntactic-opfibrational-model}
For any opfibrational theory $\mathcal T$ with the layered signature $\mathcal L$, the functor $p_{\mathcal T}:T(\mathcal T)\rightarrow B(\mathcal L)$ is a split im-opfibration.
\end{proposition}
\begin{proof}
We observe that the opcartesian morphisms in $T(\mathcal T)$ are given by those terms that are equal to a term with no non-trivial internal structure. More precisely, a morphism is opcartesian if and only if the equivalence class contains a term whose construction does not involve the rule~\ref{term:int-gen}. We now observe that if we set the internal term $x$ to be the identity term in the definition of $p_{\mathcal T}:T(\mathcal T)\rightarrow B(\mathcal L)$, then for any type $S$ and a generating morphism $f:p_{\mathcal T}(S)\rightarrow\omega$, there is a unique term $\hat f:S\rightarrow K$ in $T(\mathcal T)$ with $M_{\mathcal T}\left(\hat f\right)=f$, which is moreover opcartesian. Extending this to all morphisms gives the split opfibration structure of $p_{\mathcal T}$, which is, moreover, consistent with the (cartesian) monoidal structure of $T(\mathcal T)$. Finally, the opfibration $p_{\mathcal T}$ has indexed monoids by construction.
\end{proof}
\begin{corollary}
Any opfibrational theory $\mathcal T=(\mathcal L,E^0,E^1)$ gives rise to the opfibrational model $\left(p_{\mathcal T}:T(\mathcal T)\rightarrow B(\mathcal L),i\right)$, whose interpretation functions $i$ are given by:
\begin{itemize}
\item $i(\omega)=\omega$,
\item $i_{\omega,\tau}(f)=f$,
\item $i^{\omega}(a)= a:\omega$,
\item $i^{\omega}_{A,B}(\sigma)= \sigma : (A:\omega\mid B:\omega)$.
\end{itemize}
\end{corollary}

We have finally arrived to the main theorem about opfibrational theories.
\begin{theorem}\label{thm:opfibrational-adjoint}
The construction of opfibrational models from opfibrational theories extends to a functor $\OpFTh\rightarrow\OpFThMod$, which is moreover the left adjoint to the forgetful functor.
\end{theorem}
\begin{proof}
Any morphism of opfibrational theories $F:\mathcal T\rightarrow\mathcal K$ recursively extends to types and terms, as shown in Section~\ref{sec:types-terms}, giving the functor between the total categories $T(\mathcal T)\rightarrow T(\mathcal K)$. The induced functor between the base categories $B(\mathcal T)\rightarrow B(\mathcal K)$ is given by extending the action of $F$ on the layers and the generators to a morphism of indexed monoids. This gives a morphism of im-opfibrations.

By Proposition~\ref{prop:opfib-equations-preserved}, the term model of $\mathcal T$ is the initial object of the fibre $\OpFThMod(\mathcal T)$. Moreover, any morphism out of the term model of $\mathcal T$ into any model of $\mathcal K$ factors through the term model of $\mathcal K$, thus exhibiting the desired adjunction.
\end{proof}
\begin{corollary}\label{cor:opfibrational-completeness}
For any opfibrational theory $\mathcal T$ and layered signature $\mathcal L$, we have the following equivalences:
\begin{align*}
\OpFThMod(\mathcal T) &\simeq \quot{F(\mathcal T)}{\MonOpFib_{\mathsf{sp}}}, \\
\OpFMod(\mathcal L) &\simeq \quot{F(\mathcal L)}{\MonOpFib_{\mathsf{sp}}}.
\end{align*}
\end{corollary}

\subsection{Fibrational models}\label{subsec:fibrational-models}

The class of models characterised by fibrational theories are {\em fibrations with indexed comonoids} $\cat Y\rightarrow\Fim(\cat X)^{op}$ obtained by appropriately dualising Definition~\ref{def:opfib-indexed-mon}. The free fibrational models are obtained by imposing the structural fibrational equations (Definition~\ref{def:str-fib-eqns}) on the fibrational terms~\ref{fig:fibrational-terms}.

\section{Deflational models}\label{sec:deflational-models}

As anticipated, the deflational models are given by split monoidal deflations (Definition~\ref{def:monoidal-deflation}).

\begin{definition}[Deflational model]\label{def:deflational-model}
A {\em deflational model} of a layered signature $(\Omega,\mathcal F,\mathcal M_{\omega})$ is a split monoidal deflation $(p,\varphi):\cat Y\rightarrow\Zg(\Fim(\cat X))$ together with the following data:
\begin{itemize}
\item a function $i:\Omega\rightarrow\Ob(\cat X)$,
\item for every $\omega,\tau\in\Omega$, a function $i_{\omega,\tau}:\mathcal F(\omega,\tau)\rightarrow\cat X(i\omega,i\tau)$,
\item for every $\omega\in\Omega$, a function $i^{\omega}:C_{\omega}\rightarrow\Ob(\cat Y_{i\omega})$,
\item for every $\omega\in\Omega$ and all $a,b\in C_{\omega}^*$, a function $i^{\omega}_{a,b}:\Sigma_{\omega}(a,b)\rightarrow\cat Y_{i\omega}\left(\left(i^{\omega}\right)^*a,\left(i^{\omega}\right)^*b\right)$,
where $(i^{\omega})^*:C_{\omega}^*\rightarrow\Ob(\cat Y_{i\omega})$ is defined by extending $i^{\omega}$ to the fibrewise monoidal structure of $p$.
\end{itemize}
\end{definition}
Note that a deflational model on $(p,\varphi):\cat Y\rightarrow\Zg(\Fim(\cat X))$ is at the same time an opfibrational model on the restriction $p^*:\cat Y^*\rightarrow\Fim(\cat X)$ {\em and} a fibrational model on the restriction $p^{\circ}:\cat Y^{\circ}\rightarrow\Fim(\cat X)^{op}$. We denote a model of $(\Omega,\mathcal F,\mathcal M_{\omega})$ by $(p,\varphi,i)$.

Since any deflational model $(p,\varphi,i)$ of $\mathcal L$ restricts to an opfibrational model, it induces the function $i:\Type_{\mathcal L}\rightarrow\Ob(\cat Y)$, as $\cat Y^*$ and $\cat Y$ have the same objects. Moreover, this function is equal to the function induced by the restricted fibrational model.

The opfibrational restriction $p^*:\cat Y^*\rightarrow\Fim(\cat X)$ is likewise used to obtain the function from the basic (Figure~\ref{fig:layered-terms}), symmetry (Definition~\ref{def:symmetry-terms}) and opfibrational (Figure~\ref{fig:opfibrational-terms}) terms into $\cat Y^*$, as described in Section~\ref{sec:opfib-models}. In order to obtain the function $i:\Term_{\mathcal L}^1\rightarrow\cat Y$, it remains to interpret the fibrational terms (Figure~\ref{fig:fibrational-terms}), as well as to extend the interpretation to composite terms. Hence, let each fibrational term be interpreted as the dual of the corresponding opfibrational term, while the composition and product are interpreted as composition and the (global) monoidal product in $\cat Y$. Explicitly, we stipulate:
\begingroup
\addtolength{\jot}{1em}
\begin{align*}
\scalebox{.9}{\input{coarsen-sheet.tikz}} &\mapsto \overline{i_{\omega,\tau}(f)_{i(A:\omega)}} \\
\scalebox{.9}{\input{cap.tikz}} &\mapsto \scalebox{.9}{\begin{tikzpicture}
	\begin{pgfonlayer}{nodelayer}
		\node [style=wnode] (0) at (1, 0) {};
		\node [style=none] (1) at (-1, 0) {};
	\end{pgfonlayer}
	\begin{pgfonlayer}{edgelayer}
		\draw [style=edge] (0) to (1.center);
	\end{pgfonlayer}
\end{tikzpicture}
} & \scalebox{.9}{\input{copants.tikz}} &\mapsto \scalebox{.9}{\input{monoid-multiplication-op.tikz}} \\
\scalebox{.9}{\input{a-cup.tikz}} &\mapsto \scalebox{.9}{\begin{tikzpicture}
	\begin{pgfonlayer}{nodelayer}
		\node [style=node] (0) at (-1, 0) {};
		\node [style=none] (1) at (1, 0) {};
	\end{pgfonlayer}
	\begin{pgfonlayer}{edgelayer}
		\draw [style=edge] (0) to (1.center);
	\end{pgfonlayer}
\end{tikzpicture}
} & \scalebox{.9}{\input{pants-copy.tikz}} &\mapsto \scalebox{.9}{\input{comonoid-comultiplication-op.tikz}} \\
x;y &\mapsto i(x);i(y) & x\otimes y &\mapsto i(x)\otimes i(y). \\
\end{align*}
\endgroup

\begin{definition}\label{def:deflational-theory-model}
A {\em model of a deflational layered theory} $(\mathcal L,E^0,E^1,\eta,E^2)$ consists of: (1) a deflational model $(p,\varphi,i)$ of $\mathcal L$ (Definition~\ref{def:deflational-model}) such that the induced functions $i:\Type_{\mathcal L}\rightarrow\Ob(\cat Y)$ and $i:\Term_{\mathcal L}^1\rightarrow\cat Y_1$ preserve $E^0$ and $E^1$, and (2) for each parallel pair of terms $(x,y)\in P^1_{\mathcal L}$ with sort $(T\mid S)$ a function
$$i_{x,y}:\eta(x,y)\sqcup\eta_{\mathsf{str}}(x,y)\rightarrow\cat Y_2\left(i(x),i(y)\right)$$
from the generating 2-cells into the 2-cells $i(x)\rightarrow i(y)$, such that the structural 2-cells are mapped to the following 2-cells returned by the retrofunctor:
\begin{align*}
&\eta_x\mapsto\varphi_{i(T),i(T)}\left(\id_{i(T)},\eta_{i(x)}\right) & &\varepsilon_x\mapsto\varphi_{i(S),i(S)}\left(i(\bar x);i(x),\varepsilon_{i(x)}\right),
\end{align*}
and the function $i:\Term^2_{\mathcal L}\rightarrow\cat Y_2$ on all 2-terms (Definition~\ref{def:2terms}) recursively obtained from $\eta\sqcup\eta_{\mathsf{str}}$ preserves $E^2$.
\end{definition}

\begin{proposition}\label{prop:deflational-model-preserves-structural}
Any deflational model of a deflational layered theory preserves the structural deflational equations (Definition~\ref{def:str-defl-eqns}).
\end{proposition}
\begin{proof}
The structural opfibrational and fibrational equations are preserved by Proposition~\ref{prop:opfib-equations-preserved} (and its dual). The 2-equations are preserved since the local retrofunctor preserves both vertical and composition, and therefore equations between 2-cells.
\end{proof}

\begin{definition}\label{def:category-deflational-models}
The category of {\em models of deflational theories} $\DeflThMod$ has as objects tuples $(\mathcal T,p,\varphi,i)$ of a deflational theory $\mathcal T$ and its model $(p,\varphi):\cat Y\rightarrow\Zg(\Fim(\cat X))$ with the family of interpretation functions $i$. A morphism
$$(F,F^2,P,Q) : (\mathcal T,p,\varphi,i)\rightarrow (\mathcal K,q,\rho,j),$$
is given by a morphism of deflational theories $(F,F^2):\mathcal T\rightarrow\mathcal K$ together with a morphism of monoidal deflations $(P,Q):(p,\varphi)\rightarrow (q,\rho)$ such that the diagrams of Definition~\ref{def:category-models-opfibrational-theories} as well as the diagrams below commute, where $(p,\varphi):\cat Y\rightarrow\Zg(\Fim(\cat X))$ and $(q,\rho):\cat Y'\rightarrow\Zg(\Fim(\cat X'))$, and we write $\eta_{\mathcal T}$ and $\eta_{\mathcal K}$ for the choice of 2-cells in $\mathcal T$ and $\mathcal K$, and $(\Omega,\mathcal F,\mathcal M_{\omega})$ for the layered signature of $\mathcal T$ and $(\Psi,\mathcal G,\mathcal M_{\psi})$ for the layered signature of $\mathcal K$:
\begin{center}
\scalebox{1}{\input{category-deflational-models-diagram.tikz}},
\end{center}
where $F:\Term^1_{\mathcal L}\rightarrow\Term^1_{\mathcal K}$, $i:\Term^1\rightarrow\cat Y_1$ and $j:\Term^1_{\mathcal K}\rightarrow\cat Y'_1$ are the induced map, and $Q^2$ denotes the action of $Q$ on 2-cells.
\end{definition}

As for monoidal and opfibrational theories, we denote the full subcategory of $\DeflThMod$ of {\em deflational models of layered signatures} defined by the theories whose sets of equations and 2-cells are empty by $\DeflMod$. As in Propositions~\ref{prop:monoidal-signatures-theories-models} and~\ref{prop:opfibrational-signatures-theories-models}, below we summarise the relationship between theories, models and signatures.

\begin{proposition}\label{prop:deflational-signatures-theories-models}
The vertical forgetful functors in the diagram below are fibrations, while the horizontal forgetful functors form a morphism of fibrations:
\begin{center}
\scalebox{1}{\input{deflational-signatures-theories-models.tikz}}.
\end{center}
Moreover, the objects in the fibre $\DeflMod(\mathcal L)$ are precisely the deflational models of the layered signature $\mathcal L$, and the objects in the fibre $\DeflThMod(\mathcal T)$ are precisely the models of the deflational theory $\mathcal T$.
\end{proposition}
As in Propositions~\ref{prop:monoidal-signatures-theories-models} and~\ref{prop:opfibrational-signatures-theories-models}, we note that $\DeflTh\rightarrow\LSgn$ is also a fibration, while the functor $\DeflThMod\rightarrow\DeflMod$ fails to be a fibration.

We now turn to the construction of free deflational models. As one might anticipate from the preceding development, the free models are given by quotienting the types, terms and 2-terms by the structural equations, and any free deflational model restricts to a free opfibrational model and a free fibrational model.
\begin{definition}[Total category generated by a deflational theory]
Let $\mathcal T=(\mathcal L,E^0,E^1,\eta,E^2)$ be a deflational theory. The {\em deflational total category} $T_{\mathsf{defl}}(\mathcal T)$ generated by $\mathcal T$ has the equivalence classes of types $\Type_{\mathcal L}$ under the type congruence generated by $E^0\cup S^0$ as the 0-cells, the equivalence classes of terms $\Term^1_{\mathcal L}$ under the term congruence $E^1\cup S_{\mathsf{defl}}^1$ as the 1-cells, and the equivalence classes of terms $\Term^2_{\mathcal L}$ generated by $\eta\sqcup\eta_{\mathsf{str}}$ under the term congruence $E^2\cup S_{\mathsf{defl}}^2$ as the 2-cells.
\end{definition}

\begin{definition}[Zigzag base category generated by a layered signature]
Let $\mathcal L=(\Omega,\mathcal F,\M_{\omega})$ be a layered signature. The {\em zigzag base category} $\Zg(B(\mathcal L))$ generated by $\mathcal L$ is the zigzag 2-category of the free category with indexed monoids generated by $(\Omega,\mathcal F)$ viewed as a monoidal signature.
\end{definition}

Given a deflational theory $\mathcal T=(\mathcal L,E^0,E^1,\eta,E^2)$ with $\mathcal L=(\Omega,\mathcal F,\M_{\omega})$, define the functor $p_{\mathcal T}:T_{\mathsf{defl}}(\mathcal T)\rightarrow \Zg(B(\mathcal L))$ on objects by letting $A:\omega\mapsto\omega$ on internal types, and on all types by extension to monoidal products. On morphisms, $p$ is defined by mapping the basic, symmetry and opfibrational terms exactly as in the opfibrational case (see Page~\pageref{page:total-to-base-opf}), while each fibrational term is mapped to the dual of the image of the corresponding fibrational term. Explicitly, we define:
\begin{center}
\scalebox{.87}{\input{total-to-base-fib.tikz}}.
\end{center}
The definition of the retrofunctors
$$\left(p_{\mathcal T}, \varphi_{T,S}\right) : T_{\mathsf{defl}}(\mathcal T)(T\mid S)\rightarrow \Zg(B(\mathcal L))(pT,pS)$$
is more interesting. First, for every $f\in\mathcal F(\omega,\tau)$ define the functions
\begin{align*}
\varphi_{A:\omega,B:\omega}(-,\eta_f) : T_{\mathsf{defl}}(\mathcal T)_{\omega}(A:\omega\mid B:\omega)_0 &\rightarrow T_{\mathsf{defl}}(\mathcal T)(A:\omega\mid B:\omega)_1 \\
\varphi_{C:\tau,D:\tau}(-,\varepsilon_f) : T_{\mathsf{defl}}(\mathcal T)_{\bar f;f}(C:\tau\mid D:\tau)_0 &\rightarrow T_{\mathsf{defl}}(\mathcal T)(C:\tau\mid D:\tau)_1
\end{align*}
via the following assignments:
\begin{center}
\scalebox{1}{\input{local-retrofunctor-free-model.tikz}},
\end{center}
where $x:(A:\omega\mid B:\omega)$ is any term in the fibre above $\omega$, while $y:(C:\tau\mid f(A):\tau)$ and $z:(f(A):\tau\mid D)$ are any terms in the fibre above $\tau$: note that the terms in the fibres are precisely the internal terms. Further, note that the structural equations guarantee that it does not matter which side of the term $x$ the unit is applied at: the resulting 2-cells are equal. When defining the retrofunctor on the counit, note that indeed any term above $\bar f;f$ decomposes in the displayed way, by pushing any internal terms between $\coarsen_f$ and $\refine_f$ outside: the structural equations guarantee that it does not matter which side the terms are pushed to. The functions $\varphi_{T,S}$ are defined analogously on any opfibrational-fibrational term pair. Since the only non-trivial 2-cells in $\Zg(B(\mathcal L))$ are the units and the counits, it remains to extend the functions to composite terms, which is done by the following recursion:
\begin{align*}
\varphi_{(T,T'),(S,S')}\left(x\otimes y,\eta_{f\times g}\right) &\coloneq\varphi_{T,S}\left(x,\eta_f\right)\otimes\varphi_{T',S'}\left(y,\eta_g\right) \\
\varphi_{(T,T'),(S,S')}\left(x\otimes y,\varepsilon_{f\times g}\right) &\coloneq\varphi_{T,S}\left(x,\varepsilon_f\right)\otimes\varphi_{T',S'}\left(y,\varepsilon_g\right) \\
\varphi_{T,S}(x,\alpha;\alpha') &\coloneq \varphi_{T,S}(x,\alpha);\varphi_{T,S}(y,\alpha') \\
\varphi_{T,S}(x;y,\alpha*\alpha') &\coloneq \varphi_{T,U}(x,\alpha)*\varphi_{U,S}(y,\alpha').
\end{align*}
\begin{remark}
Note that it would have been sufficient to only define $\varphi_{A:\omega,A:\omega}\left(\id_A,\eta_f\right)$ on all the internal identity terms and, likewise, $\varphi_{f(A):\tau,f(A):\tau}\left(\coarsen_f;\refine_f,\varepsilon_f\right)$ on only the opfibrational-fibrational composites -- the above recursion then takes care of all the other terms. However, we find it more intuitive to define the whole function for a fixed (co)unit at once. This also makes it evident that the construction is well-defined with respect to the structural identities.
\end{remark}

\begin{proposition}
For any deflational theory $\mathcal T$ with the layered signature $\mathcal L$, the local retrofunctor
$$\left(p_{\mathcal T},\varphi\right):T_{\mathsf{defl}}(\mathcal T)\rightarrow \Zg(B(\mathcal L))$$ is a split monoidal deflation.
\end{proposition}
\begin{proof}
The functor $p$ is a decomposition lifting by construction: given a decomposition in the base category, its chosen lifting is given by the corresponding sequence of (op)fibrational terms (i.e.~by ``inflating'' the diagrams), interspersed with the internal terms present in the original morphism. Any two such liftings are equal, since the internal terms slide inside the diagrams. Since the codomain of $\varphi_{S,S}\left(\bar F;F,\varepsilon_f\right)$, where $F:T\rightarrow S$ is a lifting of $f$ (i.e.~some composite of opfibrational terms), is $\id_S$ by definition, we indeed have a split deflation.

To show that the deflation is monoidal, we use Lemma~\ref{lma:monoidal-deflation-opfibration-indexed-monoids}: it is equivalent to show that the opfibrational restriction has indexed monoids. We observe that the opfibrational restriction gives the syntactic opfibrational model $p_{\mathcal T}:T(\mathcal T)\rightarrow B(\mathcal L)$, which, by Proposition~\ref{prop:syntactic-opfibrational-model} has indexed monoids.
\end{proof}

\begin{corollary}
Any deflational theory $\mathcal T=(\mathcal L,E^0,E^1,\eta,E^2)$ gives rise to the deflational model $\left(p_{\mathcal T},\varphi,i\right)$ whose interpretation functions $i$ are given by:
\begin{itemize}
\item $i(\omega)=\omega$,
\item $i_{\omega,\tau}(f)=f$,
\item $i^{\omega}(a)= a:\omega$,
\item $i^{\omega}_{A,B}(\sigma)= \sigma : (A:\omega\mid B:\omega)$,
\item $i_{x,y}(\alpha)=\alpha$.
\end{itemize}
\end{corollary}
As for opfibrational theories, the syntactic models give the left adjoint to the forgetful functors from models to theories.
\begin{theorem}\label{thm:opfibrational-adjoint}
The construction of deflational models from deflational theories extends to a functor $\DeflTh\rightarrow\DeflThMod$, which is moreover the left adjoint to the forgetful functor.
\end{theorem}
\begin{proof}
Any morphism of deflational theories $F:\mathcal T\rightarrow\mathcal K$ recursively extends to types, terms and 2-terms, as shown in Section~\ref{sec:types-terms}, giving the 2-functor between the total categories $T_{\mathsf{defl}}(\mathcal T)\rightarrow T_{\mathsf{defl}}(\mathcal K)$. The induced functor between the base categories $\Zg(B(\mathcal T))\rightarrow\Zg(B(\mathcal K))$ is given by extending the action of $F$ on the layers and the generators to a morphism of indexed monoids, and then further extending to a morphism of zigzag 2-categories. This gives a morphism of deflations compatible with the morphism of theories, and hence a morphism in $\DeflThMod$.

By Proposition~\ref{prop:deflational-model-preserves-structural}, the syntactic model of $\mathcal T$ is the initial object of the fibre $\OpFThMod(\mathcal T)$. Moreover, any morphism out of the term model of $\mathcal T$ into any model of $\mathcal K$ factors through the term model of $\mathcal K$, thus exhibiting the desired adjunction.
\end{proof}
\begin{corollary}\label{cor:deflational-completeness}
For any deflational theory $\mathcal T$ and layered signature $\mathcal L$, we have the following equivalences:
\begin{align*}
\DeflThMod(\mathcal T) &\simeq \quot{F(\mathcal T)}{\MonDefl_{\mathsf{sp}}}, \\
\DeflMod(\mathcal L) &\simeq \quot{F(\mathcal L)}{\MonDefl_{\mathsf{sp}}}.
\end{align*}
\end{corollary}

\subsection{Deflational theories and (op)indexed monoidal categories}\label{subsec:defl-th-opin-moncat}

Given a deflational theory, there are two inequivalent ways to obtain an opindexed monoidal category from it\footnote{Note that diagram~\ref{eq:deflational-thy-to-indexed-monoidal} does not commute!}:
\begin{equation}\label{eq:deflational-thy-to-indexed-monoidal}
\scalebox{1}{\input{deflational-thy-to-indexed-monoidal.tikz}}.
\end{equation}
First, one can take the underlying opfibrational theory, form its syntactic category, and apply the equivalence of opfibrations with indexed monoids and opindexed monoidal categories (the down-right path in diagram~\eqref{eq:deflational-thy-to-indexed-monoidal}). Second, one can form the syntactic category of the deflational theory, take its opfibrational restriction and then apply the equivalence (the right-down path). The resulting opindexed monoidal categories will not, in general, be equivalent, as the first approach ignores information (the 2-cells). We have already seen an important example of this in Proposition~\ref{prop:cobox-iff-faithful}. Here, however, we focus on the case when the two approaches {\em are} equivalent.

A deflational theory is {\em minimal} if it has no generating 2-cells. Thus, a minimal deflational theory only has the structural 2-cells. Clearly, for a minimal deflational theory the two ways of obtaining an indexed monoidal category are equivalent.

Now consider a (not necessarily minimal) deflational theory $\mathcal T$. Recall that $\mathcal T_{\varepsilon}$ stands for the theory obtained from $\mathcal T$ by adding a section to each counit 2-cell in Figure~\ref{fig:structural-twocells-adjoints}.
\begin{proposition}
For any deflational theory $\mathcal T$, both $\mathcal T$ and $\mathcal T_{\varepsilon}$ result in the same (op)indexed monoidal category via first constructing the syntactic deflational category, then restricting to an im-opfibration (the right-down path in diagram~\eqref{eq:deflational-thy-to-indexed-monoidal}). Moreover, if $\mathcal T$ is minimal, then $\mathcal T_{\varepsilon}$ gives rise to the same indexed monoidal category via both paths in diagram~\eqref{eq:deflational-thy-to-indexed-monoidal}.
\end{proposition}
\begin{proof}
Requiring the generating counits in the total category to be surjective corresponds to the components $\varepsilon_{fa,fa}$ of the counits as defined in Proposition~\ref{prop:prof-embedding-adjoints} having sections. But such components always have a section by Remark~\ref{rem:counit-surjective-twocells}, so the requirement is already satisfied.
\end{proof}

The discussion in this subsection shows that (op)indexed monoidal categories are inadequate as a semantics for deflational theories, in the sense that they fail to distinguish between distinct theories. This is one of the reasons we chose to interpret our theories in monoidal deflations rather than indexed monoidal categories. An interesting question thus arises: what are the appropriate 2-cells between (op)indexed monoidal categories (or their embedding into $\Prof$) that would remedy this (see Remark~\ref{rem:minimal-deflation})?

\chapter{Discussion and future work}\label{ch:layered-discussion}
We have introduced layered monoidal signatures (Definition~\ref{def:layered-signature}) and theories (Definition~\ref{def:layered-theory}), as well as provided three kinds of semantics to interpret these, the most general semantics being monoidal deflations (Definition~\ref{def:monoidal-deflation}). The search for semantics for terms representing the bidirectional functor boundaries has lead us to formulate the definition of a deflation~\ref{def:deflation}, which acquires the appropriate monoidal structure by promoting it to a {\em monoidal} deflation. In addition to proving soundness and completeness with respect to three kinds of models via the means of free-forgetful adjunctions, we have started exploring what kinds of constructions are definable {\em within} deflational theories by constructing functor boxes and coboxes in Chapter~\ref{ch:functor-boxes}, and showing how they can be used to detect properties about the monoidal theories and translations at hand. Much more, however, remains to be done. In Chapter~\ref{ch:layered-examples}, we have provided ample examples of how we envision applications of the general theory, hopefully convincing the reader that the theoretical developments have been worthwhile.

We conclude with a discussion of related work in Section~\ref{sec:related-work}, as well as potential future developments. We split the latter in two parts: Section~\ref{sec:theoretical-developments} discusses theoretical questions inherent to the interplay of mathematical structures hitherto considered, while Section~\ref{sec:potential-applications} discusses potential applications of layered monoidal theories outside of the structures needed to define them.

\section{Related work}\label{sec:related-work}

This work connects to two lines of research: {\em internal string diagrams}, especially their usage to reason about profunctors, and the study of monoidal structure on fibrations.

\subsection{Internal string diagrams}

The internal string diagram notation was first introduced in the study of topological quantum field theories (TQFTs) by Bartlett, Douglas, Schommer-Pries and Vicary~\cite{modular-categories}. The notation was further exploited in the study of profunctors and traced monoidal categories by Hu~\cite{hu-thesis} and Hu \& Vicary~\cite{hu-vicary21}. The current work can be seen as providing the generators and relations for the internal string diagrams, constructed via more semantic means in previous work.

\subsection{Open diagrams for coend calculus}

Rom\' an~\cite{roman} has introduced so-called {\em open diagrams} that represent elements inside coends. The semantics of the diagrams is given by {\em pointed profunctors}, which are closely related to collages (Definition~\ref{def:collage}).

\subsection{Collages of string diagrams}

Braithwaite and Rom\' an~\cite{braithwaite-roman23} have studied {\em collages of string diagrams}, which is closely related to the notion of a collage (Definition~\ref{def:collage}) discussed here. Their notion ``flattens'' -- or creates a collage of -- a {\em bimodular category}: a category with compatible left and right monoidal actions. The terms~\ref{term:monoid} and~\ref{term:comonoid} can be thought of as the left and right monoidal actions of a monoidal category on itself. This gives a potential direction for a generalisation of this work: what kinds of diagrams account for an action of distinct monoidal categories?

\subsection{String diagrams with holes}

One way to think about the cobox (Definition~\ref{def:cobox}) is as a ``hole'' into which any morphism in the codomain can be plugged. There are many instances of holey string diagrams appearing in the literature, such as in the definition of context-free languages of string diagrams by Earnshaw and Rom\' an~\cite{earnshaw-roman24}. Particularly interesting are the holes in the study of quantum combs, as these have been characterised by using internal string diagrams by Hefford and Comfort~\cite{hefford-comfort23}.

\subsection{Ponto-Shulman string diagrams for monoidal fibrations}

Ponto and Shulman~\cite{ponto-shulman12} have introduced string diagrams to reason about the situation where every category in the image of an $\cat X$-indexed monoidal category (Definition~\ref{def:opindexed-monoidal-category}) is symmetric, and the indexing category $\cat X$ has finite products. Moeller and Vasilakopoulou~\cite{monoidal-gro} have shown that in such case the fiberwise and the ``global'' monoidal structure coincide. It would be interesting to see whether such diagrams can be obtained as a layered monoidal theory.

\subsection{Free adjoint construction}

A zigzag 2-category (Definiton~\ref{def:zigzag-category}) is obtained from a 1-category by freely adding a right adjoint to each morphism: for each morphism, we add a new morphism in the opposite direction, as well as a minimal amount of 2-cells to make the newly added morphism the right adjoint of the original one. This is somewhat brute force, as the 2-cells are simply postulated by {\em fiat}. A very similar construction is given by Dawson, Paré and Pronk~\cite{adjoining-adjoints03}, where the 2-cells, rather than being additional structure, are diagrams of a particular shape in the original 1-category called {\em fences}. The authors show that this construction indeed makes each arrow of the original 1-category (seen as embedded into the constructed 2-category) a left adjoint, and that the construction is, moreover, universal with respect to functors that send morphisms to left adjoints. It would, therefore, be interesting to investigate whether the construction given here is equivalent to the one of~~\cite{adjoining-adjoints03}.

\section{Further theoretical developments}\label{sec:theoretical-developments}

We list some theoretical developments dictated by the internal mathematical workings of the structures discussed here.

\subsection{2-cells of deflations}

The most pressing theoretical issue was already alluded to in Remark~\ref{rem:minimal-deflation} and Subsection~\ref{subsec:defl-th-opin-moncat}. Namely, it is not clear what do the 2-cells of (monoidal) deflations correspond to in terms of (op)indexed categories, or their embeddings into $\Prof$.

\subsection{Non-strict theories}

As mentioned in the Preliminaries (Subsection~\ref{subsec:opfibrations-indexed-monoidal}), it would be interesting to study graphically the non-strict case of the equivalence between (op)fibrations with monoids and (op)indexed monoidal categories (Theorem~\ref{thm:moeller-and-vasilakopoulou}). Layered monoidal theories provide enough generality to make this possible: one would have to define a theory where both 0-equations and 1-equations are empty, and account for all identifications by isomorphic 1-cells.

\subsection{Extending the definition of a deflation}

Currently, the definition of a deflation requires the codomain be a zigzag 2-category $\Zg(\cat X)$, which seems like an {\em ad hoc} restriction. The improved definition should, instead, require existence, preservation and reflection of finite products, as well as that every 1-cell in the codomain has a right adjoint, akin to various flavours of definitions hyperdoctrines. Closely related to this is investigating the case of global monoidal products, which we address next.

\subsection{Global monoidal products}

When defining string diagrams for fiberwise monoidal products, we noticed that it is convenient to freely add cartesian products to the base category, and impose the condition on the opfibration that it preserves and reflects products. This gives a hint on how the string diagrams for the ``global'' monoidal structure should look like. In this case, both total and the base category are monoidal, the (op)fibration is a strict monoidal functor, while the monoidal product of the total category preserves (op)cartesian liftings~\cite{monoidal-gro}.

\section{Potential applications}\label{sec:potential-applications}

Here we list potential applications of layered monoidal theories, both within mathematics and to other fields.

\subsection{Connections with linear logic}

The connection between linear logic and profunctors has been explored in the literature several times, for example by Dunn~\cite{dunn-thesis} and Comfort~\cite{comfort-profunctors-linear-logic}. We hope that the syntactic approach presented here will help in making precise the connection between internal string diagrams and proof nets~\cite{comfort-profunctors-linear-logic}.

\subsection{Infinite number of layers}

All the examples of layered monoidal theories we have seen in Chapter~\ref{ch:layered-examples} -- with the exception of digital circuits in Section~\ref{sec:digital-circuits} -- have had a finite number of layers. There is, of course, no reason to restrict one's attention to the finite case: the digital circuits example already gives a glimpse of why it is natural to consider an infinite number of layers in certain contexts. Other situations which the author believes are possible to capture with an infinite number of layers are:
\begin{enumerate}
\item {\em infinite monoidal products}: following the construction of infinite tensor products by Fritz and Rischel~\cite{fritz-rischel20}, it would be interesting to explore the case where the layers consist of the finite subsets of an infinite set representing the infinite monoidal product; the resulting layered theory would describe the ``finite stages of construction'' of the infinite tensor product,
\item {\em distances between string diagrams}: following the quantale enriched string diagrams introduced by Lobbia, Różowski, Sarkis and Zanasi~\cite{quantitative-monoidal-algebra}, it would be interesting to see whether the distance structure can arise from indexing rather than enrichment; the layers would be upward closed sets of a quantale ordered by inclusion, and passing from one layer to another could be seen as either losing or gaining precision,
\item more generally, it would be interesting to study order or domain theoretic structure on the indexing category, and see how it interacts with the monoidal structure.
\end{enumerate}

\subsection{Time evolution}

The examples we have seen are static, in the sense that there is no notion of change, and all layers are assumed to be simultaneously accessible. One could introduce {\em time dependence} by having an order on the layers, and some notion of dynamics, determining which future layers are accessible from the current one. This would yield a picture in which a system evolves over time, and even the rules (i.e.~equations and 2-cells) that govern the system may differ depending on the time instance.

\subsection{Microstates versus macrostates}

Even though in the introduction (Section~\ref{sec:layers-abstraction}) we excluded systems where several different microstates correspond to the same macrostate (such as statistical mechanics and thermodynamics) from the discussion, this was mostly in order to limit the scope of the current work. Generally, there is nothing in the mathematics of layered monoidal theories that prevents modelling such scenarios. Especially the ideas combining multiple layers with time evolution seem particularly fruitful in such case.

\subsection{Functional versus mechanistic explanations in biology}

Given the abundance of biological data from different abstraction levels, there is a need for formal tools in systems biology that are able to rigorously separate the levels of description, most notably the {\em mechanistic} and {\em functional} rules~\cite{rosen-life-itself,modular1999,krivine-siglog}. Differentiating the levels has at least two advantages. First, it gives a way to incorporate the vast amount of experimental data (that is growing at a fast rate) into mechanistic models, where the empirically inferred rules take the role of a higher level description~\cite{krivine-siglog}. Second, it provides a formal framework for {\em counterfactual models}, which has become an important aspect when reasoning about complex systems~\cite{counterfactual, pearl-causality, blaming2015}: it is relevant to know not only that some process $A$ caused some phenomenon $P$ to happen, but also whether $P$ could still happen if $A$ did not occur (or some other process $C$, which potentially inhibits $P$, did occur). Roughly, the lower levels can be thought of as {\em causes} for the {\em effects} at higher levels, which allows to mix-and-match the processes at lower levels to answer counterfactual questions. Layered monoidal theories may provide a suitable framework to study such questions.

\part{A synthetic approach to synthetic chemistry}\label{part:chemistry}

\chapter{Introduction}\label{ch:intro-background}
A chemical reaction can be understood as a rule which tells us what the outcome molecules (or molecule-like objects, such as ions) are when several molecules are put together. If, moreover, the reaction records the precise proportions of the molecules as well as the conditions for the reaction to take place (temperature, pressure, concentration, presence of a solvent etc.), it can be seen as a precise scientific prediction, whose truth or falsity can be tested in a lab, making the reaction reproducible. Producing complicated molecules, as required e.g.~by the pharmaceutical industry, requires, in general, a chain of several consecutive reactions in precisely specified conditions. The general task of synthetic chemistry is to come up with reproducible reaction chains to generate previously unknown molecules (with some desired properties)~\cite{organic-synthesis}. Successfully achieving a given synthetic task requires both understanding of the chemical mechanisms and the empirical knowledge of existing reactions. Both of these are increasingly supported by computational methods~\cite{strieth2020machine}: rule-based and dynamical models are used to suggest potential reaction mechanisms, while database search is used to look for existing reactions that would apply in the context of interest~\cite{compuuter-aided2022}. The key desiderata for such tools are tunability and specificity. Tunability endows a synthetic chemist with tools to specify a set of goals (e.g.~adding or removing a functional group\footnote{Part of a molecule that is known to be responsible for certain chemical function.}), while by specificity we mean maximising yield and minimising side products.

\section{Retrosynthetic analysis}

In this part of the thesis, we focus on the area of synthetic chemistry known as {\em retrosynthesis}~\cite{corey1988robert,compuuter-aided2022,warren1991designing}. While reaction prediction asks what reactions will occur and what outcomes will be obtained when some molecules are allowed to interact, retrosynthesis goes backwards: it starts with a target molecule that we wish to produce, and it proceeds in the ``reverse'' direction by asking what potential reactants would produce the target molecule. While many automated tools for retrosynthesis exist (see e.g.~\cite{route-designer,Coley2017-similarity,coley2018machine,Lin2020,Chen2021,ucak2022,dong2022deep}), there is no uniform mathematical grounding for compositional reasoning about retrosynthesis. The primary contribution of this part is to provide such a mathematically sound framework relevant for the retrosynthetic practice. Indeed, all three kinds of transformations of chemical graphs we consider (reactions, reaction schemes and disconnection rules) appear in the automated retrosynthetic tools. By formalising the methodology at this level of mathematical generality, we are able to provide insights into incorporating features that the current automated retrosynthesis tools lack: these include modelling chirality, the reaction environment, and the protection-deprotection steps (see for example~\cite{filice2010recent}), which are all highly relevant to practical applications. Our formalism, therefore, paves the way for new automated retrosynthesis tools, accounting for the aforementioned features.

Mathematically, we shall formulate retrosynthesis as a {\em layered monoidal theory} developed in Part~\ref{part:layered} of the thesis. In the context of chemistry, the layers play a threefold role: first, each layer represents a reaction environment, second, the morphisms in different layers are taking care of different synthetic tasks, and third, the rules that are available in a given layer reflect the structure that is deemed relevant for the next retrosynthetic step. The latter can be seen as a kind of coarse-graining, where by deliberately restricting to a subset of all available information, we reveal some essential features about the system. Additionally, organising chemical processes into layers allows us to include conditions that certain parts of a molecule are to be kept intact. The presentation here is self-contained up to and including Chapter~\ref{ch:disc-rules}, and in particular, does not require familiarity with the theory developed in Part~\ref{part:layered}. It is only in Chapter~\ref{ch:retrosynthesis} that a particular layered monoidal theory will be constructed.

\section{Threefold view of chemical graphs}

Throughout the presentation, we take three perspectives on chemical processes (Figure~\ref{fig:chemical-trinity}), and discuss the ways in which they are interlinked. The first perspective is that of {\em reaction schemes} (Section~\ref{sec:reaction-schemes}), which encode how bonds and charges change when two parts of chemical compounds interact. Reaction schemes give rise to {\em reactions} via {\em double pushout graph rewriting} (Section~\ref{sec:morphisms}). The category of formal reactions (Definition~\ref{def:category-reactions}) is our second perspective: reactions can be thought of as combinatorial rearrangements of molecules that preserve matter and charge. The third perspective are the local graph rewrites used in retrosynthesis -- known as {\em disconnection rules} -- which capture any possible local change in charge or connectivity. The disconnection rules can be seen as a subset of reaction schemes, and as an axiomatisation of reactions: while they provide a fine-grained view of the chemical transformations, they allow us to recursively define a functor to reactions.
\begin{figure}[h]
\centering
\scalebox{1}{\input{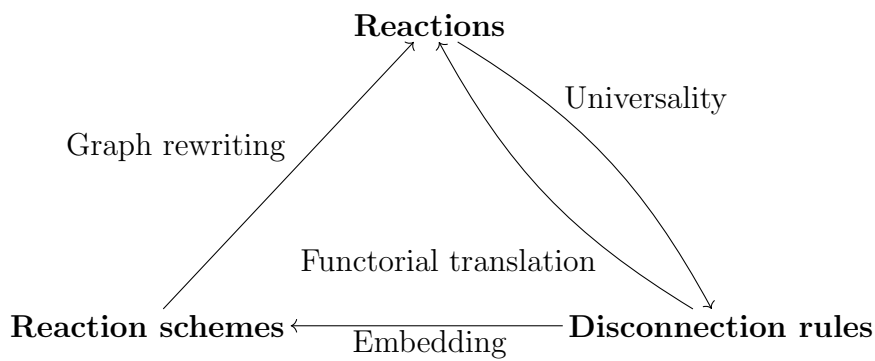}}
\caption{The three perspectives on chemical processes\label{fig:chemical-trinity}}
\end{figure}

Whereas chemical reactions have been studied formally before, a mathematical description of disconnection rules has received far less attention~\cite{rules-to-term2006,approach-ictac,disc-rules2024,tcs-cat-model-org-chem}. Our approach takes a novel perspective on the basic units of retrosynthetic analysis -- the disconnection rules -- by making them first-class citizens of reaction representation. The mathematical and conceptual justification for doing so lies in the fact that, as we show, both disconnection rules and reactions can be arranged into (monoidal) categories~\cite{approach-ictac}, such that there is a functor taking each sequence of disconnection rules to a reaction. Our main result concerning the disconnection rules states that, under a certain axiomatisation, the functor is faithful and full up to isomorphism. Such a categorical perspective provides a precise mathematical meaning to the claim that disconnection rules are sound, complete and universal with respect to the reactions. This implies that every reaction can be decomposed into a sequence of disconnection rules (universality) in an essentially unique way (completeness). More broadly, our contribution incorporates disconnection rules within the framework of applied category theory~\cite{SevenSketches}, which emphasises compositional modelling as a means to uniformly study systems across various disciplines of science.

\section{Structure of Part~\ref{part:chemistry} of the thesis}

The remaining chapters are structured as follows.
\begin{itemize}
\item[Chapter~\ref{ch:chem-background}:] We give a brief overview of the methodology of retrosynthetic analysis, as well as of the existing tools for automating it.
\item[Chapter~\ref{ch:chem-graphs}:] We define the labelled graphs that we use to represent molecular entities: {\em pre-chemical graphs} do not impose any valence conditions on the bonds, while {\em chemical graphs} require the valence of each atom to be respected. The latter form the objects of all three categories appearing in Figure~\ref{fig:chemical-trinity}.
\item[Chapter~\ref{ch:reactions-schemes}:] This chapter focusses on the categorical constructions needed for representing reactions as double pushout graph rewriting, which are used in Section~\ref{sec:reaction-schemes} to define reactions and reaction schemes.
\item[Chapter~\ref{ch:disc-rules}:] We formalise retrosynthetic disconnection rules, from which we define a functor to reactions in Section~\ref{sec:disc-to-react}, where it is moreover proved to be faithful and full up to isomorphism. This provides the soundness, completeness and universality theorem of disconnection rules with respect to the reactions.
\item[Chapter~\ref{ch:retrosynthesis}:] All the perspectives on chemical reactions considered in the previous chapters are put together in the layered theory defined here. We, moreover, describe how to reason about retrosynthesis within the resulting layered theory.
\item[Chapter~\ref{ch:conclusion}:] We conclude and sketch the prospects for future work.
\end{itemize}

\section{Summary of contributions}\label{sec:summary-contributions-chem}

The material in the second part of the thesis has been published in a series of three articles~\cite{approach-ictac,disc-rules2024,tcs-cat-model-org-chem}. We list the novel contributions per chapter.
\begin{itemize}
\item[Chapter~\ref{ch:chem-graphs}:] While the constructions of various labelled graphs is somewhat original, it is simply formalising existing usage of labelled graphs for chemistry.
\item[Chapter~\ref{ch:reactions-schemes}:] Notion of morphisms of pre-chemical graphs as functions that do not remove existing charge or matter is new, as is the proof that the resulting category is $(\mathcal M,\mathcal N)$-adhesive (Theorem~\ref{thm:adhesive}). Likewise, the definition of the category of reactions (Definition~\ref{def:category-reactions}) is new, as is the proof that all reactions arise as double pushout diagrams of reaction schemes (Theorem~\ref{thm:matching-react}, Proposition~\ref{prop:reaction-concrete}).
\item[Chapter~\ref{ch:disc-rules}:] Mathematical formalisation of retrosynthetic disconnection rules (Definition~\ref{def:disc-rules}) and their equational axiomatisation (Definition~\ref{def:disc-cat}) are new. The ensuing construction of a functorial translation from the disconnection rules to reactions, and the proofs of completeness (Theorem~\ref{thm:completeness}) and universality (Theorem~\ref{thm:universality}) are likewise new.
\item[Chapter~\ref{ch:retrosynthesis}:] The mathematical description of retrosynthesis is new.
\end{itemize}

\chapter{Background: Retrosynthesis}\label{ch:chem-background}
Retrosynthetic analysis starts with a target molecule we wish to produce but do not know how. The aim is to ``reduce'' the target molecule to known (commercially available) outcome molecules in such a way that when the outcome molecules react, the target molecule is obtained as a product. This is done by (formally) partitioning the target molecule into functional parts referred to as {\em synthons}, and finding actually existing molecules that are chemically equivalent to the synthons; these are referred to as {\em synthetic equivalents}~\cite{logic-chemical,organic-synthesis,organic-chemistry}. If no synthetic equivalents can be found that actually exist, the partitioning step can be repeated, this time using the synthetic equivalents themselves as the target molecules, and the process can continue until either known molecules are found, or a maximum number of steps is reached and the search is stopped. Note that the synthons themselves do not refer to any molecule as such, but are rather a convenient formal notation for parts of a molecule. For this reason, passing from synthons to synthetic equivalents is a non-trivial step involving intelligent guesswork and chemical know-how of how the synthons {\em would} react if they were independent chemical entities.
\begin{figure}
  \centering
    \scalebox{.8}{\input{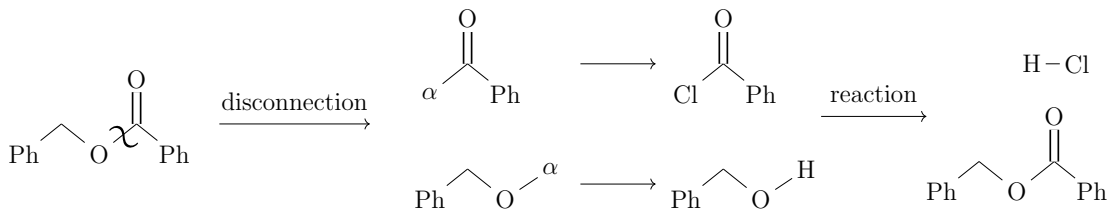}}
  \caption{A retrosynthetic sequence for benzyl benzoate \label{fig:clayden}}
\end{figure}

Clayden, Warren and Greeves~\cite{organic-chemistry} give the example in Figure~\ref{fig:clayden} when introducing retrosynthesis (a synthesis of benzyl benzoate). Here the molecule on the left-hand side (benzyl benzoate) is the target, while the resulting two parts with the symbol $\alpha$ are the synthons. We use the symbol $\alpha$ to indicate where the cut has been made, and hence which atoms have unpaired electrons. Replacing the symbols $\alpha$ in the synthons with $\mathtt{Cl}$ and $\mathtt{H}$, we obtain the candidate synthetic equivalents (benzoyl chloride and benzyl alcohol) shown one step further to the right. As the next step, we find a reaction using the synthetic equivalents as reactants and having the target amongst the products (this can be done by e.g.~looking up a reaction database). This is the simplest possible instance of a retrosynthetic sequence, in the sense that a reaction pathway is found in one iteration. In general, the interesting sequences are much longer, and, importantly, contain information under what conditions the reactions will take place.

\section{Existing tools}
Many tools for automatic retrosynthesis have been successfully developed starting from the 1960s~\cite{route-designer,Coley2017-similarity,Lin2020,Chen2021,ucak2022}. They can be divided into two classes~\cite{compuuter-aided2022}: {\em template-based}~\cite{fortunato2020data,yan2022retrocomposer} and {\em template-free}~\cite{Lin2020,somnath2020learning}. Template-based tools contain a rule database (the {\em template}), which is either manually encoded or automatically extracted. Given a molecule represented as a graph, the model checks whether any rules are applicable to it by going through the database and comparing the conditions of applying the rule to the subgraphs of the molecule~\cite{compuuter-aided2022}. Choosing the order in which the rules from the template and the subgraphs are tried are part of the model design. Template-free tools, on the other hand, are data-driven and treat the retrosynthetic rule application as a translation between graphs or their representations as strings: the suggested transforms are based on learning from known transforms, avoiding the need for a database of rules~\cite{compuuter-aided2022,ucak2022}.

While successful retrosynthesic sequences have been predicted by the computational retrosynthesis tools, they lack a rigorous mathematical foundation, which makes them difficult to compare, combine or modify. Other common drawbacks of the existing approaches include not including the reaction conditions or all cases of chirality as part of the reaction template~\cite{compuuter-aided2022,Lin2020}, as well as the fact that the existing models are unlikely to suggest protection-deprotection steps. Additionally, the template-free tools based on machine learning techniques sometimes produce output that does not correspond to molecules in any obvious way, and tend to reproduce the biases present in the literature or a data set~\cite{compuuter-aided2022}.

For successful prediction, the reaction conditions are, of course, crucial. These include such factors as temperature and pressure, the presence of a solvent (a compound which takes part in the reaction and whose supply is essentially unbounded), the presence of a reagent (a compound without which the reaction would not occur, but which is not the main focus or the target), as well as the presence of a catalyst (a compound which increases the rate at which the reaction occurs, but is itself unaltered by the reaction). The above factors can change the outcome of a reaction dramatically~\cite{matwijczuk2017effect,cook2007determination}. There have indeed been several attempts to include reaction conditions into the forward reaction prediction models~\cite{marcou2015,gao2018,walker2019,maser2021}. However, the search space in retrosynthesis is already so large that adding another search criterion should be done with caution. A major challenge for predicting reaction conditions is that they tend to be reported incompletely or inconsistently in the reaction databases~\cite{coley2017-outcomes}.

Chirality (mirror-image asymmetry) of a molecule can alter its chemical and physiological properties, and hence constitutes a major part of chemical information pertaining to a molecule. While template-based methods have been able to successfully suggest reactions involving chirality (e.g.~\cite{Coley2017-similarity}), the template-free models have difficulties handling it~\cite{Lin2020}. This further emphasises usefulness of a framework which is able to handle both template-based and template-free models.

The protection-deprotection steps are needed when more than one functional group of a molecule $A$ would react with a molecule $B$. To ensure the desired reaction, the undesired functional group of $A$ is first ``protected'' by adding a molecule $X$, which guarantees that the reaction product will react with $B$ in the required way. Finally, the protected group is ``deprotected'', producing the desired outcome of $B$ reacting with the correct functional group of $A$. So, instead of having a direct reaction $A+B\rightarrow C$ (which would not happen, or would happen imperfectly, due to a ``competing'' functional group), the reaction chain is:
\begin{align*}
&\text{(1) } A + X \rightarrow A'\text{ (protection)}, \\
&\text{(2) } A' + B \rightarrow C', \\
&\text{(3) } C' + Y \rightarrow C\text{ (deprotection).}
\end{align*}
The trouble with the protection-deprotection steps is that they temporarily make the molecule larger, which means that an algorithm whose aim is to make a molecule smaller will not suggest them.

\chapter{Chemical graphs}\label{ch:chem-graphs}
We define a {\em chemically labelled graph} (Definitions~\ref{def:chemlabgraph}) as a labelled graph whose edge labels indicate the bond type (either covalent or ionic), and whose vertex labels are either atoms or binding sites together with a charge (Definitions~\ref{def:chemlabgraph},~\ref{def:prechemgraph} and~\ref{def:chemgraph}). The intermediate notion of {\em pre-chemical graph} (Definitions~\ref{def:prechemgraph}) further imposes some conditions on the ionic bonds and the binding sites. The most restrictive notion is that of a {\em chemical graph}, which further requires that the valence of the atoms is respected, and the charge of the binding sites is either $0$ or $1$. In the latter case, the binding site may be thought of as an unpaired electron.

Chemical graphs are the objects of the reaction category (Definition~\ref{def:category-reactions}) and the disconnection category (Definition~\ref{def:disc-cat}), as well as the targets on which the reaction schemes operate (Definition~\ref{def:reaction} and Theorem~\ref{thm:matching-react}).

In order to account for chirality, we add spatial information to chemical graphs, promoting them to {\em oriented (pre-)chemical graphs} (Definition~\ref{def:orient-molpart}) in Section~\ref{sec:chirality}.

\section{Chemically labelled graphs}

Let us fix the following notions needed for Definitions~\ref{def:chemlabgraph},~\ref{def:prechemgraph} and~\ref{def:chemgraph}:
\begin{itemize}
\item a countable set of {\em vertex names} $\VS$,
\item a set of {\em vertex labels} $\Atset$ such that
\begin{enumerate}[label=(\arabic*)]
\item $\Atset$ is finite,
\item $\Atset$ contains the special symbol $\alpha$,
\item $\Atset\setminus\{\alpha\}$ has at least two elements,
\end{enumerate}
\item a {\em valence function} $\mathbf v:\Atset\rightarrow\N$ such that $\mathbf v(\alpha)=1$,
\item the set of {\em edge labels} $\Lab\coloneqq\{0,1,2,3,4,\ib\}$,
\item the functions $\cov,\ion:\Lab\rightarrow\N$ defined by
\begin{itemize}
\item if $x\in\{0,1,2,3,4\}$, then $\cov(x)\coloneqq x$ and $\ion(x)\coloneqq 0$,
\item $\cov(\ib)\coloneqq 0$ and $\ion(\ib)\coloneqq 1$.
\end{itemize}
\end{itemize}

We denote the vertex names by either positive integers or lowercase Latin letters, as appropriate to the situation. While we only make three formal assumptions about the set of vertex labels, in all the examples we shall assume that $\Atset$ contains a symbol for each main-group element of the periodic table: $\{H,C,O,P,\dots\}\sse\Atset$. For this reason, we will also refer to $\Atset$ as the {\em atom labels}. The special symbol $\alpha$ may be thought of as representing an unpaired electron or a free charge. Similarly, we shall assume in the examples that the valence of an element symbol is the number of electrons in its outer electron shell. The integers $\{0,1,2,3,4\}$ in the set of edge labels stand for covalent bonds, while $\ib$ stands for an ionic bond.
\begin{remark}
The reason for choosing such level of generality for the atom labels and their valencies is the ability to model elements which exhibit different valence depending on the context. For instance, one could have separate atom labels for nitrogen whose valence is $5$ (all outer shell electrons are shared or take part in a reaction) or $3$ (two of the outer shell electrons pair with each other).
\end{remark}

\begin{definition}[Chemically labelled graph]\label{def:chemlabgraph}
A {\em chemically labelled graph} is a triple $(V,\tau,m)$, where $V\sse\VS$ is a finite set of {\em vertices}, $\tau:V\rightarrow\Atset\times\Z$ is a {\em vertex labelling function}, and $m:V\times V\rightarrow\Lab$ is an {\em edge labelling function} satisfying $m(v,v)=0$ and $m(v,w)=m(w,v)$ for all $v,w\in V$.
\end{definition}
Thus, a chemically labelled graph is irreflexive (we interpet the edge label $0$ as no edge) and symmetric, and each of its vertices is labelled with an element of $\Atset$, together with an integer indicating the charge. Given a chemically labelled graph $A$, we write $(V_A,\tau_A,m_A)$ for its vertex set and the labelling functions. We abbreviate the vertex labelling function followed by the first projection as $\tau_A^{\At}$, and similarly we write $\tau_A^{\crg}$ for composition with the second projection.

Given a chemically labelled graph $A$ and vertex names $u,v\in\VS$ such that $u\in V_A$ but $v\notin V_A\setminus\{u\}$, we denote by $A(u\mapsto v)$ the chemically labelled graph whose vertex set is $(V_A\setminus\{u\})\cup\{v\}$, and whose vertex and edge labelling functions agree with those of $A$, treating $v$ as if it were $u$. Further, we define the following special subsets of vertices:
\begin{itemize}
\item {\em $\alpha$-vertices}, whose label is the special symbol: $\alpha(A)\coloneqq\tau_A^{-1}(\alpha,\Z)$,
\item {\em chemical vertices}, whose label is not $\alpha$: $\Chem A\coloneqq V_A\setminus\alpha(A)$,
\item {\em neutral vertices}, whose charge is zero: $\Neu A\coloneqq\tau_A^{-1}(\Atset,0)$,
\item {\em charged vertices}, which have a non-zero charge: $\Crg A\coloneqq V_A\setminus\Neu A$,
\item {\em negative vertices}, which have a negative charge:
$$\Crgn A\coloneqq\{v\in V_A : \tau_A^{\crg}(v)<0\},$$
\item {\em positive vertices}, which have a positive charge:
$$\Crgp A\coloneqq\{v\in V_A : \tau_A^{\crg}(v)>0\}.$$
\end{itemize}
The {\em net charge} of a subset $U\sse V_A$ is the integer $\Net U\coloneqq\sum_{v\in U}\tau_A^{\crg}(v)$.

When drawing chemically labelled graphs, we adopt the following conventions:
\begin{enumerate}[label=(\arabic*)]
\item the vertex label from $\Atset$ is drawn at the centre of a vertex,
\item the vertex name is drawn as a superscript on the left (so within a single graph, no left superscript appears twice),
\item a non-zero charge is drawn as a superscript on the right (hence the lack of a right superscript indicates zero charge)
\item the charge $-1$ is abbreviated as $-$, and similarly the charge $1$ as $+$,
\item $n$-ary covalent bonds are drawn as $n$ parallel lines,
\item ionic bonds are drawn as dashed lines.
\end{enumerate}

\begin{example}\label{ex:labelled-graph}
We give three examples of chemically labelled graphs: \textbf{A}, \textbf{B} (carbonate anion) and \textbf{C} (sodium cloride):
\begin{center}
\scalebox{1}{\input{thesis-ch2/example-labelling.tikz}}.
\end{center}
Below we give a table with different kinds of vertex subsets for the graphs:
\begin{center}
\begin{tabular}{ c | c | c | c }
                  & \textbf{A} & \textbf{B} & \textbf{C} \\ \hline
$\alpha$-vertices & $\{a,b\}$ & $\eset$ & $\eset$ \\ \hline
chemical vertices & $\{r,u\}$ & $V_B$ & $V_C$ \\ \hline
neutral vertices & $V_A$ & $\{u,v\}$ & $\eset$ \\ \hline
charged vertices & $\eset$ & $\{w,z\}$ & $V_C$ \\ \hline
negative vertices & $\eset$ & $\{w,z\}$ & $\{v\}$ \\ \hline
positive vertices & $\eset$ & $\eset$ & $\{u\}$ \\ \hline
net charge & $0$ & $-2$ & $0$
\end{tabular}
\end{center}
\end{example}

\begin{definition}[Neighbours]
Given a chemically labelled graph $A$ and a vertex $u\in V_A$, we define the sets of {\em neighbours} $\Nbr_A(u)$, {\em covalent neighbours} $\CN_A(u)$ and {\em ionic neighbours} $\IN_A(u)$ of $u$ as follows:
\begin{align*}
\Nbr_A(u) &\coloneqq\{v\in V_A : m_A(u,v)\neq 0\}, \\
\CN_A(u) &\coloneqq\{v\in V_A : \cov(m_A(u,v))\neq 0\}, \\
\IN_A(u) &\coloneqq\{v\in V_A : \ion(m_A(u,v))\neq 0\}.
\end{align*}
\end{definition}

\begin{definition}[Pre-chemical graph]\label{def:prechemgraph}
A {\em pre-chemical graph} $A=(V_A,\tau_A,m_A)$ is a chemically labelled graph satisfying the following additional conditions:
\begin{enumerate}
\item for all $v\in\alpha(A)$ and $w\in V_A$ we have
\begin{enumerate}
\item $\tau_A^{\crg}(v)\in\{-1,0,1\}$,\label{cgraph:alpha1}
\item $m_A(v,w)\in\{0,1,\ib\}$,\label{cgraph:alpha2}
\item $\Nbr_A(v)$ has at most one element, and if $w\in\Nbr_A(v)$, then $w\in\Chem{A}$,\label{cgraph:alpha3}
\end{enumerate}
\item for all $v\in\Chem A$ we have
\begin{enumerate}
\item either $\IN_A(v)=\{u\}$ for some $u\in\Chem A$, or $\IN_A(v)\sse\alpha(A)\cap\Crgp A$, or $\IN_A(v)\sse\alpha(A)\cap\Crgn A$,\label{cgraph:ion1}
\item if $\IN_A(v)\neq\eset$, then $v\in\Crg A$ and $\tau_A^{\crg}(v)=-\Net{\IN_A(v)}$.\label{cgraph:ion2}
\end{enumerate}
\end{enumerate}
\end{definition}
Conditions~\ref{cgraph:alpha1}-\ref{cgraph:alpha3} say that a vertex labelled by $\alpha$ is either neutral or has charge $\pm 1$, has at most one neighbour, which is necessarily chemical and to which it is connected either via an ionic or a single covalent bond. Conditions~\ref{cgraph:ion1}-\ref{cgraph:ion2} say that an edge with label $\ib$ only connects vertices which have opposite charges such that at least one is chemical and the net charges are equal in magnitude.

We say that a pre-chemical graph $A$ is {\em valence-complete} if for all $v\in V_A$, we have
$$\left|\tau_A^{\crg}(v)\right|+\sum_{u\in V_A}\cov\left(m_A(u,v)\right) = \mathbf v\tau_A^{\At}(v).$$

\begin{definition}\label{def:chemgraph}
A {\em chemical graph} is a valence-complete pre-chemical graph $A$ such that $\alpha(A)\cap\Crgp A=\eset$.
\end{definition}

A {\em synthon} is a chemical graph which is moreover connected. The collection of chemical graphs is, therefore, generated by the disjoint unions of synthons. A {\em molecular graph} is a chemical graph with no $\alpha$-vertices. A {\em molecular entity} is a connected molecular graph.
\begin{example}\label{ex:synthon-molecular}
The chemically labelled graphs in Example~\ref{ex:labelled-graph} are, in fact, chemical graphs with the standard valences of the atoms (i.e.~$\mathbf v(\text{C})=4$, $\mathbf v(\text{O})=2$ and $\mathbf v(\text{H})=\mathbf v(\text{Cl})=\mathbf v(\text{Na})=1$). Since the three graphs are connected, they are all synthons. Moreover, \textbf{B} and \textbf{C} are molecular entities.
\end{example}
\begin{example}\label{ex:prechemgraph}
We give an example of a pre-chemical graph which fails to be a chemical graph (the valence of the vertices $1$, $2$, $4$ and $6$ is not correct). This graph appears as part of the reaction scheme for glucose phosphorylation in Example~\ref{ex:reaction-scheme}.
\begin{center}
\scalebox{1}{\input{thesis-ch2/example-prechemgraph.tikz}}.
\end{center}
\end{example}

\section{Chirality}\label{sec:chirality}
An important part of chemical data is stereochemistry, that is, spatial orientation of the molecule: many molecules of interest (like pharmaceuticals) possess chiral enantiomers (i.e.~molecules that have the same atoms and connectivity, but are mirror images of each other due to spatial orientation) which have different properties. We therefore wish to incorporate (rudimentary) spatial information into (pre-)chemical graphs. The idea is to record for each triple of atoms whether they are on the same line or not, and similarly, for each quadruple of atoms whether they are in the same plane or not.

While the results of the following sections do not account for orientation and only manipulate the connectivity of a graph, we feel that spatial orientation is such an important aspect of a chemical compound that it has to be accounted for as the immediate next step in this line of work (see Subsection~\ref{sec:future-work}). We therefore include this subsection, outlining how to incorporate spatial orientation at the level of objects.

  \begin{definition}[Triangle relation]\label{def:plane-rel}
    Let $S$ be a set. We call a ternary relation $\mathcal P\sse S\times S\times S$ a {\em triangle relation} if the following hold for all elements $A$, $B$ and $C$ of $S$: (1) $ABB\notin\mathcal P$, and (2) if $\mathcal P(ABC)$ and $\mathfrak p(ABC)$ is any permutation of the three elements, then $\mathcal P(\mathfrak p(ABC))$.
  \end{definition}
  \begin{definition}[Tetrahedron relation]\label{def:tet-rel}
    Let $S$ be a set, and let $\mathcal P$ be a fixed triangle relation on $S$. We call a quaternary relation $\mathcal T\sse S\times S\times S\times S$ a {\em tetrahedron relation} if the following hold for all elements $A$, $B$, $C$ and $D$ of $S$: (1) if $\mathcal T(ABCD)$, then $\mathcal P(ABC)$, and (2) if $\mathcal T(ABCD)$ and $\mathfrak p(ABCD)$ is any even permutation of the four elements, then $\mathcal T(\mathfrak p(ABCD))$.
  \end{definition}

Unpacking the above definitions, a triangle relation is closed under the action of the symmetric group $S_3$ such that any three elements it relates are pairwise distinct, and a tetrahedron relation is closed under the action of the alternating group $A_4$ such that if it relates some four elements, then the first three are related by some (fixed) triangle relation (this, inter alia, implies that any related elements are pairwise distinct, and any $3$-element subset is related by the fixed triangle relation).

\begin{figure}
  \centering
    \input{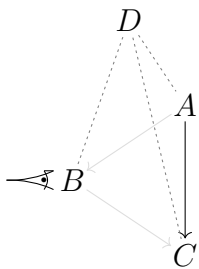}
  \caption{Observer looking at the edge $AC$ from $B$ sees $D$ on their right. \label{fig:observer}}
\end{figure}

The intuition is that the triangle and tetrahedron relations capture the spatial relations of (not) being on the same line or plane: $\mathcal P(ABC)$ stands for $A$, $B$ and $C$ not being on the same line, that is, determining a triangle; similarly, $\mathcal T(ABCD)$ stands for $A$, $B$, $C$ and $D$ not being in the same plane, that is, determining a tetrahedron. The tetrahedron is moreover oriented: $\mathcal T(ABCD)$ does not, in general, imply $\mathcal T(DABC)$. We visualise $\mathcal T(ABCD)$ in Figure~\ref{fig:observer} by placing an ``observer'' at $B$ who is looking at the edge $AC$ such that $A$ is above $C$ for them. Then $D$ is on the right for this observer. Placing an observer in the same way in a situation where $\mathcal T(DABC)$ (which is equivalent to $\mathcal T(CBAD)$), they now see $D$ on their left.
\begin{remark}
We chose not to include the orientation of the triangle, which amounts to the choice of $S_3$ over $A_3$ in the definition of a triangle relation (Definition~\ref{def:plane-rel}). This is because we assume that our molecules float freely in space (e.g.~in a solution), so that there is no two-dimensional orientation.
\end{remark}

The following example demonstrates that the triangle and tetrahedron relations indeed capture triangles and tetrahedrons in the Euclidean setting.
\begin{example}
Let us define the triangle relation $\mathcal P$ on the $3$-dimensional Euclidean space $\R^3$ by letting $\mathcal P(abc)$ if and only if $(b-a)\times (c-a)\neq 0$, where $\times$ denotes the vector product. We then have $\mathcal P(abc)$ precisely when $c$ does not lie on the line determined by $a$ and $b$, that is, when the three points uniquely determine a plane in $\R^3$.

With respect to the above triangle relation, let us define the tetrahedron relation $\mathcal T$ by letting $\mathcal T(abcd)$ if and only if $\overline{(b-a)(c-a)(d-a)}>0$, where the bar denotes the scalar triple product. We then have $\mathcal T(abcd)$ precisely when the points $a$, $b$, $c$ and $d$ are vertices of a non-degenerate (a non-zero volume) tetrahedron in such a way that $d$ lies on that side of the plane determined by $a$, $b$ and $c$ to which the vector $(b-a)\times (c-a)$ points (see Figure~\ref{fig:observer}).
\end{example}

\begin{definition}[Oriented chemically labelled graph]\label{def:orient-molpart}
An {\em oriented chemically labelled graph} is a tuple $(A,\mathcal P,\mathcal T)$ where $A$ is a chemically labelled graph, $\mathcal P$ is a triangle relation on $V_A$ and $\mathcal T$ is a tetrahedron relation on $V_A$ with respect to $\mathcal P$.
\end{definition}
An {\em oriented (pre-)chemical graph} is a chemically labelled graph, which is also a (pre-)chemical graph (Definitions~\ref{def:prechemgraph} and~\ref{def:chemgraph}).
\begin{definition}[Preservation and reflection of orientation]
Let $(M,\mathcal P_M,\mathcal T_M)$ and $(N,\mathcal P_N,\mathcal T_N)$ be oriented chemically labelled graphs, and let $f:M\rightarrow N$ be a labelled graph isomorphism. We say that $f$ {\em preserves orientation} (or is {\em orientation-preserving}) if for all vertices $A$, $B$, $C$ and $D$ of $M$ we have:
\begin{enumerate}[label=(\arabic*)]
\item $\mathcal P_M(ABC)$ if and only if $\mathcal P_N(fA,fB,fC)$, and
\item $\mathcal T_M(ABCD)$ if and only if $\mathcal T_N(fA,fB,fC,fD)$.
\end{enumerate}

Similarly, we say that $f$ {\em reflects orientation} (or is {\em orientation-reflecting}) if for all vertices $A$, $B$, $C$ and $D$ of $M$ we have:
\begin{enumerate}[label=(\arabic*)]
\item $\mathcal P_M(ABC)$ if and only if $\mathcal P_N(fA,fB,fC)$, and
\item $\mathcal T_M(ABCD)$ if and only if $\mathcal T_N(fD,fA,fB,fC)$.
\end{enumerate}
\end{definition}
Note that an orientation-reflecting isomorphism differs from an orientation-preserving one only in that preservation requires the tetrahedron relations to be the same (up to an even permutation), while reflection requires them to be the same up to an odd permutation.

\begin{definition}[Chirality]\label{def:chirality}
We say that two oriented chemically labelled graphs are {\em chiral} if there exists an orientation-reflecting isomorphism, but no orientation-preserving isomorphism between them.
\end{definition}
\begin{example}\label{ex:2butanol}
Consider 2-butanol, whose molecular structure we draw in two different ways at the top of Figure~\ref{fig:example1}\footnote{Figure~\ref{fig:example1}, as well as the figure in Example~\ref{ex:13dichloroallene}, have been created by Ella Gale, for which I'm very grateful.}. Here we adopt the usual chemical convention for drawing spatial structure: a dashed wedge indicates that the bond points ``into the page'', and a solid wedge indicates that the bond points ``out of the page''. The spatial structure is formalised by defining the tetrahedron relation for the graph on the left-hand side as the closure under the action of $A_4$ of $\mathcal T(1234)$, and for the one on the right-hand side as (the closure of) $\mathcal T(4123)$. In both cases, the triangle relation is dictated by the tetrahedron relation, so that any three-element subset of $\{1,2,3,4\}$ is in the triangle relation. Now the identity map (on labelled graphs) reflects orientation. It is furthermore not hard to see that every isomorphism restricts to the identity on the vertices labelled with superscripts, so that there is no orientation-preserving isomorphism. Thus the two molecules are chiral according to Definition~\ref{def:chirality}.

By slightly modifying the structures, we obtain two configurations of isopentane, drawn at the bottom of Figure~\ref{fig:example1}. However, in this case we can find an orientation-preserving isomorphism (namely the one that swaps vertices $2$ and $4$), so that the molecules are not chiral.
\end{example}
\begin{figure}
  \centering
    \includegraphics[scale=0.23]{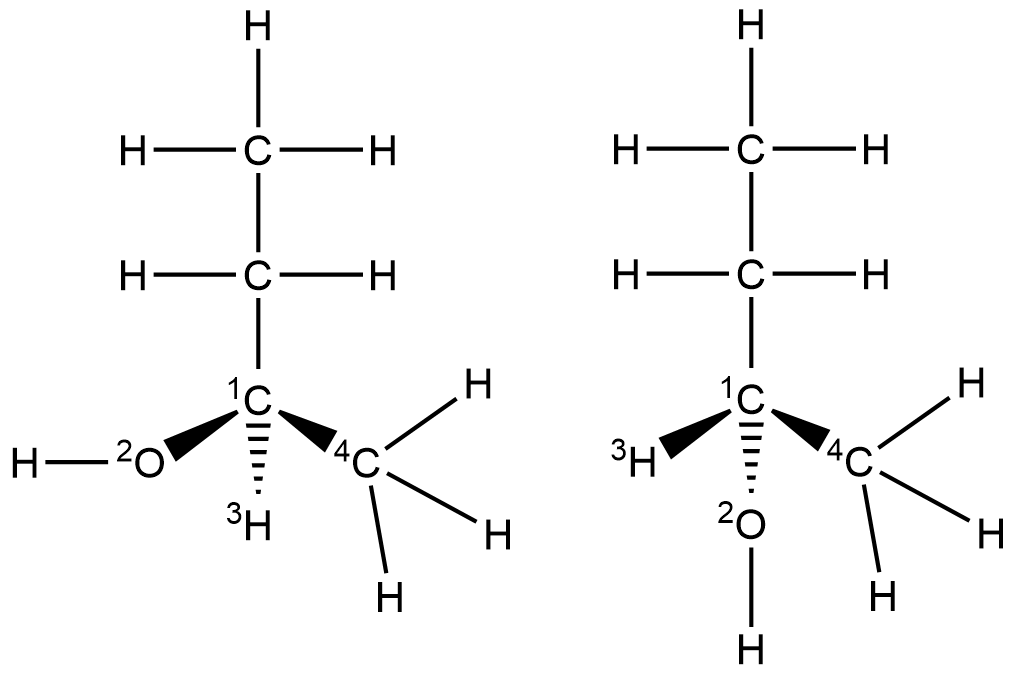}

    \vspace{5pt}

    \includegraphics[scale=0.23]{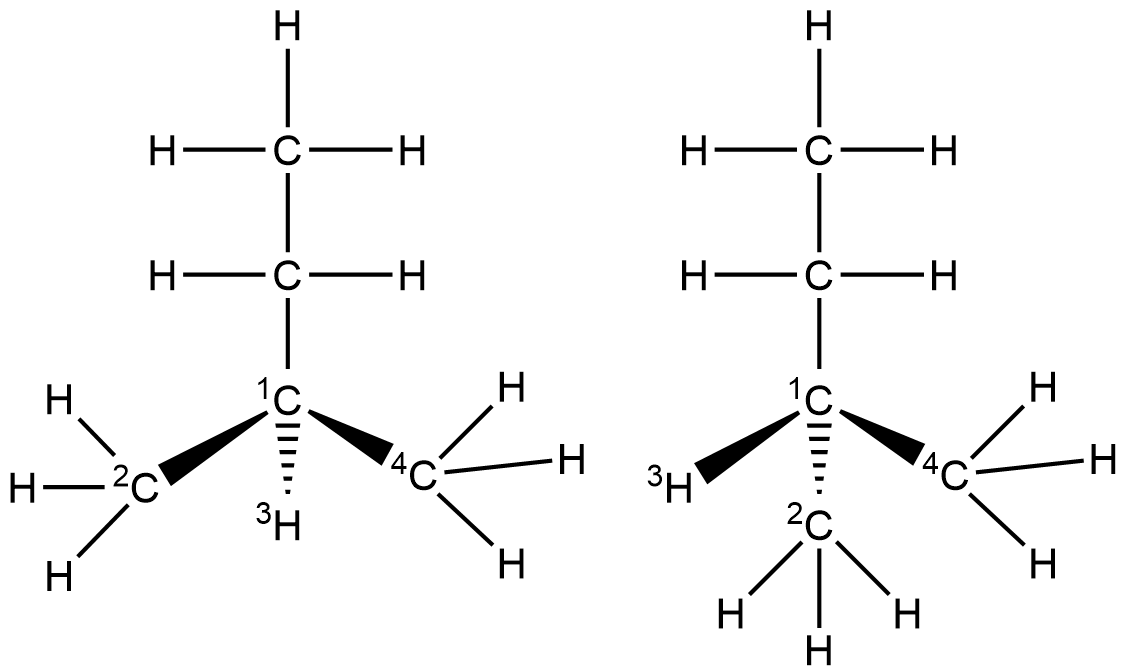}
    \caption{Top: two configurations of 2-butanol. Bottom: two configurations of isopentane.\label{fig:example1}}
\end{figure}
\begin{example}\label{ex:13dichloroallene}
Example~\ref{ex:2butanol} with 2-butanol demonstrated how to capture central chirality using Definition~\ref{def:chirality}. Now consider 1,3-dichloroallene as an example of axial chirality. We draw two versions, as before:
\begin{center}
\includegraphics[scale=0.23]{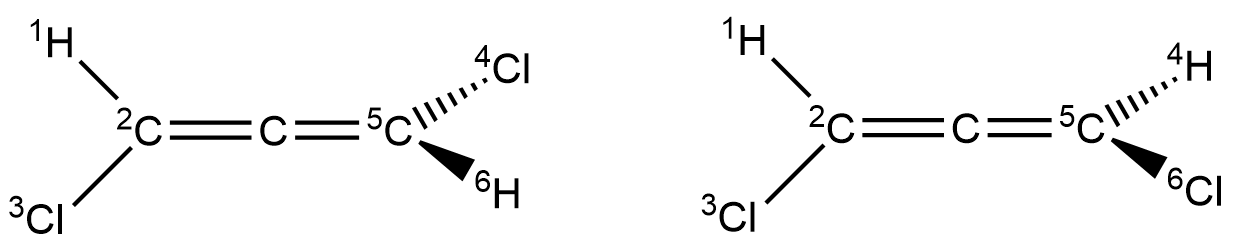}
\end{center}
The tetrahedron relation is generated by $\mathcal T(1234)$ and $\mathcal T(6123)$ for both molecules (note, however, that the vertices $4$ and $6$ have different labels). Now the isomorphism which swaps vertices $4$ and $6$ and is identity on all other vertices is orientation-reflecting, but not orientation-preserving. The only other isomorphism is $\{1\mapsto 4, 2\mapsto 5, 3\mapsto 6, 4\mapsto 3, 5\mapsto 2, 6\mapsto 1\}$, which does not preserve orientation. Thus the two molecules are indeed chiral.
\end{example}

\chapter{Reactions and reaction schemes}\label{ch:reactions-schemes}
In this chapter we equip pre-chemical graphs with morphisms, making them into a category. We identify two classes of morphisms, {\em vertex embeddings} (Definition~\ref{def:vertex-embedding}) and {\em matchings} (Definition~\ref{def:matching}), which are used for double pushout rewriting in the next section. As a mathematical prerequisite for double pushout graph rewriting, we prove in Section~\ref{sec:adhesivity} that vertex embeddings and matchings give the category of pre-chemical graphs an adhesive structure (Theorem~\ref{thm:adhesive}).

A morphism of pre-chemical graphs is a function which, intuitively, preserves any resources and structure (matter, charge and bonds) present in the domain. While the main notion of a chemically meaningful transformation is that of a reaction (Definition~\ref{def:category-reactions}), it is not a natural notion of a graph morphism, as it only captures a subclass of bijective maps. The morphisms are needed to capture reaction schemes (Definition~\ref{def:reaction-scheme}) as well as their instances (Definition~\ref{def:reaction}), which we show to be in one-to-one correspondence with reactions (with some mild assumptions on the subsets changed by the reaction) in Proposition~\ref{prop:reaction-concrete}.

\section{Morphisms of pre-chemical graphs}\label{sec:morphisms}

Recall that a pre-chemical graph (Definition~\ref{def:prechemgraph}) is a chemically labelled graph such that the charge of an $\alpha$-vertex in $\{-1,0,1\}$, the $\alpha$-vertices may only be connected to chemical vertices via a single covalent or ionic bond, and ionic bonds connect (sets of) vertices with equal but opposite net charge.

\begin{definition}[Morphism of pre-chemical graphs]\label{def:morphism}
A {\em morphism} of pre-chemical graphs $f:A\rightarrow B$ is a function $f:V_A\rightarrow V_B$ such that its restriction to the chemical vertices $f|_{\Chem A}$ is injective, the images $f(\Chem A)$ and $f(\alpha(A))$ are disjoint, and for all $v,u\in V_A$ and $w,z\in V_B$ we have
\begin{enumerate}
\item if $v\in\Chem A$, then $\tau_B^{\At}(fv)=\tau_A^{\At}(v)$,
\item if $v\in\Crgp A$, then $f(v)\in\Crgp B$; and if $v\in\Crgn A$, then $f(v)\in\Crgn B$,
\item if $\Net{f^{-1}(w)}\neq 0$, then $\Net{f^{-1}(w)}=\tau_B^{\crg}(w)$,\label{morph:crg-alpha}
\item if $m_A(v,u)=\ib$, then $m_B(fv,fu)=\ib$,
\item if $v,u\in\Chem A$ and $m_A(v,u)\neq 0$, then $m_B(fv,fu)=m_A(v,u)$,\label{morph:chem-edge}
\item if $w\in f(\alpha(A))$ and $z=f(b)$ for some $b\in\Chem A$ such that
$$k\coloneqq\sum_{a\in f^{-1}(w)} \cov(m_A(a,b))\neq 0,$$
then $k=\cov(m_B(w,z))$.\label{morph:chem-alpha}
\end{enumerate}
\end{definition}
\begin{example}\label{ex:morphism}
The idea of a morphism is that it preserves all the atoms, bonds and charges present in the domain, potentially more being present in the codomain. We give an example below, where we use superscripts to indicate the underlying function: each vertex in the domain is mapped to the vertex in the codomain with the same superscript:
\begin{center}
\scalebox{1}{\input{thesis-ch2/morphism-example.tikz}}.
\end{center}
\end{example}
Let us denote by $\PChem$ the category of pre-chemical graphs and their morphisms. Note that $\PChem$ has a symmetric monoidal structure given by the disjoint union of chemical graphs.
\begin{proposition}\label{prop:covalent-nbh-subset}
If $f:A\rightarrow B$ is a morphism of pre-chemical graphs, then for all $a,v\in V_A$, we have
\begin{enumerate}[label=(\arabic*)]
\item $\cov\left(m_A(a,v)\right)\leq\cov\left(m_B(fa,fv)\right)$,\label{morphism-edge1}
\item $\sum_{u\in V_A}\cov\left(m_A(a,u)\right)\leq\sum_{u\in V_B}\cov\left(m_B(fa,u)\right)$,\label{morphism-edge2}
\item $f\left(\CN_A(a)\right)\sse\CN_B(fa)$.\label{morphism-edge3}
\end{enumerate}
\end{proposition}
\begin{proof}
If $\cov\left(m_A(a,v)\right)=0$, then~\ref{morphism-edge1} is immediate. Hence let $v\in\CN_A(a)$. By condition~\ref{cgraph:alpha3} of being a pre-chemical graph (Definition~\ref{def:prechemgraph}), at least one of $a$ and $v$ is chemical. If both are chemical, then condition~\ref{morph:chem-edge} of Definition~\ref{def:morphism} yields $m_B(fa,fv)=m_A(a,v)$, so that~\ref{morphism-edge1} holds. Hence suppose that exactly one is an $\alpha$-vertex: without loss of generality, suppose that $a\in\alpha(A)$ and $v\in\Chem A$. Condition~\ref{morph:chem-alpha} Definition~\ref{def:morphism} then yields
$$m_B(fa,fv) = \sum_{u\in f^{-1}f(a)} m_A(u,v)\geq m_A(a,v),$$
as is required for~\ref{morphism-edge1}. Items~\ref{morphism-edge2} and~\ref{morphism-edge3} are easy consequences of~\eqref{morphism-edge1}.
\end{proof}

\begin{definition}[Vertex embedding]\label{def:vertex-embedding}
We say that a morphism $f:A\rightarrow B$ in $\PChem$ is a {\em vertex embedding} if it is injective, bijective on chemical vertices, and for all $u\in V_A$ we have $\tau_B^{\At}(fu)=\tau_A^{\At}(u)$.
\end{definition}
\begin{example}\label{ex:vertex-embedding}
A vertex embedding is an injective morphism that preserves all the atom labels (including $\alpha$) such that all the chemical vertices of the codomain are in its image. In other words, it can only add new charges, edges or $\alpha$-vertices to the domain graph. We give an example below:
\begin{center}
\scalebox{1}{\input{thesis-ch2/embedding-example.tikz}}.
\end{center}
\end{example}
We denote the class of vertex embeddings by $\E$, and will often refer to the elements of $\E$ simply as {\em embeddings}.

\begin{definition}[Ion-closed subset]\label{def:ion-closed-set}
Let $A$ be a pre-chemical graph. A subset $U\sse V_A$ is {\em ion-closed} if for all vertices $u,v\in V_A$, if $u\in U$ and $m_A(u,v)=\ib$, then also $v\in U$.
\end{definition}

\begin{definition}[Matching]\label{def:matching}
A {\em matching} is a morphism $f:A\rightarrow C$ in $\PChem$ such that the conditions~\eqref{morph:crg-alpha},~\eqref{morph:chem-edge} and~\eqref{morph:chem-alpha} of being a morphism (Definition~\ref{def:morphism}) hold without the exception for the zero charge or bond case, and the image $f(V_A)$ is ion-closed.
\end{definition}
\begin{example}\label{ex:matching}
A matching is a morphism with the further restrictions that all charges are preserved (including the zero charge), no new bonds can be added between existing vertices. We thus think of a matching as identifying the domain as a substructure of the codomain. We slightly modify the morphism in Example~\ref{ex:morphism} to obtain a matching:
\begin{center}
\scalebox{1}{\input{thesis-ch2/matching-example.tikz}}.
\end{center}
\end{example}

Let us denote the class of matchings by $\M$. We wish to characterise the images of matchings between valence-complete pre-chemical graphs, which we achieve in Proposition~\ref{prop:matching-matchable}. To this end, we need to define the {\em valence completion} (Definition~\ref{def:valence-completion}) and {\em charge decomposition} (Definition~\ref{def:charge-decomp}) of a subset of a valence-complete pre-chemcial graph. We begin with the following observation.
\begin{lemma}\label{lma:match-closed}
If $m:A\rightarrow C$ is a matching whose domain and codomain are valence-complete, then for every $u\in\Chem A$ we have $m\left(\CN_A(u)\right)=\CN_C(mu)$.
\end{lemma}
\begin{proof}
The inclusion $m\left(\CN(u)\right)\sse \CN(mu)$ is part~\ref{morphism-edge3} of Proposition~\ref{prop:covalent-nbh-subset}. Since $\tau_A(u)=\tau_C(mu)$, we have
$$\sum_{w\in \CN(u)} m_A(u,w) = \sum_{w\in \CN(mu)} m_C(mu,w).$$
If $w\in \CN(u)$ is chemical, then $m_C(mu,mw)=m_A(u,w)$, while if it is an $\alpha$-vertex, then $m_C(mu,mw)=\sum_{z\in m^{-1}m(w)}m_A(u,z)$. It follows that
\begin{align*}
\sum_{w\in \CN(u)} m_A(u,w) &= \sum_{w\in \CN(u)\cap\Chem A} m_A(u,w) + \sum_{w\in \CN(u)\cap\alpha(A)} m_A(u,w) \\
&= \sum_{w\in \CN(u)\cap\Chem A} m_C(mu,mw) \\
&+ \sum_{w\in m(\CN(u)\cap\alpha(A))}\sum_{z\in m^{-1}(w)} m_A(u,z) \\
&= \sum_{w\in m\left(\CN(u)\cap\Chem A\right)} m_C(mu,w) + \sum_{w\in m\left(\CN(u)\cap\alpha(A)\right)} m_C(mu,w) \\
&= \sum_{w\in m\left(\CN(u)\right)} m_C(mu,w).
\end{align*}
Combining this with the first equality, we get that
$$\sum_{w\in \CN(mu)} m_C(mu,w)=\sum_{w\in m\left(\CN(u)\right)} m_C(mu,w).$$
Since $m\left(\CN(u)\right)\sse \CN(mu)$, it follows that the sets are, in fact, equal.
\end{proof}

\begin{definition}[Valence completion]\label{def:valence-completion}
Let $A$ be a valence-complete pre-chemical graph and let $U\sse\Chem A$. We define the sets of formal symbols called the {\em indexed covalent neighbours}, and the {\em indexed ionic neighbours} as follows:
\begin{align*}
\mathcal{CN}(U)\coloneqq\left\{v^u_j : u\in U, v\in\CN_A(u)\cap (V_A\setminus U)\text{ and } j=1,\dots,m_A(u,v)\right\}, \\
\mathcal{IN}(U)\coloneqq\left\{v^{u,\ib}_j : u\in U, v\in\IN_A(u)\cap (V_A\setminus U)\text{ and } j=1,\dots,\left|\tau_A^{\crg}(v)\right|\right\}.
\end{align*}
The {\em valence completion} of $U$ is the pre-chemical graph\footnote{We assume a unique choice of a vertex name for every element in $\mathcal{CN}(U)\cup\mathcal{IN}(U)$, disjoint from the vertex names of $V_A$. For legibility, we omit this technical detail.}
$$U^{\alpha}\coloneqq(U\cup\mathcal{CN}(U)\cup\mathcal{IN}(U),\tau_{\alpha},m_{\alpha})$$
defined by the following labelling functions: for $u,w\in U$ and $z\in\mathcal{CN}(U)\cup\mathcal{IN}(U)$, let $\tau_{\alpha}(u)\coloneqq\tau_A(u)$, $m_{\alpha}(u,w)\coloneqq m_A(u,w)$ and $\tau^{\At}_{\alpha}(z)\coloneqq\alpha$, and for the sets defined above we let
\begin{align*}
\tau^{\crg}_{\alpha}(v^u_j) &\coloneqq 0, \\
\tau^{\crg}_{\alpha}(v^{u,\ib}_j) &\coloneqq \begin{cases}
                                             1 \text{ if } \tau_A^{\crg}(v)>0, \\
                                             -1 \text{ if } \tau_A^{\crg}(v)<0,
                                             \end{cases} \\
m_{\alpha}(u,v^u_j) &\coloneqq 1, \\
m_{\alpha}(u,v^{u,\ib}_j) &\coloneqq \ib.
\end{align*}
\end{definition}

\begin{example}\label{ex:valence-completion}
Consider the chemical graph below left. The valence completion of the dashed subset is given on the right. Note that there is a matching from the valence completion into the original graph given by the identity on the vertices.
\begin{center}
\scalebox{1}{\input{thesis-ch2/valence-completion-example.tikz}}.
\end{center}
\end{example}

\begin{definition}[Charge decomposition]\label{def:charge-decomp}
Let $A$ be a valence complete pre-chemical graph and let $B\sse\Crg A$ be a subset of charged vertices. We define the following sets of formal symbols: {\em indexed positive vertices}, and {\em indexed negative vertices} as follows:
\begin{align*}
\mathcal{PV}(B)\coloneqq\left\{b^{+}_j : b\in B\cap\Crgp A\text{ and } j=1,\dots,\left|\tau_A^{\crg}(b)\right|\right\}, \\
\mathcal{NV}(B)\coloneqq\left\{b^{-}_j : b\in B\cap\Crgn A\text{ and } j=1,\dots,\left|\tau_A^{\crg}(b)\right|\right\}.
\end{align*}
The {\em charge decomposition} of $U$ is the pre-chemical graph
$$B^{\crg}\coloneqq\left(\mathcal{PV}(B)\cup\mathcal{NV}(B),\tau_{\crg},m_{\crg}\right)$$
defined by $\tau_{\crg}\left(b^{+}_j\right)\coloneqq (\alpha,1)$, $\tau_{\crg}\left(b^{-}_j\right)\coloneqq (\alpha,-1)$, and $m_{\crg}$ is the constantly zero function.
\end{definition}

\begin{example}\label{ex:charge-decomp}
Consider the chemical graph below left. The charge decomposition of the oxygen vertex with vertex name $6$ is given by the single negatively charged $\alpha$-vertex drawn on the right. As for the charge decomposition, note that the identity mapping is a matching.
\begin{center}
\scalebox{1}{\input{thesis-ch2/charge-decomposition-example.tikz}}.
\end{center}
\end{example}

We now define the notion of a {\em matchable subset}, which lists the conditions under which a subset of a valence-complete chemical graph comes from a matching.
\begin{definition}[Matchable subset]\label{def:match-subset}
Let $A$ be a valence-complete pre-chemical graph. We say that a subset $U\sse V_A$ is {\em matchable} if it is ion-closed, and for every $u\in U$, either $u\in\Chem U$ and $\Nbr(u)\sse U$, or $u\in\Crg A$, or there is a $v\in\Chem U\cap\Nbr(u)$ with $\Nbr(v)\sse U$.
\end{definition}
Thus a subset is matchable if all boundary vertices and $\alpha$-vertices are either charged or connected to at least one interior chemical vertex. We first observe that being matchable is a necessary condition for a subset to be an image of a matching.
\begin{lemma}\label{lma:matching-matchable}
First, $m(A)$ is ion-closed by definition. Let $m:A\rightarrow B$ be a matching such that $A$ and $B$ are valence-complete. Then $m(A)$ is a matchable subset of $B$.
\end{lemma}
\begin{proof}
If $u\in\Chem A$, then $\Nbr_B(mu)=m(\Nbr_A(u))\sse m(A)$. If $u\in\alpha(A)\cap\Crg A$, then $m(u)\in\Crg B$. If $u\in\alpha(A)\cap\Neu A$, then there is a $v\in\Chem A$ with $m_A(u,v)=1$, so that $m(v)\in\Chem{m(A)}\cap\Nbr_B(mu)$.
\end{proof}

\begin{proposition}\label{prop:matching-matchable}
Let $A$ be a valence-complete pre-chemical graph. A subset $S\sse V_A$ is matchable if and only if there is a matching $m:C\rightarrow A$ with a valence-complete domain such that $m(C)=S$.
\end{proposition}
\begin{proof}
The `if' direction is Lemma~\ref{lma:matching-matchable}. Hence suppose $S$ is matchable. We define the following sets:
\begin{align*}
U &\coloneqq \{ s\in\Chem S : \Nbr(s)\sse S\}, \\
B &\coloneqq \Crg{S\setminus (U\cup\IN(U))}.
\end{align*}
It follows that $U\cup\Nbr(U)\cup B=S$, so that the image of the matching $U^{\alpha}+B^{\crg}\rightarrow A$ is precisely $S$.
\end{proof}

\section{Adhesivity}\label{sec:adhesivity}

Adhesive categories were introduced by Lack and Soboci\' nski~\cite{lack-sobocinski-adhesive2004,lack-sobocinski-adhesive2005} as a categorical setting where pushouts along monomorphisms are well-behaved. The main motivation is to provide an abstract mathematical framework for double pushout graph rewriting. Adhesive categories have been generalised to $\mathcal M$-adhesive categories by Ehrig, Golas and Hermann~\cite{m-adhesive}, allowing for good behaviour of pushouts along a specified subclass of monomorphisms. A further generalisation restricts the class of morphisms that have pushouts (still along a restricted class of monomorphisms), resulting in $(\mathcal M,\mathcal N)$-adhesive categories of Habel and Plump~\cite{mn-adhesive2012}. Here we prove that the category of pre-chemical graphs $\PChem$ is $(\E,\M)$-adhesive (Theorem~\ref{thm:adhesive}), where $\E$ are the vertex embeddings, and $\M$ are the matchings. This result enables double pushout rewriting in $\PChem$, which we will use in the next section to define reaction schemes and their instances.

We begin by stating the definition of an $(\mathcal M,\mathcal N)$-adhesive category. For the sake of brevity, we do so without a detailed discussion of the terms appearing in the definition. We refer the reader to Habel and Plump~\cite{mn-adhesive2012}, Castelnovo, Gadducci and Miculan~\cite{cgm-adhesive2022}, and Castelnovo and Miculan~\cite{cm-axioms-adhesive2024} for the details.
\begin{definition}[$(\mathcal M,\mathcal N)$-adhesive category. Definition~1 in~\cite{mn-adhesive2012}]
Let $\cat C$ be a category, $\mathcal M$ a class of monomorphisms and $\mathcal N$ a class of morphisms in $\cat C$. We say that $\cat C$ is {\em $(\mathcal M,\mathcal N)$-adhesive} if the following properties hold:
\begin{enumerate}
\item $\mathcal M$ and $\mathcal N$ contain all isomorphisms and are closed under composition and decomposition. Moreover, $\mathcal N$ is closed under $\mathcal M$-decomposition: if $g\circ f\in\mathcal N$ and $g\in\mathcal M$, then $f\in\mathcal N$.
\item $(\mathcal M,\mathcal N)$-pushouts and pullbacks along $\mathcal M$-morphisms exist in $\cat C$. Also, $\mathcal M$ and $\mathcal N$ are stable under $(\mathcal M,\mathcal N)$-pushouts and $\mathcal M$-pullbacks.
\item $(\mathcal M,\mathcal N)$-pushouts are $(\mathcal M,\mathcal N)$-van Kampen squares.
\end{enumerate}
\end{definition}

Let us denote by $\mathcal P_{fin}(\VS)$ the category whose objects are the finite subsets of $\VS$, and whose morphisms are functions (so the category is equivalent to the usual category of finite sets). Let $U:\PChem\rightarrow\mathcal P_{fin}(\VS)$ denote the evident forgetful functor. We also write $U:\PChem\rightarrow\Set$ for the forgetful functor into the category of sets and functions. Now define a functor $F:\mathcal P_{fin}(\VS)\rightarrow\PChem$ by letting $F(V)=(V,\tau^{\alpha},m^{0})$, where $\tau^{\alpha}$ and $m^{0}$ are constant functions sending every element to $(\alpha,0)$ and $0$, respectively. A function $f:V\rightarrow W$ is mapped to itself: note that since $F(V)$ does not have any chemical or charged vertices, all the conditions of Definition~\ref{def:morphism} trivialise, so that any function with $F(V)$ as the domain is a morphism of chemical graphs.

\begin{proposition}
The functors
$$F:\mathcal P_{fin}(\VS)\rightleftarrows\PChem :U$$
are adjoint with $F\dashv U$.
\end{proposition}
\begin{proof}
The unit $\eta_V:V\rightarrow UF(V)$ is given by the identity, as $UF(V)=V$. The counit $\varepsilon_A : FU(A)\rightarrow A$ is given by the identity function: the conditions of being a morphism are automatically verified. It is then straightforward to verify the triangle identities.
\end{proof}
\begin{corollary}\label{cor:forget-preserve}
The forgetful functor $U:\PChem\rightarrow\Set$ preserves all finite limits that exist in $\PChem$.
\end{corollary}

\begin{proposition}\label{prop:pb-embedding}
The pullbacks along vertex embeddings exist in $\PChem$. Moreover, the vertex embeddings are stable under pullbacks, and matchings are stable under $\E$-pullbacks.
\end{proposition}
\begin{proof}
Consider the cospan $A\xrightarrow{f}B\xleftarrow{e}C$ where $e\in\E$. Define the chemical graph $Z$ as follows:
\begin{itemize}
\item $V_Z\coloneqq\{a\in V_A : f(a)\in e(V_C)\}$
\item for all $a\in V_Z$, let $\tau_Z^{\At}(a)\coloneqq\tau_A^{\At}(a)$,
\item for $a\in\Chem Z$, if $\tau_A^{\crg}(a)=\tau_C^{\crg}(e^{-1}f(a))\eqqcolon n$, then let $\tau_Z^{\crg}(a)\coloneqq n$, otherwise let $\tau_Z^{\crg}(a)\coloneqq 0$,
\item for $a\in\alpha(Z)$, if both $a\in\Crg A$ and $e^{-1}f(a)\in\Crg C$, then let $\tau_Z^{\crg}(a)\coloneqq\tau_A^{\crg}(a)$, otherwise let $\tau_Z^{\crg}(a)\coloneqq 0$,
\item for $a,q\in\Chem Z$, if $m_A(a,q) = m_C(e^{-1}f(a),e^{-1}f(q))\eqqcolon k$, let $m_Z(a,q)\coloneqq k$, otherwise let $m_Z(a,q)\coloneqq 0$,
\item for $a\in\alpha(Z)$ and $q\in\Chem Z$, if both $a\in\CN_A(q)$ and $e^{-1}f(a)\in\CN_C(e^{-1}f(a))$, then let $m_Z(a,q)\coloneqq 1$, otherwise let $m_Z(a,q)\coloneqq 0$.
\end{itemize}
Define the map $e^*:Z\rightarrow A$ as the identity on vertices, and the map $f^*:Z\rightarrow C$ by the action of $e^{-1}f$ on vertices. Then the resulting square commutes by construction, and the universal property of a pullback is readily verified. Evidently, $e^*\in\E$.

If $f\in\M$, then $f^*(V_Z)$ is ion-closed since $f(V_A)$ is, and $f^*$ preserves all charges and bonds since $f$ does and by construction of $Z$, so that $f^*\in\M$.
\end{proof}

\begin{lemma}\label{lma:e-pbs-reflect}
The forgetful functor $U:\PChem\rightarrow\Set$ preserves and reflects $\E$-pullbacks. In more detail, the commutative square on the left with $e,e'\in\E$ is a pullback in $\PChem$ if and only if its image on the right is a pullback in $\Set$:
\begin{center}
\input{thesis-ch2/pres-refl-pullback.tikz}.
\end{center}
\end{lemma}
\begin{proof}
Preservation is a special case of Corollary~\ref{cor:forget-preserve}. For reflection, suppose that the square on the right is a pullback in $\Set$. This means there is a unique isomorphism from $u:U(A)\rightarrow U(Z)$ to the pullback of $C\xrightarrow f D\xleftarrow e B$ as constructed in the proof of Proposition~\ref{prop:pb-embedding} such that $Uf^*\circ u=Uf'$ and $Ue^*\circ u=Ue'$. Using the properties of the pullback and the fact that $e'$ is an embedding, one then checks that $u$ is an isomorphism in $\PChem$.
\end{proof}

\begin{proposition}\label{prop:em-pushouts}
$\PChem$ has $(\E,\M)$-pushouts, and the classes $\E$ and $\M$ are stable under $(\E,\M)$-pushouts.
\end{proposition}
\begin{proof}
Consider the span $B\xleftarrow{m}A\xrightarrow{e}C$ where $m\in\M$ and $e\in\E$. Define the chemical graph $Y$ whose vertex set is that of $B$ together with the complement of the image of $e$:
$$V_Y\coloneqq V_B\cup (V_C\setminus e(V_A)).$$
Note that $V_C\setminus e(V_A)$ only contains $\alpha$-vertices. The labelling functions are defined as follows for all $b,p\in V_B$ and $c\in V_C\setminus e(V_A)$:
\begin{itemize}
\item $\tau_Y^{\At}(b)\coloneqq\tau_B^{\At}(b)$ and $\tau_Y(c)\coloneqq\tau_C(c)$,
\item if $b\in m(A)$, then $\tau_Y^{\crg}(b)\coloneqq\sum_{d\in em^{-1}(b)}\tau_C^{\crg}(d)$, and $\tau_Y^{\crg}(b)\coloneqq\tau_B^{\crg}(b)$ otherwise,
\item if $b,p\in m(\Chem A)$, then $m_Y(b,p)\coloneqq m_C(em^{-1}(b),em^{-1}(p))$; and if $b\in m(\Chem A)$ and $p\in m(\alpha(A))$, then
$$m_Y(b,p)\coloneqq\sum_{d\in em^{-1}(p)} m_C(em^{-1}(b),d),$$
and $m_Y(b,p)=m_B(b,p)$ otherwise,
\item if $b\in m(\Chem A)$, then $m_Y(b,c)\coloneqq m_C(em^{-1}(b),c)$ and $m_Y(b,c)\coloneqq 0$ otherwise.
\end{itemize}
The map $e^*:B\rightarrow Y$ is defined as inclusion on vertices, and the map $m^*:C\rightarrow Y$ as $me^{-1}$ on the image of $e$, and as inclusion otherwise. Then, by construction, the resulting square commutes and we have $e^*\in\E$ and $m^*\in\M$, and the universal property of a pushout is readily verified.
\end{proof}

\begin{lemma}\label{lma:em-pos-reflect}
The forgetful functor $U:\PChem\rightarrow\Set$ preserves and reflects $(\E,\M)$-pushouts. In more detail, the commutative square on the left with $e,e'\in\E$ and $m,m'\in\M$ is a pushout in $\PChem$ if and only if its image on the right is a pushout in $\Set$:
\begin{center}
\input{thesis-ch2/pres-refl-pushout.tikz}.
\end{center}
\end{lemma}
\begin{proof}
For preservation, observe that $U(Y)$ is the pushout of $U(B)\xleftarrow{Um}U(A)\xrightarrow{Ue}U(C)$, where $Y$ is the pushout of $B\xleftarrow{m}A\xrightarrow{e}C$ as constructed in the proof of Proposition~\ref{lma:em-pos-reflect}. For reflection, suppose that the square on the right is a pushout in $\Set$. This means there is a unique isomorphism from $u:U(Y)\rightarrow U(D)$ such that $u\circ Um^*=Um'$ and $u\circ Ue^*=Ue'$. Using the properties of the pushout and the facts that $e'$ is an embedding and $m'$ is a matching, one then checks that $u$ is an isomorphism in $\PChem$.
\end{proof}

\begin{theorem}\label{thm:adhesive}
The category $\PChem$ is $(\E,\M)$-adhesive.
\end{theorem}
\begin{proof}
It is straightforward to see that both $\E$ and $\M$ contain all isomorphisms, and are closed under composition and decomposition.

To see that $\M$ is closed under $\E$-decomposition, suppose that $g\in\E$ with type $g:B\rightarrow C$ and for some composable morphism $f:A\rightarrow B$ we have $gf\in\M$. The fact that $f(V_A)$ is ion-closed follows from ion-closedness of $gf(V_A)$ and from injectivity of $g$. Since $gf$ preserves the labels of chemical vertices (and of edges between chemical vertices), and the only charge (edge label) that can be mapped by $g$ to zero is zero, we conclude that $f$ must preserve the labels of chemical vertices and of the edges between them. For the last two conditions of being a matching, let $w\in f(\alpha(A))$ and $b\in\Chem A$. Since $gf$ is a matching and $g$ is an embedding, we have
$$\tau_C^{\crg}(gw) = \sum_{a\in f^{-1}(w)}\tau_A^{\crg}(a)\eqqcolon n.$$
If $n=0$, then $\tau_C^{\crg}(gw)=\tau_B^{\crg}(w)$, so that the desired equality holds. If $n<0$, then $\tau_B^{\crg}(w)\leq n$ as $f$ is a morphism, but also $n\leq\tau_B^{\crg}(w)$ as $g$ is a morphism, so that $\tau_B^{\crg}(w)=n$ and the desired equality holds. The case for $n<0$ is symmetric.

For the last condition, we again use the facts that $gf$ is a matching and $g$ is an embedding to obtain
$$m_C\left(g(w),gf(b)\right) = \sum_{a\in f^{-1}(w)}m_A(a,b)\eqqcolon d.$$
If either $d=0$ or $w\in\Chem B$, then $m_C\left(g(w),gf(b)\right)= m_B\left(w,f(b)\right)$, so that the desired equality holds. Hence suppose that $d>0$ and $w\in\alpha(B)$. It follows that there is exactly one $a\in f^{-1}(w)$ such that $m_A(a,b)=1$, and $d=m_B\left(w,f(b)\right)=1$, so that the desired equality holds. Thus indeed $f\in\M$.

Existence of $\E$-pullbacks, as well as stability of $\E$ and $\M$ under ($\E$-)pullbacks was shown in Proposition~\ref{prop:pb-embedding}. Existence of $(\E,\M)$-pushouts and stability of $\E$ and $\M$ under such pushouts is Proposition~\ref{prop:em-pushouts}. Thus it remains to show that $(\E,\M)$-pushouts are van Kampen squares. Consider the commutative cube
\begin{center}
\scalebox{1}{\input{thesis-ch2/van-kampen.tikz}}
\end{center}
whose bottom face is a pushout of an embedding $e$ and a matching $m$, whose back faces are both pullbacks, and whose all vertical morphisms are embeddings. We thus have that $m^*,m'\in\M$ and $e^*,e'\in\E$. We have to show that the top face is a pushout if and only if the front faces are pullbacks. By Lemmas~\ref{lma:e-pbs-reflect} and~\ref{lma:em-pos-reflect}, the top face is a pushout in $\PChem$ if and only if it is a pushout in $\Set$, and the front faces are pullbacks in $\PChem$ if and only if they are pullbacks in $\Set$. Lemma~\ref{lma:em-pos-reflect}, the bottom face is a pushout in $\Set$. Since in $\Set$ pushouts along monomorphisms are van Kampen squares~\cite{lack-sobocinski-adhesive2005}, we indeed have that the top face is a pushout in $\Set$ if and only if the front faces are pullbacks in $\Set$, thus completing the proof.
\end{proof}

\begin{corollary}[Pushout complements]\label{cor:po-complements}
The solid diagram below, where $e\in\E$ and $m\in\M$, can be uniquely completed to a pushout square, with $\hat e\in\E$ and $\hat m\in\M$:
\begin{center}
\input{thesis-ch2/po-complement.tikz}.
\end{center}
\end{corollary}
\begin{proof}
The pre-chemical graph $Z$ is defined to have the vertex set $V_C\setminus m(V_A\setminus e(V_B))$. The atom labels are inherited from $C$, and likewise for the charge and edge labels on $V_C\setminus m(V_A)$, while on $m(V_A)$ the charge and edge labels are those of $B$. Uniqueness is a consequence of $(\E,\M)$-adhesivity~\cite{mn-adhesive2012}.
\end{proof}

\section{Reaction schemes}\label{sec:reaction-schemes}

In this section, we first give a formalisation of chemical reactions using double pushout rewriting -- our first perspective on chemical processes. Our approach is very similar, and inspired by, that of Andersen, Flamm, Merkle and Stadler~\cite{inferring-rule-composition}, with some important differences, such as having more strict requirements on the graphs representing molecular entities, and allowing for the free and unpaired electrons, represented by the symbol $\alpha$. After this, we characterise all the possible graph transformations resulting from such double pushout rewriting as certain partial bijections (Proposition~\ref{prop:reaction-concrete}). This gives rise to our second perspective on chemical processes -- the category of reactions (Definition~\ref{def:category-reactions}).

\begin{definition}[Reaction scheme]\label{def:reaction-scheme}
A {\em reaction scheme} is a span $A\xleftarrow f K\xrightarrow g B$ in $\PChem$ such that $A$ and $B$ are valence-complete and have the same net charge, $f$ and $g$ are vertex embeddings, and the span is terminal in the subcategory of spans with boundaries $A$ and $B$ and legs in $\E$.
\end{definition}

\begin{example}\label{ex:reaction-scheme}
The rule shown below appears in the reaction describing glucose phosphorylation. It is a reaction scheme in the sense of Definition~\ref{def:reaction-scheme}:
\begin{center}
\scalebox{1}{\input{thesis-ch2/chem-rule.tikz}}.
\end{center}
\end{example}

\begin{definition}[Reaction instance]\label{def:reaction}
A {\em reaction instance} is a double pushout diagram
\ctikzfig{thesis-ch2/reaction}
in $\PChem$ such that the top span $A\leftarrow K\rightarrow B$ is a reaction scheme, $m,m',m''\in\M$ are matchings, and $C$ and $E$ are chemical graphs.
\end{definition}

\begin{example}
The following reaction (glucose phosphorylation) is an instance of the reaction scheme in Example~\ref{ex:reaction-scheme}; we have labelled the vertices in the images of matchings on both sides (the sets $U_A$ and $U_B$), and we use the convention from chemistry where an unlabelled vertex is a carbon atom with an appropriate number of hydrogen atoms attached:
\begin{center}
\scalebox{.75}{\input{thesis-ch2/phosphorylation.tikz}}.
\end{center}
Note that the graphs appearing as the boundaries of the span of the reaction scheme are the disjoint unions of the valence completion and charge decomposition of the graphs in the above reaction (cf.~Examples~\ref{ex:valence-completion} and \ref{ex:charge-decomp}). We make this observation precise in Proposition~\ref{prop:reaction-concrete}.
\end{example}

\begin{theorem}\label{thm:matching-react}
Let $C$ be a chemical graph. Given a matching and a reaction scheme as below left, the diagram can be uniquely completed to the reaction instance on the right.
\ctikzfig{thesis-ch2/reaction-matching}
\end{theorem}
\begin{proof}
The pre-chemical graph $D$ can be constructed as the pushout complement (Corollary~\ref{cor:po-complements}). The pre-chemical graph $E$ is then obtained as the pushout, which exists by Proposition~\ref{prop:em-pushouts}.
\end{proof}

A reaction instance can be presented in a more concrete (yet equivalent) way, which involves mappings that are not morphisms in $\PChem$.
\begin{proposition}\label{prop:reaction-concrete}
Let $C$ and $E$ be chemical graphs. The data of a reaction instance $C\rightarrow E$ can be equivalently (up to an isomorphism) presented as a tuple $(U_C,U_E,b,i)$ where $U_C\sse V_C$ and $U_E\sse V_E$ are matchable subsets with equal net charge, $b:\Chem{U_C}\rightarrow \Chem{U_E}$ is a bijection preserving the atom labels, and $i:V_C\setminus U_C\rightarrow V_E\setminus U_E$ is an isomorphism of pre-chemical graphs such that for all $u\in\Chem{U_A}$ and $a\in V_A\setminus U_A$ we have $m_A(u,a)=m_B(bu,ia)$.
\end{proposition}
\begin{proof}
Given a reaction instance
\begin{center}
\input{thesis-ch2/reaction.tikz},
\end{center}
we obtain the desired tuple as $(m(A),m''(B),b,i)$, where $b$ and $i$ are both appropriate restrictions of $g'(f')^{-1}$.

Conversely, given a tuple as in the statement of the proposition, we obtain the following reaction instance
\begin{center}
\input{thesis-ch2/reaction-modified.tikz},
\end{center}
where $m$ and $m'$ are the matchings defined from matchable subsets in the proof of Proposition~\ref{prop:matching-matchable}. In order to define the graph $K$, we first note that the bijection $b:\Chem{U_C}\rightarrow \Chem{U_E}$ induces an atom label preserving bijection
$$\bar b : m^{-1}(\Chem C)\cap\alpha(U_C^*)\rightarrow (m')^{-1}(\Chem E)\cap\alpha(U_E^*)$$
as follows. For every $c\in\Chem{U_C}\cap m(\alpha(U_C^*))$, we define $m^{-1}(c)\rightarrow (m')^{-1}(bc)$ by the following procedure:
\begin{enumerate}
\item Let $N_c\coloneqq a_1,\dots,a_n$ and $C_c\coloneqq b_1,\dots,b_m$ be lists all the neutral and charged vertices of $m^{-1}(c)$, respectively. Similarly, let $N_{bc}\coloneqq c_1,\dots,c_q$ and $C_{bc}\coloneqq d_1,\dots,d_p$ be lists all the neutral and charged vertices of $(m')^{-1}(bc)$. Note that we have $n+m=p+q$.
\item For each $i=1,\dots,n$, let $n_i$ be the unique covalent neighbour of $a_i$. If there is a $c_j\in N_{bc}$ such that $c_j\in\CN_{U_E^*}(bn_i)$, define $a_i\mapsto c_j$, and remove $a_i$ from $N_c$ and $c_j$ from $N_{bc}$.
\item For each $i=1,\dots,m$, if there is a $d_j\in C_{bc}$ with the same charge as $b_i$ such that $b\left(\IN_{U_C^*}(b_i)\right)=\IN_{U_E^*}(d_j)$, define $b_i\mapsto d_j$, and remove $b_i$ from $C_c$ and $d_j$ from $C_{bc}$.
\item The remaining vertices in $N_c$ and $C_c$ may be mapped to any remaining vertices in $N_{bc}$ and $C_{bc}$ (as long as we have a bijection).
\end{enumerate}
Next, define the following subset of $\alpha$-vertices $A\sse V_{U_C}$. For every $u\in\Chem{U_C}$, let $A^{\mathtt c}_u$ be any subset of $\CN_C(u)\cap\alpha(U_C)$ of size
$$\min\left(\left|\CN_C(u)\cap\alpha(U_C)\right|,\left|\CN_E(bu)\cap\alpha(U_E)\right|\right).$$
Let $A^{\mathtt i-}_u$ be any subset of $\IN_C(u)\cap\alpha(U_C)\cap\Crgn{U_C}$ of size
$$\min\left(\left|\IN_C(u)\cap\alpha(U_C)\cap\Crgn{U_C}\right|,\left|\IN_E(bu)\cap\alpha(U_E)\cap\Crgn{U_E}\right|\right).$$
Similarly, let $A^{\mathtt i+}_u$ be any subset of positive ionic $\alpha$-neighbours of $u$, whose size is that of positive ionic $\alpha$-neighbours of $u$ or positive ionic $\alpha$-neighbours of $b(u)$, whichever is smaller. Let $A^{-}$ be any subset of isolated negative $\alpha$-vertices in $U_C$, whose size is the smaller of isolated negative $\alpha$-vertices in $U_C$ and isolated negative $\alpha$-vertices in $U_E$. The set $A^{+}$ is defined similarly, but for isolated positive $\alpha$-vertices. We then define
$$A\coloneqq A^{-}\cup A^{+}\cup\bigcup_{u\in\Chem{U_C}} A^{\mathtt c}_u\cup A^{\mathtt i-}_u\cup A^{\mathtt i+}_u.$$
Note that there is an injection $\iota : A\rightarrow U_E^*$ that takes each vertex $a\in A$ to an $\alpha$-vertex $\iota(a)$ with the same charge, such that the neighbour of $a$ is mapped by $b$ to the neighbour of $\iota(a)$. We denote by
$$\hat b : \Chem{U_C^*}\cup\left(m^{-1}(\Chem C)\cap\alpha(U_C^*)\right)\cup A\rightarrow U_E^*$$
the map that acts as $b+\bar b+\iota$ on the disjoint summands. We now define the graph $K$ as follows:
\begin{itemize}
\item the vertex set is $V_K\coloneqq\Chem{U_C^*}\cup\left(m^{-1}(\Chem C)\cap\alpha(U_C^*)\right)\cup A$,
\item the atom labels are the same as in $U_C^*$,
\item for every $k\in V_K$, if $\tau^{\crg}_{U_C^*}(k)=\tau^{\crg}_{U_E^*}(\hat bk)\eqqcolon n$, then define $\tau^{\crg}_K(k)\coloneqq n$, otherwise define $\tau^{\crg}_K(k)\coloneqq 0$,
\item for all $k,t\in V_K$, if $m_{U_C^*}(k,t)=m_{U_E^*}(\hat bk,\hat bt)\eqqcolon n$, then define $m_K(k,t)\coloneqq n$, otherwise define $m_K(k,t)\coloneqq 0$.
\end{itemize}
The graph $D$ is then constructed as the pushout complement for both $m$ and $m'$.
\end{proof}

Proposition~\ref{prop:reaction-concrete} motivates the following definition:
\begin{definition}[Category of reactions]\label{def:category-reactions}
We denote by $\React$ the {\em category of reactions}, whose
\begin{itemize}
\item objects are chemical graphs,
\item morphisms $A\rightarrow B$ are tuples $(U_A,U_B,b,i)$, where
\begin{itemize}
\item $U_A\sse V_A$ and $U_B\sse V_B$ are subsets with $\Net{U_A}=\Net{U_B}$,
\item $b:\Chem{U_A}\rightarrow\Chem{U_B}$ is a bijection preserving the atom labels,
\item $i:V_A\setminus U_A\rightarrow V_B\setminus U_B$ is an isomorphism of pre-chemical graphs,
\end{itemize}
such that for all $u\in\Chem{U_A}$ and $a\in V_A\setminus U_A$ we have
$$m_A(u,a)=m_B(bu,ia),$$
\item the composition of $(U_A,U_B,b,i):A\rightarrow B$ and $(W_B,W_C,c,j):B\rightarrow C$ is given by
$$(Z_A,Z_C,(c+j)(b+i),ji):A\rightarrow C,$$
where $Z_A\coloneqq U_A\cup i^{-1}(W_B\setminus U_B))$ and $Z_C\coloneqq W_C\cup j(U_B\setminus W_B))$,
\item for a molecular graph $A$, the identity is given by $(\eset,\eset,!,\id_A)$, where $!$ is the unique endomorphism on the empty set.
\end{itemize}
\end{definition}
Note that the composition in $\React$ is {\em not} the composition in the usual category of partial bijections: instead, it crucially relies on the fact that there is an isomorphism between the unchanged parts of the graph. The category $\React$ has a dagger structure: the dagger of $(U_A,U_B,b,i):A\rightarrow B$ is given by $(U_B,U_A,b^{-1},i^{-1}):B\rightarrow A$. Given a morphism $r\in\React$, we will denote its dagger by $\overline r$.

\begin{remark}
Note that the morphisms in $\React$ are slightly more general than tuples arising from reaction schemes in Proposition~\ref{prop:reaction-concrete}. Namely, we do not require the subsets in a morphism in $\React$ to be matchable. This generalisation is, however, merely technical, as we can always extend any subset to the smallest matchable one while keeping the same information about which bonds and charges are reconnected and exchanged. The reason for allowing more morphisms in $\React$ is to obtain an exact correspondence with the disconnection rules (Chapter~\ref{ch:disc-rules}), which operate only on a single edge or vertex at a time, and hence are unable to capture global conditions, such as being matchable.
\end{remark}

\chapter{Disconnection rules}\label{ch:disc-rules}
A {\em disconnection rule} is a partial endofunction on the set of chemical graphs. We define four classes of disconnection rules, all of which have a clear chemical significance: two versions of {\em electron detachment}, {\em ionic bond breaking} and {\em covalent bond breaking}. These local chemical transformations are our third perspective on chemical processes. The reader may want to check Figure~\ref{fig:disc-rules} before Definition~\ref{def:disc-rules} below, as it gives an intuitive explanation of our approach.

For the purposes of mathematical precision, our set of four disconnection rules is more fine-grained than what one would see in a typical textbook on retrosynthesis, where movement of electrons is usually implicitly modelled in the same step as disconnecting a bond, rather than including electron detachment as a separate step (see, for instance, the discussion on the choice of polarity in~\cite[p.~9]{organic-synthesis}).

\section{The category of disconnection rules}\label{sec:disc-rules}

We treat the disconnection rules as syntax, which generate the {\em terms} (Definition~\ref{def:terms}), whose equivalence classes under the equations of Figure~\ref{fig:disc-axioms} form the morphisms in the {\em disconnection category} (Definition~\ref{def:disc-cat}). The payoff such a syntactic presentation is an axiomatic view of chemical reactions: in Section~\ref{sec:disc-to-react}, we construct a functor from the disconnection category to the category of reactions, and show that every reaction can be represented as a sequence of disconnection rules in an essentially unique way.

\begin{figure}
\centering
\scalebox{1}{\input{disc-rules-as-rewrites.tikz}}
\caption{The four disconnection rules.\label{fig:disc-rules}}
\end{figure}

\begin{definition}[Disconnection rules]\label{def:disc-rules}
Let $u,v,a,b\in\VS$ be pairwise distinct vertex names. Let $U\in\{u,uv\}$ and $D\in\{\eset,ab\}$ range over the specified lists of vertex names. The four {\em disconnection rules} are defined by the tables in Figure~\ref{fig:disc-rules-partial-fns} as follows: a chemical graph $A$ is in the domain of $d^U_D$ if $U\sse V_A$ but $D\cap V_A=\eset$, and the additional conditions of the first column (top table) hold; the output chemical graph $d(A)$ has the vertex set $V_A\cup D$, and the labelling functions on $U\cup D$ are defined by the remaining columns (vertex labelling in the top table, edge labelling in the bottom table), while the labelling functions agree with those of $A$ on $V_A\setminus U$.
\end{definition}
\begin{figure}
\centering
\begin{tabular}{| c || c | c | c | c | c | }
\hline
$d^U_D$ & $A\in\dom(d)$ & $\tau^{\crg}_{d(A)}(u)$ & $\tau^{\crg}_{d(A)}(v)$ & $\tau_{d(A)}(a)$ & $\tau_{d(A)}(b)$ \rule{0pt}{10pt}\rule[-5pt]{0pt}{0pt} \\ \hhline{|=#=|=|=|=|=|}
$E^u_{ab}$ & \makecell{$u\in\Chem A$ \\ $u\in\Crgn A$} & $\tau_A^{\crg}(u)+1$ & N/A & $(\alpha,0)$ & $(\alpha,-1)$ \\ \hline
$E^{uv}$ & \makecell{$u\in\Chem A$ \\ $u\notin\Crgn A$ \\ $v\in\alpha(A)$ \\ $m_A(u,v)=1$} & $\tau_A^{\crg}(u)+1$ & $-1$ & N/A & N/A \\ \hline
$I^{uv}$ & \makecell{$m_A(u,v)=\ib$ \\ $u\in\Crgp A$ \\ $v\in\Crgn A$} & $\tau_A^{\crg}(u)$ & $\tau_A^{\crg}(v)$ & N/A & N/A \\ \hline
$C^{uv}_{ab}$ & \makecell{$u,v\in\Chem A$ \\ $m_A(u,v)\notin\{0,\ib\}$} & $\tau_A^{\crg}(u)$ & $\tau_A^{\crg}(v)$ & $(\alpha,0)$ & $(\alpha,0)$ \\ \hline
\end{tabular}
\begin{tabular}{| c || c | c | c |}
\hline
$d^U_D$ & $m_{d(A)}(u,v)$ & $m_{d(A)}(u,a)$ & $m_{d(A)}(v,b)$ \rule{0pt}{10pt}\rule[-5pt]{0pt}{0pt} \\ \hhline{|=#=|=|=|}
$E^u_{ab}$ & N/A & $1$ & N/A \\ \hline
$E^{uv}$ & $0$ & N/A & N/A \\ \hline
$I^{uv}$ & $0$ & N/A & N/A \\ \hline
$C^{uv}_{ab}$ & $m_A(u,v)-1$ & $1$ & $1$ \\ \hline
\end{tabular}
\caption{The disconnection rules defined as partial functions.\label{fig:disc-rules-partial-fns}}
\end{figure}

Note that the disconnection rules look a lot like (a subset of) reaction schemes (Definition~\ref{def:reaction-scheme}): indeed, each disconnection rule can be realised as a collection of reaction schemes. In Section~\ref{sec:reaction-schemes}, we already saw that reactions arise from reaction schemes (Proposition~\ref{prop:reaction-concrete}). In Section~\ref{sec:disc-to-react}, we shall strengthen this result by showing that the disconnection rules generate and axiomatise all the reactions as morphisms in $\React$.

We observe that each disconnection rule is injective (as a partial function), and hence has an inverse partial function. We use the disconnection rules to define the {\em terms}, which will be used to define the morphisms in the disconnection category.
\begin{definition}[Terms]\label{def:terms}
The set of {\em terms} with types is generated by the following recursive procedure:
\begin{itemize}
\item for every chemical graph $A$, let $\id:A\rightarrow A$ be a term,
\item for every chemical graph $A$ and every $u\in V_A$, let $S^u:A\rightarrow A$ be a term,
\item for every chemical graph $A$, every $u\in\alpha(A)$ and every $v\in\VS$ such that $v\notin V_A\setminus\{u\}$, let $R^{u\mapsto v}:A\rightarrow A(u\mapsto v)$ be a term,
\item for every disconnection rule $d$ and every chemical graph $A$ in the domain of $d$, both $d:A\rightarrow d(A)$ and $\bar d:d(A)\rightarrow A$ are terms,
\item if $\mathtt t:A\rightarrow B$ and $\mathtt s:B\rightarrow C$ are terms, then $\mathtt t;\mathtt s:A\rightarrow C$ is a term.
\end{itemize}
\end{definition}
The first and the fifth items take care of the usual categorical structure, while the terms $S^u$ generated by the second item correspond to ``touching'' the vertex $u$ without changing the structure of the graph, and the terms $R^{u\mapsto v}$ rename an existing $\alpha$-vertex $u$ into a fresh vertex $v$.

We refer to the terms of the form $d:A\rightarrow B$ and $\bar d:B\rightarrow A$ generated by the fourth item as {\em disconnections} and {\em connections}, respectively. More specifically, we use the symbols $E^{<0}$, $E^{\geq 0}$, $I$ and $C$ to denote the disconnections corresponding to the specific disconnection rules, and similarly the symbols $\bar{E}^{<0}$, $\bar{E}^{\geq 0}$, $\bar{I}$ and $\bar{C}$ refer to the corresponding connections. Similarly, $S$ and $R$ refer to the terms generated by the second and third items. The same letters in the typewriter type font ($\mathtt E^{<0}$, $\mathtt E^{\geq 0}$, $\mathtt I$, $\mathtt C$, $\bar{\mathtt E}^{<0}$, $\bar{\mathtt E}^{\geq 0}$, $\bar{\mathtt I}$, $\mathtt S$ and $\mathtt R$) are used to denote a sequence of terms of the corresponding kind.

Let us define the endofunction $\overline{()}$ on terms by the following recursion:
\begin{itemize}
\item $\left(\id:A\rightarrow A\right) \mapsto \left(\id:A\rightarrow A\right)$,
\item $\left(S^u:A\rightarrow A\right) \mapsto \left(S^u:A\rightarrow A\right)$,
\item $\left(R^{u\mapsto v}:A\rightarrow A(u\mapsto v)\right) \mapsto \left(R^{v\mapsto u}:A(u\mapsto v)\rightarrow A\right)$,
\item $\left(d:A\rightarrow B\right) \mapsto \left(\bar d:B\rightarrow A\right)$,
\item $\left(\bar d:A\rightarrow B\right) \mapsto \left(d:B\rightarrow A\right)$,
\item $\overline{\mathtt t;\mathtt s}\coloneqq\overline{\mathtt s};\overline{\mathtt t}$.
\end{itemize}

For defining equations, it will be useful to allow untyped terms: the equations (Figure~\ref{fig:disc-axioms}) capture interactions between local graph transformations (i.e.~the disconnection rules), so that the same equation should hold for a whole class of chemical graphs.
\begin{definition}[Untyped terms, well-typedness]\label{def:untyped-terms}
An {\em untyped term} is an element of the free monoid on the set
$$\{\id,S^u,R^{a\mapsto b},E^{ua},E^u_{ab},C^{uv}_{ab},I^{uv},\bar E^{ua},\bar E^u_{ab},\bar C^{uv}_{ab},\bar I^{uv} : u,v,a,b\in\VS\},$$
where we use the symbol $;$ to indicate the multiplication of the monoid.

Given an untyped term $\mathtt t$ and chemical graphs $A$ and $B$, we say that the expression $\mathtt t : A\rightarrow B$
is {\em well-typed} if it is in fact a term, that is, if it can be constructed using the recursive procedure of Definition~\ref{def:terms}.
\end{definition}
We define the binary relation $\leq$ on the set of untyped terms by letting $\mathtt t\leq\mathtt s$ if whenever $\mathtt t : A\rightarrow B$ is well-typed, then so is $\mathtt s : A\rightarrow B$.

The endofunction $\overline{()}$ on the untyped terms is defined in exactly the same way as for the terms with types, simply ignoring the types. Note that $\mathtt t : A\rightarrow B$ is well-typed if and only if $\overline{\mathtt t} : B\rightarrow A$ is. Moreover, observe that $\leq$ defines a preorder on the untyped terms. Consequently, we have $\mathtt t\leq\mathtt s$ if and only if $\overline{\mathtt t}\leq\overline{\mathtt s}$.

Given an untyped term $\mathtt t$, there are either no chemical graphs such that $\mathtt t : A\rightarrow B$ is well-typed, or there are infinitely many such graphs. The latter case is the reason for introducing the untyped terms: we want certain equalities to hold {\em whenever} both sides are well-typed.
\begin{definition}[Term equality]\label{def:term-id}
Let $\approx$ be an equivalence relation on the set of untyped terms. This induces the equivalence relation $\equiv$ on the set of terms as follows: for two terms $\mathtt t,\mathtt s:A\rightarrow B$ with the same type, we let $\mathtt t\equiv\mathtt s$ if either $\mathtt t\approx\mathtt s$ or $\overline{\mathtt t}\approx\overline{\mathtt s}$ as untyped terms.
\end{definition}
Given an equivalence relation $\approx$ on the untyped terms, let us introduce the following shorthand binary relations on the untyped terms:
\begin{itemize}
\item $\mathtt t\lesssim\mathtt s$ if $\mathtt t\approx\mathtt s$ and $\mathtt t\leq\mathtt s$,
\item $\mathtt t\simeq\mathtt s$ if $\mathtt t\lesssim\mathtt s$ and $\mathtt s\lesssim\mathtt t$.
\end{itemize}

\begin{definition}[Disconnection category]\label{def:disc-cat}
The {\em disconnection category} $\Disc$ has as objects the chemical graphs. The set of morphisms $\Disc(A,B)$ is given by the terms of type $A\rightarrow B$, subject to the usual associativity and unitality equations of a category, together with the identities $\equiv$ induced (in the sense of Definition~\ref{def:term-id}) by the equivalence relation defined in Figure~\ref{fig:disc-axioms}.
\end{definition}
\begin{figure}[p]
\centering
\begin{minipage}{0.52\textwidth}
\fbox{
\begin{minipage}{\textwidth}
\begin{align}
R^{u\mapsto z};R^{z\mapsto w} &\lesssim R^{u\mapsto w}\label{disc-eq:trans} \\
R^{u\mapsto z};R^{v\mapsto w} &\approx R^{v\mapsto w};R^{u\mapsto z}\label{disc-eq:rcomm} \\
R^{u\mapsto u} &\lesssim S^u\label{disc-eq:refl} \\
R^{b\mapsto z};R^{a\mapsto b} &\approx S^b;R^{a\mapsto z} \label{disc-eq:rsymm} \\
R^{u\mapsto v};S^w &\approx S^w;R^{u\mapsto v}\label{disc-eq:sr1} \\
R^{u\mapsto v};S^v &\simeq S^u;R^{u\mapsto v} \simeq R^{u\mapsto v}\label{disc-eq:sr2}
\end{align}
\end{minipage}}\vspace{0.2cm}
\fbox{
\begin{minipage}{\textwidth}
\begin{align}
R^{u\mapsto v};d^U_D &\approx d^U_D;R^{u\mapsto v}\label{disc-eq:rd1} \\
R^{u\mapsto v};E^{wv} &\simeq E^{wu};R^{u\mapsto v}\label{disc-eq:rd2} \\
d^U_{D[u]};R^{u\mapsto v} &\simeq d^U_{D[v/u]}\label{disc-eq:rd3} \\
d^{U'}_{ij};\bar h^U_{ab};R^{i\mapsto c};R^{j\mapsto d} &\lesssim \bar h^U_{ab};d^{U'}_{cd}\label{disc-eq:rd4}
\end{align}
\end{minipage}}\vspace{0.2cm}
\fbox{
\begin{minipage}{\textwidth}
\begin{align}
d^U_{ab};\bar d^U_{cd} &\approx S^U;R^{c\mapsto a};R^{d\mapsto b}\label{disc-eq:ddbar1} \\
d^U_{ab};\bar d^U_{cb} &\approx S^U;R^{c\mapsto a}\label{disc-eq:ddbar2} \\
d^U_{ab};\bar d^U_{ad} &\approx S^U;R^{d\mapsto b}\label{disc-eq:ddbar3} \\
d^U_D;\bar d^U_D &\lesssim S^U\label{disc-eq:ddbar4} \\
\bar d^U_D;d^U_D &\lesssim S^U;S^D\label{disc-eq:ddbar4-2} \\
E^{ua};\bar E^{ub} &\approx S^u;R^{a\mapsto z};R^{b\mapsto a};R^{z\mapsto b}\label{disc-eq:eebar} \\
\bar d^{uv};d^{wz} &\approx d^{wz};\bar d^{uv}\label{disc-eq:comm2}
\end{align}
\end{minipage}}%
\end{minipage}%
\begin{minipage}{0.04\textwidth}
\hspace{0.1cm}
\end{minipage}%
\begin{minipage}{0.36\textwidth}
\fbox{
\begin{minipage}{\textwidth}
\begin{align}
S^u;S^v &\simeq S^v;S^u\label{disc-eq:scomm} \\
S^u;S^u &\simeq S^u\label{disc-eq:sidem} \\
S^u;d^U_D &\lesssim d^U_D;S^u\label{disc-eq:sd1} \\
d^{U[v]}_D;S^v &\simeq d^{U[v]}_D\label{disc-eq:sd2} \\
C^{uv}_{ab} &\simeq C^{vu}_{ba}\label{disc-eq:cs}
\end{align}
\end{minipage}}\vspace{0.2cm}
\fbox{
\begin{minipage}{\textwidth}
\begin{align}
d^U_D;d^{U'}_{D'} &\simeq d^{U'}_{D'};d^U_D\label{disc-eq:comm1} \\
C^{uv}_{ab};I^{wz} &\simeq I^{wz};C^{uv}_{ab}\label{disc-eq:comm3} \\
E^u_{ab};I^{wz} &\lesssim I^{wz};E^u_{ab}\label{disc-eq:comm4} \\
E^{uv};I^{wz} &\lesssim I^{wz};E^{uv}\label{disc-eq:comm5} \\
\bar E^{uv};I^{wz} &\lesssim I^{wz};\bar E^{uv}\label{disc-eq:comm6} \\
\bar E^u_{ab};I^{wz} &\lesssim I^{wz};\bar E^u_{ab}\label{disc-eq:comm7} \\
\bar C^{uv}_{ab};I^{wz} &\lesssim I^{wz};\bar C^{uv}_{ab}\label{disc-eq:comm8} \\
E^u_{ab};C^{wz}_{cd} &\simeq C^{wz}_{cd};E^u_{ab}\label{disc-eq:comm9} \\
E^{uv};C^{wz}_{cd} &\lesssim C^{wz}_{cd};E^{uv}\label{disc-eq:comm10} \\
\bar E^{uv};C^{wz}_{cd} &\simeq C^{wz}_{cd};\bar E^{uv}\label{disc-eq:comm11} \\
E^{uv};E^w_{cd} &\lesssim E^w_{cd};E^{uv}\label{disc-eq:comm12} \\
\bar E^{uv};E^w_{cd} &\simeq E^w_{cd};\bar E^{uv}\label{disc-eq:comm13}
\end{align}
\end{minipage}}
\end{minipage}
\caption{The equivalence relation $\approx$ inducing the identities in the disconnection category. Here $d$ and $h$ range over $\{E,C,I\}$, while $S^U$ stands for the sequence $S^u;S^w$ if $U=uw$. Given vertex names $a,b\in\VS$, the notation $D[a]$ means $a$ occurs in $D$, and $D[b/a]$ means the occurrence of $a$ in $D$ is replaced with $b$. Note that we use the shorthand relations $\lesssim$ and $\simeq$: these are strictly speaking not part of the definition, but are used to provide the extra information of when well-typedness of one side of an identity implies well-typedness of the other.\label{fig:disc-axioms}}
\end{figure}

Note that the assignment $\overline{()}:\Disc\rightarrow\Disc$ is functorial, thus making $\Disc$ a dagger category~\cite{selinger-dagger,heunen-vicary-book}.

\begin{proposition}\label{prop:disc-ids}
The following identities are derivable in $\Disc$:
\begin{align}
d^U_{D[a]};S^a &\simeq d^U_{D[a]},\label{prop:disc-sabsorb} \\
\bar d^U_{ab};d^U_{cd} &\approx S^U;R^{a\mapsto j};R^{b\mapsto d};R^{j\mapsto c},\label{prop:disc-ids0} \\
R^{z\mapsto c};R^{w\mapsto d};d^U_{ab} &\approx R^{z\mapsto a};R^{w\mapsto b};d^U_{cd},\label{prop:disc-ids1} \\
R^{z\mapsto c};d^U_{ab} &\approx R^{z\mapsto a};d^U_{cb},\label{prop:disc-ids2} \\
R^{w\mapsto d};d^U_{ab} &\approx R^{w\mapsto b};d^U_{ad},\label{prop:disc-ids3} \\
d^U_{ab};d^U_{cd} &\simeq d^U_{ad};d^U_{cb}.\label{disc-eq:ddindex}
\end{align}
\end{proposition}
\begin{proof}
We compute~\eqref{prop:disc-sabsorb} by applying equations~\eqref{disc-eq:refl} and~\eqref{disc-eq:rd3}:
$$d^U_{D[a]};S^a \simeq d^U_{D[a]};R^{a\mapsto a} \simeq d^U_{D[a]}.$$
Equality~\eqref{prop:disc-ids0} is derived as follows:
\begin{align*}
\bar d^{uv}_{ab};d^{uv}_{cd} &\approx d^{uv}_{ij};\bar d^{uv}_{ab};R^{i\mapsto c};R^{j\mapsto d}\tag{by~\eqref{disc-eq:rd4}} \\
&\approx S^u;S^v;R^{a\mapsto i};R^{b\mapsto j};R^{i\mapsto c};R^{j\mapsto d}\tag{by~\eqref{disc-eq:ddbar1}} \\
&\approx S^u;S^v;R^{a\mapsto j};R^{b\mapsto d};R^{j\mapsto c}.\tag{by~\eqref{disc-eq:rcomm} and~\eqref{disc-eq:trans}}
\end{align*}
For~\eqref{prop:disc-ids1}, we first use~\eqref{prop:disc-ids0} to get
$$d^{uv}_{cd};\bar d^{uv}_{zw};d^{uv}_{ab}\simeq d^{uv}_{cd};S^u;S^v;R^{z\mapsto a};R^{w\mapsto b},$$
where on the right-hand side we used~\eqref{disc-eq:rcomm} and~\eqref{disc-eq:trans}. Observing that~\eqref{disc-eq:ddbar1} applies on the left-hand side, and simplifying using~\eqref{disc-eq:sr1},~\eqref{disc-eq:sd2} and~\eqref{disc-eq:rd1}, we obtain precisely~\eqref{prop:disc-ids1}. Identities~\eqref{prop:disc-ids2} and~\eqref{prop:disc-ids3} are derived similarly, by letting $w=d$ and $z=c$, respectively.

Identity~\eqref{disc-eq:ddindex} is derived as follows:
\begin{align*}
d^U_{ab};d^U_{cd} &\simeq d^U_{ad};R^{d\mapsto b};d^U_{cd}\tag{by~\eqref{disc-eq:rd3}} \\
&\simeq d^U_{ad};R^{d\mapsto d};d^U_{cb}\tag{by~\eqref{prop:disc-ids3}} \\
&\simeq d^U_{ad};d^U_{cb}.\tag{by~\eqref{disc-eq:rd3}}
\end{align*}
\end{proof}

\section{Normal form}\label{sec:normal-form}

In this section, we define a normal form (Definition~\ref{def:normal-form}), and show that every term is equal to a term in a normal form under the equalities of $\Disc$ (Proposition~\ref{prop:normal-form-existence}). We also identify a class of syntactic manipulations of terms in a normal form (Definition~\ref{def:nf-equivalence}) that both keep the normal form and preserve equality (Lemma~\ref{lma:nf-equivalence}). These results are used in the next section to prove completeness.

\begin{definition}[$ICE$-form]\label{def:ICE-form}
We say that a term is in an {\em $ICE$-form} if it is either an identity term, or if it has the following structure:
$$\mathtt I;\mathtt C;\mathtt E^{<0};\mathtt E^{\geq 0};\bar{\mathtt E}^{\geq 0};\bar{\mathtt E}^{<0};\bar{\mathtt C};\bar{\mathtt I};\mathtt R;\mathtt S,$$
where every letter is a sequence of generating terms of the corresponding kind.
\end{definition}

\begin{proposition}\label{prop:ICE-form}
Any term is equal to a term in an $ICE$-form.
\end{proposition}
\begin{proof}[Proof sketch]
The proof proceeds by repeated inductions: one first shows that all $I$-terms can always be commuted to the left, then that all $C$-terms can be commuted to the left of anything that is not an $I$-term, and so on. We give the full details of the induction in the Appendix (\ref{lma:I-commutes}-\ref{cor:ICE-form}).
\end{proof}

\begin{definition}[Renaming form]\label{def:renaming-form}
A well-typed sequence of renaming terms \mbox{$\mathtt R:H\rightarrow G$} is in a {\em renaming form} if there are sets $A=\{a_1,\dots,a_n\}$, $B=\{b_1,\dots,b_n\}$, $C=\{c_1,\dots,c_m\}$ and $D=\{d_1,\dots,d_m\}$ of vertex names such that
\begin{enumerate}[label=(\arabic*)]
\item $\mathtt R$ can be split into two sequences $\mathtt R = \mathtt A;\mathtt B$ with
$$\mathtt A = R^{a_1\mapsto b_1};\dots;R^{a_n\mapsto b_n}\quad \text{and}\quad \mathtt B = R^{c_1\mapsto d_1};\dots;R^{c_m\mapsto d_m},$$
where $\mathtt B$ can be possibly empty,\label{renaming-form-1}
\item $A\cap B=\eset$,\label{renaming-form-2}
\item $C\sse B$,\label{renaming-form-3}
\item $D\sse A$,\label{renaming-form-4}
\item if $c_i\in C$ and $b_j\in B$ is the unique element such that $b_j=c_i$, then $\Nbr_H(a_j)\neq\Nbr_H(d_i)$.\label{renaming-form-5}
\end{enumerate}
\end{definition}

\begin{lemma}\label{lma:renaming-form}
Any well-typed sequence of renaming terms is equal to a term $\mathtt R;\mathtt S$, where $\mathtt R = \mathtt A;\mathtt B$ is in a renaming form and $\mathtt S$ is a sequence of $S$-terms.
\end{lemma}
\begin{proof}
The idea is that if a vertex $a$ is to be renamed to $b$, then $a\in A$, and we have two cases: (1) $b$ does not already occur in the original chemical graph, and (2) $b$ does occur in the original graph. If (1), then $R^{a\mapsto b}\in\mathtt A$ and $b\in B\setminus C$. If (2), then we first rename $a$ using some ``dummy'' name $c$, so that $R^{a\mapsto c}\in\mathtt A$, $R^{c\mapsto b}\in\mathtt B$, $c\in C$ and $b\in D$. Note that condition~\ref{renaming-form-4} of the renaming form is satisfied, as $b$ must itself be renamed in order for the vertex name become free. Any term of the form $R^{a\mapsto a}$ is replaced by $S^a$. The formal proof proceeds by induction on the length of the original sequence.

The term $R^{a\mapsto b}$ is equal to $S^a$ if $a=b$, or is already in a renaming form by taking $A=\{a\}$, $B=\{b\}$ and $C=D=\eset$ if $a\neq b$.

Suppose that the statement of the lemma holds for all sequences of renaming terms of length at most $n$. Let $\mathtt R$ be such a sequence of length $n$ such that $\mathtt R;R^{a\mapsto b}$ is well-typed. By the inductive hypothesis, we may assume that $\mathtt R=\mathtt A;\mathtt B;\mathtt S$ where $\mathtt A;\mathtt B$ is a renaming form with vertex name sets $A$, $B$, $C$ and $D$ as in Definition~\ref{def:renaming-form}. Using the equations for $S$- and $R$-terms, we may commute $R^{a\mapsto b}$ past $\mathtt S$, possibly changing the vertex name $a$, so that it suffices to show that the lemma holds for $\mathtt A;\mathtt B;R^{a\mapsto b}$. If $a=b$, the sequence is equal to $\mathtt A;\mathtt B;S^a$ and we are done; hence assume that $a\neq b$. Note that it follows that $a\notin A\setminus D$ and $a\notin C$, as every vertex name in $A\setminus D$ or $C$ is removed, without being reintroduced. Similarly, we have that $b\notin D$ and $b\notin B\setminus C$. Moreover, if $b\in C$, rename the occurrence of $b$ in both $\mathtt A$ and $\mathtt B$ with a fresh vertex name, updating the sets $C$ and $B$ accordingly. Thus we may assume that $b\notin B$. The remaining cases are as follows.

\noindent\textbf{Case 1:} $a\notin A\cup B$.

\textbf{Subcase 1.1:} $b\notin A$. We rewrite the term to $\mathtt A;R^{a\mapsto b};\mathtt B$ and update the sets $A\mapsto A\cup\{a\}$ and $B\mapsto B\cup\{b\}$.

\textbf{Subcase 1.2:} $b\in A$. It follows that $b\in A\setminus D$, so that $R^{b\mapsto z}\in\mathtt A$ and $a,b$ do not appear in $\mathtt B$. If $\Nbr(a)\neq\Nbr(b)$, let $c$ be a fresh vertex name. We rewrite the term to $\mathtt A;R^{a\mapsto c};R^{c\mapsto b};\mathtt B$ and update the sets $A\mapsto A\cup\{a\}$, $B\mapsto B\cup\{c\}$, $C\mapsto C\cup\{c\}$ and $D\mapsto D\cup\{b\}$. If $\Nbr(a)=\Nbr(b)$, we use equation~\eqref{disc-eq:rsymm} to rewrite $R^{b\mapsto z};R^{a\mapsto b}$ to $S^b;R^{a\mapsto z}$, which reduces the number of $R$-terms to $n$, so that the inductive hypothesis applies.

\noindent\textbf{Case 2:} $a\in A$. It follows that $a\in D$. Now $R^{a\mapsto b}$ commutes with all other terms in $\mathtt B$ except for the unique term $R^{c_i\mapsto d_i}$ where $d_i=a$. But $R^{c_i\mapsto a};R^{a\mapsto b}\equiv R^{c_i\mapsto b}$, which reduces the length of the sequence to $n$, so it is has a renaming form by the inductive hypothesis.

\noindent\textbf{Case 3:} $a\in B$. It follows that $a\in B\setminus C$.

\textbf{Subcase 3.1:} $b\notin A$. Now $R^{a\mapsto b}$ commutes with all the terms in $\mathtt B$, and with all other terms in $\mathtt A$ except for the unique term $R^{a_i\mapsto b_i}$ where $b_i=a$. But $R^{a_i\mapsto a};R^{a\mapsto b}\equiv R^{a_i\mapsto b}$, which reduces the length of the sequence to $n$, so it is has a renaming form by the inductive hypothesis.

\textbf{Subcase 3.2:} $b\in A$. It follows that $b\in A\setminus D$. Now $R^{a\mapsto b}$ commutes with all the terms in $\mathtt B$, and with all other terms in $\mathtt A$ except for the terms $R^{a_i\mapsto a}$ and $R^{b\mapsto b_j}$. There are two options: (1) $a_i=b$ and $b_j=a$, so that these are the same term, (2) the terms are distinct, in which case they commute. In both cases, we use the substitution $R^{a_i\mapsto a};R^{a\mapsto b}\equiv R^{a_i\mapsto b}$ to reduce the length of the sequence, so that the inductive hypothesis applies.

This completes the induction.
\end{proof}

A term is said to be in an {\em $ICER$-form} if it is in an $ICE$-form whose sequence of renaming terms is in a renaming form (or is empty).

\begin{definition}[Normal form]\label{def:normal-form}
Let
$$\mathtt t = \mathtt I;\mathtt C;\mathtt E^{<0};\mathtt E^{\geq 0};\bar{\mathtt E}^{\geq 0};\bar{\mathtt E}^{<0};\bar{\mathtt C};\bar{\mathtt I};\mathtt A;\mathtt B;\mathtt S$$
be a term in an $ICER$-form. Let us denote the sets of vertex names in the renaming form by $A_{\mathtt t}$, $B_{\mathtt t}$, $C_{\mathtt t}$ and $D_{\mathtt t}$. Let us additionally define the following sets of vertex names occurring in $\mathtt t$:
\begin{itemize}
\item $D^{add}_{\mathtt t}\coloneqq\left\{a\in\VS : d^U_{D[a]}\in\mathtt t\right\}$ -- the vertex names appearing as subscripts in the disconnections,
\item $D^{remove}_{\mathtt t}\coloneqq\left\{a\in\VS : \bar d^U_{D[a]}\in\mathtt t\right\}$ -- the vertex names appearing as subscripts in the connections,
\item $U_{\mathtt t}\coloneqq\left\{v\in\VS : d^{U[v]}_D\in\mathtt t\text{ or } \bar d^{U[v]}_D\in\mathtt t\right\}$ -- the vertex names appearing as superscripts of the (dis)connections,
\item $S_{\mathtt t}\coloneqq\left\{u\in\VS : S^u\in\mathtt t\right\}$ -- the vertex names appearing in the $S$-terms.
\end{itemize}
We say that a term $\mathtt t$ is in a {\em normal form} if it is in an $ICER$-form as above, and additionally the following conditions hold:
\begin{enumerate}[label=(\arabic*)]
\item for every $u\in S_{\mathtt t}$, the term $S^u$ occurs in $\mathtt t$ exactly once,\label{nf:S0}
\item $\left(U_{\mathtt t}\cup A_{\mathtt t}\cup B_{\mathtt t}\right)\cap S_{\mathtt t}=\eset$,\label{nf:S1}
\item $D^{add}_{\mathtt t}\setminus D^{remove}_{\mathtt t}\sse A_{\mathtt t}\setminus D_{\mathtt t}$,\label{nf:discrename}
\item $D^{add}_{\mathtt t}\cap B_{\mathtt t}=\eset$,\label{nf:dbdisjoint}
\item if a connection $\bar d^U_{D[a]}:A\rightarrow B$ and a renaming term $R^{z\mapsto a}$ both occur, then $A$ is not in the domain of $\bar d^U_{D[z/a]}$,\label{nf:connrename}
\item if $d\neq I$ and a disconnection $d^U_D$ occurs in $\mathtt t$, then the connections $\bar d^U_{F}$ and $\bar d^{U^r}_{F}$ do not occur in $\mathtt t$ for any $F$ (here $U^r$ denotes the reverse string),\label{nf:disconnection}
\item if the disconnection $E^{uv}$ occurs in $\mathtt t$, then for any vertex name $w\in\VS$, the connection $\bar E^{uw}$ does not occur in $\mathtt t$,\label{nf:electron}
\item if the disconnection $I^{uv}$ and the connection $\bar I^{uv}$ both occur in $\mathtt t$, then one of the terms $E^v_D$, $\bar E^v_D$, $E^{va}$ or $\bar E^{va}$ occurs in $\mathtt t$.\label{nf:ion}
\end{enumerate}
\end{definition}

\begin{proposition}\label{prop:normal-form-existence}
In $\Disc$, any term is equal to a term in normal form.
\end{proposition}
\begin{proof}
By Propositions~\ref{prop:ICE-form} and~\ref{lma:renaming-form}, every term is equal to a term in an $ICER$-form: let us fix such a term
$$\mathtt t = \mathtt I;\mathtt C;\mathtt E^{<0};\mathtt E^{\geq 0};\bar{\mathtt E}^{\geq 0};\bar{\mathtt E}^{<0};\bar{\mathtt C};\bar{\mathtt I};\mathtt A;\mathtt B;\mathtt S.$$

Conditions~\ref{nf:S0} and~\ref{nf:S1} are obtained by absorbing the ``excess'' $S$-terms into other terms using equations~\eqref{disc-eq:sr2},~\eqref{disc-eq:sidem},~\eqref{disc-eq:sd2} and~\eqref{prop:disc-sabsorb}. Conditions~\ref{nf:discrename} and~\ref{nf:dbdisjoint} are obtained by treating all the vertex names in $D^{add}_{\mathtt t}$ as ``dummy'' names, which are removed either by a connection or a renaming term.

For~\ref{nf:connrename}, suppose that $\bar d^U_{D[a]}:A\rightarrow B$ and $R^{z\mapsto a}$ both occur, and moreover $A$ is in the domain of $\bar d^U_{D[z/a]}$. We commute the renaming term to the left to obtain $\bar d^U_{D[a]};R^{z\mapsto a}$. But this is equal to $\bar d^U_{D[z/a]};R^{a\mapsto a}$ by equations~\eqref{prop:disc-ids2} and~\eqref{prop:disc-ids3}, which gets rid of the renaming term by $R^{a\mapsto a} \equiv S^a$.

For~\ref{nf:disconnection} and~\ref{nf:electron}, we consider the three cases separately.

\noindent\textbf{Case 1:} $E^{uv}$ and $\bar E^{uw}$ occur in $\mathtt t$. We commute the terms so that they occur one after the other $E^{uv};\bar E^{uw}$, which by equations~\eqref{disc-eq:ddbar4} and~\eqref{disc-eq:eebar} is equal to some combination of $S$- and $R$-terms. Noticing that the rewriting procedure to obtain an ICE-form either commutes or absorbs $S$- and $R$-terms (specifically, the results from Lemma~\ref{lma:Enonnegbar-commutes} onwards apply), we conclude that $\mathtt t$ has an ICE-form without $E^{uv}$ and $\bar E^{uw}$.

\noindent\textbf{Case 2:} $E^u_{ab}$ and $\bar E^u_{cd}$ occur in $\mathtt t$. In combination with Case~1, it follows that the terms $E^{uv}$ and $\bar E^{uw}$ do not occur, so that there are no obstructions for commuting $E^u_{ab}$ next to $\bar E^u_{cd}$, obtaining $E^u_{ab};\bar E^u_{cd}$. By equations~\eqref{disc-eq:ddbar1},~\eqref{disc-eq:ddbar2} and~\eqref{disc-eq:ddbar3}, this is equal to some combination of $S$- and $R$-terms, which we commute to the right as in Case~1.

\noindent\textbf{Case 3:} $C^{uv}_{ab}$ and either $\bar C^{uv}_{cd}$ or $\bar C^{vu}_{cd}$ occur in $\mathtt t$. The latter case simply reduces to the former by equation~\eqref{disc-eq:cs}. Thus suppose that $\bar C^{uv}_{cd}$ occurs. First, we use equation~\eqref{disc-eq:comm9} to commute $C^{uv}_{ab}$ to the right past all the $E^{<0}$-terms, and $\bar C^{uv}_{cd}$ to the left past all the $\bar E^{<0}$-terms. Next, we commute any terms of the form $E^{ui}$ and $E^{vj}$ to the right past the $\bar E^{\geq 0}$-terms and $\bar C^{uv}_{cd}$ using the fact that by Case~1 the terms $\bar E^{ui}$ and $\bar E^{vj}$ do not occur, together with the equations~\eqref{disc-eq:comm2} and~\eqref{disc-eq:comm11}. Now there are no obstructions for commuting $C^{uv}_{ab}$ to the right past all the $E^{\geq 0}$- and $\bar E^{\geq 0}$-terms, obtaining $C^{uv}_{ab};\bar C^{uv}_{cd}$. As in Case~2, equations~\eqref{disc-eq:ddbar1},~\eqref{disc-eq:ddbar2} and~\eqref{disc-eq:ddbar3} yield that this is equal to some combination of $S$- and $R$-terms, which we commute to the right as in Case~1. Finally, we return the terms of the form $E^{ui}$ and $E^{vj}$ back to the left past all the $\bar E^{\geq 0}$-terms.

For~\ref{nf:ion}, suppose that both $I^{uv}$ and $\bar I^{uv}$ occur in $\mathtt t$ such that no $E$- or $\bar E$ term containing $u$ occurs. It follows that no $E$- or $\bar E$ term containing $v$ occurs either: application of $\bar I^{uv}$ requires for $u$ and $v$ to have equal and opposite charge, so if the charge of $u$ is unchanged, so is the charge of $v$; moreover, by~\ref{nf:disconnection} and~\ref{nf:electron}, we may assume that no change can be reversed, so that we indeed cannot have any $E$- or $\bar E$ term containing $v$. But now there are no obstructions for commuting $I^{uv}$ all the way to the right until we obtain $I^{uv};\bar I^{uv}$, which is equal to $S^u;S^v$ by~\eqref{disc-eq:ddbar4}.
\end{proof}

\begin{definition}[Normal form equivalence]\label{def:nf-equivalence}
Let
$$\mathtt t = \mathtt I;\mathtt C;\mathtt E^{<0};\mathtt E^{\geq 0};\bar{\mathtt E}^{\geq 0};\bar{\mathtt E}^{<0};\bar{\mathtt C};\bar{\mathtt I};\mathtt A;\mathtt B;\mathtt S$$
be a term in a normal form. Define the following syntactic manipulations of $\mathtt t$:
\begin{enumerate}
\item commuting the terms inside each of the named sequences in the normal form,\label{nf-eq:comm}
\item permuting vertex names in $C$-terms: if the term $C^{uv}_{ab}$ occurs, we may substitute it with $C^{vu}_{ba}$,\label{nf-eq:permute}
\item if $d\in\{C,E,\bar C,\bar E\}$ such that $d^U_{ab};d^U_{cd}$ occurs, we may substitute $d^U_{ab};d^U_{cd}\mapsto d^U_{ad};d^U_{cb}$,\label{nf-eq:same-swap}
\item renaming of vertices that are introduced and removed: if $a\in D^{add}_{\mathtt t}\cup C_{\mathtt t}$ and $z\in\VS$ does not occur in $\mathtt t$ or its domain, then we may substitute $\mathtt t\mapsto\mathtt t[z/a]$,\label{nf-eq:rename}
\item exchanging vertex names between renaming terms: if both $R^{a\mapsto b}$ and $R^{c\mapsto d}$ occur in $\mathtt A$ such that $\Nbr(a)=\Nbr(c)$, we may swap $a$ and $c$,\label{nf-eq:r-swap}
\item exchanging vertex names between connections and renaming terms: if $d\in\{E,C\}$, and $\bar d^U_{D[a]}:A\rightarrow B$ and $R^{z\mapsto b}$ both occur such that $A$ is in the domain of $\bar d^U_{D[z/a]}$, then we may substitute $\bar d^U_{D[a]}\mapsto \bar d^U_{D[z/a]}$ and $R^{z\mapsto b}\mapsto R^{a\mapsto b}$.\label{nf-eq:r-exchange}
\end{enumerate}
We say that two terms $\mathtt t$ and $\mathtt s$ in a normal form are {\em equivalent}, written $\mathtt t\sim\mathtt s$, if one can be obtained from the other by a sequence of the syntactic manipulations defined above.
\end{definition}
Observing that each syntactic manipulation in Definition~\ref{def:nf-equivalence} is reversible, we see that $\sim$ is an equivalence relation on the set of terms in normal form.
\begin{lemma}\label{lma:nf-equivalence}
Let $\mathtt t$ and $\mathtt s$ be terms in normal forms such that $\mathtt t\sim\mathtt s$. Then $\mathtt t \equiv \mathtt s$.
\end{lemma}
\begin{proof}
This follows by noticing that every syntactic manipulation of Definition~\ref{def:nf-equivalence} keeps the term in a normal form, and moreover preserves the equality $\equiv$:
\begin{enumerate}
\item the terms may be commuted by equations~\eqref{disc-eq:comm1},~\eqref{disc-eq:rcomm} and~\eqref{disc-eq:scomm},
\item vertex names in $C$-terms may be permuted by~\eqref{disc-eq:cs},
\item the indices in repeated (dis)connections may be exchanged by~\eqref{disc-eq:ddindex},
\item if $a\in D^{add}_{\mathtt t}$, so that there is a disconnection $d^U_{D[a]}$, we use equation~\eqref{disc-eq:rd3} to obtain $d^U_{D[a]} \equiv d^U_{D[z/a]};R^{z\mapsto a}$ to introduce the desired fresh variable $z$; the renaming term can then be absorbed into the second occurrence of $a$, hence replacing $a$ with $z$ (the case when $a\in C_{\mathtt t}$ is similar),
\item the exchange of $\alpha$-vertices with the same neighbour is obtained by equation~\eqref{disc-eq:rsymm}:
$$R^{a\mapsto b};R^{c\mapsto d} \equiv S^c;R^{a\mapsto b};R^{c\mapsto d} \equiv R^{c\mapsto b};R^{a\mapsto c};R^{c\mapsto d} \equiv R^{c\mapsto b};R^{a\mapsto d},$$
\item the last syntactic manipulation is obtained by equations~\eqref{prop:disc-ids2} and~\eqref{prop:disc-ids3}.
\end{enumerate}
\end{proof}

\section{From disconnections to reactions, functorially}\label{sec:disc-to-react}

This section establishes a tight link between the category of disconnection rules (Definition~\ref{def:disc-cat}) and the category of reactions (Definition~\ref{def:category-reactions}) by constructing a functor 
$$R:\Disc\rightarrow\React,$$
establishing soundness of the disconnection rules with respect to the reactions. Moreover, we prove that $R$ is faithful and full up to isomorphism\footnote{See Remark~\ref{rem:equiv-subcat} for the discussion on why $R$ is not full.}: a fact that entails completeness and universality (Theorems~\ref{thm:completeness} and~\ref{thm:universality}). In combination, the results of this section allow for algebraic reasoning about the reactions using the equations for the disconnection rules (Figure~\ref{fig:disc-axioms}).

We define a function $R$ from terms to morphisms in $\React$ as follows. Given a term $\mathtt t:A\rightarrow B$, the morphism $R(\mathtt t):A\rightarrow B$ has the form
$$R(R_1(\mathtt t),R_2(\mathtt t),\id,\id),$$
where $\id:\Chem{R_1(\mathtt t)}\rightarrow\Chem{R_2(\mathtt t)}$ and $\id:V_A\setminus R_1(\mathtt t)\rightarrow V_B\setminus R_2(\mathtt t)$ are both identity maps. Since all the terms are mapped to morphisms whose bijection and isomorphism parts are the identities, we omit these, and simply write $R(\mathtt t)=(R_1(\mathtt t),R_2(\mathtt t))$. The recursive definition of this mapping is given below:
\begin{align*}
R(\id_A) &\coloneqq (\eset,\eset) & R(E^{uv}) &\coloneqq (\{u,v\},\{u,v\}) \\
R\left(S^u\right) &\coloneqq (\{u\},\{u\}) & R(I^{uv}) &\coloneqq (\{u,v\},\{u,v\}) \\
R\left(R^{u\mapsto v}\right) &\coloneqq (\{u\},\{v\}) & R(C^{uv}_{ab}) &\coloneqq (\{u,v\},\{u,v,a,b\}) \\
R(E^{u}_{ab}) &\coloneqq (\{u\},\{u,a,b\}) & R\left(\bar d^{uv}_{ab}\right) &\coloneqq\overline{R\left(d^{uv}_{ab}\right)} \\
 &\phantom{\coloneqq} & R(\mathtt t;\mathtt s) &\coloneqq R(\mathtt t);R(\mathtt s).
\end{align*}
Observe that for all the disconnections we have $R(d^U_D)=(U,U\cup D)$.

Soundness of disconnection rules with respect to reactions is expressed as functoriality:
\begin{proposition}\label{prop:r-dagger-functor}
The assignment $R:\Disc\rightarrow\React$ is a dagger functor.
\end{proposition}
\begin{proof}
Functoriality and preservation of dagger structure follow immediately from the definition. We have to show that $R$ preserves the equalities in $\Disc$, generated by the identities in Figure~\ref{fig:disc-axioms}. Most of these follow immediately by assuming that the expressions on both sides of the equality have the same type, and showing that they are mapped to the same pair of sets by $R$. Hence we only give the cases that are less obvious or require more computation. To further simplify the notation, we omit the curly brackets of set-builder notation as well as the commas separating vertex names from each other: so e.g.~$(uv,uvab)$ stands for $(\{u,v\},\{u,v,a,b\})$.

For~\eqref{disc-eq:rd1}, suppose that $R^{u\mapsto v};d^U_D\equiv d^U_D;R^{u\mapsto v}$. In particular, it follows that $u,v\notin U\cup D$. Let us write $R(d^U_D)=(R_1,R_2)$, so that $u,v\notin R_1,R_2$. We use this to show that both sides of the equality evaluate to the same map:
$$(u,v);(R_1,R_2) = (\{u\}\cup R_1, R_2\cup\{v\}) = (R_1,R_2);(u,v).$$

For~\eqref{disc-eq:rd4}, suppose that $d^{U'}_{ij};\bar h^U_{ab};R^{i\mapsto c};R^{j\mapsto d} \equiv \bar h^U_{ab};d^{U'}_{cd}$. Denote $R(d^{U'}_{ij})=(u'v',u'v'ij)$, $R(d^{U'}_{cd})=(u'v',u'v'cd)$ and $R(\bar h^{U}_{ab})=(uvab,uv)$. Note that $d$ and $h$ are not $E^{\geq 0}$-terms, whence it follows that $i,j\notin U$ and $a,b\notin U'$. From the fact that the left-hand side is defined, we obtain that $\{i,j\}$ and $\{a,b\}$ are disjoint. The left-hand side is thus translated to
\begin{flalign*}
(u'v',u'v'ij);(uvab,uv);(i,c);(j,d) &= (u'v'uvab,uvu'v'ij);(i,c);(j,d) \\
                              &= (u'v'uvab,cuvu'v'j);(j,d) \\
                              &= (u'v'uvab,dcuvu'v') \\
                              &= (uvu'v'ab,uvu'v'cd) \\
                              &= (uvab,uv);(u'v',u'v'cd),
\end{flalign*}
which we recognise as the translation of the right-hand side.

For~\eqref{disc-eq:eebar}, suppose $E^{ua};\bar E^{ub} \equiv S^u;R^{a\mapsto z};R^{b\mapsto a};R^{z\mapsto b}$. We start from the translation of the right-hand side:
\begin{flalign*}
(u,u);(a,z);(b,a);(z,b) &= (ua,zu);(bz,ba) \\
                        &= (uab,bau) \\
                        &= (uab,uba) \\
                        &= (ua,ua);(ub,ub),
\end{flalign*}
which we recognise as the translation of the left-hand side.

For~\eqref{disc-eq:comm1}, write $R(d^U_D)=(U,U\cup D)$ and $R(d^{U'}_{D'})=(U',U'\cup D')$, so that we get
$$R(d^U_D;d^{U'}_{D'}) = (U\cup U', U\cup U'\cup D\cup D') = R(d^{U'}_{D'};d^U_D).$$
\end{proof}

Recall the syntactic manipulations of terms in normal form we introduced in Definition~\ref{def:nf-equivalence}. We have seen that these manipulations preserve equality (Lemma~\ref{lma:nf-equivalence}). The following lemma is the core of the completeness argument.
\begin{lemma}\label{lma:nf-equality-equivalence}
Let $\mathtt t$ and $\mathtt s$ be terms in a normal form such that $R(\mathtt t)=R(\mathtt s)$. Then $\mathtt t\sim\mathtt s$.
\end{lemma}
\begin{proof}
Let us write
\begin{align*}
\mathtt t &= \mathtt I_{\mathtt t};\mathtt C_{\mathtt t};\mathtt E^{<0}_{\mathtt t};\mathtt E^{\geq 0}_{\mathtt t};\bar{\mathtt E}^{\geq 0}_{\mathtt t};\bar{\mathtt E}^{<0}_{\mathtt t};\bar{\mathtt C}_{\mathtt t};\bar{\mathtt I}_{\mathtt t};\mathtt A_{\mathtt t};\mathtt B_{\mathtt t};\mathtt S_{\mathtt t}, \\
\mathtt s &= \mathtt I_{\mathtt s};\mathtt C_{\mathtt s};\mathtt E^{<0}_{\mathtt s};\mathtt E^{\geq 0}_{\mathtt s};\bar{\mathtt E}^{\geq 0}_{\mathtt s};\bar{\mathtt E}^{<0}_{\mathtt s};\bar{\mathtt C}_{\mathtt s};\bar{\mathtt I}_{\mathtt s};\mathtt A_{\mathtt s};\mathtt B_{\mathtt s};\mathtt S_{\mathtt s}.
\end{align*}
Similarly, let us denote the vertex name sets in the respective renaming forms by $A_{\mathtt t},B_{\mathtt t},C_{\mathtt t},D_{\mathtt t}$ and $A_{\mathtt s},B_{\mathtt s},C_{\mathtt s},D_{\mathtt s}$. Let us denote the morphism $R(\mathtt t)=R(\mathtt s)$ by $(R_1,R_2):A\rightarrow B$.

First, we observe that if $E^{ua}\in\mathtt E^{\geq 0}_{\mathtt t}$, then condition~\ref{nf:electron} of normal form (Definition~\ref{def:normal-form}) implies that the charge of $u$ cannot be increased, whence there is a vertex name $b\in\VS$ such that $E^{ub}\in\mathtt E^{\geq 0}_{\mathtt s}$. Similarly, by condition~\ref{nf:disconnection} of normal form, if $E^u_{ab}\in\mathtt E^{<0}_{\mathtt t}$, then $E^u_{cd}\in\mathtt E^{<0}_{\mathtt s}$ for some $c,d\in\VS$; and if $C^{uv}_{ab}\in\mathtt C_{\mathtt t}$, then $C^{uv}_{cd}\in\mathtt C_{\mathtt s}$ for some $c,d\in\VS$. By condition~\ref{nf:ion} of normal form, if $I^{uv}\in\mathtt I_{\mathtt t}$ then $I^{uv}\in\mathtt I_{\mathtt s}$. Since the connections cannot undo the disconnections, a similar inclusion up to $\alpha$-vertices holds for them. Thus we obtain that the sequences of disconnections and connections must coincide, up to renaming the $\alpha$-vertices.

Next, suppose that $R^{a\mapsto b}\in\mathtt A_{\mathtt t}$, so that $a\in A_{\mathtt t}$ and $b\in B_{\mathtt t}$. There are four cases.

\noindent\textbf{Case 1:} $a\notin D^{add}_{\mathtt t}$ and $b\notin C_{\mathtt t}$. This implies that $a\in R_1$ and $b\in R_2$. Moreover, if $b\in R_1$, then condition~\ref{nf:connrename} of normal form yields that $\Nbr(a)\neq\Nbr(b)$. It follows that either $R^{a\mapsto b}\in\mathtt A_{\mathtt s}$, or both $\bar d^U_{D[a]}$ and $R^{z\mapsto b}$ occur in $\mathtt s$ such that $\bar d^U_{D[z/a]}$ is defined. But in the latter case the vertex names $a$ and $z$ may be exchanged by syntactic manipulation~\ref{nf-eq:r-exchange} (Definition~\ref{def:nf-equivalence}), so that we may assume $R^{a\mapsto b}\in\mathtt A_{\mathtt s}$.

\noindent\textbf{Case 2:} $a\notin D^{add}_{\mathtt t}$ and $b\in C_{\mathtt t}$. This means that $R^{b\mapsto d}\in\mathtt B_{\mathtt t}$ for some $d\in D_{\mathtt t}$, and for some $x\in\VS$, we have $R^{d\mapsto x}\in\mathtt A_{\mathtt t}$. Condition~\ref{nf:discrename} of normal form implies that $d\notin D^{add}_{\mathtt t}$, so that we have $a\in R_1$ and $d\in R_1\cap R_2$. If $x\notin C_{\mathtt t}$, then by Case~1, $R^{d\mapsto x}\in\mathtt A_{\mathtt s}$, so that also $R^{a\mapsto z}\in\mathtt A_{\mathtt s}$ and $R^{z\mapsto d}\in\mathtt B_{\mathtt s}$ for some $z\in\VS$. If $x\in C_{\mathtt t}$, then we inductively repeat Case~2. By syntactic manipulation~\ref{nf-eq:rename}, we may assume that $R^{a\mapsto b}\in\mathtt A_{\mathtt s}$ and $R^{b\mapsto d}\in\mathtt B_{\mathtt s}$.

\noindent\textbf{Case 3:} $a\in D^{add}_{\mathtt t}$ and $b\notin C_{\mathtt t}$. Thus there is a disconnection $d^U_{D[a]}\in\mathtt t$, so that $d^U_{D[x/a]}\in\mathtt s$ for some $x\in\VS$. This implies $a\notin R_1\cup R_2$ and $b\in R_2$. As in Case~1, it follows that $R^{z\mapsto b}\in\mathtt A_{\mathtt s}$ for some $z\in\VS$. Moreover, in this case we must have $\Nbr(x)=\Nbr(z)$, whence by syntactic manipulation~\ref{nf-eq:r-swap} we may assume that $R^{x\mapsto b}\in\mathtt A_{\mathtt s}$ with $x\in D^{add}_{\mathtt s}$. By syntactic manipulation~\ref{nf-eq:rename}, we may assume that $R^{a\mapsto b}\in\mathtt A_{\mathtt s}$ and $d^U_{D[a]}\in\mathtt s$.

\noindent\textbf{Case 4:} $a\in D^{add}_{\mathtt t}$ and $b\in C_{\mathtt t}$. We thus have $a,b\notin R_1\cup R_2$. This means that $R^{b\mapsto d}\in\mathtt B_{\mathtt t}$ for some $d\in D_{\mathtt t}$, so that $R^{d\mapsto x}\in\mathtt A_{\mathtt t}$ for some $x\in\VS$. Condition~\ref{nf:discrename} of normal form implies that $d\notin D^{add}_{\mathtt t}$, so that $d\in R_1\cap R_2$ and either Case~1 or Case~2 applies to $R^{d\mapsto x}$. In both cases we conclude that $R^{d\mapsto x}\in\mathtt A_{\mathtt s}$. Since $d\in R_2$, we have $R^{w\mapsto d}\in\mathtt B_{\mathtt s}$ for some $w\in\VS$. Since there is a disconnection $d^U_{D[a]}\in\mathtt t$, we have $d^U_{D[y/a]}\in\mathtt s$ for some $y\in\VS$. Note that we have $\Nbr(y)=\Nbr(w)$. Consequently, by syntactic manipulation~\ref{nf-eq:r-swap}, we may assume that $R^{y\mapsto w}\in\mathtt A_{\mathtt s}$ with $y\in D^{add}_{\mathtt s}$ and $w\in C_{\mathtt s}$. By syntactic manipulation~\ref{nf-eq:rename}, we may assume that $R^{a\mapsto b}\in\mathtt A_{\mathtt s}$, $R^{b\mapsto d}\in\mathtt B_{\mathtt s}$ and $d^U_{D[a]}\in\mathtt s$.

Thus we have shown that $\mathtt t$ and $\mathtt s$ have the same renaming sequences (up to $\sim$), and up to the syntactic manipulations, $D^{add}_{\mathtt t}=D^{add}_{\mathtt s}$ and $D^{remove}_{\mathtt t}=D^{remove}_{\mathtt s}$.

If $S^u\in\mathtt S_{\mathtt t}$, then $u\in R_1\cap R_2$ and, by conditions~\ref{nf:S0} and~\ref{nf:S1} of normal form, $u$ does not occur anywhere else in $\mathtt t$. The argument so far entails that $u\notin U_{\mathtt s}\cup A_{\mathtt s}\cup B_{\mathtt s}$, so that $S^u\in\mathtt S_{\mathtt s}$. Thus $\mathtt S_{\mathtt t} =\mathtt S_{\mathtt s}$.

Now the only difference left between $\mathtt t$ and $\mathtt s$ is in which order the vertex names are introduced and removed. This is taken care of precisely by syntactic manipulations~\ref{nf-eq:comm} and~\ref{nf-eq:same-swap}.
\end{proof}

Combining the above lemma with the results from the previous section, we conclude that the functor $R:\Disc\rightarrow\React$ is faithful. We spell this out in detail in the following:
\begin{theorem}[Completeness]\label{thm:completeness}
For all terms $\mathtt t$ and $\mathtt s$, we have $\mathtt t\equiv\mathtt s$ in $\Disc$ if and only if $R(\mathtt t)=R(\mathtt s)$ in $\React$.
\end{theorem}
\begin{proof}
The `only if' direction is functoriality (Proposition~\ref{prop:r-dagger-functor}). The `if' direction follows from the fact that every term is equal to a term in normal form (Proposition~\ref{prop:normal-form-existence}) and from Lemmas~\ref{lma:nf-equality-equivalence} and~\ref{lma:nf-equivalence}.
\end{proof}

The argument for universality turns out to be much simpler than that for completeness. However, in combination with Theorem~\ref{thm:completeness}, it gives a rather strong representation result for reactions: not only can every reaction be decomposed into a sequence of disconnection rules, but this sequence is also unique, up to changing the vertex names and up to the equations in $\Disc$. In abstract terms, the statement of universality is that the functor $R:\Disc\rightarrow\React$ is full up to isomorphism in $\React$. As for completeness, we spell out the details:
\begin{theorem}[Universality]\label{thm:universality}
Given a reaction $r:A\rightarrow C$ in $\React$, there is a term $\mathtt t:A\rightarrow B$ in $\Disc$ and an isomorphism $\iota:B\xrightarrow{\sim} C$ in $\React$ such that $R(\mathtt t);\iota = r$.
\end{theorem}
\begin{proof}
Observe that every reaction $r:A\rightarrow C$ factorises as
$$(U_A,U_B,\id,\id);(\eset,\eset,!,\iota),$$
where $(U_A,U_B):A\rightarrow B$ is some reaction and $\iota:B\rightarrow C$ is an isomorphism of labelled graphs. Now, we may disconnect all possible bonds inside $U_A$, and then connect all possible bonds to obtain $U_B$. The fact that $U_A$ and $U_B$ have the same atom vertices and the same net charge guarantee that this can always be done. Precisely, the sought-after term $\mathtt t:A\rightarrow B$ is then given by
\begin{flalign*}
&\prod_{\substack{u\in\Crgp{U_A} \\ v\in\Crgn{U_A}}}\left(I^{uv}\right)^{\ion(m_A(u,v))}; \prod_{u,v\in\Chem{U_A}}\prod_{i=1}^{\cov(m_A(u,v))}C^{uv}_{a_ib_i}; \\
&\prod_{u\in\Crgn{U_A}}\prod_{i=1}^{-\tau^{\crg}_A(u)}E^u_{a_ib_i}; \prod_{u\in\Chem{U_A}}\prod_{i=1}^{\mathbf v\tau^{\At}_A(u) - \max\left(\tau^{\crg}_A(u),0\right)}E^{ua_i}; \\
&\prod_{u\in\Chem{U_B}}\prod_{i=1}^{\mathbf v\tau^{\At}_B(u) - \max\left(\tau^{\crg}_B(u),0\right)}\bar E^{ua_i}; \prod_{u\in\Crgn{U_B}}\prod_{i=1}^{-\tau^{\crg}_B(u)}\bar E^u_{a_ib_i}; \\
&\prod_{u,v\in\Chem{U_B}}\prod_{i=1}^{\cov(m_B(u,v))}\bar C^{uv}_{a_ib_i}; \prod_{\substack{u\in\Crgp{U_B} \\ v\in\Crgn{U_B}}}\left(\bar I^{uv}\right)^{\ion(m_B(u,v))}; \\
&\prod_{a\in\alpha(U_A)\setminus D}R^{a\mapsto b_a}; \prod_{b\in\alpha(U_B)}R^{a_b\mapsto b}; \prod_{u\in U_B}S^u,
\end{flalign*}
where the vertex names introduced by the $C$- and $E^{<0}$-terms are chosen so that they do not appear anywhere in $A$ or $B$, and their set is denoted by $I$. The vertex names removed by the $\bar C$- and $\bar E^{<}$-terms are chosen from $U_A\cup I$ such that the connection is well-typed: their set is denoted by $D$. Similarly, the $\alpha$-vertices appearing in the $E^{\geq 0}$- and $\bar E^{\geq 0}$-terms are chosen from $U_A\cup I$ such that the terms are well-typed. The vertex names introduced by the $R^{a\mapsto b_a}$-terms, where $a\in \alpha(U_A)\setminus D$, are chosen so that they do not appear in $A$, $B$ or $I$: their set is denoted by $R$. Finally, the vertex names removed by the $R^{a_b\mapsto b}$-terms, where $b\in\alpha(U_B)$ are chosen from $I\cup R$ in such a way that the terms are well-typed.

Note that while the term we obtain is in an $ICE$-form, it will not, in general, be in normal form.
\end{proof}
\begin{remark}\label{rem:equiv-subcat}
Theorem~\ref{thm:universality} entails that $\Disc$ is equivalent to the subcategory of $\React$ whose morphisms are of the form $(U_A,U_B,\id,\id)$, which is the image of the functor $R:\Disc\rightarrow\React$. The functor is not itself an equivalence, as morphisms which change labels of chemical vertices cannot be captured by the disconnection rules. The theorem, moreover, states that any morphism in $\React$ factorises as a morphism in the image of $R$ followed by an isomorphism. Thus, the disconnection rules capture all the reactions, up to renaming of the chemical vertices.
\end{remark}

\begin{example}\label{ex:disconnection-sequence}
Consider the reaction from Figure~\ref{fig:clayden} (formation of benzyl benzoate from benzoyl chloride and benzyl alcohol), which we redraw with vertex names below. Here both $b$ and $i$ are identity maps, and $\mathtt{Ph}$ stands for the phenyl group:
\begin{center}
\scalebox{.9}{\input{thesis-ch2/example-reaction.tikz}}.
\end{center}
Following the procedure of Theorem~\ref{thm:universality}, the reaction decomposes into the following sequence of (dis)connection rules:
\begin{align*}
& C^{zu}_{ab};C^{vw}_{cd};C^{ru}_{ij};C^{ru}_{nm};E^{vc};E^{wd};E^{za};E^{ub};E^{ri};E^{uj};E^{rn};E^{um}; \\
& \bar E^{vc};\bar E^{wd};\bar E^{za};\bar E^{ub};\bar E^{ri};\bar E^{uj};\bar E^{rn};\bar E^{um};\bar C^{ru}_{ij};\bar C^{ru}_{nm};\bar C^{wz}_{da};\bar C^{uv}_{bc}; \\
& S^z;S^u;S^v;S^w;S^r.
\end{align*}
The normal form of the above sequence is given by:
$$C^{zu}_{ab};C^{vw}_{cd};\bar C^{wz}_{da};\bar C^{uv}_{bc};S^r.$$
\end{example}

\chapter{Retrosynthesis as a layered theory}\label{ch:retrosynthesis}
This penultimate chapter puts to use all three perspectives on chemical processes developed so far -- reaction schemes, reactions, and disconnection rules -- as well as the theory developed in Part~\ref{part:layered} of the thesis. We combine the three perspectives into a layered theory to propose a mathematical framework for retrosynthesis. The layers of the layered theory we propose as the habitat for retrosynthesis all share the same set of objects -- namely, the chemical graphs. The morphisms of a layer are either matchings, disconnection rules or reactions, parameterised by environmental molecules (these can act as solvents, reagents or catalysts).

\section{The layers}

Given a finite set $M$ of molecular entities, let us enumerate the molecular entities in $M$ as $M_1,\dots,M_k$. Given a list natural numbers $n=(n_1,\dots,n_k)$, we denote by $n_1M_1 + \dots + n_kM_k$ the molecular graph obtained by taking the disjoint union of $n_i$ copies of $M_i$ for all $i=1,\dots,k$. We define three classes of symmetric monoidal categories parameterised by finite sets of molecular entities as follows.
\begin{definition}\label{def:para-cats}
Let $M$ be a finite set of molecular entities. We define the categories $\MMatch$, $\MReact$ and $\MDisc$ as having the chemical graphs as objects. The morphisms are defined as follows:
\begin{itemize}
\item in $\MMatch$, a morphism $A \xrightarrow{m,r} B$ is given by a matching $m:A\rightarrow B$ together with an injection $r:r_1M_1 + \dots + r_kM_k\rightarrow B$ preserving the atom labels such that $\im(m)\cup\im(r) = B$, and $\im(m)\cap\im(b)=m(\alpha(A))$. The composition of $A\xrightarrow{m,r}B\xrightarrow{n,s}C$ is given by
$$m;n:A\rightarrow C \text{ and } r;n+s:(r_1+s_1)M_1 + \dots + (r_k+s_k)M_k\rightarrow C.$$
\item in $\MReact$, a morphism $A\rightarrow B$ is a reaction $n_1M_1 + \dots + n_kM_k + A\xrightarrow r B$ (i.e.~a morphism in $\React$). Given another reaction $m_1M_1 + \dots + m_kM_k + B\xrightarrow s C$, the composite $A\rightarrow C$ is given by
$$(r + \id_{m_1M_1 + \dots + m_kM_k});s : (n_1+m_1)M_1 + \dots + (n_k+m_k)M_k + A\rightarrow C.$$
\item in $\MDisc$, a morphism $A\rightarrow B$ is given by a morphism $n_1M_1 + \dots + n_kM_k + A\xrightarrow d B$ in $\Disc$. Given another morphism $m_1M_1 + \dots + m_kM_k + B\xrightarrow h C$ the composite $A\rightarrow C$ is given by
$$d;h : (n_1+m_1)M_1 + \dots + (n_k+m_k)M_k + A\rightarrow C.$$
\end{itemize}
If $M=\eset$, we may omit the prefix.
\end{definition}
The idea is that the set $M$ models the reaction environment: the parametric definitions above capture the intuition that there is an unbounded supply of these molecules in the environment. The categories $\MReact$ and $\MDisc$ are the parameterised~\cite{fong2019backprop} versions of $\React$ and $\Disc$: a morphism $A\rightarrow B$ implicitly has a finite number of copies of molecules from $M$ in its domain. A morphism $A\rightarrow B$ in $\MMatch$ may be seen as a reaction which preserves the structure of $A$ as it is, while potentially breaking up and rearranging the molecular entities in $M$. We proceed to give an example of this.
\begin{example}\label{ex:mmatch}
A morphism in $\MMatch$ is a matching such that the environment contains enough ``building material'' to cover the complement of the image of the matching. We give an example below, taking $M=\{\texttt{HCl}\}$, where the horizontal map is the matching, and the vertical map is the injection (cf.~Figure~\ref{fig:clayden}):
\begin{center}
\scalebox{1}{\input{thesis-ch2/mmatch-example.tikz}}.
\end{center}
\end{example}
We formalise the fact that morphisms in $\MMatch$ look like special cases of reactions by noticing that for every $M$ there is an identity-on-objects functor
$$\MEmb : \MMatch\rightarrow\MReact$$
defined by $(m,r)\mapsto m|_{\Chem A} + r$, where $A$ is the domain of $(m,r)$. Thus we have the following situation, for every finite set of molecular entities $M$:
\begin{equation}\label{eq:m-functors}
\scalebox{1}{\input{thesis-ch2/m-functors.tikz}},
\end{equation}
where $\MR$ is defined by the action of the functor $R$ constructed in Section~\ref{sec:disc-to-react}. Additionally, for every pair of finite sets of molecular entities such that $M\sse N$, there is an inclusion functor for each of the three classes of categories.

\section{Retrosynthetic steps and sequences}

With the definitions of the previous section, and with the notation for deflational theories from Section~\ref{sec:opfib-defl-theories}, we are able to give mathematical definitions of retrosynthetic steps and sequences internal to the deflational theory defined by the diagram~\eqref{eq:m-functors}.

\begin{definition}[Retrosynthetic step]\label{def:retro-step}
A {\em retrosynthetic step} consists of
\begin{itemize}
\item molecular graphs $T$ and $B$, called the {\em target}, and the {\em byproduct},
\item a finite set of molecular entities $M$, called the {\em environment},
\item a chemical graph $S$, whose connected components are called the {\em synthons},
\item a molecular graph $E$, whose connected components are called the {\em synthetic equivalents},
\item morphisms $d\in\Disc(T,S)$, $m\in\MMatch(S,E)$, and $r\in\MReact(E,T+B)$.
\end{itemize}
\end{definition}
\begin{proposition}\label{prop:step-graphically}
The data of a retrosynthetic step are equivalent to existence of the following term in the deflational theory generated by the diagram~\eqref{eq:m-functors}:
\begin{center}
\scalebox{1}{\input{thesis-ch2/retro-layered-prop.tikz}}.
\end{center}
\end{proposition}
The morphism in the above proposition should be compared to the informal diagram in Figure~\ref{fig:clayden}.

The immediate advantage of presenting a retrosynthetic step as a term in a layered theory is that it illustrates how the different parts of the definition fit together in a highly procedural manner. Equally importantly, this presentation is fully compositional: note that the three morphisms constituting a retrosynthetic step can be divided between several parties (e.g.~different labs or computers), so long as their boundaries match in the specified way. Moreover, one can reason about different components of the step while preserving a precise mathematical interpretation, so long as one sticks to the rewrites (2-equations and 2-cells) of the deflational theory: we illustrate this in the following proposition.
\begin{proposition}\label{prop:step-equiv}
The term of Proposition~\ref{prop:step-graphically} is equal (up to a bidirectional 2-cell) to the following term:
\begin{center}
\scalebox{1}{\input{thesis-ch2/retro-layered-prop-equiv.tikz}}.
\end{center}
\end{proposition}

While a retrosynthetic step represents a single step in a reaction search, the final outcome should consist of reactions that may be performed one after another, terminating with the target (together with potential byproducts). This is captured by the notion of a {\em retrosynthetic sequence}.

\begin{definition}[Retrosynthetic sequence]\label{def:retro-sequence}
A {\em retrosynthetic sequence} for a target molecular entity $T$ is a sequence of morphisms $r_1,r_2,\dots,r_n$ with
\begin{align*}
r_1 &\in M_1\mdash\React(E_1,T+B_0) \\
r_2 &\in M_2\mdash\React(E_2,E_1+B_1) \\
&\vdots \\
r_n &\in M_1\mdash\React(E_n,E_{n-1}+B_{n-1})
\end{align*}
such that the codomain of $r_{i+1}$ is the disjoint union of the domain of $r_i$ with some other molecular graph:
\begin{center}
\scalebox{1}{\input{thesis-ch2/retro-sequence.tikz}}.
\end{center}
\end{definition}
Thus a retrosynthetic sequence is a chain of reactions, together with reaction environments, such that the products of one reaction can be used as the reactants for the next one, so that the reactions can occur one after another (assuming that the products can be extracted from the reaction environment, or one environment transformed into another one). In the formulation of a generic retrosynthesis procedure below, we shall additionally require that each reaction in the sequence comes from ``erasing'' everything but the rightmost cell in a retrosynthetic step.

\section{Formalisation of retrosynthetic analysis}

We are now ready to formulate step-by-step retrosynthetic analysis. The procedure is a high-level mathematical description that, we suggest, is flexible enough to capture all instances of retrosynthetic algorithms. As a consequence, it can have various computational implementations. Let $T$ be some fixed molecular entity. We initialise by setting $i=0$ and $E_0\coloneqq T$.
\begin{enumerate}
\item Choose a subset $\mathcal D$ of sequences of disconnection rules,
\item Provide at least one of the following:
\begin{enumerate}
\item a finite set of reaction schemes $\mathcal S$,
\item a function $\mathfrak F$ from molecular graphs to finite sets of molecular graphs,
\end{enumerate}
\item Search for a retrosynthetic step with morphisms $d\in\Disc(E_i,S)$, $m\in\MMatch(S,E)$, and $r\in\MReact(E,E_i+B_i)$ such that $d\in\mathcal D$, and we have at least one of the following:
\begin{enumerate}
\item there is an $s\in\mathcal S$ such that the reaction $r$ is an instance of $s$,
\item $E_i + B_i\in\mathfrak F(E)$;
\end{enumerate}
if successful, set $E_{i+1}\coloneqq E$, $M_{i+1}\coloneqq M$, $r_{i+1}\coloneqq r$ and proceed to Step~4; if unsuccessful, stop,
\item Check if the molecular entities in $E_{i+1}$ are known (commercially available): if yes, terminate; if no, increment $i\mapsto i+1$ and return to Step~1.
\end{enumerate}
Note how our framework is able to incorporate both template-based and template-free retrosynthesis, corresponding to the choices between (a) and (b) in Step~2: the set $\mathcal S$ is the template, while the function $\mathfrak F$ can be a previously trained algorithm, or other unstructured empirical model of reactions. We can also consider hybrid models by providing both $\mathcal S$ and $\mathfrak F$, hence allowing for combinations of existing algorithms.

We take the output retrosynthetic sequence to always come with a specified reaction environment for each reaction. Currently existing tools rarely provide this information (mostly for complexity reasons), and hence, in our framework, correspond to the set $M$ always being empty in Step~3.

The retrosynthetic steps outputted by the above procedure are highly tunable: the choice of the set $\mathcal D$ determines what kinds of bonds are disconnected (one could, for example, put an upper bound to the number of disconnected covalent bonds), while the set $\mathcal S$ can be used to enforce the presence of a functional group. Introducing negative application conditions for double pushout rewriting~\cite{ehrig-constraints2004,machado2015} would further allow enforcing an absence of a functional group\footnote{I thank an anonymous reviewer of {\em Theoretical Computer Science} for bringing up this point.}. A rudimentary form of specificity is obtained by minimising the size of the byproduct molecular graph $B_i$. Further adjustments increasing the yield include choosing the environment $M_i$ (which, inter alia, can function as a catalyst) as well as introducing protection-deprotection steps.

Steps~1 and~2 both require making some choices. Two approaches to reduce the number of choices, as well as the search space in Step~3, have been proposed in the automated retrosynthesis literature: to use molecular similarity~\cite{Coley2017-similarity}, or machine learning~\cite{Lin2020}. Chemical similarity can be used to determine which disconnection rules, reactions and environment molecules are actually tried: e.g.~in Step~1, disconnection rules that appear in syntheses of molecules similar to $T$ can be prioritised.

Ideally, each unsuccessful attempt to construct a retrosynthetic step in Step~3 should return some information on why the step failed: e.g.~if the codomain of a reaction fails to contain $E_i$, then the output should be the codomain and a measure of how far it is from $E_i$. Similarly, if several reactions are found in Step~3, some of which result in products $O$ that do not contain $E_i$, the step should suggest minimal alterations to $E$ such that these reactions do not occur. This can be seen as a {\em deprotection} step: the idea is that in the next iteration the algorithm will attempt to construct (by now a fairly complicated) $E$, but now there is a guarantee this is worth the computational effort, as this prevents the unwanted reactions from occurring ({\em protection} step). Passing such information between the layers would take the full advantage of the layered formalism.

\chapter{Discussion and future work}\label{ch:conclusion}
We have discussed in detail three mathematical perspectives on chemical processes: reaction schemes, category of reactions and disconnection rules. Reaction schemes are a compact way to store chemical reaction data, and generate all the formal chemical reactions via double pushout rewriting. The category of reactions captures combinatorially all the theoretically possible chemical transformations of graphs, as well as providing a uniform notion of composition for reactions. Disconnection rules provide the fine-grained, low level syntax of all chemically feasible local graph transformations. The completeness and universality results of Section~\ref{sec:disc-to-react} show that the disconnection rules and reactions are tightly linked, further motivating the use of disconnection rules, hitherto only appearing informally in the (computational) retrosynthesis literature, for both storing reaction data and as part of retrosynthetic analysis.

Universality can be thought of as a consistency result for reactions: their definition captures exactly those rearrangements of chemical graphs which result from local, chemically motivated rewrite rules. Completeness says that there is no redundancy in the representation: treating the (dis)connection rules as terms, the terms can be endowed with equations such that the terms describing the same reaction are identified. As the decomposition of a reaction into a sequence of (dis)connection rules is algorithmic, these results can be used to automatically break up a reaction (or its part) into smaller components: the purpose can be, {\em inter alia}, retrosynthetic analysis or storing reaction data in a systematic way.

The main conceptual contributions achieved by formulating retrosynthesis as a layered theory are the explicit mathematical descriptions of retrosynthetic steps (Definition~\ref{def:retro-step}) and sequences (Definition~\ref{def:retro-sequence}), which allows for a precise formulation of the entire process, as well as of more fine-grained concepts.

\section{Future work}\label{sec:future-work}

The future developments can be approached from three perspectives: applications to chemistry, computational implementations, and further mathematical developments.

\subsection{Chemical questions}

While stereochemistry is relatively straightforward to account for on the level of chemical graphs and reactions (Section~\ref{sec:chirality}), it is unclear how to do this for the disconnection rules, as they only operate at one or two vertices at a time. A more straightforward extension of disconnection category would introduce energy and dynamics into the disconnection rules by quantifying how much energy each (dis)connection (in a particular context) requires to occur.

While in the current presentation we showed how to account for the available disconnection rules, reactions and environmental molecules as part of the retrosynthetic reaction search, the general formalism of layered theories immediately suggests how to account for other environmental factors (e.g.~temperature and pressure). Namely, these should be represented as posets which control the morphisms that are available between the chemical compounds. One idea for accounting for the available energy is via the disconnection rules: the higher the number of bonds that we are able to break in one step, the more energy is required to be present in the environment\footnote{I thank Fredrik Dahlqvist for this suggestion.}.

\subsection{Computational questions}

Given the algorithmic nature of both completeness and universality proofs, the next step is to implement both. The first algorithm would take an arbitrary reaction as an input, and output a sequence of disconnection rules representing it. The second algorithm would decide whether two terms are equal or not, implementing the normalisation procedure. Another direction for connecting this work with more standard approaches to computational chemistry would be translating our formalism to a widely used notation such as SMILES~\cite{smiles88,daylight-smiles}.

On the side of retrosynthetic design, the crucial next step is to take existing retrosynthesis algorithms and encode them in our framework. This requires implementing the terms of the layered theory in the previous section in some software. As the terms in a layered theory are represented by string diagrams, one approach is to use proof formalisation software specific to string diagrams and their equational reasoning, such as~\cite{cartographer}. Alternatively, these terms could be encoded in a programming language like python or Julia. The latter is especially promising, as there are modules formalising category theory available for it~\cite{catlab2020,AlgebraicJulia}. As a lower level description, the disconnection rules and the reactions presented could be encoded in some graph rewriting language, such as Kappa~\cite{kappa-language,kappa2004,krivine-siglog,rewriting-life21}, which is used to model systems of interacting agents, or M\O{}D~\cite{mod-language,intermediate-level,chem-trans-motifs,rewriting-life21}, which represents molecules as labelled graphs and generating rules for chemical transformations as spans of graphs (akin to this work). Another promising direction is to treat the terms arising from disconnection rules as {\em proofs} that a certain reaction is possible. This opens the possibility of formalising the sequences of disconnection rules in a modern proof assistant such as Lean~\cite{lean4}, which would allow using tactics and automation for deriving (potential) chemical reactions as proof terms.

\subsection{Mathematical questions}

An important mathematical development is to introduce monoidal terms into the disconnection category, so as to allow parallel reactions, making the discussion in Chapter~\ref{ch:retrosynthesis} mathematically rigorous, as well as allowing for the usage of graphical calculi for monoidal categories. We have, in fact, already started this development as one of our case studies (Section~\ref{sec:glucose}) in Part~\ref{part:layered}.

Another mathematical question is whether the categories $\Disc$ and $\React$ have any interesting categorical structure, such as being restriction categories~\cite{cocket-lack02}.

At the level of the layered formalism, the next step is to model translations between the reaction environments as functors of the form $M\mdash\React\rightarrow N\mdash\React$. This would allow presenting a retrosynthetic sequence as a single, connected diagram, closely corresponding to actions to be taken in a lab. Similarly, we note that the informal algorithmic description in Chapter~\ref{ch:retrosynthesis} should ideally be presented internally in the layered theory: Steps 1 and 2 amount to choosing subcategories of $\Disc$ and $\React$.

\newpage
\phantomsection
\addcontentsline{toc}{part}{Bibliography}
\bibliographystyle{./sty/leosbibstyle}
\bibliography{bibliography}

\appendix

\part*{Appendices}\label{part:appendices}
\addcontentsline{toc}{part}{\nameref{part:appendices}}

\chapter{Parametric categories}\label{ch:para-copara}
Here we cover the graphical notation for the {\em (co)para construction}. The construction has appeared several times in the (applied) category theory literature, explicitly under this name in~\cite{fong2019backprop,gavranovic19,cruttwell21}.

\begin{definition}\label{def:copara}
Given a strict braided monoidal category $(\cat C,\otimes,1)$, we denote by $\Copara(\cat C)$ the {\em coparametric category over $\cat C$}: the objects are those of $\cat C$, the monoidal generators are
\begin{center}
\scalebox{1}{\input{copara-generators.tikz}},
\end{center}
subject to the equations in Figure~\ref{fig:copara-eqns}. Additionally, it will be convenient to omit the box parameterised by the monoidal unit, and to allow ``deparemeterisation'':
\begin{center}
\scalebox{1}{\input{copara-shorthand.tikz}}.
\end{center}
\end{definition}
\begin{figure}
  \centering
  \scalebox{1}{\input{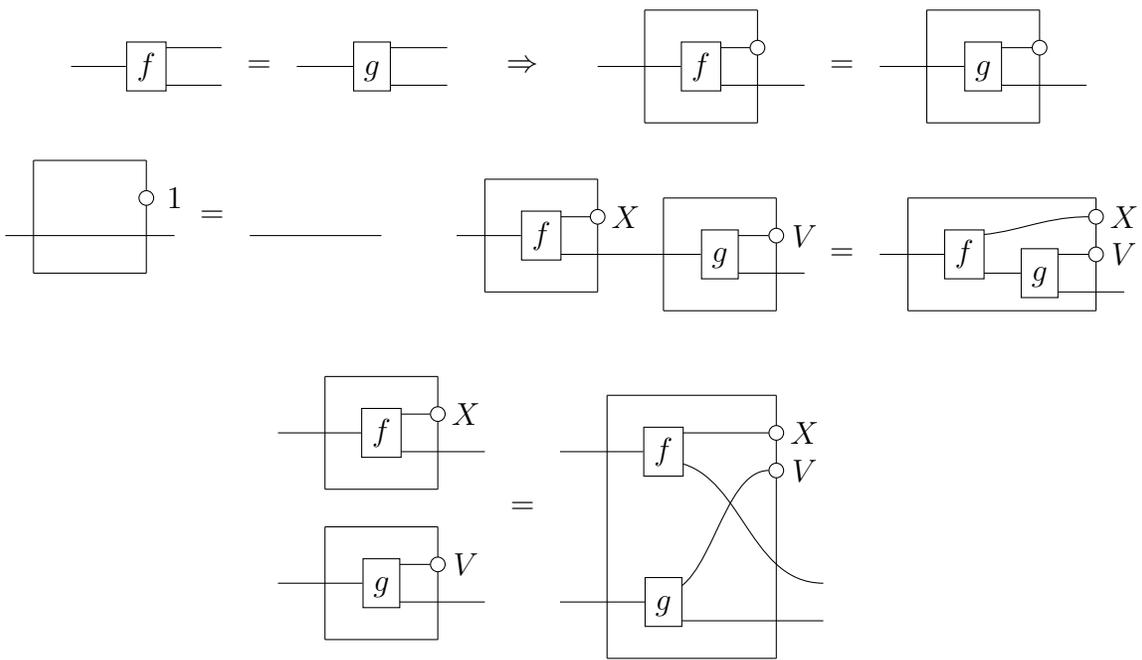}}
  \caption{Equations of the coparametric category\label{fig:copara-eqns}}
\end{figure}

\begin{proposition}
There is a faithful functor $\cat C\rightarrow\Copara(\cat C)$ taking each object and morphism to itself.
\end{proposition}

The dual construction, with the generators flipped, is the {\em parametric category}, denoted by $\Para(\cat C)$. A variant of this construction is given by $\ParaSwap(\cat C)$, where the equation defining composition is replaced by
\begin{center}
\scalebox{1}{\input{paraswap-composition.tikz}}.
\end{center}

\chapter{ZX-calculus and measurement based quantum computing}\label{ch:zx-mbqc}
The basic unit of quantum computation is a {\em qubit}, represented by the Hilbert space $\C^2$. The {\em state} of a qubit is a ray in $\C^2$, and can, therefore, be represented by a linear combination of the standard basis vectors $\ket{0}\coloneq (1,0)$ and $\ket{1}\coloneq (0,1)$, where two linear combinations represent the same state if they are equal up to multiplication by a non-zero scalar. As is standard in the quantum information literature, we represent the basis states in the tensor products of qubits by listing the elements in each component: e.g.~$\ket{00}=\ket{0}\otimes\ket{0}$ is a state in $\C^2\otimes\C^2$. Similarly, the linear maps between qubits are represented by reversing the notation for states: e.g.~$\bra{0}:\C^2\rightarrow\C$ is the linear functional defined by taking the inner product with $\ket{0}$, and $\ketbra{0}{0}:\C^2\rightarrow\C^2$ is the projection to the state $\ket{0}$ (this is known as Dirac's {\em bra-ket notation}). In the context of quantum computing, the standard orthonormal basis $\left\{\ket{0},\ket{1}\right\}$ is often referred to as the {\em computational basis}, to distinguish it from the {\em Hadamard basis} $\left\{\ket{+},\ket{-}\right\}$, defined by $\ket{+}\coloneq\frac{1}{\sqrt 2}\left(\ket{0}+\ket{1}\right)$ and $\ket{-}\coloneq\frac{1}{\sqrt 2}\left(\ket{0}-\ket{1}\right)$. The ZX-calculus can then be seen as an axiomatisation of how these two bases interact.

There is a strict monoidal functor from $\ZX$ (with generators~\eqref{eq:zx-generators} and equations~\eqref{eq:zx-rules}) to the category of finite-dimensional Hilbert spaces given by $n\mapsto\bigotimes_{i=1}^n\C^2$ (the n-fold tensor product of $\C^2$ with itself) on objects, and the Z- and X-spiders are, respectively, mapped to
$$\ketbra{0\cdots 0}{0\cdots 0} + e^{i\alpha}\ketbra{1\cdots 1}{1\cdots 1}\text{ and } \ketbra{+\cdots +}{+\cdots +} + e^{i\alpha}\ketbra{-\cdots -}{-\cdots -}.$$
Note that equation \HadamardDef in~\eqref{eq:zx-rules} then forces the interpretation of the Hadamard gate to be (up to a scalar) the transformation that maps the computational basis to the Hadamard basis: $\ket{0}\mapsto\ket{+}$ and $\ket{1}\mapsto\ket{-}$. As an example of this functorial interpretation, the computational basis states $\ket{0}$ and $\ket{1}$ are now represented by the X-spiders
$$\begin{tikzpicture}
	\begin{pgfonlayer}{nodelayer}
		\node [style=none] (1) at (1.5, 0) {};
		\node [style=X dot] (2) at (0, 0) {};
	\end{pgfonlayer}
	\begin{pgfonlayer}{edgelayer}
		\draw (2) to (1.center);
	\end{pgfonlayer}
\end{tikzpicture}
\quad\text{and }\quad\begin{tikzpicture}
	\begin{pgfonlayer}{nodelayer}
		\node [style=none] (1) at (1.5, 0) {};
		\node [style=X phase dot] (3) at (0, 0) {$\pi$};
	\end{pgfonlayer}
	\begin{pgfonlayer}{edgelayer}
		\draw (3) to (1.center);
	\end{pgfonlayer}
\end{tikzpicture}
,$$
and, likewise, the Hadamard basis states $\ket{+}$ and $\ket{-}$ are represented by the Z-spiders
$$\begin{tikzpicture}
	\begin{pgfonlayer}{nodelayer}
		\node [style=Z dot] (0) at (0, 0) {};
		\node [style=none] (1) at (1.5, 0) {};
	\end{pgfonlayer}
	\begin{pgfonlayer}{edgelayer}
		\draw (0) to (1.center);
	\end{pgfonlayer}
\end{tikzpicture}
\quad\text{and }\quad\begin{tikzpicture}
	\begin{pgfonlayer}{nodelayer}
		\node [style=Z phase dot] (0) at (0, 0) {$\pi$};
		\node [style=none] (1) at (1.5, 0) {};
	\end{pgfonlayer}
	\begin{pgfonlayer}{edgelayer}
		\draw (0) to (1.center);
	\end{pgfonlayer}
\end{tikzpicture}
.$$

The above functor is full, corresponding to universality of the calculus: any linear map between finite tensor products of qubits can be expressed as a ZX-diagram. Under the equations we chose to present, the functor is not quite faithful, so the calculus is not complete. However, restricting the diagrams to the {\em Clifford fragment}, i.e.~the diagrams where the phase $\alpha$ is a multiple of $\frac{\pi}{2}$ makes the calculus complete, so long as we ignore the non-zero scalars~\cite{backens-2014}. Adding more equations, the calculus has been shown complete for the Clifford+T fragment~\cite{jeandel-univ-complete}, and finally for the full ZX-calculus~\cite{hadzihasanovic-zx18}.

Recall that $\mathcal P\coloneq\{\XYplane,\XZplane,\YZplane\}$ denotes the set of {\em measurement planes}.
\begin{definition}[Measurement pattern~\cite{measurement-calculus}]\label{def:measurement-pattern}
Let $(V,I,O)$ be a triple of a finite set $V$ with two chosen subsets $I,O\sse V$. The set $\mathcal S$ of {\em signals} consists of the formal sums of $0$, $1$ and $s_i$, where $i\in V\setminus O$. A {\em measurement pattern} over $(V,I,O)$ is an element in the free monoid generated by the set
$$\left\{N_j, E_{nm}, M^{\lambda,\alpha}_k, X_n^s, Z_n^s\right\},$$
where $j\in V\setminus I$, $k\in V\setminus O$, $n,m\in V$ with $n\neq m$, $\lambda\in\mathcal P$, $\alpha\in [0,2\pi)$, and $t,s\in\mathcal S$, and for $C\in\{X,Z\}$ we identify $C_n\coloneq C_n^1$ and $C_n^0=\varepsilon$.
\end{definition}
A measurement pattern on $(V,I,O)$ can be thought of as a sequence of commands that takes in $|I|$ qubits, executes the commands one at a time in the order they appear, and outputs $|O|$ qubits. The intuition behind the commands is as follows (see also Table~\ref{tab:MBQC-to-ZX} for the translation of measurement patterns to the ZX-calculus):
\begin{itemize}
\item $N_j$ prepares the non-input qubit $j$ in the state $\ket{+}$,
\item $E_{nm}$ entangles two distinct qubits $n$ and $m$ by applying a CZ-gate to them,
\item $M^{\lambda,\alpha}_k$ destructively measures the non-output qubit $k$ by projecting it to the Hadamard basis in the plane $\lambda$ in the Bloch sphere shifted by the angle $\alpha$,
\item the corrections $X_n^s$ and $Z_n^s$ act as the Pauli-X and Pauli-Z operators on the qubit $n$ if the signal $s$ evaluates to $1$ in $\Z_2$, and as the identities otherwise.
\end{itemize}
Once a measurement on qubit $k$ is performed, the qubit $k$ is either in the (shifted) state $\ket{+}$, which is modelled by setting the signal $s_k$ to $0$, or in the (shifted) state $\ket{-}$, which is modelled by setting the signal $s_k$ to $1$. A correction is said to {\em depend} on the outcome of measuring a qubit $k$ if the signal $s_k$ appears in its expression.

Clearly, not all measurement patterns are physically realisable: e.g.~no correction should depend on an outcome of a qubit that has not yet been measured, and commands should only act on qubits that have been either prepared or are input qubits. Patterns satisfying such realisability conditions are called {\em runnable}~\cite{browne-gflow,thereandback}. Any runnable pattern can be written in a standard form, without changing the semantics in terms of linear maps~\cite{measurement-calculus,browne-gflow}. In the standard form, the classes of commands appear in the following order: preparations, entanglement commands, measurements, corrections. Note that this corresponds to the physical order in which the commands would be executed in the experimental implementation. We now proceed to define the semantics of patterns.

Given a runnable measurement pattern with $n$ measurement commands, a {\em branch} of the pattern is a sequence of $0$'s and $1$'s of length $n$. The interpretation of a branch is the experimental scenario in which $i$th measurement has yielded the outcome specified by the corresponding entry in the sequence. A {\em partial branch} is a sequence of $0$'s, $1$'s and symbols $x$ of length $n$. The interpretation of a partial branch is that only some of the qubits were measured, with the outcomes specified by the $0$ and $1$ entires in the sequence. Given a partial branch $b=b_1\cdots b_n$, define the {\em support} of $b$ as $\supp(b)\coloneq\{i : b_i\neq x\}$. We say that a partial branch $c$ {\em complements} a branch $b$ if $\supp(b)\cap\supp(c)=\eset$ and $\supp(b)\cup\supp(c)=\{1,\dots,n\}$. We note that a branch is a partial branch with full support, i.e.~with no symbols $x$, and that any two partial branches $b$ and $c$ that complement each other can be uniquely combined into a branch $b*c$.

\begin{definition}[Pattern coarsening]
Given a measurement pattern $P$ and a partial branch $b=b_1\cdots b_n$, we obtain a new pattern $P_b$ by modifying $P$ according to each entry $b_k$ as follows:
\begin{itemize}
\item if $b_k=x$, do nothing,
\item if $b_k=0$, substitute $s_k\mapsto 0$,
\item if $b_k=1$, substitute $M^{\lambda,\alpha}_k\mapsto M^{\lambda,\alpha+\pi}_k$ and $s_k\mapsto 1$.
\end{itemize}
We say that the pattern $P_b$ is the {\em coarsening} of the pattern $P$ by the partial branch $b$.
\end{definition}
Coarsening a runnable pattern by a branch results in a pattern with all signals replaced with formal sums of $0$'s and $1$'s.

Any branch $b$ of a runnable pattern $P$ realises a linear map from $|I|$ qubits to $|O|$ qubits~\cite{measurement-calculus} as follows:
\begin{enumerate}[label={(\arabic*)}]
\item coarsen the pattern, obtaining the pattern $P_b$ with no signal symbols $s_k$,
\item replace each signal in $P_b$ with either $0$ or $1$ by evaluating the formal sums in $\Z_2$,
\item translate the pattern to a ZX-diagram according to the translation procedure in Table~\ref{tab:MBQC-to-ZX},
\item translate the resulting diagram to a linear map.
\end{enumerate}
The conditions for a runnable pattern guarantee that the procedure defined in Table~\ref{tab:MBQC-to-ZX} is well-defined and terminating. Given a runnable measurement pattern $P=(V,I,O)$ and a branch $b$, we denote the linear map $\bigotimes_{i\in I}\C^2\rightarrow\bigotimes_{o\in O}\C^2$ realised by $b$ by $L(P_b)$. We remark that, since the ZX-calculus is universal and can be made complete, the last step of the translation can often be omitted, and we can reason directly with ZX-diagrams rather than linear maps.

 \begin{table}
  \centering
  \renewcommand{\arraystretch}{2.5}
  \begin{tabular}{c|c}
   command / vertex & diagram \\ \hline
   $i\in I$ & \input{id-input.tikz} \\ \hline
   $N_j$ & \input{Z-zero.tikz} \\ \hline
   $E_{nm}$ & \input{cz.tikz} \\ \hline
   $M^{\XYplane,\alpha}_k$ & \input{XY-effect.tikz} \\ \hline
   $M^{\XZplane,\alpha}_k$ & \input{XZ-effect.tikz} \\ \hline
   $M^{\YZplane,\alpha}_k$ & \input{YZ-effect.tikz} \\ \hline
   $X_n$ & \input{Pauli-X.tikz} \\ \hline
   $Z_n$ & \input{Pauli-Z.tikz} \\ \hline
   $o\in O$ & \input{id-output.tikz}
  \end{tabular}
  \renewcommand{\arraystretch}{1}
  \caption{Translation from measurement patterns to ZX-diagrams: match the labels on the left to labels that have already been drawn, append the specified diagram, and replace the labels with the ones on the right.\label{tab:MBQC-to-ZX}}
 \end{table}

In general, different branches of a runnable pattern may result in different linear maps. For useful algorithms and computations, one is interested in patterns where the linear maps realised by different branches are related. This gives rise to different flavours of {\em determinism}.
\begin{definition}[Determinism~\cite{measurement-calculus,browne-gflow}]\label{def:determinism}
A runnable measurement pattern $P$ is
\begin{itemize}
\item {\em deterministic} if for any two branches $b$ and $c$, the linear maps they realise are pointwise proportional: for any input $q$, the vectors $L(P_b)(q)$ and $L(P_c)(q)$ are scalar multiples of each other,
\item {\em strongly deterministic} if any two branches result in the same linear map, up to a non-zero scalar,
\item {\em uniformly deterministic} if it is deterministic for any choice of the measurement angles,
\item {\em stepwise deterministic} if for any two partial branches $b$ and $c$ with the same support, there is a sequence $C$ of correction commands, which depend only on the qubits that are not in the common support of $b$ and $c$, making the coarsened patterns $P_b$ and $P_c$ equivalent in the following sense: for any partial branches $b'$ and $c'$ that complement $b$ and $c$, the resulting linear maps $L\left((C\cdot P)_{b*b'}\right)$ and $L\left(P_{c*c'}\right)$ are pointwise proportional.
\end{itemize}
\end{definition}

\begin{example}
Consider the following pattern\footnote{Note that the usual order of commands in a measurement pattern is right to left, corresponding to the usual order of composition of operators (functions). We, however, choose to write the commands from left to right, so as to match the diagrammatic order of composition.} on $(\{1,2\},\{1\},\{1\})$
$$N_2E_{12}M_2^{\XYplane,0}X_1^{s_2}.$$
It has two branches, corresponding to the measurement outcomes of $\ket{+}$ ($s_2=0$) and $\ket{-}$ ($s_2=1$), resulting in the following ZX-diagrams upon translating according to Table~\ref{tab:MBQC-to-ZX}:
\begin{center}
\scalebox{1}{\input{bitflip-correction-1.tikz}}.
\end{center}
Simplifying using the equations of the ZX-calculus, we get that the above diagrams are equal to
\begin{center}
\scalebox{1}{\input{bitflip-correction-2.tikz}},
\end{center}
which are $\ketbra{0}{0}$ and $\ketbra{0}{1}$, so that the outcome state is always proportional to $\ket{0}$. Thus the pattern is deterministic, but not strongly deterministic.
\end{example}

\begin{example}
Consider the following pattern on $(\{1,2\},\{1\},\{2\})$
$$N_2E_{12}M_1^{\XYplane,0}X_2^{s_1}.$$
The two branches result in the following diagrams:
\begin{center}
\scalebox{1}{\input{hadamard-pattern.tikz}},
\end{center}
both of which simplify to the Hadamard gate. Thus the pattern is strongly deterministic, and each branch implements the Hadamard gate.
\end{example}

\chapter{ICE-form}\label{ch:ice-form}
Here we give the detailed inductive proof of the fact that any term has an equivalent term in an $ICE$-form (Proposition~\ref{prop:ICE-form}).

\begin{lemma}\label{lma:I-commutes}
Let $\mathtt t$ be a term such that the term $\mathtt t;I^{uv}$ is defined. Then there exists a term $\mathtt t'$ such that $\mathtt t$ and $\mathtt t'$ have the same number of $I$-terms, and one of the following holds:
\begin{enumerate}[label={(\arabic*)}]
\item $\mathtt t;I^{uv}\equiv \mathtt t'$, or
\item there is a disconnection $I^{ab}$ such that $\mathtt t;I^{uv}\equiv I^{ab};\mathtt t'$.
\end{enumerate}
\end{lemma}
\begin{proof}
We proceed by induction on the structure of $\mathtt t$.
\paragraph{Base cases}
\begin{align*}
\id;I^{uv} &\equiv I^{uv};\id, \\
S^w;I^{uv} &\equiv I^{uv};S^w, \tag{by~\eqref{disc-eq:sd1}} \\
R^{w\mapsto z};I^{uv} &\equiv I^{uv};R^{w\mapsto z}, \tag{by~\eqref{disc-eq:rd1}} \\
I^{wz};I^{uv} &\phantom{\equiv} \text{ is already in the right form} \\
C^{wz}_{ab};I^{uv} &\equiv I^{uv};C^{wz}_{ab}, \tag{by~\eqref{disc-eq:comm3}} \\
E^{w}_{ab};I^{uv} &\equiv I^{uv};E^{w}_{ab}, \tag{by~\eqref{disc-eq:comm4}} \\
E^{wz};I^{uv} &\equiv  I^{uv};E^{wz}, \tag{by~\eqref{disc-eq:comm5}} \\
\bar I^{wz};I^{uv} &\equiv \begin{cases} S^u;S^v \text{ if } w=u, z=v \\
                                    I^{uv};\bar I^{wz} \text{ otherwise,}
                      \end{cases} \tag{by~\eqref{disc-eq:ddbar4-2} and~\eqref{disc-eq:comm2}} \\
\bar C^{wz}_{ab};I^{uv} &\equiv I^{uv};\bar C^{wz}_{ab}, \tag{by~\eqref{disc-eq:comm8}} \\
\bar E^{w}_{ab};I^{uv} &\equiv I^{uv};\bar E^{w}_{ab}, \tag{by~\eqref{disc-eq:comm7}} \\
\bar E^{wz};I^{uv} &\equiv I^{uv};\bar E^{wz}. \tag{by~\eqref{disc-eq:comm6}}
\end{align*}
\paragraph{Inductive case} Let $\mathtt t:A\rightarrow B$ and $\mathtt s:B\rightarrow C$ be terms such that the statement of the lemma holds. Suppose that the term $\mathtt t;\mathtt s;I^{uv}$ is defined. Then also the term $\mathtt s;I^{uv}$ is defined, so by the inductive hypothesis for $\mathtt s$, there is a term $\mathtt s'$ with the same number of $I$-terms as $\mathtt s$ such that either (1) $\mathtt s;I^{uv}\equiv\mathtt s'$, or (2) $\mathtt s;I^{uv}\equiv I^{ab};\mathtt s'$ for some $I$-term $I^{ab}$. In the first case, we have
$$\mathtt t;\mathtt s;I^{uv} \equiv \mathtt t;\mathtt s',$$
and since $\mathtt t;\mathtt s'$ has the same number of $I$-terms as $\mathtt t;\mathtt s$, it is the sought-after term for the inductive case satisfying (1). In the second case, we have
$$\mathtt t;\mathtt s;I^{uv} \equiv \mathtt t;I^{ab};\mathtt s',$$
so that $\mathtt t;I^{ab}$ is defined. By the inductive hypothesis for $\mathtt t$, there is a term $\mathtt t'$ with the same number of $I$-terms as $\mathtt t$ such that either (1) $\mathtt t;I^{ab}\equiv\mathtt t'$, or (2) $\mathtt t;I^{ab}\equiv I^{wz};\mathtt t'$ for some $I$-term $I^{wz}$. In the first case, we get
$$\mathtt t;\mathtt s;I^{uv} \equiv \mathtt t';\mathtt s',$$
satisfying (1) for the inductive case, as $\mathtt t';\mathtt s'$ and $\mathtt t;\mathtt s$ have the same number of $I$-terms. In the second case, we obtain
$$\mathtt t;\mathtt s;I^{uv} \equiv I^{wz};\mathtt t';\mathtt s',$$
satisfying (2) for the inductive case.
\end{proof}

\begin{corollary}\label{cor:I-form}
Any term is equal to a term of the form $\mathtt I;\mathtt t$, where $\mathtt I$ is a sequence of $I$-terms, and the term $\mathtt t$ contains no $I$-terms.
\end{corollary}

\begin{lemma}\label{lma:C-commutes}
Let $\mathtt t$ be a term not containing any $I$-terms such that the term $\mathtt t;C^{uv}_{ab}$ is defined. Then there exists a term $\mathtt t'$ not containing any $I$-terms such that $\mathtt t$ and $\mathtt t'$ have the same number of $C$-terms, and one of the following holds:
\begin{enumerate}[label={(\arabic*)}]
\item $\mathtt t;C^{uv}_{ab}\equiv\mathtt t'$, or
\item there is a disconnection $C^{wz}_{cd}$ such that $\mathtt t;C^{uv}_{ab}\equiv C^{wz}_{cd};\mathtt t'$.
\end{enumerate}
\end{lemma}
\begin{proof}
By induction on $\mathtt t$.
\paragraph{Base cases}
\begin{align*}
\id;C^{uv}_{ab} &\equiv C^{uv}_{ab};\id, \\
S^w;C^{uv}_{ab} &\equiv C^{uv}_{ab};S^w, \tag{by~\eqref{disc-eq:sd1}} \\
R^{w\mapsto z};C^{uv}_{ab} &\equiv \begin{cases} C^{uv}_{ab};R^{w\mapsto z} \text{ if } w\neq\{a,b\}, \\
                                            C^{uv}_{kb};R^{a\mapsto z};R^{k\mapsto a} \text{ if } w=a, \\
                                            C^{uv}_{ak};R^{b\mapsto z};R^{k\mapsto b} \text{ if } w=b,
                              \end{cases} \tag{by~\eqref{disc-eq:rd1} and~\eqref{disc-eq:rd3}} \\
C^{wz}_{cd};C^{uv}_{ab} &\phantom{\equiv} \text{ is already in the right form,} \\
E^{w}_{cd};C^{uv}_{ab} &\equiv C^{uv}_{ab};E^{w}_{cd}, \tag{by~\eqref{disc-eq:comm9}} \\
E^{wz};C^{uv}_{ab} &\equiv C^{uv}_{ab};E^{wz}, \tag{by~\eqref{disc-eq:comm10}} \\
\bar I^{wz};C^{uv}_{ab} &\equiv C^{uv}_{ab};\bar I^{wz}, \tag{by~\eqref{disc-eq:comm8}} \\
\bar C^{wz}_{cd};C^{uv}_{ab} &\equiv \begin{cases} S^w;S^z;R^{c\mapsto j};R^{d\mapsto b};R^{j\mapsto a} \text{ if } w=u, z=v, \\
                                                   S^w;S^z;R^{c\mapsto j};R^{d\mapsto a};R^{j\mapsto b} \text{ if } w=v, z=u, \\
                                                   C^{uv}_{ij};\bar C^{wz}_{cd};R^{i\mapsto a};R^{j\mapsto b} \text{ otherwise,}
                                     \end{cases} \tag{by~\eqref{prop:disc-ids0} and~\eqref{disc-eq:rd4}} \\
\bar E^{w}_{cd};C^{uv}_{ab} &\equiv C^{uv}_{ij};\bar E^{w}_{cd};R^{i\mapsto a};R^{j\mapsto b}, \tag{by~\eqref{disc-eq:rd4}} \\
\bar E^{wz};C^{uv}_{ab} &\equiv C^{uv}_{ab};\bar E^{wz}. \tag{by~\eqref{disc-eq:comm11}}
\end{align*}
The inductive case is very similar to that of Lemma~\ref{lma:I-commutes}.
\end{proof}

\begin{corollary}\label{cor:C-form}
Any term is equal to a term of the form $\mathtt I;\mathtt C;\mathtt t$, where $\mathtt I$ and $\mathtt C$ are sequences of $I$-terms and $C$-terms, and the term $\mathtt t$ contains no $I$-terms or $C$-terms.
\end{corollary}

\begin{lemma}\label{lma:Eneg-commutes}
Let $\mathtt t$ be a term not containing any $I$- or $C$-terms such that the term $\mathtt t;E^{u}_{ab}$ is defined. Then there exists a term $\mathtt t'$ not containing any $I$- or $C$-terms such that $\mathtt t$ and $\mathtt t'$ have the same number of $E^{<0}$-terms, and one of the following holds:
\begin{enumerate}[label={(\arabic*)}]
\item $\mathtt t;E^{u}_{ab}\equiv\mathtt t'$, or
\item there is a disconnection $E^{w}_{cd}$ such that $\mathtt t;E^{u}_{ab}\equiv E^{w}_{cd};\mathtt t'$.
\end{enumerate}
\end{lemma}
\begin{proof}
By induction on $\mathtt t$.
\paragraph{Base cases}
\begin{align*}
\id;E^{u}_{ab} &\equiv E^{u}_{ab};\id, \\
S^w;E^{u}_{ab} &\equiv E^{u}_{ab};S^w, \tag{by~\eqref{disc-eq:sd1}} \\
R^{w\mapsto z};E^{u}_{ab} &\equiv \begin{cases} E^{u}_{ab};R^{w\mapsto z} \text{ if } w\neq\{a,b\}, \\
                                           E^{u}_{kb};R^{a\mapsto z};R^{k\mapsto a} \text{ if } w=a, \\
                                           E^{u}_{ak};R^{b\mapsto z};R^{k\mapsto b} \text{ if } w=b,
                             \end{cases} \tag{by~\eqref{disc-eq:rd1} and~\eqref{disc-eq:rd3}} \\
E^{w}_{cd};E^{u}_{ab} &\phantom{\equiv} \text{ is already in the right form,} \\
E^{wz};E^{u}_{ab} &\equiv E^{u}_{ab};E^{wz}, \tag{by~\eqref{disc-eq:comm12}} \\
\bar I^{wz};E^{u}_{ab} &\equiv E^{u}_{ab};\bar I^{wz}, \tag{by~\eqref{disc-eq:comm7}} \\
\bar C^{wz}_{cd};E^{u}_{ab} &\equiv E^{u}_{ij};\bar C^{wz}_{cd};R^{i\mapsto a};R^{j\mapsto b}, \tag{by~\eqref{disc-eq:rd4}} \\
\bar E^{w}_{cd};E^{u}_{ab} &\equiv \begin{cases} S^w;R^{c\mapsto j};R^{d\mapsto b};R^{j\mapsto a} \text{ if } w=u, \\
                                                 E^u_{ij};\bar E^u_{cd};R^{i\mapsto a};R^{j\mapsto b} \text{ otherwise},
                                   \end{cases} \tag{by~\eqref{prop:disc-ids0} and~\eqref{disc-eq:rd4}} \\
\bar E^{wz};E^u_{ab} &\equiv E^u_{ab};\bar E^{wz}. \tag{by~\eqref{disc-eq:comm13}}
\end{align*}
The inductive case is very similar to that of Lemma~\ref{lma:I-commutes}.
\end{proof}

\begin{corollary}\label{cor:Eneg-form}
Any term is equal to a term of the form $\mathtt I;\mathtt C;\mathtt E^{<0};\mathtt t$, where $\mathtt I$, $\mathtt C$ and $\mathtt E^{<0}$ are sequences of $I$-, $C$-, and $E^{<0}$-terms, and the term $\mathtt t$ contains no $I$-, $C$-, or $E^{<0}$-terms.
\end{corollary}

\begin{lemma}\label{lma:Enonneg-commutes}
Let $\mathtt t$ be a term not containing any $I$-, $C$-, or $E^{<0}$-terms such that the term $\mathtt t;E^{uv}$ is defined. Then there exists a term $\mathtt t'$ not containing any $I$-, $C$-, or $E^{<0}$-terms such that $\mathtt t$ and $\mathtt t'$ have the same number of $E^{\geq 0}$-terms, and one of the following holds:
\begin{enumerate}[label={(\arabic*)}]
\item $\mathtt t;E^{uv}\equiv\mathtt t'$, or
\item there is a disconnection $E^{wz}$ such that $\mathtt t;E^{uv}\equiv E^{wz};\mathtt t'$.
\end{enumerate}
\end{lemma}
\begin{proof}
By induction on $\mathtt t$.
\paragraph{Base cases}
\begin{align*}
\id;E^{uv} &\equiv E^{uv};\id, \\
S^w;E^{uv} &\equiv E^{uv};S^w, \tag{by~\eqref{disc-eq:sd1}} \\
R^{w\mapsto z};E^{uv} &\equiv \begin{cases} E^{uv};R^{w\mapsto z} \text{ if } z\neq v, \\
                                       E^{uw};R^{w\mapsto v} \text{ if } z=v,
                         \end{cases} \tag{by~\eqref{disc-eq:rd1} and~\eqref{disc-eq:rd2}} \\
E^{wz};E^{uv} &\phantom{\equiv} \text{ is already in the right form,} \\
\bar I^{wz};E^{uv} &\equiv E^{uv};\bar I^{wz}, \tag{by~\eqref{disc-eq:comm6}} \\
\bar C^{wz}_{cd};E^{uv} &\equiv E^{uv};\bar C^{wz}_{cd}, \tag{by~\eqref{disc-eq:comm11}} \\
\bar E^{w}_{cd};E^{uv} &\equiv E^{uv};\bar E^{w}_{cd}, \tag{by~\eqref{disc-eq:comm13}} \\
\bar E^{wz};E^{uv} &\equiv \begin{cases} S^u;S^v \text{ if } w=u \text{ and } z=v, \\
                                    E^{uv};\bar E^{wz} \text{ otherwise.}
                      \end{cases} \tag{by~\eqref{disc-eq:ddbar4-2} and~\eqref{disc-eq:comm2}}
\end{align*}
The inductive case is very similar to that of Lemma~\ref{lma:I-commutes}.
\end{proof}

\begin{corollary}\label{cor:Enonneg-form}
Any term is equal to a term of the form
$$\mathtt I;\mathtt C;\mathtt E^{<0};\mathtt E^{\geq 0};\mathtt t,$$
where $\mathtt I$, $\mathtt C$, $\mathtt E^{<0}$ and $\mathtt E^{\geq 0}$ are sequences of $I$-, $C$-, $E^{<0}$, and $E^{\geq 0}$-terms, and the term $\mathtt t$ contains no $I$-, $C$-, $E^{<0}$, or $E^{\geq 0}$-terms.
\end{corollary}

\begin{lemma}\label{lma:Enonnegbar-commutes}
Let $\mathtt t$ be a term not containing any $I$-, $C$- or $E$-terms such that the term $\mathtt t;\bar E^{uv}$ is defined. Then there exists a term $\mathtt t'$ not containing any $I$-, $C$- or $E$-terms such that $\mathtt t$ and $\mathtt t'$ have the same number of $\bar E^{\geq}$-terms, and one of the following holds:
\begin{enumerate}[label={(\arabic*)}]
\item $\mathtt t;\bar E^{uv}\equiv\mathtt t'$, or
\item there is a connection $\bar E^{wz}$ such that $\mathtt t;\bar E^{uv}\equiv\bar E^{wz};\mathtt t'$.
\end{enumerate}
\end{lemma}
\begin{proof}
By induction on $\mathtt t$.
\paragraph{Base cases}
\begin{align*}
\id;\bar E^{uv} &\equiv\bar E^{uv};\id, \\
S^w;\bar E^{uv} &\equiv \bar E^{uv};S^w, \tag{by~\eqref{disc-eq:sd1}} \\
R^{w\mapsto z};\bar E^{uv} &\equiv \begin{cases} \bar E^{uv};R^{w\mapsto z} \text{ if } z\neq v, \\
                                            \bar E^{uw};R^{w\mapsto v} \text{ if } z=v,
                              \end{cases} \tag{by~\eqref{disc-eq:rd1} and~\eqref{disc-eq:rd2}} \\
\bar I^{wz};\bar E^{uv} &\equiv \bar E^{uv};\bar I^{wz}, \tag{by~\eqref{disc-eq:comm5}} \\
\bar C^{wz}_{cd};\bar E^{uv} &\equiv \bar E^{uv};\bar C^{wz}_{cd}, \tag{by~\eqref{disc-eq:comm10}} \\
\bar E^{w}_{cd};\bar E^{uv} &\equiv \bar E^{uv};\bar E^{w}_{cd}, \tag{by~\eqref{disc-eq:comm12}} \\
\bar E^{wz};\bar E^{uv} &\phantom{\equiv} \text{ is already in the right form.}
\end{align*}
The inductive case is very similar to that of Lemma~\ref{lma:I-commutes}.
\end{proof}

\begin{corollary}\label{cor:Enonnegbar-form}
Any term is equal to a term of the form
$$\mathtt I;\mathtt C;\mathtt E^{<0};\mathtt E^{\geq 0};\bar{\mathtt E}^{\geq 0};\mathtt t,$$
where the term $\mathtt t$ contains no $I$-, $C$-, $E$-, or $\bar E^{\geq 0}$-terms.
\end{corollary}

\begin{lemma}\label{lma:Enegbar-commutes}
Let $\mathtt t$ be a term not containing any $I$-, $C$-, $E$-, or $\bar E^{\geq 0}$-terms such that the term $\mathtt t;\bar E^{u}_{ab}$ is defined. Then there exists a term $\mathtt t'$ not containing any $I$-, $C$-, $E$-, or $\bar E^{\geq 0}$-terms such that $\mathtt t$ and $\mathtt t'$ have the same number of $\bar E^{<0}$-terms, and one of the following holds:
\begin{enumerate}[label={(\arabic*)}]
\item $\mathtt t;\bar E^{u}_{ab}\equiv\mathtt t'$, or
\item there is a connection $\bar E^{w}_{cd}$ such that $\mathtt t;\bar E^{u}_{ab}\equiv\bar E^{w}_{cd};\mathtt t'$.
\end{enumerate}
\end{lemma}
\begin{proof}
By induction on $\mathtt t$.
\paragraph{Base cases}
\begin{align*}
\id;\bar E^{u}_{ab} &\equiv\bar E^{u}_{ab};\id, \\
S^w;\bar E^{u}_{ab} &\equiv \bar E^{u}_{ab};S^w, \tag{by~\eqref{disc-eq:sd1}} \\
R^{w\mapsto z};\bar E^{u}_{ab} &\equiv \begin{cases} \bar E^{u}_{ab};R^{w\mapsto z} \text{ if } z\notin\{a,b\}, \\
                                                \bar E^{u}_{wb} \text{ if } z=a, \\
                                                \bar E^{u}_{aw} \text{ if } z=b,
                                  \end{cases} \tag{by~\eqref{disc-eq:rd1} and~\eqref{disc-eq:rd3}} \\
\bar I^{wz};\bar E^{u}_{ab} &\equiv \bar E^{u}_{ab};\bar I^{wz}, \tag{by~\eqref{disc-eq:comm4}} \\
\bar C^{wz}_{cd};\bar E^{u}_{ab} &\equiv \bar E^{u}_{ab};\bar C^{wz}_{cd}, \tag{by~\eqref{disc-eq:comm9}} \\
\bar E^{w}_{cd};\bar E^{u}_{ab} &\phantom{\equiv} \text{ is already in the right form.}
\end{align*}
The inductive case is very similar to that of Lemma~\ref{lma:I-commutes}.
\end{proof}

\begin{corollary}\label{cor:Enegbar-form}
Any term is equal to a term of the form
$$\mathtt I;\mathtt C;\mathtt E^{<0};\mathtt E^{\geq 0};\bar{\mathtt E}^{\geq 0};\bar{\mathtt E}^{<0};\mathtt t,$$
where the term $\mathtt t$ contains no $I$-, $C$-, $E$- or $\bar E$-terms.
\end{corollary}

\begin{lemma}\label{lma:Cbar-commutes}
Let $\mathtt t$ be a term not containing any $I$-, $C$-, $E$- or $\bar E$-terms such that the term $\mathtt t;\bar C^{uv}_{ab}$ is defined. Then there exists a term $\mathtt t'$ not containing any $I$-, $C$-, $E$- or $\bar E$-terms such that $\mathtt t$ and $\mathtt t'$ have the same number of $\bar C$-terms, and one of the following holds:
\begin{enumerate}[label={(\arabic*)}]
\item $\mathtt t;\bar C^{uv}_{ab}\equiv\mathtt t'$, or
\item there is a connection $\bar C^{wz}_{cd}$ such that $\mathtt t;\bar C^{uv}_{ab}\equiv\bar C^{wz}_{cd};\mathtt t'$.
\end{enumerate}
\end{lemma}
\begin{proof}
By induction on $\mathtt t$.
\paragraph{Base cases}
\begin{align*}
\id;\bar C^{uv}_{ab} &\equiv \bar C^{uv}_{ab};\id, \\
S^w;\bar C^{uv}_{ab} &\equiv \bar C^{uv}_{ab};S^w, \tag{by~\eqref{disc-eq:sd1}} \\
R^{w\mapsto z};\bar C^{uv}_{ab} &\equiv \begin{cases} \bar C^{uv}_{ab};R^{w\mapsto z} \text{ if } z\notin\{a,b\}, \\
                                                 \bar C^{uv}_{wb} \text{ if } z=a, \\
                                                 \bar C^{uv}_{aw} \text{ if } z=b,
                                   \end{cases} \tag{by~\eqref{disc-eq:rd1} and~\eqref{disc-eq:rd3}} \\
\bar I^{wz};\bar C^{uv}_{ab} &\equiv \bar C^{uv}_{ab};\bar I^{wz}, \tag{by~\eqref{disc-eq:comm3}} \\
\bar C^{wz}_{cd};\bar C^{uv}_{ab} &\phantom{\equiv} \text{ is already in the right form.}
\end{align*}
The inductive case is very similar to that of Lemma~\ref{lma:I-commutes}.
\end{proof}

\begin{corollary}\label{cor:Cbar-form}
Any term is equal to a term of the form
$$\mathtt I;\mathtt C;\mathtt E^{<0};\mathtt E^{\geq 0};\bar{\mathtt E}^{\geq 0};\bar{\mathtt E}^{<0};\bar{\mathtt C};\mathtt t,$$
where the term $\mathtt t$ contains no $I$-, $C$-, $E$-, $\bar E$- or $\bar C$-terms.
\end{corollary}

\begin{lemma}\label{lma:Ibar-commutes}
Let $\mathtt t$ be a term not containing any $I$-, $C$-, $E$-, $\bar E$- or $\bar C$-terms such that the term $\mathtt t;\bar I^{uv}$ is defined. Then there exists a term $\mathtt t'$ not containing any $I$-, $C$-, $E$-, $\bar E$- or $\bar C$-terms such that $\mathtt t$ and $\mathtt t'$ have the same number of $\bar I$-terms, and one of the following holds:
\begin{enumerate}[label={(\arabic*)}]
\item $\mathtt t;\bar I^{uv}\equiv\mathtt t'$, or
\item there is a connection $\bar I^{ab}$ such that $\mathtt t;\bar I^{uv}\equiv\bar I^{ab};\mathtt t'$.
\end{enumerate}
\end{lemma}
\begin{proof}
By induction on $\mathtt t$.
\paragraph{Base cases}
\begin{align*}
\id;\bar I^{uv} &\equiv\bar I^{uv};\id, \\
S^w;\bar I^{uv} &\equiv \bar I^{uv};S^w, \tag{by~\eqref{disc-eq:sd1}} \\
R^{w\mapsto z};\bar I^{uv} &\equiv \bar I^{uv};R^{w\mapsto z}, \tag{by~\eqref{disc-eq:rd1}} \\
\bar I^{wz};\bar I^{uv} &\phantom{\equiv} \text{ is already in the right form.}
\end{align*}
The inductive case is very similar to that of Lemma~\ref{lma:I-commutes}.
\end{proof}

\begin{corollary}\label{cor:Ibar-form}
Any term is equal to a term of the form
$$\mathtt I;\mathtt C;\mathtt E^{<0};\mathtt E^{\geq 0};\bar{\mathtt E}^{\geq 0};\bar{\mathtt E}^{<0};\bar{\mathtt C};\bar{\mathtt I};\mathtt t,$$
where the term $\mathtt t$ contains only $S$-, $R$-, and identity terms.
\end{corollary}

\begin{lemma}\label{lma:S-commutes}
Let $\mathtt t$ be a term containing only $S$-, $R$-, and identity terms such that the term $S^{u};\mathtt t$ is defined. Then there exists a term $\mathtt t'$ containing only $S$-, $R$-, and identity terms such that $\mathtt t$ and $\mathtt t'$ have the same number of $S$-terms, and one of the following holds:
\begin{enumerate}[label={(\arabic*)}]
\item $S^{u};\mathtt t\equiv\mathtt t'$, or
\item there is a term $S^{v}$ such that $S^{u};\mathtt t\equiv\mathtt t';S^{v}$.
\end{enumerate}
\end{lemma}
\begin{proof}
By induction on $\mathtt t$.
\paragraph{Base cases}
\begin{align*}
S^u;\id &\equiv \id;S^u, \\
S^u;S^w &\phantom{\equiv} \text{ is already in the right form,} \\
S^u;R^{w\mapsto z} &\equiv \begin{cases} R^{w\mapsto z};S^u \text{ if } u\neq w, \\
                                    R^{u\mapsto z} \text{ if } u=w.
                      \end{cases} \tag{by~\eqref{disc-eq:sr1} and~\eqref{disc-eq:sr2}}
\end{align*}
The inductive case is very similar to that of Lemma~\ref{lma:I-commutes}.
\end{proof}

\begin{corollary}[Proposition~\ref{prop:ICE-form}]\label{cor:ICE-form}
Any term is equal to a term in an $ICE$-form.
\end{corollary}

\end{document}